%% file: main.tex
\def\confversion{0}
\def\ifconf{\ifnum\confversion=1}
\def\ifnotconf{\ifnum\confversion=0}
\documentclass[12pt, notitlepage, letterpaper]{report}
\def\showauthornotes{1}
\def\showkeys{0}
\def\showdraftbox{1}
\usepackage{silence}
\usepackage[utf8]{inputenc}
\usepackage{setspace}
\usepackage{pdfsync}
\input{Include/defn.tex}
\usepackage[stable]{footmisc}

\newcommand{\deffont}{\sf}
\newcommand{\defnt}[1]{ {\deffont #1} }

\crefname{ineq}{Ineq.}{Inequalities}
\creflabelformat{ineq}{#2{\upshape(#1)}#3}

\newcommand{\N}{{\mathbb{N}}}
\newcommand{\Q}{{\mathbb{Q}}}
\newcommand{\R}{{\mathbb R}}
\def\eps{\varepsilon}
\newcommand{\vartwo}[2]{x_{(#1, #2)}}
\newcommand{\vartwoempty}{x_{(\emptyset,\emptyset)}}

\newcommand{\C}{\calC}

\newcommand{\gauss}[3]{\gamma_{#1}\inparen{#2, #3}}

\newcommand{\sat}{{\sf sat}}

\newcommand{\bnu}{\bar{\nu}}

\DeclareMathOperator{\lpopt}{\operatorname{FRAC}}
\DeclareMathOperator{\sdpopt}{\operatorname{FRAC}}

\newcommand{\dist}{\ensuremath{\mathcal D}}
\newcommand{\dzero}{\dist^{(0)}}
\newcommand{\treedist}{\overline{\cal D}}
\newcommand{\component}[1]{\mathcal{C}\inparen{#1}}

\newcommand{\Hyp}{{\cal H}_k\inparen{m,n,n_0,\Gamma}}
\newcommand{\dmu}{\ensuremath{d_\mu}}

\newcommand{\cl}{\operatorname{cl}}
\newcommand{\clR}{\cl_R}

\newcommand{\ttg}{\mathtt{g}}
\newcommand{\girth}{\ttg}
\newcommand{\littleoh}{\operatorname{o}}
\newcommand{\qary}{\ensuremath{[q]}}

\DeclareMathOperator{\poly}{\operatorname{poly}}
\newcommand{\cliques}{\calC}
\newcommand{\ind}[1]{\1\insquare{#1}}
\newcommand{\parts}{\mathcal{P}}
\newcommand{\lsim}{\lesssim}
\newcommand{\triangles}{\Delta}

\newcommand{\barz}{\overline{z}}

\newcommand{\hatf}{\widehat{f}}
\newcommand{\hatg}{\widehat{g}}

\newcommand{\bff}{{\mathbf f}}
\newcommand{\bfgg}{{\mathbf g}}

\newcommand{\bfhatf}{{\bf \widehat{f}}}
\newcommand{\bfhatg}{{\bf \widehat{g}}}
\newcommand{\bfhatP}{{\bf \widehat{P}}}
\newcommand{\bfhatL}{{\bf \widehat{L}}}

\newcommand{\ftwo}[1]{\norm{\infty}{#1}}
\newcommand{\fsp}[1]{\norm{\mathrm{sp}}{#1}}

\newcommand{\tetmat}{\mathfrak{M}}

\newcommand{\indicator}[1]{\mathds{1}_{#1}}

\newcommand{\nbr}{\mathsf{N}}
\newcommand{\sfM}{\mathsf{M}}
\newcommand{\sfA}{\mathsf{A}}

\newcommand{\sfX}{\mathsf{X}}

\newcommand{\sfZ}{\mathsf{Z}}

\newcommand{\sfI}{\mathsf{I}}
\newcommand{\sfT}{\mathsf{T}}

\newcommand{\hA}{\widehat{\mathsf{A}}}

\newcommand{\hM}{\widehat{\mathsf{M}}}
\newcommand{\hN}{\widehat{\mathsf{N}}}
\newcommand{\hJ}{\widehat{\mathsf{J}}}

\newcommand{\pth}[2][\!]{#1\left({#2}\right)}%
\newcommand{\Mat}[1]{\mathop{\mathbf{Mat}}\pth{#1}}
\newcommand{\Vector}[1]{\mathop{\mathsf{Vec}}\pth{#1}}

\newcommand{\Tr}[1]{\mathop{\mathsf{Tr}}\pth{#1}}

\newcommand{\cardin}[1]{\left| {#1} \right|}%

\newcommand{\Sym}{\mathbb{S}}

\newcommand{\orbit}[1]{\mathscr{O}{\,\pth{#1}}}

\DeclareMathOperator*{\E}{\ensuremath{\mathbb{E}}}
\newcommand{\multichoose}[2]{\ensuremath{\left(\kern-.3em\left(\genfrac{}{}{0pt}{}{#1}{#2}\right)\kern-.3em\right)}}

\def\restrict#1{\raise-.5ex\hbox{\ensuremath|}_{#1}}

\newcommand{\emphi}[1]{\emph{\textbf{#1}}}

\newcommand{\barq}{\overline{q}}

\newcommand{\one}{\mathds{1}}

\newcommand{\real}[1]{\mathrm{Re}\pth{#1}}
\newcommand{\im}[1]{\mathrm{Im}\pth{#1}}

\newcommand{\SSS}{\mathbb{S}}

\newcommand{\1}{\ensuremath{\mathbb{1}}}

\newcommand{\spread}[1]{\mathop{\mathbf{spread}}{\pth{ #1 }}}
\newcommand{\spreadA}[1]{\mathop{\mathbf{spread}_A}{\pth{ #1 }}}
\newcommand{\spreadbA}[1]{\mathop{\mathbf{spread}_{\overline{A}}}{\pth{ #1 }}}
\newcommand{\spreadB}[1]{\mathop{\mathbf{spread}_B}{\pth{ #1 }}}
\newcommand{\spreadC}[1]{\mathop{\mathbf{spread}_C}{\pth{ #1 }}}
\newcommand{\spreadD}[1]{\mathop{\mathbf{spread}_D}{\pth{ #1 }}}
\newcommand{\spreadABCD}[1]{\mathop{\mathbf{spread}_{ABCD}}{\pth{ #1 }}}
\newcommand{\spreadbABCD}[1]{\mathop{\mathbf{spread}_{\overline{A}BCD}}{\pth{ #1 }}}

\newcommand{\hgm}[2]{\mathsf{M}_{\textrm{hyp}}^{#1}(#2)}

\newcommand{\overbar}[1]{\mkern 1.2mu\overline{\mkern-1.2mu#1\mkern-1.2mu}\mkern 1.2mu}
\newcommand{\mi}[1]{\alpha({#1})}

\newcommand{\PR}[1]{{\RR}_{#1}[x]}
\newcommand{\pr}[1]{\PR{#1}}
\newcommand{\NPR}[1]{{\RR}^+_{#1}[x]}
\newcommand{\npr}[1]{\NPR{#1}}
\newcommand{\FPR}[2]{{\pth{{\RR}_{#2}\![x]}\!_{#1}}\![x]}

\newcommand{\NFPR}[2]{{\pth{{\RR}^+_{#2}\![x]}\!_{#1}}\![x]}

\newcommand{\degmindex}[1]{{\NN}_{#1}^{n}}
\newcommand{\udmindex}[1]{{\NN}_{\!\leq #1}^{n}}

\newcommand{\degbmindex}[1]{\{0,1\}_{#1}^{n}}
\newcommand{\multif}[1]{F_{#1}}
\newcommand{\supp}[1]{\mathsf{S}(#1)}
\newcommand{\fold}[2]{\overbar{#1}_{#2}}
\newcommand{\collapse}[2]{\mathsf{C}_{#1}\pth{#2}}
\newcommand{\unfold}[1]{\mathsf{U}(#1)}

\newcommand{\hscsos}[2]{\Lambda_{#1}\left(#2\right)}
\newcommand{\hssos}[1]{\Lambda\left(#1\right)}
\newcommand{\fmax}[1]{#1_{\max}}

\newcommand{\bx}{\overbar{x}}

\def\bfx {{\bf x}}
\def\bfy {{\bf y}}
\def\bfz {{\bf z}}
\def\bfu{{\bf u}}
\def\bfv{{\bf v}}
\def\bfw{{\bf w}}
\def\bfa{{\bf a}}
\def\bfb{{\bf b}}
\def\bfc{{\bf c}}
\def\bfg{{\bf g}}

\newcommand{\cip}[1]{\ensuremath{\left\langle #1 \right\rangle}}
\newcommand{\mysmalldot}{\innerprod}
\newcommand{\onorm}[2]{\ensuremath{#1\to #2}\xspace}
\newcommand{\groth}[2]{\ensuremath{(#1,#2)}-Grothendieck\xspace}
\newcommand{\sep}[1]{\,\left|\, {#1} \right.}%

\title{Some Applications and Limitations of Convex Optimization Hierarchies for Discrete and Continuous
Optimization Problems.}
\author{Mrinalkanti Ghosh}
\makeatletter\let\Title\@title\makeatother

\begin{document}
\pagenumbering{roman}
\input{title.tex}
\input{abs.tex}

\input{acknowledgement.tex}
\tableofcontents

\newpage
\pagenumbering{arabic}
\setcounter{page}{1}

\input{intro.tex}

\input{sa_csp/csp_concat.tex}

\input{poly_opt/poly_opt_concat.tex}

\input{PtoQ/PtoQ_concat.tex}

\bibliographystyle{alpha}
\bibliography{bibs,sa_csp,polynomials,p-to-q-refs,p-to-q_algo-refs}
\end{document}

%% file: Include/defn.tex
\usepackage{amsmath,amssymb,amsfonts,amsthm,mathtools,enumerate}
\usepackage{graphicx}
\usepackage{lmodern}
\usepackage{tikz}
\usetikzlibrary{arrows,calc,positioning}
\usepackage{enumitem,enumerate}
\usepackage{nicefrac}
\usepackage{listings}
\usepackage{xcolor,xspace}

\usepackage[T1]{fontenc}
\usepackage{dsfont}
\usepackage{bbm}
\usepackage{newpxmath} %%Palladio clone for math fonts

\usepackage{newpxtext} %%Palladio clone -- new implementation
\usepackage[colorlinks,citecolor=teal]{hyperref}

\usepackage[capitalise,nameinlink]{cleveref}
\usepackage{thmtools}
\usepackage{thm-restate}
\usepackage{float}
\floatstyle{boxed}
\restylefloat{figure}
\usepackage{framed}

\ifnum\showkeys=1
\usepackage[color]{showkeys}
\fi

\newcommand{\calA}{\ensuremath{\mathcal{A}}\xspace}

\newcommand{\calC}{\ensuremath{\mathcal{C}}\xspace}
\newcommand{\calD}{\ensuremath{\mathcal{D}}\xspace}
\newcommand{\calE}{\ensuremath{\mathcal{E}}\xspace}
\newcommand{\calF}{\ensuremath{\mathcal{F}}\xspace}

\newcommand{\calH}{\ensuremath{\mathcal{H}}\xspace}

\newcommand{\calK}{\ensuremath{\mathcal{K}}\xspace}
\newcommand{\calL}{\ensuremath{\mathcal{L}}\xspace}

\newcommand{\calP}{\ensuremath{\mathcal{P}}\xspace}

\DeclareMathOperator{\sgn}{\operatorname{sgn}}
\DeclareMathOperator{\inv}{\operatorname{Inv}}

\DeclareMathOperator*{\argmax}{\arg\!\max}

\DeclareMathOperator{\Vtx}{\operatorname V}
\DeclareMathOperator{\Edg}{\operatorname E}

\DeclareMathOperator{\diag}{\operatorname{diag}}
\DeclareMathOperator{\diam}{\operatorname{diam}}

\newcommand{\ee}{\ensuremath{\mathrm e}}

\newcommand{\CC}{\ensuremath{\mathbb{C}}}

\newcommand{\RR}{\ensuremath{\mathbb{R}}}
\newcommand{\NN}{\ensuremath{\mathbb{N}}}

\newcommand{\sym}{\ensuremath{\mathbf{S}}}

\renewcommand{\bar}[1]{\ensuremath{\overline{#1}}}

\ifnum\showauthornotes=1
\newcommand{\Authornote}[2]{{\sf\small\color{red}{[#1: #2]}}}
\newcommand{\Authorcomment}[2]{{\sf \small\color{gray}{[#1: #2]}}}
\newcommand{\Authorfnote}[2]{\footnote{\color{red}{#1: #2}}}
\else
\newcommand{\Authornote}[2]{}
\newcommand{\Authorcomment}[2]{}
\newcommand{\Authorfnote}[2]{}
\fi
\newtheorem{theorem}{Theorem}[section]

\newtheorem{definition}[theorem]{Definition}
\newtheorem{observation}[theorem]{Observation}
\newtheorem{lemma}[theorem]{Lemma}
\newtheorem{remark}[theorem]{Remark}
\newtheorem{proposition}[theorem]{Proposition}
\newtheorem{corollary}[theorem]{Corollary}
\newtheorem{claim}[theorem]{Claim}
\newtheorem{fact}[theorem]{Fact}

\def\FullBox{\hbox{\vrule width 6pt height 6pt depth 0pt}}

\def\qed{\ifmmode\qquad\FullBox\else{\unskip\nobreak\hfil
\penalty50\hskip1em\null\nobreak\hfil\FullBox
\parfillskip=0pt\finalhyphendemerits=0\endgraf}\fi}

\def\qedsketch{\ifmmode\Box\else{\unskip\nobreak\hfil
\penalty50\hskip1em\null\nobreak\hfil$\Box$
\parfillskip=0pt\finalhyphendemerits=0\endgraf}\fi}

\ifnotconf
\renewenvironment{proof}{\begin{trivlist} \item {\bf Proof:~~}}
   {\qed\end{trivlist}}
\fi

\newenvironment{proofsketch}{\begin{trivlist} \item {\bf
Proof Sketch:~~}}
  {\qedsketch\end{trivlist}}

\newenvironment{proofof}[1]{\begin{trivlist} \item {\bf Proof
#1:~~}}
  {\qed\end{trivlist}}

\newcommand{\defeq}{\stackrel{\mathrm{def}}=}

\newcommand{\mper}{\,.}
\newcommand{\mcom}{\,,}
\newcommand{\inparen}[1]{\ensuremath{\left(#1\right)}}
\newcommand{\inbraces}[1]{\ensuremath{\left\{#1\right\}}}
\newcommand{\insquare}[1]{\ensuremath{\left[#1\right]}}

\let\nfrac=\nicefrac

\newcommand{\abs}[1]{\ensuremath{\left\lvert #1 \right\rvert}}

\newcommand{\innerprod}[2]{\ensuremath{\left\langle #1, #2 \right\rangle}}

\makeatletter
\def\norm#1{%fancy norm command
  \@ifnextchar\bgroup%
   {\normalpnorm{#1}}%
   {\defaultnorm{#1}}%
}
\def\defaultnorm#1{%
    \renderNorm{#1}{}
}
\def\normalpnorm#1#2{%
   \@ifnextchar\bgroup%
   {\banaxnorm{#1}{#2}}%
   {\renderNorm{#2}{#1}}%\ensuremath{\left\lVert {#2} \right\rVert_{{#1}}}}%
}

\def\banaxnorm#1#2#3{%
    \renderNorm{#3}{#1\to#2 }
}

\def\renderNorm#1#2{%
    \@ifnextchar^%
    {\fixExponent{#1}{#2}}%
    {\ensuremath{\mathchoice%
        {\lVert #1 \rVert_{#2}}%
        {\lVert #1 \rVert_{#2}}%
        {\lVert #1 \rVert_{#2}}%
        {\lVert #1 \rVert_{#2}}}%
    }%
}
\def\fixExponent#1#2^#3{%
    \ensuremath{{\lVert #1 \rVert^{#3}_{#2}}}%
}

\def\enorm#1#2{%
   \@ifnextchar\bgroup%
   {\norm{L_{#1}}{L_{#2}}}%
   {\norm{L_{#1}}{#2}}%
}

\def\cnorm#1#2{%
   \@ifnextchar\bgroup%
   {\norm{\ell_{#1}}{\ell_{#2}}}%
   {\norm{\ell_{#1}}{#2}}%
}

\makeatother

\newcommand{\ie}{i.e.,\xspace}
\newcommand{\aka}{a.k.a.,\xspace}
\newcommand{\eg}{e.g.,\xspace}
\newcommand{\etal}{et al.\xspace}

\newcommand{\suchthat}{{\;\; : \;\;}}

\newcommand{\problemmacro}[1]{\textsf{#1}}

\newcommand{\zo}{\ensuremath{\{0,1\}}\xspace}
\newcommand{\on}{\ensuremath{\{-1,1\}}\xspace}
\newcommand{\maxcsp}{\problemmacro{MAX-CSP}\xspace}
\newcommand{\maxkcsp}{\problemmacro{MAX k-CSP}\xspace}
\newcommand{\maxkcspq}{\problemmacro{MAX k-CSP}$_q$\xspace}

\newcommand{\maxthreesat}{\problemmacro{MAX 3-SAT}\xspace}
\newcommand{\maxthreexor}{\problemmacro{MAX 3-XOR}\xspace}

\newcommand{\maxcut}{\problemmacro{MAX-CUT}\xspace}

\newcommand{\threesat}{\problemmacro{3-SAT}\xspace}

\newcommand{\independentset}{\problemmacro{Maximum Independent Set}\xspace}

\newcommand{\vertexcover}{\problemmacro{Minimum Vertex Cover}\xspace}
\newcommand{\clique}{\problemmacro{Maximum Clique}\xspace}

\newcommand{\labelcover}{\problemmacro{Label Cover}\xspace}

\newcommand{\gaussian}{\ensuremath{\mathcal{N}}}

\newcommand{\lp}{\problemmacro{LP}\xspace}
\newcommand{\sa}{\problemmacro{SA}\xspace}
\newcommand{\sdp}{\problemmacro{SDP}\xspace}
\newcommand{\NP}{\problemmacro{NP}\xspace}
\newcommand{\EXP}{\problemmacro{EXP}\xspace}
\newcommand{\BPP}{\problemmacro{BPP}\xspace}
\newcommand{\BPTIME}[1]{\problemmacro{BPTIME}\left( #1 \right)\xspace}
\newcommand{\DTIME}[1]{\problemmacro{DTIME}\left( #1 \right)\xspace}

\renewcommand{\P}{\problemmacro{P}\xspace}

\makeatletter
\def\sa{%
    \HierarchyRender{\problemmacro{SA}}
}
\def\sos{%
    \HierarchyRender{\problemmacro{SoS}}
}
\def\HierarchyRender#1{
  \@ifnextchar\bgroup%
  {\hierarchywithlvl{#1}}
  {\nameofhierarchy{#1}}
}
\def\nameofhierarchy#1{#1\xspace}
\def\hierarchywithlvl#1#2{%
    \ensuremath{#1^{(#2)}}\!\xspace
}
\makeatother
\def\bfa{{\bf a}}
\def\bfb{{\bf b}}
\def\bfc{{\bf c}}

\def\bfg{{\bf g}}

\def\bfu{{\bf u}}
\def\bfv{{\bf v}}
\def\bfw{{\bf w}}
\def\bfx{{\bf x}}
\def\bfy{{\bf y}}
\def\bfz{{\bf z}}
\def\bfX{{\bf X}}

\def\bfA{{\bf A}}

\def\bfF{{\bf F}}

\newcommand{\Esymb}{\mathbb{E}}
\newcommand{\Psymb}{\mathbb{P}}

\DeclareMathOperator*{\ExpOp}{\Esymb}

\makeatletter
\def\Pr{%
    \ProbabilityRender{\Psymb}%
}

\def\Ex{%
    \ProbabilityRender{\Esymb}%
}

\def\ProbabilityRender#1{%
  \@ifnextchar\bgroup%
  {\properRender{#1}}
  {#1}%Just symbol
}
\def\properRender#1#2{%fancy probability command
  \@ifnextchar\bgroup%
  {\renderwithdist{#1}{#2}}
   {\singlervrender{#1}{#2}}
}
\def\singlervrender#1#2{%
   \ensuremath{\mathchoice
       {{#1}\left[ #2 \right]}
       {{#1}[ #2 ]}
       {{#1}[ #2 ]}
       {{#1}[ #2 ]}
   }
}
\def\renderwithdist#1#2#3{%
   \@ifnextchar\bgroup
   {\superfancyrender{#1}{#2}{#3}}
   {\ensuremath{\mathchoice
      {\underset{#2}{#1}\left[ #3 \right]}
      {{#1}_{#2}[ #3 ]}
      {{#1}_{#2}[ #3 ]}
      {{#1}_{#2}[ #3 ]}
     }
   }
}
\def\superfancyrender#1#2#3#4#5{
   \ensuremath{\mathchoice
      {\underset{#1}{{#1}}\left#4 #3 \right#5}
      {{#1}_{#2}#4 #3 #5}
      {{#1}_{#2}#4 #3 #5}
      {{#1}_{#2}#4 #3 #5}
   }
}
\makeatother

\DeclareMathOperator{\opt}{\operatorname{OPT}}

\newcommand{\given}[1]{\left\vert #1 \right.}
\DeclareMathOperator{\divides}{\operatorname \vert}

\DeclareMathOperator{\bigomega}{\operatorname \Omega}
\DeclareMathOperator{\bigoh}{\operatorname O}

\ifnum\showdraftbox=1

\else

\fi

\usetikzlibrary{decorations.markings,arrows,positioning,automata}
\tikzset{->-/.style={decoration={
  markings,
  mark=at position .5 with {\arrow{>}}},postaction={decorate}}}

\newcommand{\Lovasz}{Lov\'asz\xspace}

\newcommand{\Bollobas}{Bollob\'as\xspace}

\newcommand{\Holder}{H\"{o}lder}

\newcommand{\Odonnell}{O'Donnell}
\newcommand{\Briet}{Bri\"et\xspace}          %% Added by Euiwoong on 2017-10-22.
\newcommand{\Esseen}{Ess\'een\xspace}          %% Added by Euiwoong on 2017-10-22.

%% file: title.tex
\begin{titlepage}
    \begin{center}
        \Huge
        \textbf{\Title}
            
        \vspace{0.4cm}
        \LARGE \textbf{Mrinalkanti Ghosh}
        \vspace{0.8cm}
        
        \large

        A thesis submitted\\
        in partial fulfillment of the requirements for\\
        the degree of\\
        \vspace{0.8cm}
        Doctor of Philosophy in Computer Science\\
        \vspace{0.8cm}
        at the\\

        \vspace{0.8cm}
        TOYOTA TECHNOLOGICAL INSTITUTE AT CHICAGO\\
        Chicago, Illinois\\

        \vspace{0.8cm}
        August 2025\\
            
        \vspace{0.8cm}
        \vspace{0.8cm}

        Thesis Committee:\\
        \vspace{0.4cm}
        Yury Makarychev\\
        \vspace{0.4cm}
        Aaron Potechin\\
        \vspace{0.4cm}
        Madhur Tulsiani (Thesis Advisor)\\
    \end{center}
\end{titlepage}

%% file: abs.tex
\begin{abstract}
In this thesis, we study algorithmic applications and limitations of convex relaxation hierarchies for
approximating both discrete and continuous optimization problems. Another common theme in these results
is their connection to geometry. For the discrete optimization problem of CSPs, geometry plays a crucial
role in our lower-bound proof. For the continuous optimization problem of optimizing a polynomial over
the sphere, geometry appears in an interesting way (beyond the appearance of the sphere, a geometric
object, in the problem definition). Earlier works~\cite{KN08,HLZ10,So11} use sequences of diameter
estimations of convex bodies to design their approximation algorithms. For a related continuous
optimization problem of approximating matrix $p\to q$-norm, the problem definition itself is geometric.
Therefore, geometry also appears both in our proofs of hardness results and that of the approximation
algorithm. We describe the specific results in this thesis in the following paragraphs.

We show a dichotomy of approximability of constraint satisfaction problems (CSPs) by linear programming
(LP) relaxations: for every CSP, the approximation obtained by a basic LP relaxation, is no weaker than
the approximation obtained using relaxations given by $\Omega\inparen{\frac{\log n}{\log \log n}}$ levels
of the Sherali-Adams hierarchy on instances of size $n$.

The polynomial optimization problem we consider: given an n-variate degree-d homogeneous polynomial f
with real coefficients, compute a unit vector $x \in \mathbb{R}^n$ that maximizes $|f(x)|$.  We give
approximation algorithms for this problem that offer a trade-off between the approximation ratio and
running time: in $n^{O(q)}$ time (for $2d\divides q$ and $q \le n$), we get an approximation within
factor $O_d(n/q)^{d/2-1}$ for arbitrary polynomials, $O_d(n/q)^{d/4-1/2}$ for polynomials with
non-negative coefficients, and $(m/q)^{1/2}$ for sparse polynomials with m monomials. The approximation
guarantees are with respect to the optimum of the level-q sum-of-squares (SoS) SDP relaxation of the
problem,  even though our algorithms do not rely on actually solving the SDP. For polynomials with
non-negative coefficients, we prove an $\tilde{\Omega}(n^{1/6})$ gap for the degree-$4$ case, for a
related relaxation.

The $p \to q$-norm of a matrix $A \in \mathbb{R}^{m \times n}$ is defined as $\|A\|_{p\to q}~\defeq~\sup_{x
\in \mathbb{R}^n \setminus \{0\}} \frac{\|Ax\|_q}{\|x\|_p}$.  The regime when $p < q$, known as
\emph{hypercontractive norms}. The case with $p = 2$ and $q > 2$ was studied in~\cite{BBHKSZ12}.
They proved a hardness of approximation result based on the Exponential Time Hypothesis. We
prove the first NP-hardness result for approximating hypercontractive norms: for any $1< p < q < \infty$
with $2 \notin [p,q]$, $\|A\|_{p\to q}$ is hard to approximate within $2^{O((\log n)^{1-\eps})}$ assuming
$\NP \not \subseteq \BPTIME{2^{(\log n)^{O(1)}}}$. We also prove almost tight results for the case when
$p \geq q$ with $2 \in [q,p]$. For such $p$ and $q$, we show that $\|A\|_{p\to q}$ is \NP-hard to
approximate within any factor smaller than $1/(\gamma_{p^*} \cdot \gamma_q)$, where $p^* := p/(p-1)$ is the
dual norm of $p$. We also give an approximation algorithm with approximation ratio
$(1+\epsilon_0)/(\sinh^{-1}(1)\cdot \gamma_{p^*} \,\gamma_{r^*})$ for some fixed $\epsilon_0 \le
0.00863$. For the algorithm we develop a generalization of random hyperplane rounding using
H\"{o}lder-duals of Gaussian projection, which may prove useful for rounding $\ell_p$-norm bounded
vectors arising from convex relaxations ($p\ge 2$).
\end{abstract}

%% file: acknowledgement.tex
\chapter*{Acknowledgements}
I thank my advisor, Madhur Tulsiani, for his constant support and encouragement over the years. He has been
incredibly patient and magnanimous with his time and advice, both in academic and non-academic
matters. I also thank other members of my thesis committee, Yury Makarychev and Aaron
Potechin, who have been supportive throughout the process.

I am especially thankful to my co-authors Vijay Bhattiprolu and Euiwoong Lee for their friendship.
Working with them has been both fun and enlightening. I also thank Venkatesan Guruswami for lending
us his valuable time and expertise.

TTIC is a stimulating and nurturing place for graduate students. Plenty of workshops, seminars,
and social events promote collaboration in a relaxed environment. Much of the credit goes to administrative staff
and the leadership. Special thanks to Mary who is a smiling presence day and night at the institute. Thank
you, Adam, Amy, Chrissy, Erica, Jessica and everyone else who makes navigating the system easy. I also
thank David McAllester and Avrim Blum for fostering an excellent learning environment during their tenure
as Chief Academic Officer.

I was fortunate to partake in useful theory courses at TTIC and UChicago offered by
L\'aszl\'o Babai, Alexander Razborov, Andrew Drucker, Aaron Potechin, Julia Chuzhoy, Yury Makarychev, and Madhur Tulsiani.

The students at TTIC and Uchicago made up an enjoyable social circle in Chicago. I would like to thank
Haris, Shubhendu, Behnam, Somaye, Pooya, and Denis for welcoming me to TTIC and Uchicago. I am also grateful
to many friends I met along the way, including, Fernando, Leonardo, Chris, Goutham, Pushkar, Shubham, and
Sudarshan, among others. Special thanks to Rachit for organizing many of our trips
around Chicago, and for the numerous trips to Devon. Many thanks to Shashank and Akash for their constant
friendship and encouragement to complete this thesis. 

Last but not least, I thank my family for their love. My parents always supported and
encouraged me to pursue an academic career, and I dedicate this thesis to them.

%% file: intro.tex
\chapter{Introduction}
Time complexity classes are usually defined for decision problems, \ie~to decide whether a
binary string is in a given language. Some of the most well-known complexity classes are \P,
\NP, and \EXP. Separating \P from \NP\ is one of the central open questions of complexity
theory. In contrast, optimization problems require optimizing some quantity over a
set of feasible solutions. For example, in the \clique~problem, we want to find the size of the
largest clique in a given graph $G$. Such optimization problems can easily be turned into decision
problems. For \clique, the decision version is the language consisting of binary encodings of the pairs
$(G,k)$ where $G$ is a graph and $k\in\NN$ is the size of the largest clique in $G$.  Thus, notions of \P,
\NP, \NP-hardness, and \NP-completeness of decision problems can be ported to the optimization version.

The solution spaces for optimization problems come in two forms: discrete and continuous. For a discrete
optimization problem, the space of solutions is finite for any given instance, whereas for a continuous
optimization problem, the solution space is uncountable. The \clique\ problem is a
discrete optimization problem. Finding the largest eigenvalue of a Hermitian matrix is an example of a
continuous optimization problem (which is efficiently computable).

Many natural optimization problems of interest turn out to be \NP-complete or \NP-hard. As solving these
exactly would show $\P=\NP$, these are believed to be computationally intractable. The next best thing might be
to hope for an approximate solution to these optimization problems. For an instance $I$ of an
optimization problem, we use $\opt(I)\ge 0$ to denote the value of the optimal solution. For a
maximization problem, the value $v$ is said to be $\alpha$-approximation (approximation factor for us
will be $\ge 1$) if $v\cdot \alpha\le \opt(I)$ -- approximation factor for minimization problems may be
defined similarly. For many \NP-complete problems, efficient non-trivial approximations have been
achieved. So we may ask what are the limits to such approximation. For the \NP-complete \clique\ problem,
we observe that approximating the value within factor $\nfrac{n}{(n-1)}$ for graphs on $n$ vertices would
solve the problem exactly (and thus establishing an easy \NP-completeness of approximation).
Interestingly, although all NP-complete problems are equivalent in terms of exact solvability, they can
behave very differently when it comes to approximability. For example, \threesat\ is shown to be 
\NP-hard to approximate within factor better than $\nfrac{8}{7}$ in~\cite{Hastad01}. Whereas, for
the \independentset problem, the result in~\cite{Hastad99} shows hardness of approximation within
factor $n^{1-\varepsilon}$ for any $\varepsilon>0$.

The problem of establishing \NP-completeness of approximation for various optimization problems turned
out to be a very interesting and fruitful area of research.  A powerful tool, the Probabilistically
Checkable Proof (PCP) was developed to prove \NP-hardness of approximation. PCP has a long history, and an
overview can be found in~\cite{Hastad01}. PCPs and their
variants have been used to prove optimal inapproximability of some problems, \ie\ one can prove
\NP-hardness of $\alpha$-factor approximation. Simultaneously, we have efficient algorithms, perhaps
probabilistic, to compute $\alpha+\varepsilon$-approximation for $\varepsilon>0$ (the runtime may or may
not depend on $\varepsilon$). The problem \threesat\ is such an example: while it is \NP-hard to
approximate within factor $\nfrac{8}{7}$, for any constant $\varepsilon>0$, there is an efficient,
randomized algorithm which approximates the value within factor $\nfrac{8}{7}+\varepsilon$. However, for many
other optimization problems there are significant gaps between the approximation factors achievable by
efficient algorithms and the approximation factors for which we can prove \NP-hardness. We refer to this
difference as the approximation gap. Another approach was proposed in~\cite{Khot02UGC}, known as
Unique Games Conjecture (UGC), for tightening of these approximation gaps for some optimization problems.
Here the starting point is to assume \NP-hardness of so-called Unique Games and deduce hardness of
approximation. This has shown success in improving the hardness factor. Known \NP-hardness of
approximation for \maxcut\ is given in~\cite{Hastad01} to be $\nfrac{17}{16}$ (this has been improved in
some special cases to $\nfrac{11}{8}$ in~\cite{HHMOW14}); these are not tight as the beautiful algorithm
of Goemans-Williamson~\cite{GoemansW95} achieves approximation factor $\approx0.878$. However, by using UGC,
tight approximation hardness was shown in~\cite{KKMO04}, matching the approximation factor of the
aforementioned algorithm. Also, for CSPs, optimal approximation hardness
has been obtained in~\cite{Raghavendra08}.  Still there are optimization problems for which approximation
gap remains. Two other notable assumptions that have been used for starting point of hardness reductions
for better approximation hardness are: Feige's R3SAT hypothesis proposed in~\cite{Feige02} and
Exponential Time Hypothesis (ETH) proposed in~\cite{IP01}. We note here that Feige's R3SAT hypothesis
assumes the hardness of randomly generated instances -- that is, average-case hardness as opposed to
worst-case hardness.

We may consider another weakening of requirements for hardness results:
instead of proving hardness for \emph{all} efficient algorithms, we restrict ourselves to
a large family of efficient algorithms. In other words, can we close the approximation gap for some
class of algorithms $\calA$ for some of these optimization problems? Note that such result would be
unconditional. In particular, a tight result for an optimization problem $O$ will state that there is an algorithm $A\in
\calA$ which achieves $\alpha$-approximation whereas no algorithm in $\calA$ can achieve approximation
factor $\alpha-\varepsilon$ for (some, all) $\varepsilon>0$. Some reasonable classes of such algorithms
are convex relaxation hierarchies.

As the name suggests, the convex relaxation hierarchies are sequences of successively stronger
convex relaxations. We will consider two variants of these hierarchies: the \lp~relaxation hierarchy,
\aka~the Sherali-Adams hierarchy and the \sdp~relaxation hierarchy, \aka~the Sum of Square (\sos)
hierarchy. The levels of these hierarchies are denoted $\sa{d}$ and $\sos{d}$, respectively. For
discrete optimization problems, these relaxations may be viewed as defining ``local distributions''. We
will skip the details here. Details of this relaxation for the discrete optimization problem of CSP may
be found in~\cref{sec:csp:relaxations}. Details for the \sos{d} relaxation for a continuous optimization
problem is given in~\cref{sec:poly:sos:prelims}. The constraints of the relaxation $\sa{d+1}$ contains the
constraints of $\sa{d}$ for any $d>0$. The same is true for the \sos~relaxation as well. So the higher levels
of these hierarchies provide tighter relaxations and hence provide a (possibly) better approximation factor.
However, the downside is higher levels of the hierarchy require larger description sizes (which may, in turn,
require more time to solve). These hierarchies form powerful classes of algorithms: most efficient
algorithms constructed (with some known exceptions) are known to be captured by one of these hierarchies.

Unconditional hardness results for \sa~or \sos~for an optimization problem thus reduce to 
demonstrating an integrality gap at certain levels of the hierarchy. That is, showing the existence of an
instance $I$ of the optimization problem such that the optimal value for $I$ is $s$ and the $d$-level
relaxation has value $c$. For maximization problems, this proves integrality gap factor
$\nfrac{c}{s}$; this is not quite the approximation factor. But this is used to show evidence of
hardness approximation for these convex relaxation hierarchies.

As these hierarchies constitute some of the most powerful classes of algorithms, we can approach the
problem of diminishing the approximation gap from the other end. That is, we may want to analyze the
performance of these convex relaxations and their rounding to provide approximation algorithms for
optimization problems. In~\cite{BKS17}, the authors designed a new algorithm
using \sos~hierarchy to give an approximation factor which was previously unknown.

In this thesis, we establish unconditional hardness for the discrete optimization
problem of CSP. For continuous optimization problems, we present an approximation algorithm for polynomial
optimization problem using \sos~hierarchy. For another, related, continuous optimization problem of
matrix norm estimation, we show some new and improved conditional hardness of approximations. We also show
almost tightness of our hardness by designing a convex relaxation based algorithm for some parameter
regime. The rest of the document is divided in sections for these problems. Each section provides some
details of our results and some natural questions that arise.

An interesting aspect of these problems and results is the role of geometry. In case of the matrix norm
optimization problem, the connection to geometry of Banach spaces is direct and apparent. The proofs make
essential use of embedding and projection results in Banach spaces, which, in retrospect, is not
surprising. The connections to geometry in the cases of CSPs and polynomial optimization are more subtle
and are described below.

For CSP, construction of the feasible solution in the soundness case requires building local
distributions (described in~\cref{sec:csp:relaxations}). Our approach constructs these local
distributions while remaining oblivious to the parts of the instance outside the neighborhood of the
variables of interest. Within this neighborhood, we partition the instance into forests and then randomly
satisfy each tree to obtain a local distribution (see~\cref{sec:csp:decomposition}). To ensure
consistency of the local distributions, the partitioning procedure must itself be local and oblivious to
``far variables'' (even when the set under consideration may contain them). We achieve this by using an
$\ell_2$-embeddable metric (from~\cite{CMM07Metric}) on the neighborhood and using a random
unitary-invariant partition of the space to generate the forest decomposition. A crucial advantage of
this process is that it allows us to ``step out of the instance hypergraph'' and partition it externally,
as it resides on an $\ell_2$-ball. This external viewpoint ensures the required obliviousness, which is
essential for proving local consistency of the distributions constructed. The same metric was also used
more directly in~\cite{CMM09SACSP} to establish that the \maxcut~problem has integrality gap
$\nfrac{1}{2}$ for polynomial levels of \sa.

For polynomial optimization, the geometric connection arises in a subroutine of one of the works we
extend. A common approach to polynomial optimization (over the sphere) is to decouple the variables,
reducing the problem to multilinear polynomial optimization (where each multilinear component lies on the
sphere). This is achieved by examining a monomial and replacing each occurrence of a variable with a
fresh copy of a new variable (see~\cref{sec:poly:decoupling}). This step introduces a multiplicative
factor of roughly $d^d$ to the approximation ratio for a degree $d$ homogeneous polynomial.
In~\cite{So}, the authors then recursively approximate the multilinear optimization problem, with a
crucial step in their approach relying on tools from geometry. In each recursive step of~\cite{So}, the
problem is reduced to approximating the diameter of a related convex body. By applying known tools from
convex geometry, the authors solve the inductive step, leading to an overall approximation for the
multilinear problem and, ultimately, for polynomial optimization. The connection to convex body
estimation also appears in the works of~\cite{KN08}.

%% file: sa_csp/csp_concat.tex
\newcommand{\expop}{\mathbb{E}}
\newcommand{\B}{\zo}
\chapter{Lifting LP Lower Bounds for CSPs to Sherali-Adams}
\section{The Result in Context: Background and Overview}\label{sec:csp:intro}

Given a finite alphabet $\qary = \inbraces{0, \ldots, q-1}$ and a predicate $f: \qary^k \to \B$, an
instance of the problem $\maxkcsp(f)$ consists of (say) $m$ constraints over a set of $n$ variables
$x_1, \ldots, x_n$ taking values in the set $[q]$. 
Each constraint $C_i$ is of the form
$f(x_{i_1}+b_{i_1}, \ldots, x_{i_k} + b_{i_k})$ for some $k$-tuple of variables $(x_{i_1}, \ldots
x_{i_k})$ and $b_{i_1}, \ldots, b_{i_q}  \in \qary$, and the addition is taken to be modulo
$q$. We say an assignment $\sigma$ to the variables satisfying the constraint $C_i$ if
$C_i(\sigma(x_{i_1}), \ldots, \sigma(x_{i_k})) = 1$. Given an instance $\Phi$ of the problem, the
goal is to find an assignment $\sigma$ to the variables satisfying as many constraints as possible.
The approximability of the $\maxkcsp(f)$ problem has been extensively studied for various predicates $f$
(see \eg~\cite{Hastad07:survey} for a survey), and special cases include several interesting and
natural problems such as \maxthreesat, \maxthreexor and \maxcut.

A topic of much recent interest has been the efficacy of Linear Programming (LP) and Semidefinite
Programming (SDP) relaxations. For a given instance $\Phi$ of $\maxkcsp(f)$, let $\opt(\Phi)$ denote
the \emph{fraction} of constraints satisfied by an optimal assignment, and let $\lpopt(\Phi)$ denote
the value of the convex (LP/SDP) relaxation for the problem. Then, the performance guarantee of this
algorithm is given by the {\deffont integrality gap} which equals the supremum of
$\frac{\lpopt(\Phi)}{\opt(\Phi)}$, over all instances $\Phi$. 

The study of unconditional lower bounds for general families of LP relaxations was initiated by Arora,
\Bollobas~and \Lovasz~\cite{AroraBL02} (see also~\cite{AroraBLT06}). They studied the
\Lovasz-Schrijver~\cite{LoS91} LP hierarchy and proved lower bounds on the integrality gap for
\vertexcover~(their technique also yields similar bounds for \maxcut).  De la Vega and
Kenyon-Mathieu~\cite{delaVegaK07} and Charikar, Makarychev and Makarychev~\cite{CMM09SACSP} proved a
lower bound of $2-o(1)$ for the integrality gap of the LP relaxations for \maxcut~given respectively by
$\Omega(\log \log n)$ and $n^{\Omega(1)}$ levels  of the Sherali-Adams LP hierarchy~\cite{SA90}.  Several
follow-up works have also shown lower bounds for various other special cases of the \maxkcsp~problem,
both for LP and SDP hierarchies~\cite{AroraAT05, Schoenebeck08, Tulsiani09, RaghavendraS09,
BenabbasGMT12, BarakCK15, KothariMOW17}.

A recent result by Chan \etal~\cite{ChanLRS13} shows a connection between strong lower bounds for the
Sherali-Adams hierarchy, and lower bounds on the size of LP extended formulations for the corresponding
problem.  In fact, their result proved a connection not only for a lower bound on the worst case
integrality gap, but for the entire \emph{approximability curve}. We say that $\Phi$ is 
\deffont{$(c,s)$-integrality gap instance} for a relaxation of $\maxkcsp(f)$, if we have $\lpopt(\Phi) \geq c$ and
$\opt(\Phi) < s$. And we say that $\Phi$ is \deffont{$(c,s)$-approximable} by a relaxation of
$\maxkcsp(f)$, if for instances with $\opt(\phi)<s$, we have $\lpopt(\Phi) \le c$. They showed that for
any fixed $t \in \N$, if there exist $(c,s)$-integrality gap instances of size $n$ for the relaxation
given by $t$ levels of the Sherali-Adams hierarchy, then for all $\eps > 0$ and sufficiently large $N$,
there exists a $(c -\eps,s+\eps)$ integrality gap instance of size (number of variables) $N$, for any
linear extended formulation of size at most $N^{t/2}$.  They also give a tradeoff when $t$ is a function
of $n$. This was recently improved by Kothari \etal~\cite{KothariMR16} and we describe the improved
tradeoff later.

We strengthen the above results by showing that for all $c,s \in [0,1]$,  
$(c,s)$-integrality gap instances for a ``basic LP'' can be used to construct $(c-\eps,s+\eps)$ integrality
gap instances for $\Omega_{\eps}\inparen{\frac{\log n}{\log \log n}}$ levels of the Sherali-Adams
hierarchy. The basic LP uses only a subset of the constraints in the relaxation given by $k$ levels
of the Sherali-Adams hierarchy for $\maxkcsp(f)$. In particular, this shows that a lower bound on
the integrality gap for even the basic LP, implies a similar lower bound on the integrality gap of any
polynomial size extended formulation. 
This can also be viewed as a dichotomy result showing that for any predicate $f$, either
$\maxkcsp(f)$ is $(c,s)$-approximable by the \emph{basic LP relaxation} (which is of size linear in
the size of the instance) or for all $\eps > 0$, a $(c-\eps, s+\eps)$ cannot be achieved by
\emph{any polynomial sized LP extended formulation}.
We note that both the above results 
and our result apply to all $f, q$ and all $c,s \in [0,1]$.

\paragraph{Comparison with (implications of) Raghavendra's UGC hardness result.} 
A remarkable result by Raghavendra \cite{Raghavendra08} shows that a $(c,s)$-integrality gap
instance for a ``basic SDP'' relaxation of $\maxkcsp(f)$ implies hardness of distinguishing
instances $\Phi$ with $\opt(\Phi) < s$ from instances with $\opt(\Phi) \geq c$, assuming the Unique
Games Conjecture (UGC) of Khot \cite{Khot02:unique}. The basic SDP considered by
Raghavendra involves moments for all pairs of variables, and all subsets of variables included in a
constraint. 
The basic LP we consider is weaker than this SDP and does not contain the positive
semidefiniteness constraint.

Combining Raghavendra's result with known
constructions of integrality gaps for Unique Games by Raghavendra and Steurer~\cite{RaghavendraS09},
and by Khot and Saket~\cite{KhotS09}, one can obtain a result qualitatively similar to ours, for the
mixed hierarchy. In particular, a $(c,s)$ integrality gap for the basic SDP implies a
$(c-\eps,s+\eps)$ integrality gap for $\Omega((\log \log n)^{1/4})$ levels of the mixed hierarchy. 

Note however, that the above result is incomparable to our result, since it starts with stronger
hypothesis (a basic SDP gap) and yields a gap for the mixed hierarchy as opposed to the
Sherali-Adams hierarchy. While the above can also be used to derive lower bounds for linear extended
formulations, one needs to start with an SDP gap instance to derive an LP lower bound. The basic SDP
is known to be provably stronger than the basic LP for several problems including various
2-CSPs. Also, for the worst case $f$ for $q=2$, the integrality gap of the basic SDP is
$O(2^k/k)$~\cite{CMM07UGCCSP}, while that of the basic LP is $2^{k-1}$.

A recent result by Khot and Saket~\cite{KhotS15} shows a connection between the integrality
gaps for the basic LP and those for the basic SDP.
They prove that, assuming the UGC, a $(c,s)$ integrality gap instance for the basic LP implies an
NP-hardness of distinguishing instances $\Phi$ with $\opt(\Phi) \geq \Omega\inparen{\frac{c}{k^3
    \cdot \log(q)}}$ from instances with $\opt(\Phi) \leq 4s$.
Their result also shows that a $(c,s)$ integrality gap instance for the basic LP can be used to
produce  a $\inparen{\Omega\inparen{\frac{c}{k^3 \cdot \log(q)}}, 4s}$ integrality gap instance for
the basic SDP,  and hence for $\Omega({(\log \log n)}^{1/4})$ levels of the mixed hierarchy.

\paragraph{Other related work.}
The power of the basic LP for solving valued CSPs \emph{to optimality} has been studied in several
previous works. These works consider the problem of minimizing the penalty for unsatisfied
constraints, where the penalties take values in $\Q \cup \inbraces{\infty}$. Also, they study the
problem not only in terms of single predicate $f$, but rather in terms of the constraint
language generated by a given set of (valued) predicates. 

It was shown by Thapper and \v{Z}ivn\'{y} \cite{ThapperZ13} that when the penalties are
finite-valued, if the problem of finiding the optimum solution cannot be solved by the basic LP,
then it is NP-hard. Kolmogorov, Thapper and \v{Z}ivn\'{y} \cite{KolmogorovTZ15} give a
characterization of CSPs where the problem of minimizing the penalty for unsatisfied constraints can
be solved \emph{exactly} by the basic LP. Also, a recent result by Thapper and \v{Z}ivn\'{y}
\cite{ThapperZ16} shows the valued CSP problem for a constraint language can be solved to optimality
by a bounded number of levels of the Sherali-Adams hierarchy if and only if it can be solved by a
relaxation obtained by augmenting the basic LP with contraints implied by three levels of the
Sherali-Adams hierarchy. However, the above works only consider the case when the LP gives an exact
solution, and do not focus on approximation.

The techniques from~\cite{CMM09SACSP} used in our result, were also extended by Lee~\cite{Lee15}
to prove a hardness for the Graph Pricing problem. Kenkre \etal~\cite{KenkrePPS15} also applied these to
show the optimality of a simple LP-based algorithm for Digraph Ordering.

\subsubsection*{Our results}
Our main result is the following.
\begin{restatable}{reptheorem}{maintheoremcsp}\label{thm:csp:main}
Let $f: \qary^k \to \B$ be any predicate. Let $\Phi_0$ be a $(c,s)$ integrality gap instance for basic
LP relaxation of \maxkcsp($f$). Then for every $\eps > 0$, there exists $c_{\eps} > 0$ 
such that for infinitely many $N \in \N$, there exist $(c-\eps, s+\eps)$ integrality gap
instances of size $N$ for the LP relaxation given by $c_{\eps} \cdot \frac{\log N}{\log \log N}$
levels of the Sherali-Adams hierarchy.
\end{restatable}

Combining the above with the connection between Sherali-Adams gaps and extended formulations 
by~\cite{ChanLRS13,KothariMR16} yields the following corollary:
\begin{restatable}{repcorollary}{sizecorollary}\label{cor:csp:size-lower-bound}
Let $f: \qary \to \B$ be any predicate. Let $\Phi_0$ be a $(c,s)$ integrality gap instance for basic
LP relaxation of \maxkcsp($f$). Then for every $\eps > 0$, there exists $c'_{\eps} > 0$ 
such that for infinitely $N \in \N$,  there exist $(c-\eps, s+\eps)$ integrality gap instances of
size $N$,  for every linear extended formulation of size at most $N^{c'_{\eps} \cdot
\frac{\log N}{\log \log N}}$.
\end{restatable}
As an application of our methods, we also simplify and strengthen the approximation resistance
results for LPs proved by Khot \etal~\cite{KhotTW14}. They studied predicates $f: \B^k \to \B$ and
provided a necessary and sufficient condition for the predicate to be {\deffont strongly
  approximation resistant} for the Sherali-Adams LP hierarchy. We say a predicate is strongly
approximation resistant if for all $\eps > 0$, it is hard to distinguish instances $\Phi$ for which
$\abs{\opt(\Phi) - \Ex{x \in \B^k}{f(x)}} \leq \eps$ from instances with $\opt(\Phi) \geq
1-\eps$. In the context of the Sherali-Adams hierarchy, they showed that when this condition is
satisfied, there exist instances $\Phi$ satisfying $\abs{\opt(\Phi) - \Ex{x \in
    \B^k}{f(x)}} \leq \eps$ and $\lpopt(\Phi) \geq 1-\eps$, where $\lpopt(\Phi)$ is the value of the
relaxation given by $O_{\eps}(\log \log n)$ levels of the Sherali-Adams hierarchy. We strengthen
their result (and provide a simpler proof) to prove the following.
\begin{theorem}\label{thm:csp:ktw}
Let $f: \B^k \to \B$ be any predicate satisfying the condition for strong approximation resistance
for LPs, given by~\cite{KhotTW14}. Then for every $\eps > 0$, there exists $c_{\eps} > 0$ 
such that infinitely many $N \in \N$, there exists an instance $\Phi$ of $\maxkcsp(f)$ of size $N$, satisfying
\[
  \abs{\opt(\Phi) - \Ex{x \in \B^k}{f(x)}}~\leq~\eps \qquad \lpopt(\Phi)~\geq~1-\eps \mcom
\]
where $\lpopt(\Phi)$ is the value of the relaxation given by $c_{\eps} \cdot \frac{\log N}{\log \log
  N}$ levels of the Sherali-Adams hierarchy.
\end{theorem} 
As before, the above theorem also yields a corollary for extended formulations.

\section{Preliminaries of CSPs and Their Relaxations}\label{sec:csp:prelims}
\subsection{Constraint Satisfaction Problems}\label{sec:csp:def}
We use $[n]$ to denote the set $\inbraces{1, \ldots, n}$. The only exception is $\qary$, where we
overload this notation to denote the set $\inbraces{0, \ldots, q-1}$, which corresponds to the
alphabet for the Constraint Satisfaction Problem under consideration.

\begin{definition}
Let $\qary$ denote the set $\inbraces{0, \ldots, q-1}$. 
For a predicate $f : \qary^k \rightarrow \{0,1\}$, an instance
$\Phi$ of \maxkcspq$(f)$ consists of a set of variables $\{x_1,\ldots,x_n\}$ and a set of
constraints $C_1, \ldots, C_m$. Each constraint $C_i$ is over a $k$-tuple of variables
$(x_{i_1}, \ldots, x_{i_k})$ and is of the form
\[
C_i ~\equiv~ f(x_{i_1} + b_{i_1}, \ldots, x_{i_k} + b_{i_k})
\]
for some $b_{i_1}, \ldots, b_{i_k} \in \qary$, where the addition is modulo $q$.
For an assignment $\sigma: \{x_1,\ldots,x_n\} \mapsto \qary$, let $\sat(\sigma)$ denote the fraction of
constraints satisfied by $\sigma$. 
The maximum fraction of constraints that can be simultaneously satisfied is denoted by $\opt(\Phi)$,
i.e.
\[ 
\opt(\Phi) = \max_{\sigma: \{x_1,\ldots,x_n\} \mapsto \qary}  \sat(\sigma). 
\]
\end{definition}

For a constraint $C$ of the above form, we use $x_C$ to denote the tuple of variables $(x_{i_1},
\ldots, x_{i_k})$ and $b_C$ to denote the tuple $(b_{i_1}, \ldots, b_{i_k})$. We then write
the constraint as $f(x_C + b_C)$. We also denote by $S_C$ the set of indices
$\{i_1,\ldots, i_k\}$ of the variables participating in the constraint $C$.

\subsection{The LP Relaxations for Constraint Satisfaction Problems}\label{sec:csp:relaxations}
Below we present various LP relaxations for the \maxkcspq$(f)$ problem that are relevant in this chapter.
Note that an integer solution to the problem can be given by an assignment $\sigma: [n] \to \qary$.

We start with the level-$t$ Sherali-Adams relaxation.
The intuition behind it is the following: instead of looking for an assignment to satisfy maximum number
of constraints, we look for distribution on assignments that maximize expected satisfied constraint.
Note that this, in itself, does not relax the problem as we can sample an assignment form the optimal
distribution to get an optimal assignment (the optimal distribution may be supported on one assignment).
The relaxation happens by looking at the collection of ``local distributions'' with consistency
constraints as described below. A feasible solution to our relaxation consists of local distributions for
all sets $S\subseteq [n], 1\leq \abs{S}\leq t$, where $t$ is the level of Sherali-Adams relaxation.
We define $\{0,1\}$-valued variables $\vartwo{S}{\alpha}$ for each $S \subseteq [n], 1 \leq |S| \leq t$
and $\alpha \in \qary^S$, to be the probability of $S$ getting assigned $\alpha$ in \emph{any} local
distribution of a feasible solution. In case of optimal distribution with unique support, we get
$\vartwo{S}{\alpha} = 1$ if $\sigma(S) = \alpha$ and 0 otherwise.
We also introduce a variable $\vartwoempty$, which equals to 1.
One of the consistency constraint is that the variables $\{\vartwo{S}{\alpha}\}_{\alpha \in \qary^{S}}$ give a
probability distribution over assignments to $S$, we denote it by $\dist_S$.
A further consistency constraint between these local distributions
by requiring that for $T \subseteq S$, the distribution over assignments to $S$, when marginalized
to $T$, denoted as $\dist_{S|T}$ is precisely the distribution over assignments to $T$ \ie $\dist_{S|T} = \dist_T$.
The relaxation is shown in~\cref{fig:csp:SA-lp}. Note that second and third constraint imply that for
each local distribution, probabilities sum to $1$.

\begin{figure}[t]
  \hrule
  \vline
  \begin{minipage}[t]{0.99\linewidth}
    \vspace{-5 pt}
    {\small
      \begin{align*}
        \mbox{maximize} &~~\Ex{C \in \Phi}{\sum_{\alpha \in \qary^k} f(\alpha \cdot b_C) \cdot
        \vartwo{S_C}{\alpha} }&  \\
        \vartwo{S}{\alpha} &~\geq~0
        & \forall S \subseteq [n], |S| \leq t,~\forall \alpha \in \qary^{S} \\ 
        \vartwoempty &~=~1\\
        \sum_{\substack{\alpha \in \qary^{S} \\ \alpha|_T = \beta} }\vartwo{S}{\alpha} &~=~\vartwo{T}{\beta}
        &\forall T \subseteq S \subseteq [n], |S|\leq t,~\forall \beta \in \qary^T \\
    \end{align*}}
    \vspace{-14 pt}
  \end{minipage}
  \hfill\vline
  \hrule
  \caption{Level-$t$ Sherali-Adams LP for \maxkcspq($f$)}\label{fig:csp:SA-lp}
\end{figure}

The basic LP relaxation is a reduced form of the above relaxation where only
those variables $\vartwo{S}{\alpha}$ are included for which $S  = S_C$ is the set of CSP variables
for some constraint $C$. The consistency constraints are included
only for singleton subsets of the sets $S_C$.
Note that the all the constraints for the
basic LP are implied by the relaxation obtained by level $k$ of the Sherali-Adams hierarchy.
\begin{figure}[ht]
  \hrule
  \vline
  \begin{minipage}[t]{0.99\linewidth}
    \vspace{-5 pt}
    {\small
      \begin{align*}
        \mbox{maximize} &~~\Ex{C \in \Phi}{\sum_{\alpha \in \qary^k} f(\alpha + b_C) \cdot
        \vartwo{S_C}{\alpha} }&  \\
        \sum_{j \in \qary}\vartwo{i}{b} &~=~1 &\forall i \in [n] \\
        \sum_{\substack{\alpha \in \qary^{S_C} \\ \alpha(i) = b} }\vartwo{S_C}{\alpha} &~=~\vartwo{i}{b}
        &\forall C \in \Phi, i \in S_C, b \in \qary\\
        \vartwo{S_C}{\alpha} &~\geq~0& \forall C \in \Phi,~\forall \alpha \in \qary^{S_C}
      \end{align*}
    }
    \vspace{-14 pt}
  \end{minipage}
  \hfill\vline
  \hrule
  \caption{Basic LP relaxation for \maxkcspq($f$)}\label{fig:csp:basic-lp}
\end{figure}

For an LP/SDP relaxation of \maxkcspq, and for a given instance $\Phi$ of the problem, we denote by
$\sdpopt(\Phi)$ the LP/SDP (fractional) optimum.
A relaxation is said to have a $(c, s)$-integrality gap if there exists a
CSP instance $\Phi$ such that $\sdpopt(\Phi) \geq c$ and $\opt(\Phi) < s$.

The main theorem we prove concerning the hardness of \maxkcspq\ in two different LP relaxations above:
\maintheoremcsp*

\subsubsection*{Proof Overview and Techniques}
\paragraph{The Gap Instance.} 
The construction of our gap instances is inspired by the construction by Khot \etal~\cite{KhotTW14}.
They gave a generic construction to prove integrality gaps for any approximation
resistant predicate (starting from certificates of hardness in form of certain ``vanishing
measures''), and we use similar ideas to give a construction which can start from a basic LP
integrality gap instance as a certificate, to produce a gap instance for many
levels. This construction is discussed in~\cref{sec:csp:sagaps}.

Given an integrality gap instance $\Phi_0$ on $n_0$ variables, we treat this as a
``template'' (as in Raghavendra~\cite{Raghavendra08}) and generate a random instance using
this. Concretely, we generate a new instance $\Phi$ on $n_0$ sets of $n$ variables each. To
generate a constraint, we sample a random constraint  $C_0 \in \Phi_0$, and pick a variable randomly
from each of the sets corresponding to variables in $C_0$. 
Thus, the instances generated are
$n_0$-partite random hypergraphs, with each edge being generated according to a specified ``type''
(indices of sets to chose vertices from). 

Note that previous instances of gap constructions for LP and SDP hierarchies were (hyper)graphs generated
according to the model ${\cal G}_{n,p}$ with the signs of the literals chosen independently at random.
However, proving an LP/SDP lower bound using such instances implies a strong result: it in fact proves
that the predicate $f$ is \emph{useless} for the corresponding relaxation, in the sense defined
by~\cite{AustrinH13}. Assuming the UGC, uselessness only holds for a limited class of predicates $f$
(when $f^{-1}(1)$ supports a balanced pairwise independent distribution on ${[q]}^k$), as shown
in~\cite{AustrinH13}. Thus, proving an SDP lower bound for predicates which are not expected to be
useless requires a new construction of instances, which cannot be generated uniformly at random. Our
construction provides such a generalization, and may be useful in proving new SDP lower bounds.  The
properties of random ${\cal G}_{n,p}$ hypergraphs easily carry over to our instances, and we collect
these properties in~\cref{sec:csp:hypergraphs}.

The above construction ensures that if the instance $\Phi_0$ does not have an assignment satisfying more
than an $s$ fraction of the constraints, then $\opt(\Phi) \leq s+\eps$ with high probability.  Also, it
is well-known that providing a good LP solution to the relaxation given by $t$ levels of the
Sherali-Adams hierarchy is equivalent to providing distributions $\dist_S$ on $\qary^S$ for all sets of
variables  $S$ with $\abs{S} \leq t$, such that the distributions are consistent restricted to subsets
\ie\ for all $S$ with $\abs{S} \leq t$ and all $T \subseteq S$, we have $\dist_{S|T} = \dist_{T}$.  Thus,
in our case, we need to produce such consistent local distributions such that the expected probability
that a random constraint $C \in \Phi$ is satisfied by the local distribution on the set of variables
involved in $C$ (which we denote as $S_C$) is at least $c - \eps$.

\paragraph{Local Distributions from Local Structure.}
Most works on integrality gaps for CSPs utilize the local structure of random
hypergraphs to produce such distributions. Since the girth of a sparse random hypergraph is
$\Omega(\log n)$, any induced subgraph on $o(\log n)$ vertices is simply a forest. In case the
induced (hyper)graph $G_S$ on a set $S$ is a \emph{tree}, there is an easy distribution to consider: simply
choose an arbitrary root and propagate down the tree by sampling each child conditioned on its
parent. It is also easy to see that for $T \subseteq S$, if the induced (hyper)graph  $G_T$ is a
\emph{subtree} of $G_S$, then the distributions $\dist_S$ and $\dist_T$ produced as above
are consistent.

The extension of this idea to forests requires some care. One can consider extending the
distribution to forests by propagating independently on each tree in the forest.
However, if for $T \subseteq S$ $G_T$ is a forest while $G_S$ is a tree, then a pair of vertices
disconnected in $G_T$ will have no correlation in $\dist_T$ but may be correlated in $\dist_S$.
This was handled, for example, in \cite{KhotTW14} by adding noise to the propagation and using a
large ball $B(S)$ around $S$ to define $\dist_S$. Then, if two vertices of $T$ are disconnected in
$B(T)$ but connected in $B(S)$, then they must be at a large distance from each other. Thus, because
of the noise, the correlation between them (which is zero in $\dist_T$) will be very small in
$\dist_S$. However, correcting approximate consistency to exact consistency incurs a cost which is
exponential in the number of levels (\ie the sizes of the sets), which is what limits the results in
\cite{KhotTW14, delaVegaK07} to $O(\log \log n)$ levels. This also makes the proof more
involved since it requires a careful control of the errors in consistency.

\paragraph{Consistent Partitioning Schemes.}
We resolve the above consistency issue by first partitioning the given set $S$ into a set of
clusters, each of which have diameter $\Delta_H = o(\log n)$ in the underlying hypergraph $H$.
Since each cluster has bounded diameter, it becomes a tree when we add all the missing paths
between any two vertices in the cluster. We then propagate independently on each
cluster (augmented with the missing paths). This preserves the correlation between any two vertices
in the same cluster, even if the path between them was not originally present in $G_S$.

Of course, the above plan requires that the partition obtained for $T \subseteq S$, is consistent
with the restriction to $T$ of partition obtained for the set $S$. In fact, we construct
distributions over partitions $\inbraces{\calP_S}_{\abs{S} \leq t}$, which satisfy the consistency
property $\calP_{S | T} = \calP_T$. These distributions over partitions, which we call {\deffont
  consistent partitioning schemes}, are constructed in~\cref{sec:csp:decomposition}.

In addition to being consistent, we require that the partitioning scheme cuts only a few
edges in expectation, since these contribute to a loss in the LP objective. We remark that such
low-diameter decompositions (known as \emph{separating} and \emph{padded} decompositions) have been
used extensively in the theory metric embeddings (see \eg~\cite{KrauthgamerLMN05} and the references
therein). The only additional requirement in our application is consistency.

We obtain the decompositions by proving the (easy) hypergraph extensions the results of Charikar,
Makarychev and Makarychev \cite{CMM07Metric}, who exhibit a metric which is similar to the shortest
path metric on graphs at small distances, and has the property that its restriction to any subset of
size at most $n^{\eps'}$ (for an appropriate $\eps' < 1$) 
is $\ell_2$ embeddable. This is proved in~\cref{sec:csp:hypergraphs}.
A variant of this metric was used by Charikar, Makarychev and Makarychev~\cite{CMM09SACSP} to
prove lower bounds for \maxcut, for $n^{\eps'}$ levels of the Sherali-Adams hierarchy. They used the
embedding to construct a ``local SDP solution'' for any $n^{\eps'}$ variables (with value $1-\eps'$) 
and produced the distributions required for Sherali-Adams by rounding the SDP solutions (which gives
value $1 - O(\sqrt{\eps'})$). However, rounding an SDP solution with a high value does 
not always produce a good integral solution for other CSPs.

Instead, we use these metrics in~\cref{sec:csp:decomposition} to construct the consistent partitioning
schemes as described above, by applying  a result of Charikar \etal~\cite{CharikarCGGP98} giving separating
decompositions for finite subsets of $\ell_2$.
We remark that it is the consistency requirement of the partitioning procedure that limits our results
to $O\inparen{\frac{\log n}{\log \log n}}$ levels. The separation probability in the decomposition
procedure grows with the dimension of the $\ell_2$ embedding, while (to the best of our knowledge)
dimension reduction procedures seem to break consistency.

\section{Hypergraphs}
An instance $\Phi$ of \maxkcsp defines a natural associated hypergraph $H = (V,E)$ with $V$ being
the set of variables in $\Phi$ and $E$ containing one $k$-hyperedge for every constraint $C \in \Phi$.
We remind the reader of the familiar notions of degree, paths, and cycles for the case of 
($k$-uniform) hypergraphs:
\begin{definition}
  Let $H=(V,E)$ be a hypergraph. 
\begin{itemize}
\item For a vertex $v\in V$, the {\deffont degree} of the vertex $v$ is defined to be the number of
  distinct hyperedges containing it.
\item   A {\deffont simple path} $P$ is a finite alternate sequence of distinct vertices and
  distinct edges starting and ending at vertices, \ie $P=v_1,e_1,v_2,\ldots ,v_{l},e_{l},v_{l+1}$,
  where $v_i \in V ~\forall i \in [l+1]$ and $e_i\in E~ \forall i\in [l]$.  Furthermore, $e_i$
  contains $v_i,v_{i+1}$ for each $i$. Here $l$ is called the {\deffont length} of the path $P$. 
  All paths discussed in this chapter will be simple paths.
\item A sequence $\calC = (v_1,e_1,v_2,\ldots, v_l,e_l,v_1)$ is called a cycle of length
  $l$ if the initial segment $v_1,e_1,\ldots, v_l$ is a (simple) path, $e_{l+1}\ne e_i$ for all $i\in [l]$,
  and $v_1\in e_l$. For a path $P$ (or cycle $\calC$), we use $\Vtx(P)$ (or $\Vtx(\calC)$) to 
  denote the set of all the vertices that occur in the edges, \ie the set $\{v\suchthat (\exists i\in[h])
  (v\in e_i)\}$, where $e_1,\ldots,e_h$ are the hyperedges included in $P$ (or $\calC$).
\item For a given hypergraph $H$, the length of the smallest cycle in $H$ is called the
  {\deffont girth} of $H$.
\end{itemize}
\end{definition}
To observe the difference the notions of cycle in graphs and hypergraphs, it is instructive to
consider the following example: let $u,v$ be two distinct vertices in a $k$-uniform hypergraph for
$k \geq 3$, and let $e_1,e_2$ be two distinct hyperedges both containing $u$ and $v$. Then
$u,e_1,v,e_2,u$ is a cycle of length $2$, which cannot occur in a graph.
 
We shall also need the following notion of the \emph{closure} of a set $S \subseteq V$ in a given
hypergraph $H$, defined by~\cite{CMM09SACSP} for the case of graphs. A stronger notion of closure
was also considered by~\cite{BarakCK15}.
\begin{definition}\label{def:csp:closure}
  For a given hypergraph $H$ and $R \in \N$, and a set $S\subseteq \Vtx(H)$, we denote by $\clR(S)$ 
  the {\deffont $R$-closure} of $S$ obtained by adding all the vertices in all the paths of length at most $R$
  connecting two vertices of $S$:
  \[
    \clR(S)=S\cup \bigcup_{\substack{P: P\text{ is a path in H} \\ P\text{ connects }u,v\in S \\
    \abs{P}\le R}}\Vtx(P)\mper
  \]
  For ease of notation, we use $\cl(S)$ to denote $\cl_1(S)$.
\end{definition}

\subsection{Properties of Random Hypergraphs}\label{sec:csp:hypergraphs}
In this section we collect various properties of the hypergraphs corresponding to our integrality gap
instances.
The gap instances we generate contain several disjoint collections of variables. Each
constraint in the instance has a specified ``type'', which specifies which of the collections each
of the participating $k$ variables much be sampled from. The constraint is generated by randomly 
sampling each of the $k$ variables, from the collections specified by its type. This is captured by
the generative model described below. 

In the model below and in the construction of the gap instance, the parameter $n_0$ should be
thought of as constant, while the parameters $n$ and $m$ should be though of a growing to infinity. 
We will choose $m = \gamma \cdot n$ for $\gamma = O_{k,q}(1)$.
\begin{definition}
Let $n_0, k \in \N$ with $k \geq 2$. Let $m, n > 0$ and let $\Gamma$ be a distribution on
$\insquare{n_0}^k$. We define a distribution $\Hyp$ on $k$-uniform $n_0$-partite hypergraphs with $m$
edges and $N=n_0 \cdot n$ vertices, divided in $n_0$ sets $X_1, \ldots, X_{n_0}$ of size $n$ each. A
random hypergraph $H \sim \Hyp$ is generated by sampling $m$ random hyperedges independently as
follows:
\begin{itemize}
\item Sample a random type $(i_1, \ldots, i_k) \in [n_0]^k$ from the distribution $\Gamma$.
\item For all $j \in [k]$, sample $v_{i_j}$ independently and uniformly in $X_{i_j}$.
\item Add the edge $e_i = \inbraces{v_{i_1}, \ldots, v_{i_k}}$ to $H$.
\end{itemize}
\end{definition}
Note that as specified above, the model may generate a multi-hypergraph. However, the number of such
repeated edges is likely to be small, and we will bound these, and in fact the number of cycles of
size $o(\log n)$ in~\cref{lem:csp:cycle-count}.

We will study the following metrics (similar to the ones defined in~\cite{CMM07Metric}):
\begin{definition}\label{def:csp:metric-dmu}
  Given a hypergraph $H$ with vertex set $V$, we define two metrics $d^H_\mu(\cdot\xspace,\cdot), \rho^H_\mu
  (\cdot\xspace,\cdot)$ on $V$ as
  \[
    d^H_\mu(u,v)~\coloneqq~1-{(1-\mu)}^{2\cdot d_H(u,v)}
    \qquad \text{ and } \qquad \rho^H_\mu(u,v)~\coloneqq~\sqrt{\frac{2\cdot d^H_\mu(u,v)+ \mu}{1+\mu}}\mcom
  \]
  for $u\ne v$, where $d_H(\cdot\xspace,\cdot)$ denotes the shortest path distance in $H$.
\end{definition}

We primarily need the fact the local $\ell_2$-embeddability of the metric $\rho_{\mu}$. 
The following theorem captures various properties of random hypergraphs required for our construction.
The proof of the theorem heavily uses results proved in~\cite{AroraBLT06}
and~\cite{CMM09SACSP}.
\begin{restatable}{reptheorem}{localmetric}\label{thm:csp:locally-l2}
  Let $H'\sim \Hyp$ with $m=\gamma\cdot n$ edges and let $\eps>0$. Then for large
  enough $n$, with high probability (at least $1-\eps$, over the choice of $H'$),
  there exists $\delta>0$, constant $c = c(k,\gamma,n_0,\eps)$, $\theta = \theta(k,\gamma,n_0,\eps)$ 
  and a subhypergraph $H\subset H'$ with $V(H) = V(H')$ satisfying the following:
\begin{itemize}
\item H has girth $\ttg\ge \delta\cdot \log n$.
\item $\abs{\Edg(H') \setminus\Edg(H)}\le \eps\cdot m$.
\item For all $t \leq n^{\theta}$, for $\mu \geq c \cdot \frac{\log t+\log\log n}{\log n}$, 
  for all $S\subseteq \Vtx(H')$ with
  $\abs{S} \leq t$, the metric $\rho^H_{\mu}$ restricted to $S$ is isometrically embeddable
  into the unit sphere in $\ell_2$,
\end{itemize}
\end{restatable}
This theorem is proved in the next section. We note for the reader that ideas and methods involved
in the proof of the above theorem are largely disconnected from those of the~\cref{thm:csp:main}.

\section[Local L2-embeddability of the Metric rho]{Local $\ell_2$-embeddability of the Metric
$\rho_{\mu}^H$}\label{sec:csp:appendix}
The goal of this section is to prove the following result about the local $\ell_2$-embeddability of
the metric $\rho_{\mu}^H$. 
\localmetric*
To prove the above theorem, we will use the local structure of random hypergraphs.
We first prove that with high probability for random hypergraphs (sampled from $\Hyp$)
a few hyperedges can be removed to obtain a hypergraph whose girth is $\Omega(\log n)$ and the
degree is bounded.
The following lemma shows a possible
trade-off between the degree of the hypergraph vs the number of hyperedges required to be
removed.
\begin{restatable}{replemma}{lowdegree}\label{lem:csp:low-degree}
  Let $H'\sim \Hyp$ be a random hypergraph with $m=\gamma\cdot n$ hyperedges. Then for any
  $\eps>0$, with probability $1 - \eps$, there exists a sub-hypergraph
  $H$ with $V(H) = V(H')$ such that $\forall u \in V(H)$, $\deg_H(u) \leq
  100\cdot\log\inparen{\frac{n_0}{\eps}}\cdot k \cdot \gamma$ and $\abs{\Edg(H')\setminus
  \Edg(H)}\le \eps\cdot m$.
\end{restatable}
\begin{proof}
  By linearity of expectation, the expected degree of any vertex $v$ in $H'$ is at most
  $k\cdot\gamma$. Let $D = 100\cdot\log\inparen{\frac{n_0}{\eps}}\cdot k\cdot \gamma$, and let $S$ be the
  set of all vertices $u$ with $\deg_{H'}(u) > D$. Let $E_S$ be the set of all hyperedges with at least
  one vertex in $S$. We shall take $E(H) = E(H') \setminus E_S$. 
Note that for any $u \in V(H')$, $\Pr{u \in S} = \Pr{\deg_{H'}(u) \geq D} \leq \exp(-D/4)$ by a
Chernoff-Hoeffding bound. We use this to bound the expected number of edges deleted.
\begin{align*}
\Ex{E_S} 
~\leq~ \sum_{u \in V(H')} \Ex{\deg(u) \cdot \indicator{u \in S}} 
&~=~ \sum_{u \in V(H')} \Ex{\deg(u) ~\mid~ u \in S} \cdot \Pr{u \in S} \\
&~\leq~ \sum_{u \in V(H')} \Ex{\deg(u) ~\mid~ u \in S} \cdot \exp\inparen{-D/4} \\
&~\leq~ \sum_{u \in V(H')} \inparen{D+k\gamma} \cdot \exp\inparen{-D/4} \\
&~\leq~ (n \cdot n_0) \cdot 2D \cdot \exp\inparen{-D/4} \mper
\end{align*}
The penultimate inequality uses the independence of the hyperedges in the generation process, which gives
$\Ex{\deg_{H'}(u) \mid \deg_{H'}(u) \geq D} ~\leq~ D + \Ex{\deg_{H'}(u)}$. From our choice of the
parameter $D$, we get that $\Ex{E_S} \leq \eps^2 \cdot \gamma \cdot n = \eps^2 \cdot m$. Thus, the
number of edges deleted is at most $\eps \cdot m$ with probability at least $1-\eps$.
\end{proof}
The following lemma shows that the expected number of small cycles in random hypergraph is small.
\begin{restatable}{replemma}{cyclecount}\label{lem:csp:cycle-count}
Let $H \sim \Hyp$ be a random hypergraph and for $l \geq 2$, let $Z_{l}(H)$ denote the number of
cycles of length at most $l$ in $H$. For $m, n$ and $k$ such that $k^2 \cdot (m/n) > 1$, we have
  \[
    \Ex{H \sim \Hyp}{Z_{l}(H)} \le \inparen{k^2 \cdot \frac{m}{n}}^{2l} \mper
  \]
\end{restatable}
\begin{proof}
  Let the vertices of $H$ correspond to the set $[n_0]\times [n]$.
  Suppose we contract the set of $[n_0]\times\{j\}$ vertices into a single
  vertex $j \in [n]$ to get a random multi-hypergraph $H'$ on vertex set $[n]$.
  An equivalent way to view the sampling to $H'$ is: for each $i\in [m]$, the
  $i$-th hyperedge $e_i$ of $H'$ is sampled by independently sampling $k$ vertices (with
  replacement) uniformly at random from $[n]$. Note
  that the sampling of $H'$ is independent of $\Gamma$ in the definition of
  $\Hyp$.
  Clearly, a cycle of length at most $l$ in $H$ produces a cycle of length at most $l$ in
  $H'$. Hence, it suffices to bound the expected number of cycles in $H'$

  Given any pair $(u',v')$ of vertices of $H'$, for $u' \ne v'$, the probability of the
  pair $(u',v')$ belonging together in some hyperedge of $H'$ is at most $\frac{mk^2}{n^2}$.
  Consider a given $h$-tuple  of vertices $\bfu=(u_{i_1},\cdots u_{i_h})$. Note that we
  require that hyperedges participating in a cycle be distinct. So, the probability that $\bfu$
  is part of a cycle in $H'$, \ie there exists distinct hyperedges $e_j\in H'$ for $j\in [h]$
  such that $u_{i_j},u_{i_{j+1}}\in e_j$ for $j\in [h-1]$, and $u_{i_1},u_{i_h}\in e_{h}$
  is at most $\inparen{\frac{mk^2}{n^2}}^{h}$. As a result, expected number of
  cycles of length $h$ in $H'$ is bounded above by:
  \[
     \binom{n}{h} \inparen{\frac{mk^2}{n^2}}^{h} \le n^h \inparen{\frac{mk^2}{n^2}}^{h} =
    \inparen{k^2\cdot\frac{m}{n}}^{h}
  \]
  From the geometric form of the bound, it follows that expected number of cycles of length
  at most $l$ in $H'$ is at most $\frac{\inparen{k^2\cdot\frac{m}{n}}^{l+1}}{
    \inparen{k^2\cdot\frac{m}{n}}-1}<\inparen{k^2\cdot\frac{m}{n}}^{2l}\mper$
\end{proof}
Using the above lemma, it is easy to show that one can remove all small cycles in a random
hypergraph by deleting only a few hyperedges.
\begin{corollary}\label{cor:csp:large-girth}
  Let $H \sim \Hyp$ be a random hypergraph with $m = \gamma \cdot n$ for $\gamma>1$ and $k\ge 2$.
  Then, there exists $\delta=\delta(\gamma) > 0$ such that with probability $1-n^{-1/6}$, all cycles
  of length at most $\delta \cdot \log n$ in $H$ can be removed by deleting at most $n^{2/3}$ hyperedges.
\end{corollary}
\begin{proof}
  As above, let $Z_l$ denote the number of cycles of length at most $l$. With the choice of $m,n,$ and $k$,
  we have $k^2\cdot\frac{m}{n}\ge 2$. By~\cref{lem:csp:cycle-count}, 
  $\Ex{Z_l} \leq \inparen{k^2 \cdot \frac{m}{n}}^{2l}$. Since $m = \gamma \cdot n$, there exists a
  $\ttg = \delta \cdot \log n$ such that  $\Ex{Z_l} \leq \sqrt{n}$. By Markov's inequality,
  $\Pr{Z_l \geq n^{2/3}} \leq n^{-1/6}$. Thus, with probability $1-n^{-1/6}$, one can remove
  all cycles of length at most $\delta \cdot \log n$ by deleting at most $n^{2/3}$ edges.
\end{proof}
One can also extend the analysis in \cite{AroraBLT06} to show that the hypergraphs are locally
sparse, \ie the number of hyperedges contained in a small set of vertices is small. For a hypergraph $H$
and a set $S \subseteq \Vtx(H)$, we use $E(S)$ to denote the edges contained in the set $S$. 
\begin{definition}
We say that $S \subseteq V(H)$ is $\eta$-sparse if $\abs{E(S)} \leq \frac{\abs{S}}{k-1-\eta}$. We
call a $k$-uniform hypergraph $H$ on $N$ vertices to be {\deffont $(\tau,\eta)$-sparse} if all
subsets $S\subset \Vtx(H), \abs{S}\le \tau \cdot \abs{\Vtx(H)}$, $S$ is $\eta$-sparse. 
We call $H$ to be $\eta$-sparse if it is $(1,\eta)$-sparse, \ie all subsets of vertices of $H$ are
sparse.
\end{definition}
We note here that while this notion of sparsity is a generalization of that considered
in~\cite{AroraBLT06}, it is also identical to the notions of expansion considered in works in proof
complexity (see \eg~\cite{BS01}) and later in works on integrality gaps~\cite{AroraAT05,
BenabbasGMT12,BarakCK15}.  We prove that random hypergraphs generated with our model are locally sparse:
\begin{restatable}{replemma}{locallysparse}\label{lem:csp:locally-sparse}
  Let $\eta < 1/4$ and $m=\gamma\cdot n$ for $\gamma > 1$. Then
  for $\tau \le \frac{1}{n_0} \cdot \inparen{\frac{1}{e \cdot k^{3k} \cdot
        \gamma}}^{1/\eta}$ the following holds:
  \[
    \Pr{H \sim \Hyp}{H~\text{is not}~(\tau,\eta)\text{-sparse}}~\leq~3 \cdot
    \inparen{\frac{k^{3k} \cdot \gamma}{n^{\eta/4}}}^{1/k} \mper
  \]
\end{restatable}
 We note that we will require the sparsity $\eta$ to be $O_{k,\gamma}(1/\log n)$. This gives
 sparsity only for sublinear size sets, as compared to sets of size $\Omega(n)$ in previous works
 where $\eta$ is a constant. 
For the proof of the lemma, we follow an approach similar to that of~\cref{lem:csp:cycle-count}: we collapse
the vertices of $H$ of the form $[n_0]\times\{j\}$ to vertex $j\in [n]$ to construct $H'$, and thus reducing
the problem to random multi-hypergraph form a random multipartite hypergraph.
The rest proof of the lemma is along the lines of several known proofs \cite{AroraAT05,
  BenabbasGMT12}. 
\begin{proof}
  As in the proof of~\cref{lem:csp:cycle-count}, given a random hypergraph $H$,
  we construct a hypergraph $H'$ ( by contracting all the vertices in
  $[n_0]\times\{j\}$ to $j\in [n]$ ).

  Consider a subset of vertices $S \subseteq \Vtx(H)$ and let $S' \subseteq \Vtx(H')$ be the
  corresponding contracted set in $H'$. Since each edge in $H$ corresponds to an edge in $H'$
  (counting multiplicities), we have
\[
\abs{E(S)} ~\geq~ \frac{\abs{S}}{k-1-\eta} 
\quad  \Rightarrow \quad 
\abs{E(S')} ~\geq~ \frac{\abs{S}}{k-1-\eta} ~\geq~ \frac{\abs{S'}}{k-1-\eta} \mper
\]
Thus, it suffices to show that $H'$ is $(\tau',\eta)$-sparse for $\tau' = \tau \cdot n_0$, since
$\abs{S'} \leq \tau \cdot N = (\tau \cdot n_0) \cdot n$.
Given any multiset in $[n]^k$, the probability that it corresponds to an edge in $H'$ is at most 
$(k!) \cdot (m/n^k) $. Thus, the probability that there exists a set $T$ of size at most $\tau'
\cdot n$, containing at least $\abs{T}/(k-1-\eta)$ edges (counting multiplicities) is at most
\[
\sum_{h=1}^{\tau' \cdot n} \binom{n}{h} \cdot \binom{h^k}{r} \cdot \inparen{\frac{k! \cdot
    m}{n^k}}^r \mcom
\]
where $r = \frac{h}{k-1-\eta}$. Note that  we also need to consider $h=1$ as edges in $H'$
correspond to multisets of size $k$, and so may not have all distinct vertices. Simplifying the
above using $\binom{a}{b} \leq \inparen{\frac{a \cdot e}{b}}^b$ and $k! \leq k^k$ gives
\begin{align*}
\sum_{h=1}^{\tau' \cdot n} \binom{n}{h} \cdot \binom{h^k}{r} \cdot \inparen{\frac{k! \cdot
    m}{n^k}}^r
&~\leq~
\sum_{h=1}^{\tau' \cdot n} \inparen{\frac{n \cdot e}{h}}^h \cdot \inparen{\frac{h^k \cdot e}{r}}^r \cdot \inparen{\frac{k^k \cdot
    m}{n^k}}^r \\
&~=~
\sum_{h=1}^{\tau' \cdot n} \inparen{e^{k-\eta} \cdot (k-1-\eta) \cdot k^k \cdot \gamma \cdot
  \inparen{\frac{h}{n}}^{\eta}}^{h/(k-1-\eta)} \\
&~\leq~
\sum_{h=1}^{\tau' \cdot n} \inparen{k^{3k} \cdot \gamma \cdot
  \inparen{\frac{h}{n}}^{\eta}}^{h/(k-1-\eta)} \\
\end{align*}
Let $\theta = \eta/(2k)$. We divide the above summation in two parts and first consider
\begin{align*}
\sum_{h=n^{\theta}}^{\tau' \cdot n} \inparen{k^{3k} \cdot \gamma \cdot
  \inparen{\frac{h}{n}}^{\eta}}^{h/(k-1-\eta)}
&~\leq~
\sum_{h=n^{\theta}}^{\tau' \cdot n} \inparen{k^{3k} \cdot \gamma \cdot
  \inparen{\tau'}^{\eta}}^{n^{\theta}/(k-1-\eta)} \\
&~\leq~
2 \cdot \exp\inparen{- \frac{n^{\theta}}{k}} \\
&~\leq~
2 \cdot \frac{k}{n^{\theta}}
\mcom
\end{align*}
for $\tau' \leq \inparen{e \cdot k^{3k} \cdot \gamma}^{-1/\eta}$.
Considering the first half of the summation, we get
\begin{align*}
\sum_{h=1}^{n^{\theta}} \inparen{k^{3k} \cdot \gamma \cdot
  \inparen{\frac{h}{n}}^{\eta}}^{h/(k-1-\eta)}
&~\leq~
n^{\theta} \cdot \inparen{\frac{k^{3k} \cdot \gamma}{n^{(1-\theta) \cdot \eta}}}^{1/k} \\
&~\leq~
\inparen{\frac{k^{3k} \cdot \gamma}{n^{\eta/4}}}^{1/k} 
~=~
k^3 \cdot \gamma^{1/k} \cdot n^{-\theta/2} \mper
\end{align*}
Combining the two bounds gives that the probability is at most
~$3k^3 \cdot \gamma^{1/k} \cdot n^{-\theta/2}$, which equals the desired bound.
\end{proof}

Charikar \etal~\cite{CMM07Metric} prove an analogue of~\cref{thm:csp:locally-l2} for metrics defined on
locally-sparse graphs. In fact, they  use a consequence of sparsity, which they call $\ell$-path
decomposability. To this end, we define the \emph{incidence graph}\footnote{This is the same notion
  as the constraint-variable graph considered in various works on lower bounds for CSPs.} 
associated with a hypergraph, on which we will apply their result.
\begin{definition}
Let $H = (V(H),E(H))$ be a $k$-uniform hypergraph. We define its {\deffont incidence graph} as the
bipartite graph $G_H$ defined on vertex sets $V(H)$ and $E(H)$, and edge set $\calE$ defined as
\[
\calE ~\defeq~ \inbraces{(v,e) ~\mid~ v \in V(H), ~e \in E(H), ~v \in e} \mper
\]
\end{definition}
Note that for any $u,v \in V(H)$, we have $d_{G_H}(u,v) = 2 \cdot d_H(u,v)$.
We prove that for a locally sparse hypergraph $H$, its incidence graph $G_H$ is also locally sparse.
\begin{restatable}{replemma}{sparsityvc}\label{lem:csp:sparsity-variable-constraint}
  Let $H$ be a $k$-uniform $(\tau,\eta)$-sparse hypergraph on $N$ vertices with $m=\gamma\cdot n$
  hyperedges. Then the incidence graph $G_H$ is $(\tau',\eta')$ sparse for $\tau' =
  \nfrac{\tau}{k\cdot(1+\gamma)}$ and $\eta' = \nfrac{\eta}{(1+\eta)}$.
\end{restatable}
\begin{proof}
  Let $\tau'=\nfrac{\tau}{k\cdot(1+\gamma)}$ and let $G_H$ be the incidence graph with
  $N+m = (1+\gamma) \cdot N$ vertices. Let $G'$ be is the densest subgraph of $G_H$, among all subgraphs
  of size at most $\tau' \cdot (N+m)$. 
Let the vertex set of $G'$ be $V' \cup E'$ where $V' \subseteq V(H)$ and $E' \subseteq E(H)$, and
let the edge-set be $\calE'$. There cannot be any isolated vertices in $G'$ since removing those will only
increase the density.

Let $S \subseteq V(H)$ be the set of all vertices contained in all edges in $E'$ \ie 
$S ~\defeq~ \inbraces{v \in V(H) ~\mid~ \exists e \in E'~\text{s.t.}~ v \in e}$. Note that $V'
\subseteq S$, since there are no isolated vertices,  
and $E' \subseteq E(S)$, where $E(S)$ denotes the set of hyperedges contained in $S$.

By our choice of parameters, $\abs{S} \leq k \cdot \abs{E'} \leq k \cdot \tau' \cdot (N+m) \leq \tau
\cdot N$.
Thus, using the sparsity of $H$, we have
\[
\abs{E'} ~\leq~ \abs{E(S)} ~\leq~ \frac{\abs{S}}{k-1-\eta} \mper
\]
Also, since each hyperedge of $E'$ can include at most $k$ vertices in $S$, and since each edge in
$\calE'$ is incident on a vertex in $V'$, we have
\[
\abs{S} - \abs{V'} ~\leq~ k \cdot \abs{E'} - \abs{\calE'} \mper
\]
Combining the two inequalities gives
\[
(k-1-\eta) \cdot \abs{E'} ~\leq~ \abs{V'} + k \cdot \abs{E'} - \abs{\calE'}
\quad \implies \quad
\abs{\calE'} ~\leq~ (1+\eta) \cdot \abs{E'} + \abs{V'}  \mper
\]
Hence, we get that $\abs{\calE'} \leq \frac{\abs{V'} + \abs{E'}}{(1-\eta')}$ for $\eta' = \frac{\eta}{(1+\eta)}$.
\end{proof}
Charikar \etal~\cite{CMM07Metric} defined the following structural property of a graph.
\begin{definition}[\cite{CMM07Metric}]\label{def:csp:l-path-decmpsble} 
  A graph $G$ is {\sf $\ell$-path decomposable} if every $2$-connected subgraph $G'$ of $G$,
  such that $G'$ is not an edge, contains a path of length $\ell$ such that every vertex of the
  path has degree at most $2$ in $G'$.
\end{definition}
The above property was also implicitly used by Arora \etal~\cite{AroraBLT06}, who proved the
following (see Lemma 2.12 in~\cite{AroraBLT06}):
\begin{lemma}\label{lem:csp:path-decomposable}
  Let $\ell>0$ be an integer and $0< \eta < \frac{1}{3\ell-1} < 1$.
  Let $G$ be a $\eta$-sparse graph with girth $\girth>\ell$.
  Then $G$ is $\ell$-path decomposable.
\end{lemma}
Recall that we defined the metrics $\dmu$ and $\rho_\mu$ on $H$ as (for $u\ne v$) :
\[
    d^H_\mu(u,v) ~\coloneqq~ 1-(1-\mu)^{2\cdot d_H(u,v)}
     \qquad \text{ and } \qquad \rho^H_\mu(u,v) ~\coloneqq~
     \sqrt{\frac{2\cdot d^H_\mu(u,v)+ \mu}{1+\mu}}\mcom
\]
For a graph $G$, we define the following two metrics, for $u\ne v$:
\[
  d_\mu^G(u,v)~\coloneqq~1-{(-1)}^{d_G(u,v)}{(1-\mu)}^{d_G(u,v)}
  \quad \text{ and }  \quad\rho_\mu^G(u,v)~\coloneqq~\sqrt{\frac{2\cdot d_\mu^G(u,v)+ \mu}{1+\mu}}\mper
\]
We note that if $H$ is a hypergraph and $G_H$ is its incidence graph, then the metrics 
$d_\mu^{G_H}$ and $\rho_\mu^{G_H}$ restricted to $V(H)$, coincide with the metrics $d_\mu$ and
$\rho_\mu$ defined on $H$.
Charikar \etal~proved the following theorem (see Theorem $5.2$) in~\cite{CMM09SACSP}.
\begin{theorem}[\cite{CMM09SACSP}]\label{thm:csp:l2-graph-embedding}
  Let $G$ be a graph on $n'$ vertices with maximum degree $D$. Let
  $t< \sqrt{n'}$ and $\ell > 0$ be such that for $t'= D^{\ell+1}\cdot t$, every
  subgraph of $G$ on at most $t'$ vertices is $\ell$-path decomposable. Also, let $\mu$, $t$ and
  $\ell$ satisfy the relation ${(1-\mu)}^{\ell/9} \leq \frac{\mu}{2(t+1)}$.
Then for every subset $S$ of at most $t$ vertices there exists
  a mapping $\psi_S$ from $S$ to \emph{unit sphere} in $\ell_2$ such that all $u,v \in S$:
  \[
    \norm{\psi_S(u)-\psi_S(v)}_2~=~\rho_\mu^G(u,v) \mper
  \]
\end{theorem}
We use this theorem to prove the main theorem of the section.
\begin{proofof}{of~\cref{thm:csp:locally-l2}}
  Let $H'\sim \Hyp$ with $m=\gamma\cdot n$ hyperedges and $N=n_0\cdot n$
  vertices. Given $\eps>0$, from~\cref{lem:csp:low-degree} we have that with
  high probability at least $1-\eps/2$, there exists $H_1$ such that 
  the maximum degree of $H_1$ is at most $D=100\cdot\log\inparen{\frac{2n_0}{\eps}}
  \cdot k \cdot \gamma$ with $\abs{\Edg(H')\setminus\Edg(H_1)}\le (\eps/2)\cdot m$.

Using~\cref{cor:csp:large-girth} we also have that there exists $\delta > 0$, such that with probability
at least $1-\eps/4$ (for large enough $n$) $H'$ has a sub-hypergraph $H_2$ with $\ttg \geq
\delta\cdot\log n$ and $\abs{\Edg(H') \setminus\Edg(H_2)}\le (\eps/4) \cdot m$.
By~\cref{lem:csp:locally-sparse}, there exists $\eta = \Omega_{n_0,k,\gamma,\eps}(1/(\log n))$ such
that $H'$ is $(\tau,\eta)$-sparse with probability at least $1-\eps/4$, for $\tau \geq n^{-1/4}$.
  
Hence with probability $1-\eps$, we have that $H = (V(H'), E(H_1) \cap E(H_2))$ satisfies:
  \begin{itemize}
    \item Degree of $H$ is bounded above by $D$.
    \item $H$ is $(\tau,\eta)$-sparse (for $\tau \geq n^{-1/4}$ and $\eta =
      \Omega_{n_0,k,\gamma,\eps}(1/(\log n)$).
    \item Girth of $H$ is at least $\ttg > \delta\cdot\log n$.
    \item $\abs{\Edg(H')\setminus\Edg(H)}\le \eps\cdot m$.
  \end{itemize}
We now show that the metric $\rho^H_{\mu}$ is locally $\ell_2$ embeddable.

Let $G=G_H$ be the incidence graph for the hypergraph $H$. Note that $N \leq \abs{\Vtx(G)} \leq 
  N\cdot (1+\gamma)$ and degree of $G$ is also bounded by $D$. 
  Since a cycle in $G$ is also a cycle in $H$, the girth of $G$ is at also least $\girth \geq \delta
  \cdot \log n$.

  By~\cref{lem:csp:sparsity-variable-constraint}, we have $G$ is $(\frac{\tau}{k (1+\gamma)},
  \frac{\eta}{1+\eta})$-sparse. 
By~\cref{lem:csp:path-decomposable}, any subgraph of $G$ on at most $\frac{\tau}{k(1+\gamma)}
\cdot (N+m)$ vertices is $\ell$-path decomposable for any $\ell \leq \min\{\girth,
  1/(4\eta)\}$. 
Since $D = 100 \cdot k\gamma \cdot \log(2n_0/\eps)$, there exists $\ell_0 =
\Omega_{k,\gamma,n_0,\eps}(\log n)$ such that $D^{\ell_0+1} \leq n^{1/6}$. We choose $\ell =
\min\inbraces{\girth, 1/(4\eta), \ell_0}$.

Let $\mu_0$ be the smallest $\mu$ such that $\exp\inparen{-\mu\ell/9} \leq \frac{\mu}{2(t+1)}$ (note
that $\frac{1}{\mu} \cdot \exp\inparen{-\mu\ell/9}$ is decreasing in $\mu$). 
Since we must have $\mu \geq 1/\ell$, there exists a $\mu_0$ satisfying
\[
\mu_0 ~\leq~ \frac{9}{\ell} \cdot \inparen{\ln(2(t+1)) + \ln \ell} \mper
\]
From our choice of $\ell$, there exist constants $c = c(k,\gamma,n_0,\eps)$ and $\theta =
\theta(k,\gamma,n_0,\eps) < 1/2$ such that $\mu_0 \leq c \cdot \frac{\log t + \log \log n}{\log n}
< 1$ when $t \leq n^{\theta}$. 
Then, for any $\mu \in [\mu_0,1)$, we have $(1-\mu)^{\ell/9} \leq \exp(-\mu\ell/9) \leq
\frac{\mu}{2(t+1)}$.

We can now apply~\cref{thm:csp:l2-graph-embedding} to construct the embedding. 
Given any subset $S$ of $\Vtx(H)$ of size at most $t\le n^{\theta}$, note that $S$ is also a subset
of $\Vtx(G)$. Moreover, we have $t \leq n^{\theta} \leq (N+m)^{1/2}$. Also, we have $t \cdot
D^{\ell+1} \leq n^{1/2} \cdot n^{1/6} = n^{2/3} \leq \frac{\tau}{k(\gamma+1)} \cdot (N+m)$. Thus,
any subgraph of $G$ on $t \cdot D^{\ell+1}$ vertices is $\ell$-path decomposable. 

Thus, when $\mu \geq \mu_0$, by~\cref{thm:csp:l2-graph-embedding} 
there exists a mapping $\psi_S$ from $S$ to the unit sphere, such that for all $u,v \in S$, we have
\[
\norm{\psi_S(u)-\psi_S(v)}_2~=~\rho^G_{\mu}(u,v)~=~\rho^H_{\mu}(u,v) \mcom
\]
where the last equality uses the fact that for all $u,v \in \Vtx(H)$, $\rho^H_{\mu}(u,v) =
\rho^G_{\mu}(u,v)$ since $d_{G}(u,v) = 2 \cdot d_H(u,v)$.
\end{proofof}

\section{Decompositions of Hypergraphs from Local Geometry}\label{sec:csp:decomposition}
We will construct the Sherali-Adams solution by partitioning the given subset of vertices in to trees,
and then creating a natural distribution over satisfying assignments on trees. We define below the
kind of partitions we need.
\begin{definition}\label{def:csp:partioning-scheme}
Let $X$ be a finite set. For a set $S$, let $\calP_S$ denote a distribution over partitions of $S$. For
$T \subseteq S$, let $\calP_{S|T}$ be the distribution over partitions of $T$ obtained by restricting the
partitions in $\calP_S$ to the set $T$.  We say that a collection of distributions
$\inbraces{\calP_S}_{\abs{S} \leq t}$ forms a {\deffont consistent partitioning scheme of order} $t$, if 
\[
\forall S \subseteq X, \abs{S} \leq t ~\text{and}~ \forall T \subseteq S 
\qquad \calP_T = \calP_{S|T} \mper
\]
\end{definition}
In addition to being consistent as described above, we also require the distributions to have small
probability of cutting the hyperedges for the hypergraphs corresponding to our CSP instances.
We define this property below.
\begin{definition}\label{def:csp:sparse-partitioning}
Let $H = (V,E)$ be a $k$-uniform hypergraph. Let $\inbraces{\calP_S}_{\abs{S} \leq t}$ be a consistent
partitioning scheme of order $t$ for the vertex set $V$, with $t \geq k$. We say the scheme
$\inbraces{\calP_S}_{\abs{S} \leq t}$ is $\eps$-{\deffont sparse} for $H$ if 
\[
\forall e \in E \qquad \Pr{P \sim \calP_e}{P \neq \inbraces{e}} ~\leq~ \eps \mper
\]
\end{definition}
In this section, we will prove that the hypergraphs arising from random CSP instances
admit sparse and consistent partitioning schemes. Recall that for a hypergraph $H$, we define
(\cref{def:csp:metric-dmu})  the
metrics $d^H_{\mu}$ and $\rho^H_{\mu}$ as:
\[
  d^H_\mu(u,v)~\coloneqq~1-{(1-\mu)}^{2\cdot d_H(u,v)}\qquad \text{ and }
  \qquad \rho^H_\mu(u,v)~\coloneqq~\sqrt{\frac{2\cdot d^H_\mu(u,v)+ \mu}{1+\mu}}\mcom
\]
\begin{lemma}\label{lem:csp:partitioning}
Let $H = (V,E)$ be $k$-uniform hypergraph and let $d_{\mu}$ be the metric as defined
above. Let $H$ be such that for all sets $S \subseteq V$ with $\abs{S} \leq t$, the metric induced
on $\rho_{\mu}$ on $S$ is isometrically embeddable into $\ell_2$. Then, there exists 
$\eps \leq 10k \cdot \sqrt{\mu \cdot t}$ and $\Delta_H=O(1/\mu)$ such that  $H$ admits an
$\eps$-sparse consistent partitioning scheme of order $t$, with each partition consisting of clusters of diameter
at most $\Delta_H$ in $H$.
\end{lemma}

We use the following result of Charikar \etal~\cite{CharikarCGGP98} which shows that
low-dimensional metrics have good \emph{separating decompositions} with bounded diameter, \ie~decompositions
which have a small probability of separating points at a small distance. 
\begin{theorem}[\cite{CharikarCGGP98}]\label{thm:csp:l2-decomposition}
Let $W$ be a finite
  collection of points in $\R^d$ and let $\Delta > 0$ be given. 
  Then there exists a distribution $\calP$ over partitions of $W$
  such that
\begin{itemize}
\item[-] $\forall P \in Supp(\calP)$, each cluster in $P$ has $\ell_2$ diameter at most $\Delta$.
\item[-]  For all $x,y \in W$
\[
\Pr{P \sim \calP}{P ~\text{separates}~ x ~\text{and}~ y} ~\leq~ 2\sqrt{d} \cdot \frac{\norm{x-y}_2}{\Delta} \mper
\]
\end{itemize}
\end{theorem}
We also need the observation that the partitions produced by the above theorem are consistent,
assuming the set $S$ considered above lie in a fixed bounded set (using a trivial modification of
the procedure in~\cite{CharikarCGGP98}). For the sequel, we use $B(x,\delta)$ to denote the $\ell_2$
ball around $x$ of radius $\delta$ and $B_H(u,r)$ to denote a ball of radius $r$ 
around a vertex $u \in V(H)$. Thus,
\[
  B(x,\delta)~\defeq~\inbraces{y~\mid~\norm{x-y}_2 \leq \delta}
  \qquad \text{and} \qquad B_H(u,r)~\defeq~\inbraces{v \in V~\mid~d_H(u,v) \leq r} \mper
\]
The balls $B(S,\delta)$ and $B_H(S,r)$ are defined similarly.
\begin{claim}\label{claim:csp:l2-consistency}
Let $S$ and $T$ be sets such that $T \subseteq S$. Let $W_S = \inbraces{w_u}_{u \in S}$ and
$W_T = \inbraces{w_u'}_{u \in T}$ be $\ell_2$-embeddings of $S$ and $T$  
satisfying  $\phi(W_T) \subseteq W_S \subseteq  B(0,R_0) \subset
\R^d$, for some unitary transformation $\phi$ and $R_0 > 0$. 
Let $\calP_S$ and $\calP_T$ be distributions over partitions of $S$ and $T$ respectively, 
induced by partitions on $W_S$ and $W_T$ as given  by~\cref{thm:csp:l2-decomposition}. Then
\[
\calP_{S|T}~=~\calP_T \mper
\]
\end{claim}
\begin{proof}
The claim follows simply by considering (a trivial modification of) the algorithm of~\cite{CharikarCGGP98}.
For a given set $W$ and a parameter $\Delta$, they produce a partition using
the following procedure:
\begin{itemize}
\item Let $W' = W$.
\item Repeat until $W' = \emptyset$
\begin{itemize}
\item Pick a random point $x$ in $B(W,\Delta/2)$ according to the Haar measure. Let $C_x =
  B(x,\Delta/2) \cap W'$. 
\item If $C_x \neq \emptyset$, set $W' = W' \setminus C_x$. Output $C_x$ as a cluster in the partition.
\end{itemize}
\end{itemize}
The work~\cite{CharikarCGGP98} shows that the above procedure produces a distribution over partitions
satisfying the conditions in~\cref{thm:csp:l2-decomposition}. We simply modify the procedure to sample a
random point $x$ in $B(0,R_0 + \Delta/2)$ instead of $B(S,\Delta/2)$. 
This does not affect the separation probability of any two points, since the only non-empty clusters
are still produced by the points in $B(S,\Delta/2)$. 

Let $P$ be a partition of $S$ produced by the above procedure when applied to the point set $W_S$,
and let $P'$ be a random partition produced when applied to the point set $\phi(W_T)$. 
It is easy to see from the above procedure that the distribution $\calP_T$ 
is invariant under a unitary transformation of $W_T$. 
By coupling the random choice of a
point in $B(0,R_0 + \Delta/2)$ chosen at each step in the procedures applied to $W_S$ and $\phi(W_T)
\subseteq W_S$, we get that $P(T) = P'$ \ie the partition $P$ restricted to
$T$ equals $P'$. 
Thus, we get $\calP_{S|T} = \calP_T$.
\end{proof}

We can use the above to prove~\cref{lem:csp:partitioning}.
\begin{proofof}{of~\cref{lem:csp:partitioning}}
Given a set $S$, let $W_S$ be an $\ell_2$ embedding of the metric $\rho_{\mu}$ restricted to
$S$. Since, $|S| \leq t$, we can assume $W_S \in \R^t$. We apply partitioning procedure of
Charikar \etal~from~\cref{thm:csp:l2-decomposition} with $\Delta = 1/2$. From the definition of the metric
$\rho^H_{\mu}$, we get that there exists a $\Delta_H = O(1/\mu)$ such that $\rho^H_{u,v} \leq 
1/2~\implies~d_H(u,v) \leq \Delta_H$. Moreover, for $u,v$ contained in an edge $e$, we have that
$\rho_{\mu}(u,v) \leq \sqrt{5\mu}$ and hence the probability that $u$ and $v$ are separated is
at most $10\sqrt{\mu \cdot t}$. Thus, the probability that any vertex in $e$ is separated from $u$ is at
most $10k \cdot \sqrt{\mu \cdot t}$.

Finally, for any $S \subseteq T$, if $W_S$ and $W_T$ denote the corresponding $\ell_2$ embeddings,
by the rigidity of $\ell_2$ we have that for $\phi(W_T) \subseteq W_S$ for some unitary
transformation $\phi$. Thus, by~\cref{claim:csp:l2-consistency}, we get that this is a consistent
partitioning scheme of order $t$.
\end{proofof}

\section{The Sherali-Adams Integrality Gaps Construction}\label{sec:csp:sagaps}
\subsection{Integrality Gaps from the Basic LP}
Recall that the basic LP relaxation for \maxkcspq($f$) as given in~\cref{fig:csp:basic-lp}.
In this section, we will prove\cref{thm:csp:main}. We recall the statement below.
\maintheoremcsp*
Let $\Phi_0$ be a $(c,s)$ integrality gap instance for the basic LP relaxation for
\maxkcspq($f$) with $n_0$ variables and $m_0$ constraints. We use it to construct a new integrality
gap instance $\Phi$. The construction is similar to the gap instances
constructed by Khot \etal~\cite{KhotTW14} discussed in the next section. However, we describe this
construction first since it's simpler. The procedure for constructing the instance $\Phi$ is
described in~\cref{fig:csp:gap-instance}.
\begin{figure}[htb]
\hrule
\vline
\hspace{10 pt}
\begin{minipage}[t]{0.95\linewidth}
\vspace{10 pt}
{
\underline{\textsf{Given}}: A $(c,s)$ gap instance $\Phi_0$ on $n_0$ variables, for the basic LP.

\smallskip

\underline{\textsf{Output}}: An instance $\Phi$ with $N = n \cdot n_0$ variables and $m$ constraints. 

\medskip

The variables are divided into $n_0$ sets $X_1, \ldots, X_{n_0}$, one for each variable in
$\Phi_0$. We generate $m$ constraints  independently at random as follows:
\begin{enumerate}
\item Sample a random constraint $C_0 \sim \Phi_0$. Let $S_{C_0} = \inbraces{i_1, \ldots, i_k}
  \subseteq [n_0]$  denote the set of variables in this constraint.
\item For each $j \in [k]$, sample a random variable $x_{i_j} \in X_{i_j}$. 
\item Add the constraint $f((x_{i_1}, \ldots, x_{i_k}) + b_{C_0})$ to the instance $\Phi$.
\end{enumerate}
}
\smallskip
\end{minipage}
\hfill\vline
\hrule
\caption{Construction of the gap instance $\Phi$}\label{fig:csp:gap-instance}
\end{figure}

\subsubsection*{Soundness}
We first prove that no assignment satisfies more than $s+\eps$ fraction of constraints for the
above instance.
\begin{lemma}\label{lem:csp:basic-lp-soundness}
For every $\eps > 0$, there exists $\gamma = \gamma(\eps, n_0, q)$ such that for an instance $\Phi$
generated by choosing at least $\gamma \cdot n$ constraints independently at random as above, we
have with probability $1 - \exp\inparen{-\Omega(n)}$, $\opt(\Phi) < s+\eps$.
\end{lemma}
\begin{proof}
Fix an assignment $\sigma \in \qary^{N}$. We will first consider $\Ex{\sat_{\Phi}\inparen{\sigma}}$ for a randomly
generated $\Phi$ as above.
\begin{align*}
\Ex{\Phi}{\sat_{\Phi}\inparen{\sigma}} 
&~=~ \expop_{C_0 \in {\Phi_0}} \expop_{x_{i_1} \in X_{i_1}} \cdots
                            \expop_{x_{i_k} \in X_{i_k}}[f(\sigma(x_{i_1})+b_{i_1}, \ldots,
                            \sigma(x_{i_k})+b_{i_k})] \\
                            &~=~ \expop_{C_0 \in {\Phi_0}} \expop_{Z_1,\ldots Z_{n_0}}[f(Z_{C_0} + b_{C_0})]\mcom
\end{align*}
where for each $i \in [n_0]$, $Z_i$ is an independent random variable with the distribution
\[
\Pr{Z_i = b} ~\defeq~ \Ex{x \in X_i}{\indicator{\sigma(x) = b}} \mcom
\]
and $Z_{C_0}$ denotes the collection of variables in the constraint $C_0$, \ie $Z_{C_0} =
\inbraces{Z_i}_{i \in S_{C_0}}$.
Thus, the random variables $Z_1, \ldots, Z_{n_0}$ define a random assignment to the variables in
$\Phi_0$, which gives, for any $\sigma$
\[
\Ex{\Phi}{\sat_{\Phi}\inparen{\sigma}} ~=~ \expop_{C_0 \in {\Phi_0}} \expop_{Z_1,\ldots
Z_{n_0}}[f(Z_{C_0} + b_{C_0})] ~<~ s \mper
\]
Consider a randomly added constraint $C$ to the instance $\Phi$. We have that 
\[
\Pr{C(\sigma) = 1} ~=~ \Ex{\Phi}{\sat_{\Phi}(\sigma)} ~<~ s \mcom 
\]
for any fixed $\sigma$ over random choice of the constraint $C$. Thus, for an instance $\Phi$ with
$m$ independently and randomly generated constraints, we have
\begin{align*} 
\Pr{\Phi}{\sat_{\Phi}(\sigma) ~\geq~ s + \eps}
&~\leq~
\Pr{\Phi}{\sat_{\Phi}(\sigma) 
~\geq~ \Ex{\Phi}{\sat_{\Phi}(\sigma)} + \eps} \\
&~=~ \Pr{\Phi}{\Ex{C \in \Phi}{\indicator{C(\sigma)=1}} ~\geq~ \Ex{\Phi}{\sat_{\Phi}(\sigma)} +
  \eps} \\
&~\leq~ \exp\inparen{-\Omega(\eps^2 \cdot m)} \mper
\end{align*}
Taking a union bound over all assignments, we get
\[
\Pr{\Phi}{\exists \sigma ~~\sat_{\Phi}(\sigma) ~\geq~ s + \eps}
~\leq~ q^{n \cdot n_0} \cdot \exp\inparen{-\eps^2 \cdot m} \mcom
\]
which is at most $\exp\inparen{-\Omega(n)}$ for $m = O(((\log q)/\eps^2)\cdot n \cdot n_0)$.
\end{proof}

\subsubsection*{Completeness}
To prove the completeness, we first observe that the instance $\Phi$ as constructed above is also a
gap instance for the basic LP. We will then ``boost'' this hardness to many levels of the
Sherali-Adams hierarchy.
\begin{lemma}\label{lem:csp:basic-lp-solution}
For every $\eps > 0$, there exists $\gamma = \gamma(\eps)$ such that for an instance $\Phi$
generated by choosing at least $\gamma \cdot n$ constraints independently at random as above, with probability $1 - \exp\inparen{-\Omega(n)}$ there exist distributions
$\treedist_{S_C}$ over $\qary^{S_C}$ for each $C \in \Phi$, and distributions $\treedist_{i}$ over $\qary$
for each variable $x_i \in [n \cdot n_0]$, satisfying
\begin{itemize}
\item[-] For all $C \in \Phi$ and all $i \in S_C$, $\treedist_{S_C | \{i\}} = \dist_i$.
\item[-] The distributions satisfy  $\ExpOp_{C \in \Phi} \Ex{\alpha \sim \treedist_{S_C}}{f(\alpha +
    b_C)} ~\geq~ c - \frac{\eps}{10}$.
\end{itemize}
\end{lemma}
\begin{proof}
For each $C_0 \in \Phi_0$ and each $j \in [n_0]$, let $\dzero_{S_{C_0}}$ and $\dzero_{j}$ denote the
basic LP solution satisfying
\[
\dzero_{S_{C_0} | j} ~=~ \dzero_{j} ~~\forall C_0 \in \Phi_0 ~\forall j \in S_{C_0}
\qquad
\text{and}
\qquad
\ExpOp_{C_0 \in \Phi_0} \Ex{\alpha \sim \dzero_{S_{C_0}}}{f(\alpha + b_{C_0})} ~\geq~ c \mper
\]
Each constraint $C \in \Phi$ is sampled according to some constraint $C_0 \in \Phi_0$, and we take
$\treedist_{S_C} \defeq \dzero_{S_{C_0}}$ for the corresponding contraint $C_0 \in \Phi_0$. Also,
each variable $x_i$ for $i \in [n_0 \cdot n]$, belongs to one of the sets $X_j$ for $j \in
[n_0]$, and we take $\treedist_{i} \defeq \dzero_{j}$ for the corresponding $j \in [n_0]$.

The consistency of the distributions follows immediately from the construction of the instance
$\Phi$. Let $C \in \Phi$ be any constraint and let $C_0$ be the corresponding constraint in
$\Phi_0$. If $S_{C_0} = (j_1, \ldots, j_k)$, then $S_C = (i_1, \ldots, i_k)$ where each $i_r \in
\{j_r\} \times [n]$ for all $r \in [k]$. Thus, for any $r \in [k]$,
\[
\treedist_{S_C | i_r} ~=~ \dzero_{S_{C_0} | j_r} ~=~ \dzero_{j_r} ~=~ \treedist_{i_r} \mper
\]
To bound the objective value, we again consider its expectation over a randomly generated instance
$\Phi$. Let $C$ be a random constraint added to $\Phi$. Then, if we define $\treedist_{S_C}$ as
above for this constraint, we have
\[
\ExpOp_{C} \Ex{\alpha \in \treedist_{S_C}}{f(\alpha + b_C)} ~=~ \ExpOp_{C_0 \in \Phi_0}\Ex{\alpha
  \sim \dzero}{f(\alpha + b_{C_0})} ~\geq~ c \mper
\]
Thus, the expected contribution of each constraint is at least $c$. The probability that the average
of $m$ constraints deviates by at least $\eps/10$ from the expectation, is at most
$\exp\inparen{- \Omega(\eps^2 \cdot m)}$. 
There exists $\gamma = O(1/\eps^2)$ such that for $m \geq
\gamma \cdot n$, the probability is at most $\exp(- \Omega(n))$.
\end{proof}
To construct local distributions for the Sherali-Adams hierarchy, we will consider (a slight
modification) the hypergraph $H$ corresponding to the instance $\Phi$. We first show that
distributions on hyperedges of this hypergraph can be consistently
propagated in a tree, provided they agree on intersecting vertices. 

For a set $U \subseteq \Vtx(H)$ in a hypergraph
$H$, recall that $\cl(U)$ includes all paths of lengths at most 1 between any two vertices in
$U$. Thus, $E(\cl(U)) = \inbraces{e \in E~\mid~\abs{e \cap U} \geq 2}$.
Note that~\cref{lem:csp:basic-lp-solution} implies that hyperedges forming a tree in $H$ satisfy the
hypothesis of~\cref{lem:csp:tree-propagation} below. 
\begin{lemma}\label{lem:csp:tree-propagation}
Let $H=(V,E)$ be a $k$-uniform hypergraph. Let $U \subseteq V$ and let the set of hyperedges
$E(\cl(U))$ form a tree.  For each $e \in E(\cl(U))$, 
let $\treedist_e$ be a distribution on $[q]^e$ such that for any $u \in U$ and $e_1, e_2 \in E(\cl(U))$ such
that $e_1 \cap e_2 = \{u\}$, we have $\treedist_{e_1|u} = \treedist_{e_2|u} = \treedist_u$. 
Then, 
\begin{itemize}
\item[-] there exists a distribution $\treedist_U$ on $[q]^U$ such that $\treedist_{U | e \cap U} =
  \treedist_{e | e \cap U}$ for all $e \in E(U)$. 
\item[-] If $U' \subseteq U$ is such that the hyperedges in $E(\cl(U'))$ form a subtree of $E(\cl(U))$, then 
$\treedist_{U | U'} = \treedist_{U'}$.
\end{itemize}
\end{lemma}
\begin{proof}
We define the distribution by starting with an arbitrary hyperedge and traversing the tree in an
arbitrary order. Let $e_1, \ldots, e_r$ be a traversal of the hyperedges in $E(\cl(U))$ such that for all $i$,
$\abs{\inparen{\cup_{j<i}e_j} \cap e_i} = 1$. Let $U_0 = \cup_{j<i}e_j$ be the set of vertices for
which we have already sampled an assignment and let $e_i$ be the next hyperedge in the traversal, with $u$
being the unique vertex in $e_i \cap U_0$. We sample an assignment to the vertices in
$e$, conditioned on the value for the vertex $u$. Formally, we extend the distribution $\treedist_{U_0}$
to $U_0 \cup e$ by taking, for any $\alpha \in \qary^{U_0 \cup e}$
\[
\treedist_{U_0 \cup e} (\alpha)
~=~
\treedist_{U_0}(\alpha(U_0)) \cdot
\frac{\treedist_e(\alpha(e))}{\treedist_{e|u}(\alpha(u))} 
~=~
\treedist_{U_0}(\alpha(U_0)) \cdot
\frac{\treedist_e(\alpha(e))}{\treedist_{u}(\alpha(u))} 
\mper
\]
The above process defines a distribution $\treedist_{\cl(U)}$ on $\cl(U)$, with
\[
\treedist_{\cl(U)}(\alpha) ~=~ \frac{\prod_{e \in E(U)} \treedist_e(\alpha(e))}{ \prod_{u \in \cl(U)}
  \inparen{\treedist_{u}(\alpha(u))}^{\deg(u)-1}} \mper
\]
In the above expression, we use $\deg(u)$ to denote the degree of vertex $u$ in tree formed by the
hyperedges in $E(\cl(U))$ \ie $\deg(u) = \abs{\inbraces{e \in E(\cl(U)) ~|~ u \in e}}$. We then define the
distribution $\treedist_U$ as the marginalized distribution $\treedist_{\cl(U) | U}$ \ie
\[
  \treedist_U(\alpha)~=~\sum_{\substack{\beta \in \qary^{\cl(U)} \\ \beta(U) = \alpha}}
\treedist_{\cl(U)}(\beta) \mper
\]
Note that the distribution $\treedist_{\cl(U)}$ and hence also the distribution $\treedist_U$ are independent
of the order in which we traverse the hyperedges in $E(\cl(U))$. 
Also, since the above process samples each
hyperedge according to the distribution $\treedist_e$, we have that for any $e \in E(U)$,
$\treedist_{\cl(U) | e} =
\treedist_e$. Thus, also for any $e \in E(U)$, $\treedist_{U | e \cap U} = \treedist_{e | e \cap U}$.

Let $U' \subseteq U$ be any set such that $E(\cl(U'))$ forms a subtree of $E(\cl(U))$. Then there
exists a traversal $e_1, \ldots, e_r$, and $i \in [r]$ such that $e_j \in E(\cl(U')) ~\forall j \leq i$ and $e_j
\notin E(\cl(U')) ~\forall j > i$. However, the distribution defined by the partial traversal $e_1,
\ldots, e_i$ is precisely $\treedist_{\cl(U')}$. Thus, we get that $\treedist_{\cl(U) | \cl(U')} =
\treedist_{\cl(U')}$ which implies $\treedist_{U | U'} = \treedist_{U'}$.
\end{proof}
We can now prove the completeness for our construction using consistent decompositions.
\begin{lemma}\label{lem:csp:basic-lp-completeness}
Let $\eps > 0$ and let $\Phi$ be a random instance of \maxkcspq($f$) generated by choosing $\gamma
\cdot n$ constraints independently at random as above. Then, there is a $t =
\Omega_{\eps,k,n_0}\inparen{\frac{\log n}{\log \log n}}$, such that
with probability $1 - \eps$ over the choice of $\Phi$, there exist distributions
$\inbraces{\dist_S}_{\abs{S} \leq t}$ satisfying:
\begin{itemize}
\item[-] For all $S \subseteq V$ with $\abs{S} \leq t$, $\dist_S$ is a distribution on $\qary^S$.
\item[-] For all $T \subseteq S \subseteq V$ with $\abs{S} \leq t$, $\dist_{S | T} = \dist_T$.
\item[-] The distributions satisfy
\[
\ExpOp_{C \in \Phi} ~\Ex{\alpha_C \sim \dist_{S_C}}{f(\alpha_C + b_C)} ~\geq~ c - \eps \mper
\]
\end{itemize}
\end{lemma}
\begin{proof}
By~\cref{thm:csp:locally-l2}, we know that there exists $\delta$ such that with probability
$1-\eps/4$, after removing a set of constraints $C_B$ of size at most $(\eps/4)\cdot m$, we can assume
that the remaining instance has girth at least $\girth = \delta \cdot \log n$. Also, there exists
$\theta, c > 0$ such that for all $t \leq n^{\theta}$, the metric $\rho^H_{\mu}$ restricted to any
set $S$ of size at most $t$ embeds isometrically into the unit sphere in $\ell_2$, for all $\mu \geq
c \cdot \frac{\log t + \log \log n}{\log n}$.

We choose $\mu = 2c \cdot \frac{\log \log n}{\log n}$
and $t = \frac{\eps^2}{400 k^2} \cdot \frac{1}{\mu}$ so that
\[
  \mu~\geq~c \cdot \frac{\log t + \log \log n}{\log n} \quad \text{and} \quad \sqrt{\mu \cdot
  t}~\leq~\frac{\eps}{20k} \mper
\]
Thus, by~\cref{lem:csp:partitioning}, $H$ admits an $(\eps/2)$-sparse partitioning scheme of order $t$ with
each cluster in the partition having diameter at most $\Delta_H = O(1/\mu)$. Let
$\inbraces{\calP_S}_{\abs{S} \leq t}$ denote this partitioning scheme.

Given a set $S$, the distribution $\dist_S$ is a convex combination of several distributions $\dist_{S,P}$,
corresponding to different partitions $P$ sampled from $\calP_S$. We describe the distribution
$\dist_S$ by giving the procedure to sample an $\alpha \in \qary^S$. Given the set $S$ with $\abs{S}
\leq t$:
\begin{itemize}
\item[-] Sample a partition $P = (U_1, \ldots, U_r)$ from the distribution $\calP_S$.
\item[-] For each set $U_i$, consider the set $\component{U_i}$ obtained by including the vertices
  contained in all the hyperedges in the shortest path between all $u,v \in U_i$. Note that since $U_i$
  has diameter at most $\Delta_H$ in $H$, $\component{U_i}$ is connected and in fact $\component{U}
  = \cl_{\Delta_H}(U)$. 
Also, since each vertex in an included path is within distance at most $\Delta_H/2$ of an
end-point, and $U_i$ has diameter at most $\Delta_H$, we know that the diameter of $\component{U_i}$
is at most $2 \cdot \Delta_H$. Hence, $\component{U_i}$ is a tree.
Finally, we must have $\cl(\component{U_i})~=~\component{U_i}$ since any additional path of length
1 would create a cycle of length at most $2\cdot \Delta_H + 1$.

Thus, by~\cref{lem:csp:basic-lp-solution} and~\cref{lem:csp:tree-propagation} (with probability at least
$1-\eps/4$) there exists a 
distribution $\treedist_{\component{U_i}}$ for each $U_i$, satisfying $\treedist_{\component{U_i} |
  e} = \treedist_{e}$ for all $e \in E\inparen{\component{U_i}}$. Here, $\treedist_{e}$ are the
distributions given by~\cref{lem:csp:basic-lp-solution}, which form a solution to the basic LP for
$\Phi$, with value at least $c - \eps/4$. For each $U_i$, define the distribution
\[
\dist'_{U_i}~\defeq~\treedist_{\component{U_i} | U_i} \mper
\]
\item[-] Sample $\alpha \in \qary^S$ according to the distribution
\[
\dist_{S,P} ~\defeq~ \dist'_{U_1} \times \cdots \times \dist'_{U_r} \mper
\]
\end{itemize}
Thus, we have
\[
\dist_S ~:=~ \Ex{P = (U_1, \ldots, U_r) \sim \calP_S}{\prod_{i=1}^r \dist'_{U_i}} \mcom
\]
where the distributions $\dist'_{U_i}$ are defined as above. 

We first prove the distributions are consistent on intersections \ie\ $\dist_{S | T} = \dist_T$ for
any $T \subseteq S$. Note that by~\cref{lem:csp:partitioning}, the distributions $\calP_S$ and $\calP_T$
satisfy $\calP_{S|T} = \calP_T$. Each partition $(U_1, \ldots, U_r)$ also produces a partition
$T$. For ease of notation, we assume that the first (say) $r'$ clusters have non-empty intersection
with $S$. Let $V_i = U_i \cap T$ for $1 \leq i \leq r'$ ($V_i = \emptyset$ for $i > r'$). Then, we have
\begin{align*}
  \dist_{S|T}~=~\Ex{P = (U_1, \ldots, U_r) \sim \calP_S}{\prod_{i=1}^r \dist'_{U_i | V_i}} 
  &~=~\Ex{P = (U_1, \ldots, U_r) \sim \calP_S}{\prod_{i=1}^{r'} \treedist_{\component{U_i} | V_i}} \\
  &~=~\Ex{P = (U_1, \ldots, U_r) \sim \calP_S}{\prod_{i=1}^{r'} \treedist_{\component{V_i} | V_i}} \\
  &~=~\Ex{P' = (V_1, \ldots, V_{r'}) \sim \calP_T}{\prod_{i=1}^{r'} \treedist_{\component{V_i} | V_i}} 
\end{align*}
The second to last equality above uses the fact that $\component{V_i}$ is a subtree of
$\component{U_i}$ and thus $\treedist_{\component{U_i} | \component{V_i}} =
\treedist_{\component{V_i}}$ by~\cref{lem:csp:tree-propagation}. The last equality uses the fact that
$\calP_{S|T} = \calP_{T}$ by~\cref{lem:csp:partitioning}.

We now argue that the LP solution corresponding to the above distributions
$\inbraces{\dist_S}_{\abs{S} \leq t}$ has value at least $c - \eps$. Recall that the value of the LP
solution is given by
\[
\ExpOp_{C \in \Phi} ~\Ex{\alpha \sim \dist_{S_C}}{f(\alpha + b_C)} \mper
\]
Consider any constraint $C$ in $\Phi$, with the corresponding set of variables $S_C$ and the
corresponding hyperedge $e$. When defining the distribution $\dist_{S_C}$, we will partition $S_C$
according to the distribution $\calP_{S_C}$. By~\cref{lem:csp:partitioning} and our choice of parameters
\[
\Pr{P \sim \calP_{S_C}}{P \neq \{S_C\}} ~\leq~ 10k \cdot \sqrt{\mu \cdot t} ~\leq~ \frac{\eps}{2} \mper 
\]
For a constraint set which is not in the deleted set $C_B$, 
if the hyperedge $e$ corresponding to the constraint $C$ is not split by a partition $P$
sampled according to
$\calP_{S_C}$, then by~\cref{lem:csp:tree-propagation} $\dist_{S_C, P} = \treedist_{S_C}$. Here,
$\treedist_{S_C}$ is the distribution given by~\cref{lem:csp:basic-lp-solution}. Since $f$ is Boolean,
we have that for $C \notin C_B$,
\[
\Ex{\alpha \sim \dist_{S_C}}{f(\alpha + b_C)} 
~\geq~
\Ex{\alpha \sim \treedist_{S_C}}{f(\alpha + b_C)} - \frac{\eps}{2} \mper
\]
Using~\cref{lem:csp:basic-lp-solution} again, we get
\begin{align*}
\ExpOp_{C \sim \Phi} ~\Ex{\alpha \sim \dist_{S_C}}{f(\alpha + b_C)} 
&~\geq~
\Ex{C \sim \Phi}{\inparen{1-\indicator{C \in C_B}} \cdot \inparen{\Ex{\alpha \sim \treedist_{S_C}}{f(\alpha + b_C)} - \frac{\eps}{2}}} \\
&~\geq~
\ExpOp_{C \sim \Phi} ~\Ex{\alpha \sim \treedist_{S_C}}{f(\alpha + b_C)} - \frac{\eps}{2} -  \Ex{C
  \sim \Phi}{\indicator{C \in C_B}} \\
&~\geq~
c - \frac{\eps}{4} - \frac{\eps}{2} - \frac{\eps}{4} \\
&~\geq~ 
c - \eps \mcom
\end{align*}
where the penultimate inequality uses the fact that the fraction of constraints in the initially deleted set
$C_B$ is at most $\eps/4$ (for large enough $n$).
\end{proof}

\subsection{Integrality Gaps for Resistant Predicates}
Let $f:\xspace \B^k\to \B$ be a boolean predicate and
let $\rho(f)=\frac{f^{-1}(1)}{2^k}$ be the fractions of satisfying assignments to $f$. Then
$f$ is approximation resistant if it is hard to distinguish the \maxcsp instances on $f$ between which are
at least $1-\littleoh(1)$ satisfiable vs which are at most $\rho(f)+\littleoh(1)$ satisfiable. 

In \cite{KhotTW14} the authors introduce the notion of {\sf vanishing measure} 
(on a polytope defined by $f$) and use it to characterize a variant of  approximation resistance,
called strong approximation resistance, assuming the Unique Games conjecture. 
They also show gave a \emph{weaker} notion of vanishing measures, which they used to characterize
strong approximation resistance for LP hierarchies. In particular, they proved that when the
condition in their characterization is satisfied, there exists a 
$(1-\littleoh(1),\rho(f)+\littleoh(1))$ integrality gap for $\bigoh(\log\log n)$ levels of
Sherali-Adams hierarchy for predicates $f$. 
Here, we show that using~\cref{thm:csp:main}, their result can be simplified and strengthened\footnote{The LP
  integrality gap result of Khot \etal~is in fact slightly stronger than stated
  above. They show that LP value is at least $1-o(1)$ while there is no integer solution achieving a
  value outside the range $[\rho(f)-o(1), \rho(f)+o(1)]$. It is easy to see that the same also holds
  for the instance constructed here.} 
to $O\inparen{\frac{\log n}{\log \log n}}$ levels.

Let us first recall some useful notation defined by Khot \etal~\cite{KhotTW14} before we define the
notion of vanishing measure:
\begin{definition}\label{def:csp:vanishing-measure-lp-1}
  For a predicate $f:\xspace \B^k\to \B$, let $\calC(f)$ be the convex polytope
  of \emph{first} moments (biases) of distributions supported on satisfying assignments of
  $f$ \ie\
  \[
    \calC(f)~\defeq~\inbraces{ \zeta \in \R^k~\mid~\forall i \in [k],~\zeta_i = \Ex{\alpha \sim
    \nu}{{(-1)}^{\alpha_i}},~\supp(\nu) \subseteq f^{-1}(1) } \mper
  \]
  For a measure $\Lambda$ on $\calC(f)$, $S \subseteq [k]$, $b \in \B^S$ and permutation $\pi: S \to S$,
  let $\Lambda_{S,\pi,b}$ denote the induced measure on $\R^S$ by considering vectors with coordinates
  $\inbraces{ {(-1)}^{b_{\pi(i)}} \cdot \zeta_{\pi(i)}}_{i \in S}$, where $\zeta \sim \Lambda$.
\end{definition}
We recall below the definition of vanishing measure for LPs from~\cite{KhotTW14} (see
Definition $1.3$):
\begin{definition}\label{def:csp:vanishing-measure-lp-2}
  A measure $\Lambda$ on $\calC(f)$ is called {\deffont vanishing} (for LPs) if for
  every $1\le t\le k$, the following signed measure
  \[
    \ExpOp_{\abs{S}=t}~\ExpOp_{\pi:\xspace S\to S}~\Ex{b\in \B^t}{\inparen{\prod_{i=1}^t {(-1)}^{b_i}}\cdot
    \hat{f}(S)\cdot \Lambda_{S,\pi,b}}
  \]
  is identically $0$. We say $f$ has a vanishing measure if there exists a vanishing measure $\Lambda$ on $\calC(f)$.
\end{definition}
In particular, they prove the following theorem:
\begin{theorem}\label{thm:csp:ktw-SA-gap}
  Let $f : \B^k \to \B$ be a $k$-ary boolean predicate that has a vanishing measure. Then for every $\eps>0$,
  there is a constant $c_\eps > 0$ such that for infinitely may $N \in\N$, there exists an
  instance $\Phi$ of $\maxkcsp(f)$ on $N$ variables satisfying the following:
  \begin{itemize}
    \item $\opt(\Phi)~\leq~\rho(f) + \eps$.
    \item The optimum for the LP relaxation given by $c_\eps\cdot \log\log N$ levels of Sherali-Adams hierarchy
      has $\sdpopt(\Phi) \geq 1-\bigoh(k\cdot \sqrt \eps)$.
  \end{itemize}
\end{theorem}
Combining this with our~\cref{thm:csp:main} already gives us the following stronger result:
\begin{corollary}~\label{cor:csp:ktw-SA-gap-improved}
  Let $f : \B^k \to \B$ be a $k$-ary boolean predicate that has a vanishing measure. Then for every $\eps>0$,
  there is a constant $c_\eps > 0$ such that for infinitely may $N \in\N$, there exists an
  instance $\Phi$ of $\maxkcsp(f)$ on $N$ variables satisfying the following:
  \begin{itemize}
    \item All integral assignment of $\Phi$ satisfies at most $\rho(f) +  \eps$ fraction of
      constraints.
    \item The LP relaxation given by  $c_\eps\cdot \frac{\log N}{\log\log N}$ levels of
      Sherali-Adams hierarchy has $\sdpopt(\Phi) \geq 1- O(k\sqrt{\eps})$.
  \end{itemize}
\end{corollary}

However, note that to apply~\cref{thm:csp:main}, one only needs a gap for the basic LP, which is much
weaker requirement than the $O(\log \log N)$-level gap given by~\cref{thm:csp:ktw-SA-gap}. 
We observe below that the gap for the basic LP follows very simply from the construction by
Khot \etal~\cite{KhotTW14}. One can then directly use this gap for applying~\cref{thm:csp:main} instead
of going through~\cref{thm:csp:ktw-SA-gap}.

Khot \etal~\cite{KhotTW14} use the probabilistic construction given in~\cref{fig:csp:ktw-instance}, 
for a given $\eps > 0$.
The construction actually requires $\Lambda$ to be a vanishing measure over the polytope
$\calC_{\delta}(f) \defeq (1-\delta) \cdot \calC(f)$, for $\delta = \sqrt{\eps}$. However, since
$\calC_{\delta}(f)$ is simply a scaling of $\calC(f)$, a vanishing measure over $\calC(f)$
also gives a vanishing measure over $\calC_{\delta}(f)$. Note that each $\zeta_0 \in \calC(f)$
corresponds to a distribution $\nu_0$ supported in $f^{-1}(1)$. For each $\zeta \in \calC_{\delta}$,
let $\zeta_0 = \frac{1}{1-\delta} \cdot \zeta$ be the point in $\calC(f)$ with distribution
$\nu_0$. Then the distribution $\nu = (1-\delta) \cdot \nu_0 + \delta \cdot U_k$ (where $U_k$
denotes the uniform distribution on $\B^k$) satisfies $\forall i \in [k] \Ex{\alpha \sim
  \nu}{(-1)^{\alpha_i}} = \zeta_i$.
\begin{figure}[!ht]
  \hrule
  \vline
  \hspace{10 pt}
  \begin{minipage}[t]{0.95\linewidth}{
    \vspace{10 pt}
    Let $n_0=\lceil \frac{1}{\eps}\rceil$. Partition the interval
    $[0,1]$ into $n_0+1$ disjoint intervals $I_0,I_1,\ldots ,I_{n_0}$ where $I_0=\{0\}$ and $I_i=\left(
    \nfrac{i-1}{n_0},\nfrac{i}{n_0}\right]$ for $1\le i\le n_0$. For each interval $I_i$, let $X_i$
  be a collection of $n$ variables (disjoint from all $X_j$ for $j \neq i$).

\smallskip
Generate $m$ constraints independently according to the following procedure:
    \begin{itemize}
      \item Sample $\zeta\sim \Lambda$.
      \item For each $j\in [k]$, let $i_j$ be the index of the interval which contains
        $\abs{\zeta(j)}$. Sample uniformly a variable $y_j$ from the set $X_j$.
      \item If $\zeta(j)<0$, then negate $y_j$. If $\zeta(j)=0$, then negate $y_j$ w.p.
        $\frac{1}{2}$.
      \item Introduce the constraint $f$ on the sampled $k$ tuple of literals.
    \end{itemize}
    \smallskip
  }
  \end{minipage}
  \hfill
  \vline
  \hrule
  \caption{Sherali-Adams integrality gap instance for vanishing measure}\label{fig:csp:ktw-instance}
\end{figure}

They show for a sufficiently large constant $\gamma$, an instance $\Phi$ with $m = \gamma \cdot n$
constraints satisfies with high probability, that for all assignments $\sigma$,
$\abs{\sat_{\Phi}(\sigma) - \rho(f)} \leq \eps$ (see Lemma $4.4$ in~\cite{KhotTW14}). The proof is
similar to that of~\cref{lem:csp:basic-lp-soundness}.

Additionally, we need the following claim from~\cite{KhotTW14} (see Claim $4.7$ there), which allows
one to ``round''
coordinates of the vectors $\zeta \in \calC_{\delta}(f)$ to the end-points of the intervals $I_0,
\ldots, I_{n_0}$. This ensures that any two variables in the same collection $X_i$ have the same
bias.  The proof of the claim follows simply from a hybrid argument.
\begin{restatable}{repclaim}{biasclaim}\label{claim:csp:bias-rounding}
Let $\zeta \in \calC_{\delta}(f)$ and let $\nu$ be the corresponding distribution supported in $f^{-1}(1)$
such that for all $i \in [k]$, we have $\zeta_i = \Ex{\alpha \sim \nu}{(-1)^{\alpha_i}}$. Let $t_1, \ldots,
t_k \in [0,1]$ be such that for all $i \in [k]$, $\abs{t_i - \abs{\zeta_i}} \leq \eps$ for
$\eps<\delta/2$. 
Then there exists a distribution $\nu'$ on $\B^k$ such that
\[
  \norm{\nu - \nu'}_1~=~O(k \cdot (\eps/\delta)) \qquad \text{and} \qquad \forall i \in [k],~\Ex{\alpha
  \sim \nu'}{ {(-1)}^{\alpha_i}}~=~\sgn(\zeta_i)\cdot t_{i} \mper
\]
\end{restatable}
\begin{proof}
Let $r_j = \sgn(\zeta_j) \cdot t_{j}$ be the desired bias of the $j^{th}$ coordinate. Then,
$\abs{\zeta(j) - r_j} \leq \eps$ for all $j \in [k]$ 
We construct a sequence of distributions $\nu_0,\ldots,\nu_k$ such that $\nu_0 = \nu$ and
$\nu_k = \nu'$. In $\bnu_j$, the biases are $(r_1,\ldots,r_j,\zeta_{j+1},\ldots,\zeta_k)$.

The biases in $\nu_0$ satisfy the above by definition. 
We obtain $\bnu_{j}$ from $\bnu_{j-1}$ as, 
\[
\nu_{j} = (1-\tau_j) \cdot \nu_{j-1} + \tau_j \cdot D_j \mcom
\]
where $D_j$ is the distribution in which all bits, except for the $j^{th}$ one, are set
independently according to their biases in $\bnu_{j-1}$. For the $j^{th}$ bit, we set it to
$\sgn(r_j-\zeta_j)$ (if $r_j-\zeta(j) = 0$, we can simply proceed with $\bnu_j = \bnu_{j-1}$). 
The biases for all except for the $j^{th}$ bit are unchanged. For the $j^{th}$ bit, the bias now
becomes $r_j$ if
\[
r_j = (1-\tau_j) \cdot \zeta_j + \tau_j \cdot \sgn(r_j-\zeta_j)
~\Longrightarrow~
\tau_j \cdot (\sgn(r_j-\zeta_j) - r_j) = (1-\tau_j) \cdot (r_j - \zeta_j) \mper
\]
Since $\zeta \in \calC_{\delta}(f)$, we know that 
$\abs{\sgn(r_j-\zeta(j)) - r_j} \geq \delta/2$. Also, $\abs{r_j - \zeta(j))} \leq \eps$ by
assumption. Thus, we can choose $\tau_j = O(\eps/\delta)$ which gives that $\norm{\bnu_{j} -
  \bnu_{j-1}}_1 = O(\eps/\delta)$. The final bound then follows by triangle inequality. 
\end{proof}
We can now use the above to give a simplified proof of~\cref{cor:csp:ktw-SA-gap-improved}.
\begin{proofof}{of\cref{cor:csp:ktw-SA-gap-improved}}
  Here we exhibit a solution of the basic LP~\cref{fig:csp:basic-lp} for the instance given
  in~\cref{fig:csp:ktw-instance}.
  For each variable $y_{j}$ coming from the set $X_j$ for $j\in\{0,1,\ldots,n_0\}$, we set
  the bias $t_{j}$ of the variable to be the rightmost point of the interval $I_j$ \ie 
  set $\vartwo{y_j}{-1}= \frac12 \cdot \inparen{1-\frac{i}{n_0}}$ and  $\vartwo{y_j}{1}=\frac12 \cdot
  \inparen{1+\frac{i}{n_0}}$. 
  
  For each constraint $C$ of the form $f(y_{i_1} + b_{1}, \ldots, y_{i_k}+b_{k})$, 
  let $\zeta(C) \in \calC_{\delta}(f)$ be the point used to generate it,
  and let $\nu(C)$ denote the corresponding distribution on $\B^k$. By~\cref{claim:csp:bias-rounding},
  there exists a distribution $\nu'(C)$ such that $\norm{\nu(C) - \nu'(C)}_1 = O(k\eps/\delta)$ and
  such that the biases of the \emph{literals} satisfy $\Ex{\alpha \sim \nu'(C)}{(-1)^{\alpha_j}} =
  \sgn(\zeta_j) \cdot t_{i_j}$, where $t_{i_j}$ denotes the bias for the interval to which $y_{i_j}$
  belongs. When $t_{i_j} \neq 0$, we negate a variable only when $\sgn(\zeta_j) < 0$. Thus, we have
  $\Ex{\alpha \sim \nu'(C)}{(-1)^{\alpha_j + b_j}} = t_{i_j}$, which is consistent with the bias given by the
  singleton variables $\vartwo{y_{i_j}}{1}$ and $\vartwo{y_{i_j}}{-1}$. We thus define the local
  distribution on the set $S_C$ as $\dist_{S_C}(\alpha) = (\nu'(C))(\alpha+b_C)$.

  For all $C \in \Phi$, since $\zeta(C) \in \calC_{\delta}(f)$, we have that $\Ex{\alpha \sim
    \nu(C)}{f(\alpha)} \geq 1 - \delta$. Also, since $\norm{\nu(C)-\nu'(C)}_1 =
  O(k\eps/\delta)$, we get that $\Ex{\alpha \sim \nu'(C)}{f(\alpha)} \geq 1 - \delta -
  O(k\eps/\delta)$. Thus, we have for all $C \in \Phi$, $\Ex{\alpha \sim \dist_{S_C}}{f(\alpha+b_C)} \geq 1
  - \delta - O(k\eps/\delta)$. Taking $\delta = \sqrt{\eps}$ proves the claim.
\end{proofof}

\subsection{Lower-bounds for LP Extended Formulations}
A connection between LP integrality gaps for the Sheral-Adams hierarchy, and lower bounds on the
size of LP extended formulations, was first established by Chan \etal~\cite{ChanLRS13} and later
improved by Kothari \etal~\cite{KothariMR16}. In~\cite{KothariMR16}, the authors proved the
following:
\begin{theorem}[\cite{KothariMR16}, Theorem$1.2$]\label{thm:csp:KothariMR16}
  There exist constants $0 < h < H$ such that the following holds. Consider a function
  $f:\xspace\mathbb{N}\to\mathbb{N}$. Suppose that the $f(n)$-level Sherali-Adams relaxation
  for a CSP cannot achieve a $(c,s)$-approximation on instances on $n$ variables. Then, no LP
  extended formulation (of the original LP) of size at most $n^{h\cdot f(n)}$ can achieve a
  $(c,s)$-approximation for the CSP on $n^H$ variables.
\end{theorem}

Combining~\cref{thm:csp:main} with~\cref{thm:csp:KothariMR16} yields (with $f(N)=c_\eps\cdot\frac{\log N}{\log \log N}$):
\sizecorollary*

%% file: poly_opt/poly_opt_concat.tex
\newcommand{\iprod}{\innerprod}
\def\PExc{\ProbabilityRender{\widetilde{\mathbb{E}}_C}}
\def\PEx{\ProbabilityRender{\widetilde{\mathbb{E}}}}
\newcommand{\xtensor}[1]{x^{\otimes #1}}
\newcommand{\Cp}{\mathds{C}}
\newcommand{\fmaxd}[1]{#1_{\max}}
\chapter{Optimizing Polynomials Over Sphere}
\input{poly_opt/1.intro}

\input{poly_opt/2.overview}

\input{poly_opt/3.prelims}

\input{poly_opt/4.basic_results}

\input{poly_opt/5.folding.tex}

\input{poly_opt/6.oraclelowerbound}

\input{poly_opt/7.nnc-lowerbound}

%%%
%%
\input{poly_opt/8.fsp-lower-bounds}

%% file: poly_opt/1.intro.tex
\section{Introduction and Context}
In this chapter, we study the problem of optimizing homogeneous polynomials over the unit sphere. Formally, given an
$n$-variate degree-$d$ homogeneous polynomial $f$, the goal is to compute
\begin{equation}\label{eq:poly:problem}
  \ftwo{f}~\defeq~\sup_{\|x\| = 1} \abs{f(x)} 
\end{equation}
When $f$ is a homogeneous polynomial of degree 2, this problem is equivalent computing the spectral norm
of an associated symmetric matrix $M_f$. For higher degree $d$, it defines a natural higher-order
analogue of the eigenvalue problem for matrices.  The problem also provides an important testing ground
for the development of new spectral and semidefinite programming (SDP) techniques, and techniques
developed in the context of this problem have had applications to various other constrained
settings~\cite{HLZ10, Laurent09, Lasserre09}.

Besides being a natural and fundamental problem in its own right, it has connections to widely studied
questions in many other areas. In quantum information theory~\cite{BH13, BKS14}, the problem of computing
the optimal success probability of a protocol for Quantum Merlin-Arthur games can be thought of as
optimizing certain classes of polynomials over the unit sphere.  The problem of estimating the $2
\rightarrow 4$ norm of an operator, which is equivalent to optimizing certain homogeneous degree-4
polynomials over the sphere, is known to be closely related to the Small Set Expansion Hypothesis (SSEH)
and the Unique Games Conjecture (UGC)~\cite{BBHKSZ12, BKS14}.  The polynomial optimization problem is
also very relevant for natural extensions of spectral problems, such as low-rank decomposition and PCA,
to the case of tensors~\cite{BKS15,GM15, MR14, HSS15}.  Frieze and Kannan~\cite{FK08} (see
also~\cite{BV09}) also established a connection between the problem of approximating the spectral norm of
a tensor (or equivalently, computing $\ftwo{f}$ for a polynomial $f$), and finding planted cliques in
random graphs.

% , the Small Set Expansion Hypothesis (SSEH) and the Unique Games Conjecture (UGC)~(via $2
% \rightarrow 4$ norm, see \cite{BBHKSZ12, BKS14}), tensor decomposition~\cite{BKS15,GM15}, tensor
% PCA~\cite{MR14, HSS15}, and planted clique~(via the parity tensor, see \cite{FK08, BV09}). 

%%%%% Known Results

The problem of polynomial optimization has been studied\footnote{In certain cases, the problem studied is
not to maximize $\abs{f}$, but just $f(x)$. While the two problems are equivalent for homogeneous
polynomials of odd degree, some subtle issues arise when considering polynomials of even degree. We
compare the two notions in the full version~\cite{BGGLT17}.} over various  compact sets~\cite{Lasserre09,
deKlerk08}, and is natural to ask how well polynomial time algorithms can {\em approximate} the optimum
value over a given compact set (see~\cite{deKlerk08} for a survey).
While the maximum of a degree-$d$ polynomial over the simplex admits a PTAS for every fixed $d$~\cite{deKLP06},
the problem of optimizing even a degree $3$ polynomial over the hypercube does not
admit any approximation better than $2^{ {(\log n)}^{1-\eps}}$ (for arbitrary $\eps > 0$) 
assuming NP cannot be solved in time $2^{ {(\log n)}^{O(1)}}$~\cite{HV04}. 

The approximability of polynomial optimization on the sphere is
poorly understood in comparison. It is known that the  maximum of a degree-$d$ polynomial 
can be approximated within a factor of $n^{d/2-1}$ in polynomial time~\cite{HLZ10, So11}. 
On the hardness side, Nesterov~\cite{Nesterov03} gave a reduction from Maximum Independent Set to
optimizing a homogeneous cubic polynomial over $\SSS^{n-1}$. Formally, given a graph $G$, there
exists a homogeneous cubic polynomial $f(G)$ such that 
$\sqrt{1 - \frac{1}{\alpha(G)}} = \max_{\|x\|=1} f(x)$. Combined with the hardness of Maximum
Independent Set~\cite{Hastad96}, this rules out an FPTAS for optimization over the unit sphere. 
Assuming the Exponential Time Hypothesis, Barak \etal~\cite{BBHKSZ12} proved that computing $2
\rightarrow 4$ norm of a matrix, a special case when $f$ is a degree-$4$ homogeneous polynomial, is
hard to approximate within a factor $\exp(\log^{1/2-\epsilon}(n))$ for any $\epsilon > 0$.

Optimization over $\SSS^{n-1}$ has been given much attention in the optimization community, where
for a fixed number of variables $n$ and degree $d$ of the polynomial, it is known that the estimates
produced by $q$ levels a certain hierarchy of SDPs (Sum of Squares) get arbitrarily close to the
true optimal  solution as $q$ increases (see~\cite{Lasserre09} for various applications). We refer
the reader to the recent work of Doherty and Wehner~\cite{DW12} and de Klerk, Laurent, and
Sun~\cite{dKMS14} and references therein for more information on convergence results.
These algorithms run in time $n^{O(q)}$, which is polynomial for constant $q$.
Unfortunately, known convergence results often give a non-trivial
bound only when the $q$ is linear in $n$. 

In computer science, much attention has been given to the sub-exponential runtime regime (i.e.  $q \ll
n$) since many of the target applications such as SSE, QMA and refuting random CSPs are of considerable
interest in this regime.  In addition to the polytime $n^{d/2-1}$-approximation for general
polynomials~\cite{HLZ10, So11}, approximation guarantees have been proved for several special cases
including $2 \rightarrow q$ norms~\cite{BBHKSZ12}, polynomials with non-negative
coefficients~\cite{BKS14}, some polynomials that arise in quantum information theory~\cite{BKS17,BH13},
and random polynomials~\cite{RRS16, BGL16}. Hence, there is considerable interest in tightly
characterizing the approximation guarantee achievable using sub-exponential time.

In this chapter, we develop general techniques to design and analyze algorithms for polynomial optimization
over the sphere. The sphere constraint is one of the simplest constraints for polynomial optimization and
thus is a good test-bed for techniques. Indeed, we believe these techniques will also be useful in
understanding polynomial optimization for other constrained settings.

In addition to giving an analysis the problem for arbitrary polynomials, these techniques can also be
adapted to take advantage of the structure of the input polynomial, yielding better approximations for
several special cases such as polynomials with non-negative coefficients, and sparse polynomials.
Previous polynomial time algorithms for polynomial optimization work by reducing the problem to diameter
estimation in convex bodies~\cite{So11} and seem unable to utilize structural information about the
(class of) input polynomials. Development of a method which can use such information was stated as an
open problem by Khot and Naor~\cite{KN08} (in the context of $\ell_{\infty}$ optimization). 
%
%\mtnote{Added note about the open problem in Khot and Naor. We mention later that the diameter
%estimation methods do not seem to be able to take advantage of structure, but I also wanted to
%mention this near our contributions and add the fact that this was explicitly stated as an open
%problem in their paper.}
%

Our approximation guarantees are with respect to the optimum of the well-studied Lasserre/sum-of-squares
(SoS) semidefinite programming relaxation. Such SDPs are the most natural tool to bound the optima of
polynomial optimization problems, and our results shed light on the efficacy of higher levels of the SoS
hierarchy to deliver better approximations to the optimum. We discuss the SoS connection in~\cref{sec:poly:sos},
but first turn to stating our approximation guarantees.

\subsection{Our Algorithmic Results}
For a homogeneous polynomial $h$ of even degree $q$, a matrix $M_h \in \Re^{ {[n]}^{q/2}\times  {[n]}^{q/2}}$
is called  a matrix representation of $h$ if ${(x^{\otimes q/2})}^T \cdot M_h \cdot
x^{\otimes q/2} = h(x)~\forall x \in \R^n$. Next we define the quantity, 
\begin{equation}\label{eq:poly:def-of-Lambda}
  \hssos{h}~\defeq~\inf \inbraces{\sup_{\|z\|_2 = 1} z^T M_h~z~\given~M \text{ is a representation of } h} \mper 
\end{equation}
%which can be computed in time $n^{O(t)}$ (since it can be written as an SDP). 
%
Let $\fmax{h}$ denote $\sup_{\|x\|=1} h(x)$. Clearly, $\fmax{h}\leq
\hssos{h}$, i.e. $\hssos{h}$ is a relaxation of $\fmax{h}$. However,
this does not imply that $\hssos{h}$ is a relaxation of $\ftwo{h}$,
since it can be the case that $\fmax{h}\neq \ftwo{h}$. To remedy this,
one can instead consider $\sqrt{\hssos{h^2}}$ which is a relaxation of
$\ftwo{h}$, since $\fmax{(h^2)}=\ftwo{h^2}$.  More generally, for a degree-$d$ homogeneous
polynomial $f$ and an integer $q$ divisible by $2d$, we have the upper estimate
\[
     \ftwo{f}~\leq \hssos{f^{q/d}}^{d/q} 
\]

The following result shows that $\hssos{f^{q/d}}^{d/q}$ approximates $\ftwo{f}$ within polynomial
factors, and also  gives an algorithm to approximate $\ftwo{f}$ with respect to the upper bound
$\hssos{f^{q/d}}^{d/q}$. In the statements below and the rest of this section, $O_d(\cdot)$ and
$\Omega_d(\cdot)$ notations hide $2^{O(d)}$ factors. Our algorithmic results are as follows:

\begin{theorem}\label{thm:poly:results-list}
Let $f$ be an $n$-variate homogeneous polynomial of degree-$d$, and let $q\leq n$ be an 
integer divisible by $2d$. Then,
\begin{align*}
&\text{Arbitrary $f$:} \\
&\pth{\hssos{f^{q/d}}}^{d/q}~\leq~O_d\inparen{\inparen{n/q}^{d/2 - 1}} \cdot \ftwo{f} \\ &\\
&\text{$f$ with Non-neg. Coefficients:} \\
&\pth{\hscsos{C}{f^{q/d}}}^{d/q}~\leq~O_d\inparen{\inparen{n/q}^{d/4 - 1/2}} \cdot \ftwo{f} \\ &\\
&\text{$f$ with Sparsity $m$:} \\
&\pth{\hssos{f^{q/d}}}^{d/q}~\leq~O_d\inparen{\sqrt{m/q}} \cdot \ftwo{f}. 
\end{align*} 
(where $\hscsos{C}{\cdot}$ is a related efficiently computable quantity)

Furthermore, there is a deterministic algorithm that runs in $n^{O(q)}$ time and returns $x$ such that 
\[|f(x)|\geq \frac{\hssos{f^{q/d}}^{d/q}}{O_d(c(n,d,q))}\] 
where $c(n,d,q)$ is ${(n/q)}^{d/2 - 1}$, ${(n/q)}^{d/4 - 1/2}$ 
and $\sqrt{m/q}$ respectively, for each of the above cases (the inequality uses $\hscsos{C}{\cdot}$ 
in the case of polynomials with non-negative coefficients). 
\end{theorem}
\begin{remark}
  Interestingly, our deterministic algorithms only involve computing the maximum eigenvectors of
  $n^{O(q)}$ different matrices in $\Re^{n\times n}$. We actually don't require computing
  $\hssos{f^{q/d}}^{d/q}$, even though this quantity can also be computed in $n^{O(q)}$ time by the
  sum-of-squares SDP (see~\cref{sec:poly:sos}). The quantity $\hssos{f^{q/d}}^{d/q}$ is only used in the
  analysis.
\end{remark}
\begin{remark}
  If~$m=n^{\,\rho \cdot d}$~for~$\rho < 1/3$, then for all $q \leq n^{1-\rho}$, the
  $\sqrt{m/q}$-approximation for sparse polynomials is better than the ${(n/q)}^{d/2 - 1}$ arbitrary
  polynomial approximation. 
\end{remark}
\begin{remark}
  In cases where $\ftwo{f}=\fmax{f}$ (such as when $d$ is odd or $f$ has non-negative coefficients), the
  above result holds whenever $q$ is even and divisible by $d$, instead of $2d$.
\end{remark}

A key technical ingredient en route establishing the above results is a
method to reduce the problem for arbitrary polynomials to a list of
\emph{multilinear} polynomial problems (over the same variable set). We
believe this to be of independent interest, and describe its context
and abstract its consequence (\cref{thm:poly:intro:gen:multi}) next.

Let $M_g$ be a matrix representation of a degree-$q$ homogeneous
polynomial $g$, and let $K=(I,J)\in [n]^{q/2}\times [n]^{q/2}$ have
all distinct elements. Observe that there are $q!$ distinct entries of
$M_g$ including $K$ across which, one can arbitrarily assign values
and maintain the property of representing $g$, as long as the sum
across all $q!$ entries remains the same (specifically, this is the
set of all permutations of $K$). In general for $K'=(I',J')\in
[n]^{q/2}\times [n]^{q/2}$, we define the orbit of $K'$ denoted by
$\orbit{K'}$, as the set of permutations of $K'$, i.e. the number of
entries to which 'mass' from $M_g[I',J']$ can be moved while still
representing $g$.

As $q$ increases, the orbit sizes of the entries increase, and to show
better bounds on $\hssos{f^{q/d}}$, one must exploit these additional
"degrees of freedom" in representations of $f^{q/d}$. However, a big
obstacle is that the orbit sizes of different entries can range
anywhere from $1$ to $q!$,~ two extremal examples being $((1,\dots
1),(1,\dots 1))$ and $((1,\dots q/2),(q/2+1,\dots q))$. This makes it
hard to exploit the additional freedom afforded by growing
$q$. Observe that if $g$ were multilinear, all matrix entries
corresponding to non-zero coefficients have a span of $q!$ and indeed
it turns out to be easier to analyze the approximation factor in the
multilinear case as a function of $q$ since the representations of $g$
can be highly symmetrized. However, we are still faced with the
problem of $f^{q/d}$ being highly non-multilinear.  The natural
symmetrization strategies that work well for multilinear polynomials
fail on general polynomials, which motivates the following
result:
\begin{theorem}[Informal]\label{thm:poly:intro:gen:multi}
    For even $q$, let $g(x)$ be a degree-$q$ homogeneous polynomial. Then there exist multilinear
    polynomials $g_1(x), \dots , g_m(x)$ of degree at most $q$, such that 
    \[
        \frac{\hssos{g}}{\ftwo{g}}~\leq~2^{O(q)}\cdot 
        \max_{i\in [m]} \frac{\hssos{g_i}}{\ftwo{g_i}} 
    \]
    and $m=q^{O(q)}$. 
\end{theorem}

By combining \cref{thm:poly:intro:gen:multi} (or an appropriate
generalization) with the appropriate analysis of the multilinear
polynomials induced by $f^{q/d}$, we obtain the aforementioned results
for various classes of polynomials. 

\paragraph{Weak decoupling lemmas.} 
A common approach for reducing to the multilinear case is through more general 
``decoupling'' or ``polarization'' lemmas, which also have a variety of
applications in functional analysis and probability~\cite{de2012decoupling}. However, such methods
increase the number of variables to $nq$, which would completely nullify any advantage obtained from
the increased degrees of freedom. This is because the approximation 
obtained would be of the form $(\#vars/q)^{d/2-1} = n^{d/2-1}$.

Our proof of \cref{thm:poly:intro:gen:multi} (and its generalizations)  requires only a decoupling with
somewhat weaker properties than given by the above lemmas. However, we need it to be very efficient
in the number of variables.
In analogy with ``weak regularity lemmas'' in combinatorics, which trade structural control for
complexity of the approximating object, 
we call these results  ``weak decoupling lemmas'' 
(see \cref{sec:poly:decoupling}). 
They provide a milder form of decoupling but only increase the number of 
variables to $2n$ (independently of $q$).

We believe these could be more generally applicable; in
particular to other constrained settings of polynomial optimization as
well as in the design of sub-exponential algorithms. Our techniques might also be able to yield a
full tradeoff between the number of variables and quality of decoupling.

%\subsubsection*{Sum of Squares Hierarchy}
\subsection{Connection to sum-of-squares hierarchy}\label{sec:poly:sos}
The \emph{Sum of Squares Hierarchy} (SoS) is one of the canonical and well-studied approaches to
attack polynomial optimization problems. Algorithms based on this framework are parametrized by
the degree or level $q$ of the SoS relaxation.
 For the case of optimization of a homogeneous polynomial $h$ of even degree $q$ (with some matrix
 representation $M_h$) over the unit sphere, the level $q$ SoS relaxes the non-convex program of
 maximizing $(x^{\otimes q/2})^T \cdot M_h \cdot x^{\otimes q/2} = h(x)$ over $x \in \R^n$ with
 $\|x\|_2=1$, to the semidefinite program of maximizing $\Tr{M^T_h X}$ over all positive
 semidefinite matrices $X \in \R^{[n]^{q/2} \times [n]^{q/2}}$ with $\Tr{X}=1$. (This is a
 relaxation because $X = x^{\otimes q/2}  (x^{\otimes q/2})^T$ is psd with $\Tr{X} = \|x\|_2^q$.)

It is well known (see for instance~\cite{Laurent09}) that the quantity
$\hssos{h}$ from~\cref{eq:poly:def-of-Lambda} is the dual value of this
SoS relaxation. Further, strong duality holds for the case of
optimization on the sphere and therefore $\hssos{h}$ equals the
optimum of the SoS SDP and can be computed in time $n^{O(q)}$. (See
the full version~\cite{BGGLT17} for more detailed SoS preliminaries.) 
In light of this, our results from \cref{thm:poly:results-list} can also 
be  viewed as a convergence analysis of the SoS hierarchy for 
optimization over the sphere, as a function of the number of levels $q$. 
Such results are of significant interest in the optimization community, and have been studied for
example in~\cite{DW12, dKMS14} (see \cref{sec:poly:related} for a comparison of results). 

\iffalse
One of the popular approaches to attack polynomial optimization
called the {\em Sum of Squares Hierarchy} (SoS), proceeds by replacing
a system of non-negativity constraints by a suitable sum of squares
decomposition. Algorithms based on this framework are parametrized by
the degree $q$ of their SoS decomposition.

For a homogeneous polynomial $h$ of even degree $t$, it is well known
(see~\cite{Laurent09}) that $\hssos{h}$ is the dual value of the
following degree-$d$ SoS relaxation for the problem of maximizing $h$
(see \secref{sos:prelims} for more detailed SoS preliminaries):
\[\max \,\PEx{h} ~~s.t.~\PEx{\|x\|^t}=1.\] 
It is also known that strong duality holds for the case of
optimization on the sphere and thus the two values are equal, implying
$\hssos{f}$ can be computed in time $n^{O(t)}$.

\fi

\medskip \noindent \textbf{SoS Lower Bounds.} While the approximation
factors in our upper bounds of \cref{thm:poly:results-list} are modest, there
is evidence to suggest that this is inherent.

%\paragraph{Arbitrary Polynomials.}

When $h$ is a degree-$q$ polynomial with \emph{random} i.i.d $\pm 1$
coefficients, it was shown in \cite{BGL16} that there is a 
constant $c$ such that w.h.p.  
$\Bigl( \frac{n}{q^{c+o(1)}}
\Bigr)^{q/4} \le \hssos{h} \le \Bigl( \frac{n}{q^{c-o(1)}}
\Bigr)^{q/4}$.  On the other hand, $\ftwo{h} \le O(\sqrt{nq \log q})$
w.h.p. Thus, the ratio between $\hssos{h}$ and $\ftwo{h}$ can be as
large as $\Omega_q(n^{q/4-1/2})$.

Hopkins \etal~\cite{HKPRSS16} recently proved that
degree-$d$ polynomials with random coefficients achieve a degree-$q$
SoS gap of roughly ${(n/q^{O(1)})}^{d/4 - 1/2}$ (provided
$q>n^{\epsilon}$ for some constant $\epsilon>0$). This is also a lower
bound on the ratio between $\hssos{f^{q/d}}^{d/q}$ and $\ftwo{f}$ for
the case of \emph{arbitrary} polynomials.  Note that this lower bound
is roughly square root of our upper bound from~\cref{thm:poly:results-list}.
Curiously, our upper bound for the case of polynomials with
non-negative coefficients essentially matches this lower bound for
random polynomials.

\medskip \noindent \textbf{Non-Negative Coefficient Polynomials.}
In this chapter, we give a new lower bound construction for the case of
non-negative polynomials,
To the best of our knowledge, the only previous lower
bound for this problem, was known through Nesterov's reduction~\cite{deKlerk08},
which only rules out a PTAS. We give the following polynomially large lower bound. 
The gap applies for random polynomials associated with a novel distribution 
of $4$-uniform hypergraphs, and is analyzed using subgraph counts in a
random graph.
\begin{theorem}\label{thm:poly:nnc-lowerbound}
  There exists an $n$ variate degree-4 homogeneous polynomial $f$ with non-negative coefficients such that
  \[
    \ftwo{f}~\le~{(\log n)}^{O(1)}\qquad \text{and} \qquad
    \hssos{f}~\ge~\tilde{\Omega}(n^{1/6}) \mper
  \]
\end{theorem}
For larger degree $t$, we prove an $n^{\Omega(t)}$ gap between $\ftwo{h}$ and a 
quantity $\fsp{h}$ that is closely related to $\hssos{h}$. Specifically, $\fsp{h}$ is defined by 
replacing the largest eigenvalue of matrix representations $M_h$ of $h$ in~\cref{eq:poly:def-of-Lambda}
by the \emph{spectral norm} $\norm{2}{M_h}$. (See~\cite{BGGLT17} for a formal definition.) Note
that $\fsp{h} \ge \max \{ \hssos{h}, \hssos{-h} \}$.  Like
$\hssos{\cdot}$,
$\fsp{\cdot}$ suggests a natural hierarchy of relaxations for the problem of approximating $\ftwo{h}$, 
obtained by computing $\|h^{q/t}\|_{sp}^{t/q}$ as the $q$-th level of the hierarchy. 

We prove a lower bound of $n^{q/24}/\inparen{q\cdot \log n}^{O(q)}$ on $\|f^{q/4}\|_{sp}$ where 
$f$ is as in~\cref{thm:poly:nnc-lowerbound}.  This not only gives $\fsp{\cdot}$ gaps for the degree-$q$
optimization problem on polynomials with non-negative coefficients, but also an 
$n^{1/6}/{(q \log n)}^{O(1)}$ gap on higher levels of the aforementioned $\fsp{\cdot}$ hierarchy for
optimizing degree-$4$ polynomials with non-negative coefficients. Formally we show: 
\begin{theorem}\label{thm:poly:intro:fsplb}
    Let $g:=f^{q/4}$ where $f$ is the degree-$4$ polynomial as in~\cref{thm:poly:nnc-lowerbound}.  Then
    \[
      \frac{\fsp{g}}{\ftwo{g}}~\geq~\frac{n^{q/24}}{ {(q \log n)}^{O(q)}}\mper
    \]
\end{theorem}
Our lower bound on $\|f^{q/4}\|_{sp}$  is based on a general tool that allows one to ``lift''
level-$4$ $\fsp{\cdot}$ gaps, that meet one additional condition, to higher levels.  While we derive
final results only for the weaker relaxation $\fsp{\cdot}$, 
the underlying structural result  can be used to lift SoS lower 
bounds (i.e. gaps for $\hssos{\cdot}$) as well, provided the SoS solution matrix $X$ satisfies
PSD-ness of two other matrices of appropriately related shapes to $X$ (see full version~\cite{BGGLT17})
--- this inspired us to name our tool ``Tetris theorem.''
Recently, the insightful pseudo-calibration approach \cite{BHKKMP16} has provided a recipe to give
higher level SoS lower bounds for certain \emph{average-case} problems.  We believe our lifting
result might similarly be useful in the context of \emph{worst-case} problems, where in order to get
higher degree lower bounds, it suffices to give  lower bounds for constant degree SoS with some
additional structural properties.

\subsection{Related Previous and Recent Works}\label{sec:poly:related}
Polynomial optimization is a vast area with several previous results. Below, we collect the results most
relevant for comparison with the ones in this chapter, grouped by the class of polynomials. Please see the
excellent monographs~\cite{Laurent09,Lasserre09} for more.

\medskip \noindent \textbf{Arbitrary Polynomials.}
For general homogeneous polynomials of degree-$d$, a polytime 
$O_d\inparen{n^{d/2-1}}$-approximation was given by He \etal~\cite{HLZ10}, which was improved to
$O_d\inparen{ {(n/\log n)}^{d/2-1}}$ by So~\cite{So11}.
The convergence of SDP hierarchies for polynomial optimization was
analyzed by Doherty and Wehner~\cite{DW12}. However, their result only
applies to relaxations given by $\Omega(n)$ levels of the SoS
hierarchy (Theorem 7.1 in~\cite{DW12}).
Thus, our results can be seen as giving an interpolation between the
polynomial time algorithms obtained by~\cite{HLZ10,So11} and the
exponential time algorithms given by $\Omega(n)$ levels of SoS,
although the bounds obtained by~\cite{DW12} are tighter (by a factor
of $2^{O(d)}$) for $q = \Omega(n)$ levels.

For the case of arbitrary polynomials, we believe a tradeoff between
running time and approximation quality similar to ours can also be
obtained by considering the tradeoffs for the results of Brieden \etal~\cite{BGKKLS01} used
by So~\cite{So11}. However, to the best of our
knowledge, this is not published.  In particular, So uses the
techniques of Khot and Naor~\cite{KN08} to reduce degree-$d$
polynomial optimization to $d-2$ instances of the problem of
optimizing the $\ell_2$ diameter of a convex body.  This is solved by~\cite{BGKKLS01},
who give an $O({(n/k)}^{1/2})$ approximation in time
$2^k \cdot n^{O(1)}$. We believe this can be combined with proof of
So, to yield a $O_d\inparen{ {(n/q)}^{d/2-1}}$ approximation in time $2^{q}$.
We note here that the method of Khot and Naor~\cite{KN08} cannot be
improved further (up to polylog) for the case $d = 3$ (see~\cref{sec:poly:query-lower-bound}).
Our results for the case of arbitrary polynomials show that similar
bounds can also be obtained by a very generic algorithm given by the
SoS hierarchy. Moreover, the general techniques developed here are
versatile and demonstrably applicable to various other cases (like
polynomials with non-negative coefficients, sparse polynomials,
worst-case sparse PCA) where no alternate proofs are available. The
techniques of~\cite{KN08,So11} are oblivious to the structure in the
polynomials, and it appears to be unlikely that similar results can be
obtained by using diameter estimation techniques.

\medskip \noindent \textbf{Polynomials with Non-negative
  Coefficients.}  The case of polynomials with non-negative
coefficients was considered by Barak, Kelner, and Steurer~\cite{BKS14}
who  proved that the relaxation obtained by $\Omega(d^3 \cdot \log n /
\eps^2)$ levels of the SoS hierarchy provides an $\eps \cdot
\|f\|_{BKS}$ additive approximation to the quantity $\ftwo{f}$. Here, the parameter we denote by
$\norm{BKS}{f}$ corresponds to a relaxation for $\ftwo{f}$ that is weaker than the one given by
$\fsp{f}$.\footnote{Specifically, $\norm{BKS}{f}$ minimizes the spectral norm over a smaller set of
  matrix representations of $f$ than $\fsp{f}$ which allows all matrix representations.}
Their results can be phrased as showing that a relaxation obtained by
$q$ levels of the SoS hierarchy gives an approximation ratio of
\[
    1 + \pth{\frac{d^3 \cdot \log n}{q}}^{1/2} \cdot \frac{\|f\|_{BKS}}{\ftwo{f}} \mper
\]
Motivated by connections to quantum information theory, they were
interested in the special case where $\|f\|_{BKS}/\ftwo{f}$ is bounded
by a constant. However, this result does not imply strong
multiplicative approximations outside this special case since, in
general, $\|f\|_{BKS}$ and $\ftwo{f}$ can be far apart. In particular,
we are able to establish that there exist polynomials $f$ with
non-neg. coefficients such that $\|f\|_{BKS}/\ftwo{f}\geq
n^{d/24}$. Moreover, we conjecture that the worst-case gap between
$\|f\|_{BKS}$ and $\ftwo{f}$ for polynomials with
non-neg. coefficients are as large as
$\widetilde{\Omega}_d((n/d)^{d/4-1/2})$ (note that the conjectured
$(n/d)^{d/4-1/2}$ gap for non-negative coefficient polynomials is
realizable using arbitrary polynomials, i.e. it was established in
\cite{BGL16} that polynomials with i.i.d. $\pm 1$ coefficients achieve
this gap w.h.p.).

Our results show that $q$ levels of SOS gives an $(n/q)^{d/4-1/2}$
approximation to $\ftwo{f}$ which has a better dependence on $q$ and
consequently, converges to a constant factor approximation after
$\Omega(n)$ levels.

% We also note that optimizing problems with non-negative coefficients
% is closely related to the densest sub-hypergraph problem. If all
% non-zero coefficients in a degree-$d$ polynomial are equal to 1
% (say), then the polynomial can be thought of as a describing a
% $d$-uniform hypergraph corresponding to the non-zero
% coefficients. 

\medskip \noindent \textbf{$2$-to-$4$ norm.}  It was proved in
\cite{BKS14} that for any matrix $A$, $q$ levels of the
SoS hierarchy approximates $\|A\|^4_{2\rightarrow 4} =
\ftwo{\|Ax\|_4^4}$ (i.e. the fourth power of the $2$-to-$4$-norm)
within a factor of
\[
    1 + \pth{\frac{\log n}{q}}^{1/2} \cdot 
    \frac{\|A\|^2_{2\rightarrow 2}\|A\|^2_{2\rightarrow \infty}}{\|A\|^4_{2\rightarrow 4}} \mper
\]
Brandao and Harrow \cite{BH15} also gave a nets based algorithm with
runtime $2^q$ that achieves the same approximation as above. Here
again, the cases of interest were those matrices for which
$\|A\|^2_{2\rightarrow 2}\|A\|^2_{2\rightarrow \infty}$ and
$\|A\|^4_{2\rightarrow 4}$ are at most constant apart.

We would like to bring attention to an open problem in this line of
work.  It is not hard to show that for an $m\times n$ matrix $A$ with
i.i.d.  Gaussian entries, $\|A\|^2_{2\rightarrow 2} = \Theta(m + n)$,
$\|A\|^2_{2\rightarrow \infty} = \Theta(n)$, and
$\|A\|^2_{2\rightarrow 4} = \Theta(m+n^2)$ which implies the worst
case approximation factor achieved above is $\Omega(n/\sqrt{q})$ when
we take $m=\Omega(n^2)$.

Our result for arbitrary polynomials of degree-$4$, achieves an
approximation factor of $O(n/q)$ after $q$ levels of SoS which implies
that the current best known approximation $2$-to-$4$ norm is oblivious
to the structure of the $2$-to-$4$ polynomial and seems to suggest
that this problem can be better understood for arbitrary tall
matrices. For instance, can one get a $\sqrt{m}/q$ approximation for
$(m\times n)$ matrices (note that \cite{BH15} already implies a
$\sqrt{m/q}$-approximation for all $m$, and our result implies a
$\sqrt{m}/q$-approximation when $m=\Omega(n^2)$).

\medskip \noindent \textbf{Random Polynomials.}  For the case when
$f$ is a degree-$d$ homogeneous polynomial with i.i.d. random $\pm 1$
coefficients \cite{BGL16,RRS16} showed that degree-$q$ SoS certifies
an upper bound on $\ftwo{f}$ that is with high probability at most
$\widetilde{O}((n/q)^{d/4 - 1/2})\cdot \ftwo{f}$.  Curiously, this matches our approximation
guarantee for the case of \emph{arbitrary} polynomials with non-negative
coefficients. This problem was also studied for the case of
sparse random polynomials in \cite{RRS16} motivated by applications to refuting random
CSPs.

\iffalse Hopkins \etal also showed the following lower bounds for
random polynomials which also provides lower bounds for arbitrary
polynomials:
\begin{theorem}[\cite{HKPRSS16}]
Let $f$ be a degree-$d$ polynomial with i.i.d. Gaussian coefficients. Then for every constant 
$\epsilon > 0$, with high probability over $f$, there is a degree $n^{\epsilon}$ pseudodistribution 
with $\PEx{f(x)} \geq n^{d/4 - O(d\epsilon) }$, which satisfies $\| x\|^2 = 1$.
\end{theorem}
\fi

%% file: poly_opt/2.overview.tex
\section{Overview of Proofs and Techniques}\label{section:poly:overview}

In the interest of clarity, we shall present all techniques for the special case where $f$ is an
arbitrary degree-$4$ homogeneous polynomial. We shall further assume that $\ftwo{f} = \fmax{f}$ just so
that $\hssos{f}$ is a relaxation of $\ftwo{f}$. Summarily, the goal of this section is to give an
overview of an $O(n/q)$-approximation of $\ftwo{f}$, \ie\ 
\[
    \hssos{f^{q/4}}^{4/q} \leq O(n/q)\cdot \ftwo{f}.
\]
Many of the high level ideas remain the same when considering higher degree polynomials and special
classes like polynomials with non-negative coefficients, or sparse polynomials.

\subsection{Warmup: $(n^2/q^2)$-Approximation}
We begin with seeing how to analyze constant levels of the $\hssos{\cdot}$ relaxation and will then
move onto higher levels in the next section.  The level-$4$ relaxation actually achieves an
$n$-approximation, however we will start with $n^2$ as a warmup and cover the $n$-approximation a
few sections later. 
\bigskip

\subsubsection{$n^2$-Approximation using level-$4$ relaxation}\label{section:poly:warmup:n2}
We shall establish that $\hssos{f}\leq O(n^2)\cdot \ftwo{f}$. 
Let $M_f$ be the SoS-symmetric representation of $f$, let 
$x_{i_1}x_{i_2}x_{i_3}x_{i_4}$ be the monomial whose coefficient 
in $f$ has the maximum magnitude, and let $B$ be the magnitude 
of this coefficient. Now by Gershgorin circle theorem, 
we have $\hssos{f} \leq \|M_f\|_2\leq  n^2 \cdot B$. 

It remains to establish $\ftwo{f} = \Omega(B)$. To this end, 
define the decoupled polynomial $\mathcal{F}(x,y,z,t) := 
{(x\otimes y)}^T\cdot M_f \cdot (z\otimes t)$ and define the decoupled 
two-norm as
\[
  \ftwo{\mathcal{F}}~:=~\sup_{\|x\|,\|y\|,\|z\|,\|t\| = 1} 
  \mathcal{F}(x,y,z,t).
\] 
It is well known that $\ftwo{f} = \Theta(\ftwo{\mathcal{F}})$ 
(see \cref{lem:poly:decoupled:lower:bound:pre}). Thus, we have, 
\begin{align*}
  &\ftwo{f}~=~\Omega(\ftwo{\mathcal{F}})~\geq~\Omega\inparen{|\mathcal{F}(e_{i_1},e_{i_2},e_{i_3},e_{i_4})|} \\
  &~=~\Omega(B)~=~\Omega \inparen{\hssos{f} / n^2} \mper
\end{align*}

\noindent
In order to better analyze $\hssos{f^{q/4}}^{4/q}$ we will need 
to introduce some new techniques.
\bigskip

\subsubsection{$(n^2/q^2)$-Approximation Assuming \cref{thm:poly:intro:gen:multi}}
We will next show that $\hssos{f^{q/4}}^{4/q} \leq O(n^2/q^2)\cdot\ftwo{f}$ 
(for $q$ divisible by $4$). In fact, one can show something stronger, namely that for every 
homogeneous polynomial $g$ of degree-$q$, $\hssos{g}\leq 2^{O(q)}\cdot {(n/q)}^{q/2}\cdot \ftwo{g}$ 
which clearly implies the above claim (also note that for the target $O(n^2/q^2)$-approximation 
to $\ftwo{f}$, losses of $2^{O(q)}$ in the estimate of $\ftwo{g}$ are negligible, while factors of 
the order $q^{\Omega(q)}$ are crucial). 

Given the additional freedom in choice of representation (due to the polynomial having higher degree), 
a first instinct would be to completely symmetrize, \ie{} take the SoS-symmetric representation of 
$g$, and indeed this works for multilinear $g$ (see~\cite{BGGLT17} for details). 

However, the above approach of taking the SoS-symmetric representation breaks down when 
the polynomial is non-multilinear. To circumvent this issue, we employ \cref{thm:poly:intro:gen:multi} 
which on combining with the aforementioned multilinear polynomial result, yields that 
for every homogeneous polynomial $g$ of degree-$q$, $\hssos{g}\leq {(n/q)}^{q/2}\cdot \ftwo{g}$. 
The proofs of \cref{thm:poly:intro:gen:multi} and it's generalizations (that will be required for the $n/q$ 
approximation), are quite non-trivial and are the most technically involved sections of our upper
bound results. We shall next give an outline of the proof of \cref{thm:poly:intro:gen:multi}.

\subsubsection{Reduction to Optimization of Multilinear Polynomials}\label{sec:poly:decoupling}
One of the main techniques we develop in this work, is a way of reducing the optimization problem
for general polynomials to that of multilinear polynomials, which \emph{does not increase the 
number of variables}. 
While general techniques for reduction to the multilinear case have been widely used in the 
literature~\cite{KN08, HLZ10, So11} (known commonly as decoupling/polarization techniques), 
these reduce the problem to optimizing a multilinear polynomial in $n \cdot d$ variables (when 
the given polynomial $h$ is of degree $d$). Below is one example:
\begin{lemma}[\cite{HLZ10}]\label{lem:poly:decoupled:lower:bound:pre}
  Let $\mathcal{A}$ be a SoS-symmetric $d$-tensor and let $h(x):= \iprod{\mathcal{A}}{x^{\otimes d}}$.
  Then  $\|h\|_2~\geq~2^{-O(d)}\cdot \max_{\|x^i\|=1} \iprod{\mathcal{A}}{x^1\otimes \dots \otimes x^d}$.
\end{lemma}
Since we are interested in the improvement in approximation obtained by considering $f^{q/4}$ 
for a large $q$, applying these would yield a multilinear polynomial in $n \cdot q$ variables. 
For our analysis, this increase in 
variables exactly cancels the advantage we obtain by considering $f^{q/4}$ instead of $f$ 
(\ie\ the advantage obtained by using $q$ levels of the SoS hierarchy).

We can uniquely represent a homogeneous polynomial $g$ of degree $q$ as
\begin{align}\label{eq:poly:decomposition}
  g(x) &~=~\sum_{\abs{\alpha} \leq q/2} x^{2\alpha}  \cdot G_{2\alpha}(x) \nonumber\\
  &~=~\sum_{r=0}^{q/2} \sum_{\abs{\alpha} = r} x^{2\alpha}  \cdot G_{2\alpha}(x) \nonumber\\
  &~=:~\sum_{r=0}^{q/2} g_r(x) \mcom 
\end{align}
where each $G_{2\alpha}$ is a multilinear polynomial and $g_r(x) \defeq \sum_{\abs{\alpha} = r}
x^{2\alpha}  \cdot G_{2\alpha}(x)$.
We reduce the problem to optimizing $\ftwo{G_{2\alpha}}$ for each of the polynomials
$G_{2\alpha}$. More formally, we show that 
\begin{equation}\label[ineq]{ineq:poly:do:rtm}
    \frac{\hssos{g}}{\ftwo{g}}~\leq~\max_{\alpha\in \udmindex{q/2}} 
    \frac{\hssos{G_{2\alpha}}}{\|G_{2\alpha}\|_2}\cdot 2^{O(q)}
\end{equation}
As a simple and immediate example of its applicability, \cref{ineq:poly:do:rtm} provides a simple proof 
of a polytime constant factor approximation for optimization over the simplex (actually this case is
known to admit a PTAS~\cite{deKLP06, dKLS15}). Indeed, observe that a simplex optimization problem for a 
degree-$q/2$ polynomial in the variable vector $y$ can be reduced to a sphere optimization 
by substituting $y_i = x_i^2$. Now since every variable present in a monomial has even degree 
in that monomial, each $G_{2\alpha}$ is constant, which implies a constant factor approximation 
(dependent on $q$) on applying \cref{ineq:poly:do:rtm}.  

Returning to our overview of the proof, note that given representations of each of the polynomials
$G_{2\alpha}$, each of the polynomials $g_r$ can be represented as a block-diagonal matrix with one
block corresponding to each $\alpha$. Combining this with triangle inequality and the fact that the
maximum eigenvalue of a block-diagonal matrix is equal to the maximum eigenvalue of one of the
blocks, gives the following inequality:
\begin{equation}\label[ineq]{ineq:poly:weak:sos:split}
  \hssos{g}~\leq~(1+q/2)\cdot \max_{\alpha\in \udmindex{q/2}} \hssos{G_{2\alpha}}. 
\end{equation}
We can further strengthen \cref{ineq:poly:weak:sos:split} by averaging the ``best'' representation of 
$G_{2\alpha}$ over $|\orbit{\alpha}|$ diagonal-blocks which all correspond to $x^{2\alpha}$.
We show (see~\cite{BGGLT17})
\begin{equation}\label[ineq]{ineq:poly:sos:split}
\hssos{g}~\leq~(1+q/2)\cdot \max_{\alpha\in \udmindex{q/2}} \frac{\hssos{G_{2\alpha}}}{|\orbit{\alpha}|} \mper
\end{equation}
Since $|\orbit{\alpha}|$ can be as large as $q^{\Omega(q)}$, the above strengthening is crucial.  We then
prove the following inequality, which shows that the decomposition in \cref{eq:poly:decomposition} not
only gives a block-diagonal decomposition for matrix representations of $g$, but can in fact be thought
of as a ``decomposition'' of the \emph{tensor} corresponding to $g$ (with regard to computing
$\ftwo{g}$). We show that
\begin{equation}
\label{ineq:poly:do:mult:2}
    \ftwo{g}~\geq~2^{-O(q)}\cdot \max_{\alpha\in \udmindex{q/2}} 
    \frac{\|G_{2\alpha}\|_2}{|\orbit{\alpha}|} \mper
\end{equation}
The above inequality together with \cref{ineq:poly:sos:split}, implies \cref{ineq:poly:do:rtm}. 

\bigskip

\subsubsection{A new weak decoupling lemma}\label{sec:poly:weak-decoupling}
Recall that the expansion of $g(x)$ in \cref{eq:poly:decomposition}, contains the term 
$x^{2\alpha}\cdot G_{2\alpha}(x)$. 
The key part of proving the bound in \cref{ineq:poly:do:mult:2} is to show the following ``weak
decoupling'' result for $x^{2\alpha}$ and $G_{2\alpha}$.
%
% The proof can be divided into two parts: 
%
\begin{align*}
  \forall \alpha \qquad  \ftwo{g}
  &~\geq~\max_{\|y\|=\|x\|=1}~y^{2\alpha}\cdot G_{2\alpha}(x)\cdot 2^{-O(q)} \\
  &~=~\max_{\|y\|=1}~y^{2\alpha}\cdot \ftwo{G_{2\alpha}}\cdot 2^{-O(q)}.
\end{align*} 
%
%
% \begin{compactenum}[(A)]
%     \item (Decoupling $x^{2\alpha}$ and $G_{2\alpha}(x)$)~~ We establish that 
%     \[
%     \ftwo{g}
%     ~\geq~ 
%     \max_{\|y\|=\|x\|=1}~
%     y^{2\alpha}\cdot G_{2\alpha}(x)\cdot 2^{-O(q)}
%     ~=~ 
%     \max_{\|y\|=1}~
%     y^{2\alpha}\cdot \ftwo{G_{2\alpha}}\cdot 2^{-O(q)}.
%     \] 
%    
%     \medskip
%    
%     \item (Reweighting Step)~~ Choose unit vector $y$ so that ~ $y^{2\alpha}\geq ~
%     2^{-O(q)}/\cardin{\orbit{\alpha}}.$  
% \end{compactenum}
% \bigskip
%
% \noindent
% The proof of (B) follows by considering the unit vector $y:= \sqrt{\alpha}/\sqrt{|\alpha|}$, i.e. 
% $y:= \sum_{i\in [n]} \frac{\sqrt{\alpha_i}}{\sqrt{|\alpha|}}\cdot e_i$, and a careful calculation.

The proof of \cref{ineq:poly:do:mult:2} can then be completed by considering the unit vector 
$y:= \sqrt{\alpha}/\sqrt{|\alpha|}$, i.e. $y:= \sum_{i\in [n]}
\frac{\sqrt{\alpha_i}}{\sqrt{|\alpha|}}\cdot e_i$. A careful calculation shows that 
$y^{2\alpha} \geq 2^{-O(q)}/\cardin{\orbit{\alpha}}$ which finishes the proof.

The primary difficulty in establishing the above decoupling is the possibility of cancellations.  To see
this, let $x^*$ be the vector realizing $\|G_{2\alpha}\|_2$ and substitute $z = (x^* + y)$ into $g$.
Clearly, $y^{2\alpha}\cdot G_{2\alpha}(x^*)$ is a term in the expansion of $g(z)$, however there is no
guarantee that the other terms in the expansion don't cancel out this value. 
To fix this our proof relies on multiple delicate applications of the 
first-moment method, \ie\ we consider a complex vector random variable $Z(x^*,y)$ that is a 
function of $x^*$ and $y$, and argue about $\Ex{|g(Z)|}$. 

\smallskip \noindent \textbf{The extremal case of $\alpha = 0^n$.} We first consider the extremal case of $\alpha = 0^n$,
where we define $y^{2\alpha}=1$.
This amounts to showing that  for every homogeneous polynomial $h$ of degree $t$, $\ftwo{h}\geq
\ftwo{h_m}\cdot 2^{-O(t)}$ where $h_m$ is the restriction of  $h$ to its multilinear monomials.

Given the optimizer $x^*$ of $\ftwo{h_m}$, let $z$ be a random vector such that each $Z_i = x_i^*$
with probability $p$ and $Z_i = 0$ otherwise. Then, $\Ex{h(Z)}$ is a \emph{univariate} degree-$t$
polynomial in $p$ with the coefficient of $p^t$ equal to $h_m(x^*)$. An application of Chebyshev's extremal polynomial inequality then gives 
that there exists a value of the probability $p$ such that
\begin{align*}
  \ftwo{h} &~\geq~\Ex{\abs{h(Z)}}~\geq~\abs{\Ex{h(Z)}}~\geq~2^{-O(t)} \cdot \abs{h_m(x^*)} \\
  &~=~2^{-O(t)} \cdot \ftwo{h_m}\mper
\end{align*}

For the case of general $\alpha$, we first pass to the \emph{complex version} of $\ftwo{g}$ defined as
\[
  \ftwo{g}^c~\defeq~\sup_{z \in \CC^n, \|z\| = 1} \abs{g(z)} \mper
\]
We use another averaging argument together with an application of the polarization lemma
(\cref{lem:poly:decoupled:lower:bound:pre}) to show that we do not lose much by considering
$\ftwo{g}^c$. In particular, $\ftwo{g}~\leq~\ftwo{g}^c~\leq~2^{O(q)} \cdot \ftwo{g}$.

\smallskip\noindent \textbf{The extremal case of $g = g_r$.}
In this case, the problem reduces to showing that 
for all $\alpha\in \degmindex{r}$ and for all $y \in \SSS^{n-1}$, 
\[
    \|g_r\|^c_2~\geq~y^{2\alpha} \cdot \ftwo{G_{2\alpha}} \cdot 2^{-O(q)}.
\]
Fix any $\alpha\in \degmindex{r}$, and let $\omega\in \CC^n$ be a complex vector random variable, 
such that $\omega_i$ is an independent and uniformly random $(2\alpha_i+1)$-th root of unity. 
Let $\Xi$ be a random $(q-2r+1)$-th root of unity, and let $x^*$ be the optimizer of 
$\|G_{2\alpha}\|_2$. Let $Z := \omega\circ y + \Xi\cdot x^*$, where $\omega\circ y$ denotes the
coordinate-wise product.  Observe that for any 
$\alpha',\gamma$ such that $|\alpha'|=r,~|\gamma|=q-2r,~\gamma\leq \one$, 
\[
  \Ex{\prod_{i}\omega_i \cdot \Xi \cdot Z^{2\alpha' +\gamma}}~= 
  \begin{cases}
    y^{2\alpha}\cdot {(x^*)}^{\gamma}
    &\text{if~}\alpha' = \alpha \\
    0 &\text{otherwise}
  \end{cases}
\]
By linearity, this implies 
$\Ex{\prod_{i}\omega_i \cdot \Xi \cdot g_r(Z)}  = y^{2\alpha}\cdot G_{2\alpha}(x^*)$.
The claim then follows by noting that 
\begin{align*}
    \|g_r\|^c_2 
    &~\geq~\Ex{|g_r(Z)|} 
    =~\Ex{\cardin{\prod_{i}\omega_i \cdot \Xi \cdot g_r(Z)}} \\
    &\geq~\cardin{\,\Ex{\prod_{i}\omega_i \cdot \Xi \cdot g_r(Z)}} 
    \geq~y^{2\alpha}\cdot \ftwo{G_{2\alpha}}.
\end{align*}
\noindent \textbf{The general case.} The two special cases considered here correspond to the cases when we
need to extract a specific $g_r$ (for $r = 0$), and when we need to extract a fixed $\alpha$ from a
given $g_r$. The argument for the general case uses a combination of the arguments for both these cases.
Moreover, to get an $O(n/q)$ approximation, we also need versions of such decoupling
lemmas for folded polynomials to take advantage of ``easy substructures'' as described next.

\subsection{Exploiting Easy Substructures via Folding and Improved Approximations }
To obtain the desired $n/q$-approximation to $\ftwo{f}$, we need to use the fact that the problem of 
optimizing quadratic polynomials can be solved in polynomial time, and moreover that SoS 
captures this. More generally, in this section we consider the problem of getting improved 
approximations when a polynomial contains ``easy substructures''. It is not hard to obtain improved 
guarantees when considering constant levels of SoS. The second main technical contribution of our work 
is in giving sufficient conditions under which higher levels of SoS improve on the approximation of 
constant levels of SoS, when considering the optimization problem over polynomials containing 
``easy substructures''. 

As a warmup, we shall begin with seeing how to exploit easy substructures at constant levels 
by considering the example of degree-$4$ polynomials that trivially ``contain'' quadratics. 

\bigskip 

\subsubsection{$n$-Approximation using Degree-$4$ SoS}\label{sec:poly:warmup:n}

Given a degree-$4$ homogeneous polynomial $f$ (assume $f$ is multilinear for simplicity), we 
consider a degree-$(2,2)$ folded polynomial $h$, whose unfolding yields $f$, chosen so that 
$\max_{\|y\|=1}\ftwo{h(y)} = \Theta(\ftwo{f})$ (recall that an evaluation of a folded 
polynomial returns a polynomial, i.e., for a fixed $y$, $h(y)$ is a quadratic polynomial in the 
indeterminate $x$). Such a $h$ always exists and is not hard to find based on the SoS-symmetric
representation of $f$. Also recall, 
\[
\qquad\qquad 
    h(x) = \sum_{|\beta|=2,\,\beta\leq \one} \fold{h}{\beta}(x) \cdot x^{\beta} \mcom
\]
where each $\fold{h}{\beta}$ is a quadratic polynomial (the aforementioned phrase 
``easy substructures'' is referencing the folds: $\fold{h}{\beta}$ which are easy to optimize). 
Now by assumption we have, 
\[
    \ftwo{f}
    \geq \max_{|\beta|=2,\,\beta\leq \one} \|h(\beta/\sqrt{2})\|_2
    =\max_{|\beta|=2,\,\beta\leq \one} \|\fold{h}{\beta}\|_2/2.
\]
We then apply the block-matrix generalization of 
Gershgorin circle theorem to the SoS-symmetric matrix representation of $f$ to show that 
\begin{align*}
  \hssos{f} &~\leq~\fsp{f}~\leq~n \cdot  \max_{|\beta|=2,\,\beta\leq \one} \fsp{\fold{h}{\beta}}\\
  &~=~n\cdot \max_{|\beta|=2,\,\beta\leq \one} \|\fold{h}{\beta}\|_2 \mcom
\end{align*} 
where the last step uses the fact that $\fold{h}{\beta}$ is a quadratic, and $\fsp{\cdot}$ is a tight
relaxation of $\ftwo{\cdot}$ for quadratics.  This yields the desired $n$-approximation. 
%To find a unit vector guaranteeing this approximation, we simply output $\argmax
%

\bigskip

\subsubsection{$n/q$-approximation using Degree-$q$ SoS}

Following the cue of the $n^2/q^2$-approximation, we derive the desired $n/q$ bound by proving a 
folded-polynomial analogue of every claim in the previous section (including the multilinear 
reduction tools), a notable difference being that when we consider a power $f^{q/4}$ of $f$, we 
need to consider degree-$(q - 2q/4, 2q/4)$ folded polynomials, since we want to use the fact that 
any {\bf {\em product of $q/4$ quadratic polynomials}} is ``easy'' for SoS (in contrast to
\cref{sec:poly:warmup:n} where we only used the fact quadratic polynomials are easy for SoS). 
We now state an abstraction of the general approach we use to leverage the tractability of the folds.

\medskip \noindent \textbf{Conditions for Exploiting ``Easy Substructures'' at Higher Levels of SoS.} 
Let $d:= d_1 + d_2$ and $f:= \unfold{h}$ where $h$ is a degree-$(d_1,d_2)$ folded polynomial that
satisfies 
\[\sup_{\|y\|=1}\ftwo{h(y)} = \Theta_d(\ftwo{f}) \mper
\] 
Implicit in our work (see~\cite{BGGLT17}), is the following theorem 
we believe to be of independent interest: 
\begin{theorem}\label{thm:poly:easy-substructure}
    Let $h,f$ be as above, and let 
    \[
        \Gamma
        :=~\min \inbraces{\frac{\hssos{p}}{\ftwo{p}}
        \sep{p(x)\in \textrm{span}\inparen{\fold{h}{\beta}\sep{\beta\in \degmindex{d_2}}}} 
        }.
    \]
    Then for any $q$ divisible by $2d$, 
    \[\hssos{f^{q/d}}^{d/q} 
        \leq~O_d\inparen{\Gamma \cdot {(n/q)}^{d_1/2}}\cdot \ftwo{f}.\] 
\end{theorem}
In other words, if degree-$d_2$ SoS gives a good approximation for every polynomial in 
the subspace spanned by the folds of $h$, then higher levels of SoS give an improving 
approximation that exploits this. In this work, we only apply the above with $\Gamma =1$, where exact optimization is easy for the space spanned by the folds.

While we focused on general polynomials for the overview, let us remark that in the case of polynomials
with non-negative coefficients, the approximation factor in \cref{thm:poly:easy-substructure} becomes
$O_d\inparen{\delta\cdot {(n/q)}^{d_1/4}}$.

\bigskip

\subsection{Lower Bounds for Polynomials with Non-negative Coefficients}

\subsubsection{Degree-4 Lower Bound for Polynomials with Non-Negative Coefficients}
We discuss some of the important ideas from the proof of~\cref{thm:poly:nnc-lowerbound}.
The lower bound proved by a subset of the authors in~\cite{BGL16} 
proves a large ratio $\frac{\hssos{f}}{\| f \|_2}$ by considering a random polynomial $f$ where each
coefficient of $f$ is an independent (Gaussian) random variable with bounded variance. 
The most natural adaptation of the above strategy to degree-$4$ polynomials with non-negative
coefficients is to consider a random polynomial $f$ where each coefficient  $f_{\alpha}$ is
independently sampled such that $f_{\alpha} = 1$ with probability $p$ and $f_{\alpha} = 0$ with
probability $1-p$. However, this construction fails for every choice of $p$.
If we let $\sfA \in \RR^{ {[n]}^2 \times {[n]}^2}$ be the natural matrix representation of $f$ (\ie~each
coefficient $f_{\alpha}$ is distributed uniformly among the corresponding entries of $\sfA$),
the Perron-Frobenius theorem shows that $\norm{2}{\sfA}$ is less than the maximum row sum
$\max(\tilde{O}(n^2 p), 1)$ of $\sfM$, which is also an upper bound on $\hssos{f}$.
However, we can match this bound by (within constant factors) choosing $x = (\frac{1}{\sqrt{n}},
\dots, \frac{1}{\sqrt{n}})$ when $p \geq 1/n^2$. Also, when $p < 1/n^2$, we can take any $\alpha$
with $f_{\alpha} = 1$ and set $x_i = 1/2$ for all $i$ with $\alpha_i > 0$, which achieves a value of $1/16$.

We introduce another natural distribution of random non-negative polynomials that bypasses this problem. 
Let $G = (V, E)$ be a random graph drawn from the distribution $G_{n, p}$ (where we choose $p = n^{-1/3}$
and $V = [n]$). Let $\cliques \subseteq \binom{V}{4}$ be the set of $4$-cliques in $G$. 
The polynomial $f$ is defined as 
\[
  f(x_1, \dots, x_n) := \sum_{\{ i_1, i_2, i_3, i_4 \} \in \cliques} x_{i_1} x_{i_2} x_{i_3} x_{i_4}. 
\]
Instead of trying $\Theta(n^4)$ $p$-biased random bits, we use $\Theta(n^2)$ of them. This limited
independence bypasses the problem above, since the rows of $\sfA$ now have significantly different
row sums:
 $\Theta(n^2 p)$ rows that correspond to an edge of $G$ will have row sum $\Theta(n^2 p^5)$, and all
 other rows will be zeros. 
Since these $n^2 p$ rows (edges) are chosen independently of $\binom{[n]}{2}$, they still reveal
little information that can be exploited to find an $n$-dimensional vector $x$ with large
$f(x)$. However, the proof requires a careful analysis of the trace method (to bound the spectral
norm of an ``error'' matrix).

It is simple to prove that $\|f\|_{sp} \geq \Omega\inparen{\sqrt{n^2p^5}} = \Omega(n^{1/6})$ by
considering the Frobenius norm of the $n^2p \times n^2p$ principal submatrix, over any matrix
representation (indeed, $\sfA$ is the minimizer). To prove
$\hssos{f} \geq \tilde{\Omega}(n^{1/6})$, we construct a moment matrix $\sfM$ that is SoS-symmetric,
positive semidefinite, and has a large $\langle \sfM, \sfA \rangle$ (see the dual form of
$\hssos{f}$ in~\cite{BGGLT17}). It turns out that the $n^2 p
\times n^2 p$ submatrix of $\sfA$ shares spectral properties of the adjacency matrix of a random
graph $G_{n^2p, p^4}$, and taking $\sfM := c_1\sfA + c_2 \sfI$ for some identity-like matrix $\sfI$
proves $\hssos{f} \geq \tilde{\Omega}(n^{1/6})$. An application of the trace method is needed to
bound $c_2$.

To upper bound $\| f \|_2$, we first observe that $\| f \|_2$ is the same as the following natural combinatorial problem up to an $O(\log^4 n)$ factor: find four sets $S_1, S_2, S_3, S_4 \subseteq V$ that maximize
\[
\frac{\abs{\cliques_G(S_1, S_2, S_3, S_4)}}{\sqrt{|S_1||S_2||S_3||S_4|}}
\]
where $|\cliques_G (S_1, S_2, S_3, S_4)|$ is the number of $4$-cliques $\{ v_1, \dots, v_4 \}$ of $G$ with $v_i \in S_i$ for $i = 1, \dots, 4$. 
The number of copies of a fixed subgraph $H$ in $G_{n, p}$, including its tail bound, has been actively
studied in probabilistic combinatorics~\cite{Vu01, KV04, JOR04, Chatterjee12, DK12a, DK12b, LZ16}, though
we are interested in bounding the $4$-clique density of {\em every} $4$-tuple of subsets simultaneously.
The previous results give a strong enough tail bound for a union bound, to prove that the optimal value
of the problem is $O(n^2p^6 \cdot \log^{O(1)} n)$ when $\abs{S_1} = \cdots = \abs{S_4}$, but this strategy
inherently does not work when the set sizes become significantly different. 
However, we give a different analysis for the above asymmetric case, showing that the optimum is still no
more than $O(n^2p^6 \cdot \log^{O(1)} n)$.

\bigskip

\subsubsection{Lifting Stable Degree-$4$ Lower Bounds}
For a degree-$t$ ($t$ even) homogeneous polynomial $f$, note that 
$\max \{|\hssos{f}|,|\hssos{-f}|\}$ is a relaxation of $\ftwo{f}$. 
$\fsp{f}$ is a slightly weaker (but still quite natural) relaxation 
of $\ftwo{f}$ given by 
\begin{align*}
  &\fsp{f} \defeq \\
  &\inf \inbraces{\norm{2}{M}~\given~M~\text{is a matrix representation of}~f} \mper
\end{align*}
As in the case of $\hssos{f}$, for a (say) degree-$4$ polynomial $f$, $\fsp{f^{q/4}}^{4/q}$ gives a
hierarchy of relaxations for $\ftwo{f}$, for increasing values of $q$.

We give an overview of a general method of ``lifting''
certain ``stable'' low degree gaps for $\fsp{\cdot}$ to gaps for higher levels with 
at most $q^{O(1)}$ loss in the gap. While we state our techniques for lifting 
degree-$4$ gaps, all the ideas are readily generalized to higher levels. 
We start with the observation that the dual of $\fsp{f}$ is given by the following ``nuclear norm''
program.
Here $\sfM_f$ the canonical matrix representation of $f$, and $\norm{S_1}{\sfX}$ is the Schatten
$1$-norm (nuclear norm) of $X$, which is the sum of its singular values.
\begin{align*}
  \mathsf{maximize}~\qquad \qquad \qquad \iprod{\sfM_f}{\sfX}  \\
  \mathsf{subject~to:} \qquad \qquad \quad \norm{S_1}{\sfX} = 1  \\
  {\sfX}~\text{is SoS symmetric}
\end{align*}
Now let $\sfX$ be a solution realizing a gap of $\delta$ between $\fsp{f}$ and 
$\ftwo{f}$. We shall next see how assuming reasonable conditions on $\sfX$ and $M_f$, 
one can show that $\|f^{q/4}\|_{sp}/\|f^{q/4}\|_2$ is at least $\delta^{q/4}/q^{O(q)}$. 
% (i.e., we get a large gap for $\|f^{q/4}\|_{sp}$). 

In order to give a gap for the program corresponding to $\fsp{f^{q/4}}$, a natural choice for a solution
is the symmetrized version of the matrix $\sfX^{\otimes q/4}$ normalized by its Schatten-$1$ norm \ie\ for
$Y = \sfX^{\otimes q/4}$, we take
\begin{align*}
  \sfZ &~:= Y^S/\norm{S_1}{Y^S} \\
  &~\text{where}~Y^S[K] = \Ex{\pi \in \Sym_{q}}{Y[\pi(K)]}~\forall K \in {[n]}^q \mper
\end{align*}
To get a good gap, we are now left with showing that $\norm{S_1}{Y^S}$ is not 
too large. Note that symmetrization can 
drastically change the spectrum of a matrix as for different permutations $\pi$, the matrices
 $Y^{\pi}[K] \defeq Y[\pi(K)]$ can have very different ranks (while $\norm{F}{Y} = \norm{F}{Y^{\pi}}$).
% certain permutations 
% can have low rank and certain permutations have high rank (with identical 
% frobenius norm). 
In particular, symmetrization can increase the maximum 
eigenvalue of a matrix by polynomial factors, and thus one must carefully 
count the number of such large eigenvalues in order to get a good upper bound 
on $\norm{S_1}{Y^S}$. Such an upper bound is a consequence of a 
structural result about $Y^S$ that we believe to be of independent interest. 

To state the result, we will first need some notation. 
For a matrix $M\in\Re^{{[n]}^2\times {[n]}^2}$ let $T\in \RR^{{[n]}^4}$ denote 
the tensor given by, $T[i_1,i_2,i_3,i_4] = M[(i_1,i_2),(i_3,i_4)]$. 
Also for any non-negative integers $x,y$ satisfying $x+y = 4$, let 
$M_{x,y}\in \Re^{ {[n]}^{x}\times {[n]}^{y}}$ denote the (rectangular) matrix given by, 
$M[(i_1,\dots ,i_x),(j_1,\dots j_y)] = T[i_1,\dots ,i_x,j_1,\dots j_y]$. 
Let $M\in \Re^{ {[n]}^2\times {[n]}^2}$ be a degree-$4$ SoS-Symmetric matrix, 
let $M_A:= M_{1,3}\otimes M_{4,0}\otimes M_{1,3}$, let $M_B := M_{1,3}\otimes M_{3,1}$, 
let $M_C:=M$ and let $M_D:= \Vector{M}\Vector{M}^T = M_{0,4}\otimes M_{4,0}$. 

We show (see~\cite{BGGLT17}) that ${(M^{\otimes q/4})}^S$ can be written as the 
sum of $2^{O(q)}$ terms of the form: 
\[
    C(a,b,c,d)\cdot P \cdot (M_A^{\otimes a}\otimes M_B^{\otimes b}\otimes M_C^{\otimes c}
    \otimes M_D^{\otimes d}) \cdot P
\]
where $12a+8b+4c+8d = q$,  $P$ is a matrix with spectral norm $1$ and $C(a,b,c,d) = 2^{O(q)}$. 
This implies that 
controlling the spectrum of $M_A,M_B,M$ and $M_D$  
is sufficient to control on the spectrum of ${(M^{\otimes q/4})}^{S}$. 

Using this result with $M:= \sfX$, we are able to establish that 
if $\sfX$ satisfies the additional condition of 
$\norm{S_1}{\sfX_{1,3}}\leq 1$ (note that we already know 
$\norm{S_1}{X}\leq 1$), then $\norm{S_1}{Y^S} = 2^{O(q)}$. 
Thus $\sfZ$ realizes a $\iprod{M_f^{\otimes q/4}}{Y^S}/2^{O(q)}$ gap 
for $\fsp{f^{q/4}}$.
On composing this result with the degree-$4$ gap from the previous section, we 
obtain an $\fsp{\cdot}$ gap of $n^{q/24}/\inparen{q\cdot \log n}^{O(q)}$ for degree-$q$ polynomials with 
non-neg. coefficients. We also show the $q$-th level $\fsp{\cdot}$ gap for 
degree-$4$ polynomial with non-neg. coefficients is $\tilde{\Omega}(n^{1/6})/q^{O(1)}$. 

\smallskip
Even though we only derive results for the weaker relaxation $\fsp{\cdot}$, 
the structural result above can be used to lift ``stable'' low-degree SoS lower 
bounds as well (\ie\ gaps for $\hssos{\cdot}$), albeit with a stricter notion of stability (see~\cite{BGGLT17}).
However, the problem of finding 
such stable SoS lower bounds remains open. 

There are by now quite a few results giving near-tight lower bounds on the performance of higher level
SoS relaxations for \emph{average-case} problems~\cite{BHKKMP16,KMOW17,HKPRSS16}.  However, there are few
examples in the literature of matching SoS upper/lower bounds on \emph{worst-case} problems. We believe
our lifting result might be especially useful in such contexts, where in order to get higher degree lower
bounds, it suffices to give stable lower bounds for constant degree SoS.

%% file: poly_opt/3.prelims.tex
\section{Preliminaries on the SoS Hierarchy for Polynomials}\label{sec:poly:prelims}

\subsection{Pseudoexpectations and Moment Matrices}\label{sec:poly:sos:prelims}
Let ${\Re[x]}_{\leq q}$ be the vector space of polynomials with real coefficients in variables $x =
(x_1, \ldots, x_n)$, of degree at most $q$. For an even integer $q$, the degree-$q$
pseudo-expectation operator is a linear operator $\PEx : {\Re[x]}_{\leq q} \mapsto \Re$ such that 
\smallskip
\begin{enumerate}
  \item $\PEx{1} = 1$ for the constant polynomial $1$.
  \item $\PEx{p_1 + p_2} = \PEx{p_1} + \PEx{p_2}$ for any polynomials $p_1, p_2 \in {\Re[x]}_{\leq q}$. 
  \item $\PEx{p^2} \geq 0$ for any polynomial $p \in {\Re[x]}_{\leq q/2}$. 
\end{enumerate}
The pseudo-expectation operator $\PEx$ can be described by a \defnt{moment matrix} $\hM \in
\RR^{\udmindex{q/2} \times \udmindex{q/2}}$ defined as $\hM[\alpha,\beta] = \PEx{x^{\alpha+\beta}}$
for $\alpha,\beta \in \udmindex{q/2}$. 

For each fixed $t \leq q/2$, we can also consider the principal minor of $\hM$ indexed by
$\alpha,\beta \in \degmindex{t}$. This also defines a matrix $M \in \R^{ {[n]}^t \times {[n]}^{t}}$ with
$M[I,J] = \PEx{x^{\mi{I}+\mi{J}}}$. Note that this new matrix $M$ satisfies $M[I,J] = M[K,L]$
whenever $\mi{I}+\mi{J} = \mi{K} + \mi{L}$. Recall that a matrix in $\R^{ {[n]}^t \times {[n]}^t}$ with
this symmetry is said to be \defnt{SoS-symmetric}.

We will use the following facts about the operator $\PEx$ given by the SoS hierarchy.
\begin{claim}[Pseudo-Cauchy-Schwarz~\cite{BKS14}]
  $\PEx{p_1p_2}\leq {(\PEx{p_1^2}\PEx{p_2^2})}^{1/2}$ for any $p_1,p_2$ of degree at most $q/2$. 
\end{claim}

\subsubsection{Constrained Pseudoexpectations}
For a system of polynomial constraints $C = \inbraces{f_1 = 0, \ldots, f_m = 0, g_1 \geq 0, \ldots,
  g_r \geq 0}$, we say $\PEx_C$ is a pseudoexpectation operator respecting $C$, if in addition to the
above conditions, it also satisfies
\begin{enumerate}
\item $\PExc{p \cdot f_i} = 0$,  $\forall i \in [m]$ and $\forall p$ such that $\deg(p \cdot f_i)
  \leq q$.
\item $\PExc{p^2 \cdot \prod_{i \in S} g_i} \geq 0$, $\forall S \subseteq [r]$ and $\forall p$
  such that $\deg(p^2 \cdot \prod_{i \in S} g_i) \leq q$.
\end{enumerate}
It is well-known that such constrained pseudoexpectation operators can be described as solutions to
semidefinite programs of size $n^{O(q)}$~\cite{BS14, Laurent09}. This hierarchy of semidefinite
programs for increasing $q$ is known as the SoS hierarchy.

\subsection{Matrix Representations of Polynomials and Relaxations}\label{sec:poly:hscsos}
For a homogeneous polynomial $f$ of even  degree $d$, we say a matrix $M \in \Re^{ {[n]}^{d/2}\times
{[n]}^{d/2}}$ is a degree-$d$ matrix representation of $f$ if for all $x$, $f(x) = {(x^{\otimes
d/2})}^{T} \cdot M \cdot x^{\otimes d/2}$. Recall that we consider the semidefinite program for
optimizing the quantity $\hssos{f}$, which is a relaxation for $\ftwo{f}$ when $f \geq 0$.
Let $\sfM_f \in \R^{n^{d/2} \times n^{d/2}}$ denote the unique SoS-symmetric matrix representation
of $f$. The~\cref{fig:poly:Lambda} gives the primal and dual forms of the relaxation computing
$\hssos{f}$.
It is easy to check that strong duality holds in this case, since the solution
$\PExc{x^{\alpha}} = {(1/\sqrt{n})}^{\abs{\alpha}}$ for all $\alpha \in \udmindex{d}$,
is strictly feasible and in the relative interior of the domain. 
Thus, the objective values of the two programs are equal.
\begin{figure}[htb]
  \begin{tabular}{|c|}
    \hline
    \begin{minipage}[t]{0.94\textwidth}
      \smallskip
      \underline{\textsf{Primal}}
      \[
        \hssos{f}~\defeq~\inf \inbraces{\sup_{\|z\| = 1} z^T M z~\given~M \in \sym{n^{d/2}},~{(x^{\otimes d/2})}^T
        \cdot M \cdot x^{\otimes d/2} = f(x)~\forall x \in \R^n} 
      \]
      \smallskip
    \end{minipage} \\
    \hline
    \begin{minipage}{0.95\textwidth}
      \begin{tabular}{l|l}
        \begin{minipage}[t]{0.47 \textwidth}
          \smallskip
          \underline{\textsf{Dual I}}
          \begin{align*}
            \mathsf{maximize}~\qquad \qquad \qquad \iprod{\sfM_f}{\sfX}  \\
            \mathsf{subject~to:} \qquad \qquad \quad \Tr{\sfX} = 1 \\
            {\sfX}~\text{is SoS symmetric} \\
            \sfX \succeq 0
          \end{align*}
        \end{minipage}
        &
        \begin{minipage}[t]{0.47 \textwidth}
          \smallskip
          \underline{\textsf{Dual II}}
          \begin{align*}
            \mathsf{maximize} \qquad \qquad \qquad \qquad \quad \PExc{f}  \\
            \mathsf{subject~to:}  \qquad \qquad \tilde{\mathbb{E}}_C~\text{is a degree-$d$} \\
            \text{pseudoexpectation} \\
            \tilde{\mathbb{E}}_C~\text{respects}~C \equiv \inbraces{\norm{2}{x}^d = 1}\\
          \end{align*}
        \end{minipage}  \\
      \end{tabular}
    \end{minipage} \\
    \hline
  \end{tabular}
  \caption{Primal and dual forms for the relaxation computing $\hssos{f}$}\label{fig:poly:Lambda}
\end{figure}
We will also consider a weaker relaxation of $\ftwo{f}$, which we denote by $\fsp{f}$. A somewhat
weaker version of this was used as the reference value in the work of~\cite{BKS14}. The~\cref{fig:poly:fsp}
gives the primal and dual forms of this relaxation.
\begin{figure}[htb]
  \begin{tabular}{|c|}
    \hline
    \begin{minipage}{0.94\textwidth}
      \smallskip
      \underline{\textsf{Primal}}
      \[
        \fsp{f}~\defeq~\inf \inbraces{ \norm{2}{M}~\given~M \in \sym{n^{d/2}},~{(x^{\otimes d/2})}^T \cdot 
        M \cdot x^{\otimes d/2} = f(x)~\forall x \in \R^n}
      \]
      \smallskip
    \end{minipage} \\
    \hline
    \begin{minipage}{0.94\textwidth}
      \smallskip
      \underline{\textsf{Dual}}
      \begin{align*}
        \mathsf{maximize}~\qquad \qquad \qquad \iprod{\sfM_f}{\sfX}  \\
        \mathsf{subject~to:} \qquad \qquad \quad \norm{S_1}{\sfX} = 1 \\
        {\sfX}~\text{is SoS symmetric}
      \end{align*}
    \end{minipage} \\
    \hline
  \end{tabular}
  \caption{Primal and dual forms for the relaxation computing $\lVert f\rVert_{\mathrm{sp}}$}\label{fig:poly:fsp}
\end{figure}
%

% \begin{equation}
% \label{eq:fsp}
%   \fsp{f} ~\defeq~ \inf \inbraces{\norm{2}{M} ~\given~  M \in \sym{n^{d/2}}, ~~(x^{\otimes d/2})^T \cdot M \cdot x^{\otimes d/2} = f(x) ~~\forall x \in \R^n} \mper
% \end{equation}

%
We will also need to consider constraint sets $C = \inbraces{\norm{2}{x}^2 = 1, x^{\beta_1} \geq 0,
  \ldots, x^{\beta_m} \geq 0}$. We refer to the non-negativity constraints here as \defnt{moment
  non-negativity constraints}.
When considering the maximum of $\PExc{f}$, for constraint sets $C$ containing moments
non-negativity constraints in addition to $\norm{2}{x}^2=1$, we refer to the optimum value as
$\hscsos{C}{f}$. Note that the maximum is still taken over degree-$d$ pseudoexpectations.
Also, strong duality still holds in this case since $\PExc{x^{\alpha}} =
{(1/\sqrt{n})}^{\abs{\alpha}}$ is still a strictly feasible solution.

\subsubsection{Properties of Relaxations Obtained from Constrained Pseudoexpectations}
We use the following claim, which is an easy consequence of the fact that the sum-of-squares
algorithm can produce a certificate of optimality (see~\cite{OZ13}). 
In particular, if $\max_{\PExc} \PExc{f} =
\hscsos{C}{f}$ for a degree-$q_1$ pseudoexpectation operator respecting $C$ containing
$\norm{2}{x}^2=1$ and moment non-negativity constraints for $\beta_1, \ldots, \beta_m$, then for every
$\lambda > \hscsos{C}{f}$, we have that $\lambda - f$ can be certified to be positive by 
showing that $\lambda - f \in \Sigma^{q_1}_C$. Here $\Sigma^{(q_1)}_C$ is the set of all expressions
of the form
\[
  \lambda - f~=~\sum_j p_j \cdot \pth{\norm{2}{x}^2-1} + \sum_{S \subseteq [m]} h_S(x) \cdot \prod_{i
  \in S}{x^{\beta_i}} \mcom
\]  
where each $h_S$ is a sum of squares of polynomials and the degree of each term is at most $q_1$.
\begin{lemma}\label{lem:poly:sos:replace}
  Let $\hscsos{C}{f}$ denote the maximum of $\PExc{f}$ over all degree-$d$ pseudoexpectation operators
  respecting $C$. Then, for a pseudoexpectation operator of degree $d'$ (respecting $C$) and a
  polynomial $p$ of degree at most $(d'-d)/2$, we have that 
  \[
    \PExc{p^2 \cdot f}~\leq~\PExc{p^2} \cdot \hscsos{C}{f} \mper
  \]
\end{lemma}
\begin{proof}
As described above, for any $\lambda > \hscsos{C}{f}$, we can write $\lambda - f = g$ for $g \in
\Sigma^{(d)}_C$. Since the degree of each term in $p^2 \cdot g$ is at most $d'$, we have by the
properties of pseudoexpectation operators (of degree $d'$) that
\[
  \lambda \cdot \PExc{p^2} - \PExc{p^2 \cdot f}~=~\PExc{p^2 \cdot (\lambda - f)}~=~\PExc{p^2
  \cdot g}~\geq~0 \mper
\] 
\end{proof}

The following monotonicity claim for non-negative coefficient polynomials will come in handy in later sections. 

\begin{lemma}\label{lem:poly:sos:nnc:monotonicity}
  Let $C$ be a system of polynomial constraints containing $\{\forall \beta\in \degmindex{t},
  x^{\beta}\geq 0\}$. Then for any non-negative coefficient polynomials $f$ and $g$ of degree $t$, and
  such that $f\geq g$ (coefficient-wise, i.e. $f-g$ has non-negative coefficients), we have
  $\hscsos{C}{f}\geq \hscsos{C}{g}$. 
\end{lemma}

\begin{proof}
  For any pseudo-expectation operator $\PExc$ respecting $C$, we have $\PExc{f-g}\geq 0$ because of
  the moment non-negativity constraints and by linearity. 
    
  So let $\PExc$ be a pseudo-expectation operator realizing $\hscsos{C}{g}$. Then we have, 
  \[
    \hscsos{C}{f}\geq \PExc{f} = \PExc{g}+\PExc{f-g} = \hscsos{C}{g} + \PExc{f-g}\geq 0.
  \] 
\end{proof}

% For polynomials $p_1,p_2$,
% let $p_1\succeq p_2$ denote that  $p_1-p_2$ is a sum of squares. 
% It is easy to verify that if $p_1,p_2$ are homogeneous degree $d$ polynomials and 
% there exist matrix representations $M_{p_1}$ and $M_{p_2}$ of $p_1$ and $p_2$ respectively, such that ~
% $M_{p_1}-M_{p_2}\succeq 0$, \,then ~$p_1-p_2\succeq 0$. 
%
%
% Given a polynomial $p$ of degree $d\leq q$, and the optimization problem $\max_{\|x\|=1}p(x)$, 
% we are interested in the relaxation: $~ \max~  \PEx{p} ~~ s.t. ~ \PEx{\|x\|^{q}} = 1$,~
% over all valid degree-$q$ pseudo-expectation operators $\PE$.
 
 \subsection{An additional operation on folded polynomials}
We define the following operation (and it's folded counterpart) which in the case of 
a multilinear polynomial corresponds (up to scaling) to the sum of a row of the SOS 
symmetric  matrix representation of the polynomial.
This will be useful in our result for non-negative polynomials.
\begin{definition}[Collapse]
  Let $f \in \pr{d}$ be a polynomial. The $k$-\defnt{collapse} of $f$, denoted as $\collapse{k}{f}$ is
  the degree $d-k$ polynomial $g$ given by, 
  \[
    g(x) = \sum_{\gamma \in \degmindex{d-k}} g_{\gamma} \cdot x^{\gamma} \quad \text{where} \quad
    g_{\gamma} = \sum_{\alpha \in \degmindex{k}} f_{\gamma+\alpha} \mper
  \]
  For a degree-$(d_1,d_2)$ folded polynomial $f$, we define $\collapse{k}{f}$ similarly as the
  degree-$(d_1-k,d_2)$ folded polynomial $g$ given by, 
  \[
    g(x) = \sum_{\gamma \in \degmindex{d_1-k}} \fold{g}{\gamma}(x) \cdot x^{\gamma} \quad \text{where}
    \quad \fold{g}{\gamma} = \sum_{\alpha \in \degmindex{k}} \fold{f}{\gamma+\alpha} \mper
  \]
\end{definition}

%% file: poly_opt/4.basic_results.tex
\section[Intermediate Results for Degree d (general) Polynomials and NNC Polynomials]{Intermediate Results for Polynomials in $\pr{d}$ and $\npr{d}$}\label{sec:poly:results}

\subsection{Reduction to Multilinear Polynomials}\label{sec:poly:gen-to-multi}

\begin{lemma}\label{lem:poly:split:gen:mult}
  Any homogeneous $n$-variate degree-$d$ polynomial $f(x)$ has a unique representation of the form 
  \[
    \sum_{\mathclap{\alpha \in \udmindex{d / 2}}}\multif{2\alpha}(x)\cdot x^{2 \alpha}
  \] 
  where for any $\alpha \in \udmindex{d / 2}$, $\multif{2\alpha}$ is a homogeneous multilinear  
  degree-$(d-2|\alpha|)$ polynomial. 
\end{lemma}

We would like to approximate $\ftwo{f}$ by individually approximating $\ftwo{\multif{2\alpha}}$ for each
multilinear polynomial $\multif{2\alpha}$. This section will establish the soundness of this goal. 

\subsubsection{Upper Bounding $\hssos{f}$ in terms of $\hssos{\multif{2\alpha}}$}

We first bound $\hssos{f}$ in terms of $\max_{\alpha \in \udmindex{d / 2}} \hssos{\multif{2\alpha}}$.
The basic intuition is that any matrix $M_f$ such that $\inparen{\xtensor{(d/2)}}^T \cdot M_f \cdot
\xtensor{(d/2)}$ for all $x$ (called a \defnt{matrix representation} of $f$) can be written as a sum of
matrices $M_{t,f}$ for each $t\leq d/2$, each of which is block-diagonal matrix with blocks corresponding
to matrix representations of the polynomials $M_{\multif{2\alpha}}$ for each $\alpha$ with $\abs{\alpha}
= 2t$.

\begin{lemma}\label{lem:poly:gen:mult:sp}
  Consider any homogeneous $n$-variate degree-$d$ polynomial $f(x)$. We have,
  \[
    \hssos{f}~\leq~\max_{\alpha \in \udmindex{d/2}} \frac{\hssos{\multif{2\alpha}}}{|\orbit{\alpha}|}\,(1+d/2) 
  \]
\end{lemma}

\begin{proof}
  We shall start by constructing an appropriate matrix representation $M_f$ of $f$ that will give us the
  desired upper bound on $\hssos{f}$. To this end, for any $\alpha \in \udmindex{d/2}$, let
  $M_{\multif{2\alpha}}$ be the matrix representation of $\multif{2\alpha}$ realizing
  $\hssos{\multif{2\alpha}}$. For any $0\leq t\leq d/2$, we define $M_{(t,f)}$ so that for any $\alpha
  \in \degmindex{t}$ and $I\in \orbit{\alpha}$, $M_{(t,f)}[I,I] :=
  M_{\multif{2\alpha}}/|\orbit{\alpha}|$, and $M_{(t,f)}$ is zero everywhere else. Now let $M_f :=
  \sum_{t\in [d/2]} M_{(t,f)}$. As for validity of $M_f$ as a representation of $f$ we have,
  \begin{align*}
    \iprod{M_f}{x^{\otimes d/2}{(x^{\otimes d/2})}^T}
    &= \sum_{0\leq t\leq \frac{d}{2}}\iprod{M_{(t,f)}}{x^{\otimes d/2}{(x^{\otimes d/2})}^T} \\
    &= \sum_{\alpha \in \udmindex{d/2}}\sum_{I\in \orbit{\alpha}}\iprod{M_{(|\alpha|,f)}[I,I]}
    {x^{\otimes (d/2-|\alpha|)}{(x^{\otimes (d/2-|\alpha|)})i}^T}x^{2\alpha} \\
    &= 
    \sum_{\alpha \in \udmindex{d/2}}\sum_{I\in \orbit{\alpha}}\frac{1}{|\orbit{\alpha}|}
    \iprod{M_{\multif{2\alpha}}}{x^{\otimes (d/2-|\alpha|)}{(x^{\otimes (d/2-|\alpha|)})}^T}x^{2\alpha} \\
    &= \sum_{\alpha \in \udmindex{d/2}} x^{2\alpha}\cdot \iprod{M_{\multif{2\alpha}}}
    {x^{\otimes (d/2-|\alpha|)} {(x^{\otimes (d/2-|\alpha|)})}^T}\\
    &= \sum_{\alpha \in \udmindex{d/2}} \multif{2\alpha}(x) x^{2\alpha} \\
    &= f(x)
  \end{align*}
	
  Now observe that $M_{(t,f)}$ is a block-diagonal matrix (up to simultaneous permutation of its rows
  and columns). Thus, we have $\|M_{(t,f)}\|\leq \max_{\alpha \in
  \degmindex{t}}\|M_{\multif{2\alpha}}\|/|\orbit{\alpha}|$. Thus on applying triangle inequality, we
  obtain $\|M_{f}\|\leq \max\limits_{\alpha \in \udmindex{d/2}}(1+d/2)\,\|M_{\multif{2\alpha}}\|/|\orbit{\alpha}|$
\end{proof}

\subsubsection[Lower Bounding Polynomials in terms of Multilinear Components (NNC)]{Lower Bounding
$\ftwo{f}$ in terms of $\ftwo{\multif{2\alpha}}$ (NNC)}
We first bound $\ftwo{f}$ in terms of $\max_{\alpha \in \udmindex{d / 2}} \ftwo{\multif{2\alpha}}$,
when every coefficient of $f$ is non-negative.  If $x^*$ is the optimizer of $\multif{2\alpha}$, then it
is easy to see that $x^* \geq 0$.  Setting $y =  x^* + \frac{\sqrt{\alpha}}{|\alpha|}$ ensures that $\| y
\|_2 \leq 2$ and $f(y)$ is large, since $f(y)$ recovers a significant fraction (up to a $2^{O(d)}\cdot
|\orbit{\alpha}|$ factor) of $\multif{2\alpha}(x^*)$. 

\begin{lemma}\label{lem:poly:nnc:mult:2}
  Let $f(x)$ be a homogeneous $n$-variate degree-$d$ polynomial 
  with non-negative coefficients. Consider any $\alpha\in \udmindex{d/2}$. Then 
  \[
    \ftwo{f}~\geq~\frac{\ftwo{\multif{2\alpha}}}{2^{O(d)}\,|\orbit{\alpha}|}.
  \]
\end{lemma}

\begin{proof}
  Consider any $0\leq t\leq d/2$, and any $\alpha\in \degmindex{t}$.  Let $x^*_{\alpha} :=
  \mathrm{argmax} \ftwo{\multif{2\alpha}}$ (note $x^*_{\alpha}$ must be non-negative). Let \[y^* :=
  x^*_\alpha~+~\frac{\sqrt{\alpha}}{\sqrt{t}} \] and let $x^* := y^*/\|y^*\|$. The second term is a
  unit vector since $\norm{2}{\sqrt{\alpha}}^2 = t$.  Thus, $\|y^*\| = \Theta(1)$ since $y^*$ is the sum
  of two unit vectors. This implies $f(x^*)\geq f(y^*)/2^{O(d)}$. Now we have, 
  \begin{align*}
    f(y^*) &= \sum_{\beta \in \udmindex{d/2}} \multif{2\beta}(y^*) \cdot {(y^*)}^{2\beta}
    &&(\text{by~\cref{lem:poly:split:gen:mult}}) \\
    &\geq \multif{2\alpha}(y^*)\cdot {(y^*)}^{2\alpha} &&(\text{by non-negativity of coefficients}) \\
    &\geq \multif{2\alpha}(y^*)~\frac{1}{t^{t}} \prod_{\ell\inv\supp{\alpha}} \alpha_\ell^{\alpha_\ell}
    &&(y^* \geq \frac{\sqrt{\alpha}}{\sqrt{t}}~\mbox{entry-wise})\\
    &\geq \multif{2\alpha}(y^*)~\frac{1}{2^{O(t)}\,t!} \prod_{\ell\in \supp{\alpha}}
    \alpha_\ell^{\alpha_\ell}  \\
    &\geq \multif{2\alpha}(y^*)~\frac{\prod_{\ell\in \supp{\alpha}}\alpha_\ell!}{2^{O(t)}\,t!} \\
    &\geq \multif{2\alpha}(y^*)~\frac{1}{2^{O(t)}\,|\orbit{\alpha}|} \\
    &\geq \multif{2\alpha}(x^*)~\frac{1}{2^{O(t)}\,|\orbit{\alpha}|} &&(y^*~\text{is entry-wise at
    least}~x^*) \\
    &= \frac{\ftwo{\multif{2\alpha}}}{2^{O(t)}\,|\orbit{\alpha}|}\mper
  \end{align*}
  This completes the proof. 
\end{proof}

\begin{theorem}\label{thm:poly:sos:nnc:mult}
  Consider any homogeneous $n$-variate degree-$d$ polynomial $f(x)$ with non-negative coefficients. Then 
  \[
    \frac{ \hssos{f} }{\ftwo{f}}~\leq~2^{O(d)}\max_{\alpha \in \udmindex{d/2}}
    \frac{\hssos{\multif{2\alpha}}}{\ftwo{\multif{2\alpha}}} \mper
  \]
\end{theorem}

\begin{proof}
  Combining~\cref{lem:poly:gen:mult:sp} with~\cref{lem:poly:nnc:mult:2} yields the claim. 
\end{proof}

We will next generalize the~\cref{thm:poly:sos:nnc:mult} by proving a more general version of
the~\cref{lem:poly:nnc:mult:2}. 

\subsubsection[Lower Bounding Polynomial in terms of Multilinear Components (General Case)]{Lower Bounding
  $\ftwo{f}$ in terms of $\ftwo{\multif{2\alpha}}$ (General  Case)}
We lower bound $\ftwo{f}$ in terms of $\ftwo{\multif{2\alpha}}$ for all polynomials. 
We will first recollect and establish some polynomial identities that will be used in the proof of the
generalized version of the~\cref{lem:poly:nnc:mult:2} (\ie~\cref{lem:poly:gen:mult:2}). 

\paragraph{Polynomial Identities}
\begin{lemma}[Chebyshev's Extremal Polynomial Inequality]\label{lem:poly:chebyshev}
  Let $p(x)$ be a univariate degree-$d$ polynomial and let $c_d$ be 
  it's leading coefficient. Then we have, $\max_{x\in [0,1]} |p(x)| \geq 2|c_d|/4^{d}$.
\end{lemma}

\begin{lemma}[\cite{HLZ10}]\label{sem:poly:decoupling}
  Let $x^1,x^2,\dots x^d \in \Re^{n}$ be arbitrary, let $\mathcal{A}\in \Re^{{[n]}^{d}}$ be a
  SoS-symmetric $d$-tensor, and let $\xi_1,\dots ,\xi_d$ be independent Rademacher random variables.
  Then 
  \[
    \Ex{\prod_{i\in [d]}\xi_i~\iprod{\mathcal{A}}{{(\xi_1 x^1 + \cdots + \xi_d x^d)}^{\otimes d}}}
    = d!\,\iprod{\mathcal{A}}{x^1\otimes \dots \otimes x^d}.
  \]
\end{lemma}

This lemma implies:

\begin{lemma}[\cite{HLZ10}]\label{lem:poly:decoupled:lower:bound}
  Let $\mathcal{A}$ be a SoS-symmetric $d$-tensor and let $f(x):= \iprod{\mathcal{A}}{x^{\otimes d}}$. Then  
  \[
    \ftwo{f}~\geq~\frac{1}{2^{O(d)}}\max_{\|x^i\|=1}\iprod{\mathcal{A}}{x^1\otimes \dots \otimes x^d}\mper
  \]
\end{lemma}

\begin{lemma}\label{lem:poly:complex:to:real}
  Let $f$ be an $n$-variate degree-$d$ homogeneous polynomial. Let $\ftwo{f}^c := 
  \max\limits_{\substack{z\in \Cp^{n}\\ \|z\|=1}}  |f(z)|$, then 
  \[
    \frac{\ftwo{f}^c}{2^{O(d)}} \leq \ftwo{f} \leq \ftwo{f}^c\mper
  \]
\end{lemma}

\begin{proof}
  Let $\mathcal{A}$ be the SoS-symmetric tensor representing $f$. Let $z^* = a^* + ib^*$ be the complex
  unit vector realizing $f(z^*) = \ftwo{f}^c$. Then we have,
  \begin{align*}
    f(z^*) 
    &= \iprod{\mathcal{A}}{(z^*)^{\otimes d}} \\
    &= \iprod{\mathcal{A}}{(a^*+ib^*)^{\otimes d}} \\
    &= \sum_{c^1,\dots c^d \in \{a^*,ib^*\}} \iprod{\mathcal{A}}{\bigotimes_{j\in [d]}c^j} \\
    \Rightarrow	\real{f(z^*)}
    &=\sum_{\substack{c^1,\dots c^d \in \{a^*,b^*\}, \\|\{j|c^j = b^*\}|\%4 = 0}}
    \iprod{\mathcal{A}}{\bigotimes_{j\in [d]}c^j}-
    \sum_{\substack{c^1,\dots c^d \in \{a^*,b^*\}, \\|\{j|c^j = b^*\}|\%4 = 2}}
    \iprod{\mathcal{A}}{\bigotimes_{j\in [d]}c^j}, \\ 
    \im{f(z^*)}	&= \sum_{\substack{c^1,\dots c^d \in \{a^*,b^*\}, \\|\{j|c^j = b^*\}|\%4 = 1}}
    \iprod{\mathcal{A}}{\bigotimes_{j\in [d]}c^j} -
    \sum_{\substack{c^1,\dots c^d \in \{a^*,b^*\}, \\|\{j|c^j = b^*\}|\%4 = 3}} 
    \iprod{\mathcal{A}}{\bigotimes_{j\in [d]}c^j}
  \end{align*}
  which implies that there exists $c^1, \dots ,c^d\in \{a^*,b^*\}$ such that
  $|\iprod{\mathcal{A}}{\bigotimes_{j\in [d]}c^j}| \geq \ftwo{f}^c/2^{O(d)}$. Lastly,
  applying~\cref{lem:poly:decoupled:lower:bound} implies the claim. 
\end{proof}

\paragraph{Some Probability Facts}
\begin{lemma}\label{lem:poly:bernoulli:moments}
  Let $X_1, \ldots X_k$ be i.i.d. $\mathrm{Bernoulli}(p)$ random variables. Then for any $t_1,\ldots
  ,t_k\in \mathds{N}$,\\
  $\Ex{X_1^{t_1}\dots X_k^{t_k}} = p^k$. 
\end{lemma}

\begin{lemma}\label{lem:poly:rand:rtou:moments}
  Let $\zeta$ be a uniformly random $p$-th root of unity. Then for any $t\in [p-1]$, $\Ex{\zeta^t} = 0$.
  Also, clearly $\Ex{\zeta^p} = 1$. 
\end{lemma}

We finally lower bound $\ftwo{f}$ in terms of $\multif{2\alpha}$.  Fix $\alpha \in \udmindex{d/2}$ and,
let $x^*$ be the optimizer of $\multif{2\alpha}$.  Setting \,$y =  x^* +
\frac{\sqrt{\alpha}}{|\alpha|}$\, as in the non-negative coefficient case does not work since terms from
$\multif{2\beta}$ may be negative.  We bypass this issue by first lower bounding $\ftwo{f}^c$ in terms
of $\multif{2\alpha}$ and using~\cref{lem:poly:complex:to:real}.  For $\ftwo{f}^c$, we use random roots of
unity and Bernoulli random variables, together with~\cref{lem:poly:chebyshev}, to extract nonzero contribution
only from the monomials that are multiples of $x^{\alpha}$ times multilinear parts. 

\begin{lemma}\label{lem:poly:gen:mult:2}
  Let $f(x)$ be a homogeneous $n$-variate degree-$d$ polynomial. Then for any $\alpha\in \udmindex{d/2}$, 
  \[
    \ftwo{f}~\geq~\frac{\ftwo{\multif{2\alpha}}}{2^{O(d)}\,|\orbit{\alpha}|}\mper
  \]
\end{lemma}

\begin{proof}
  Fix any $\alpha\in \udmindex{d/2}$, let $t:=|\alpha|$ and let $k:= d-2t$.  For any $i\in [n]$, let
  $\zeta_i$ be an independent and uniformly randomly chosen $(2\alpha_{i}+1)$-th root of unity, and let
  $\Xi$ be an independent and uniformly randomly chosen $(k+1)$-th root of unity. 
		
  Let $\bx := \mathrm{argmax} \ftwo{\multif{2\alpha}}$. Let $p\in [0,1]$ be a parameter to be fixed later,
  let $b_1,\dots ,b_n$ be i.i.d.  $\mathrm{Bernoulli}(p)$ random variables, let $\zeta:= (\zeta_1,\dots
  ,\zeta_n),~b:= (b_1,\dots ,b_n)$ and finally let  
  \[
    z~:=~\Xi\cdot b\circ \frac{1}{2\alpha + \one}\circ \bx~+~\frac{\zeta\circ
    \sqrt{\alpha}}{\sqrt{t}}\mper
  \] 
  Since $\sum_{\ell\in \supp{\alpha}} \alpha_\ell = t$ and roots of unity have 
  magnitude one, $z$ has length $O(1)$. Now consider any fixed $\gamma\in \degbmindex{k}$. 
  We have, 
  \begin{align*}\label{eq:poly:monomial:exp}	    
    &~\Ex{z^{2\alpha + \gamma}\cdot \Xi\cdot \prod_{i\in [n]} \zeta_{i}} \\
    &=\text{coefficient of~}\Xi^{k}\cdot \prod_{i\in [n]} \zeta_{i}^{2\alpha_i}
    \text{~in~}\Ex{z^{2\alpha+\gamma}}\qquad\qquad\text{(by~\cref{lem:poly:rand:rtou:moments})} \\
    &=\text{coefficient of~}\Xi^{k}\cdot \prod_{i\in [n]} \zeta_{i}^{2\alpha_i}\text{~in~}
    \Ex{\prod_{i\in [n]} \inparen{\zeta_i\cdot \frac{\sqrt{\alpha_i}}{\sqrt{t}}
    + \Xi\cdot \frac{b_i\cdot \bx_i}{2\alpha_i+1}}^{2\alpha_i+\gamma_i}} \\
    &=\prod_{i\in [n]}
    \text{~coefficient of~}\Xi^{\gamma_i}\cdot \zeta_{i}^{2\alpha_i} 
    \text{~in~}\Ex{\inparen{\zeta_i\cdot \frac{\sqrt{\alpha_i}}{\sqrt{t}}
    + \Xi\cdot \frac{b_i\cdot \bx_i}{2\alpha_i+1}}^{2\alpha_i+\gamma_i}}
    \qquad (\text{since}~\gamma\in \degbmindex{k}) \\
    &= p^{k}\cdot \prod_{\mathclap{i\in \supp{\alpha}}} \,\frac{\alpha_i^{\alpha_i}}{t^{\alpha_i}} \cdot 
    \bx_i^{\gamma_i} \qquad (\text{by~\cref{lem:poly:bernoulli:moments}}) \\
    &= p^{k}\cdot \bx^{\gamma} \cdot \prod_{\mathclap{i\in \supp{\alpha}}} \,\frac{\alpha_i^{\alpha_i}}{t^{\alpha_i}} 
  \end{align*}
  Thus we have, 
  \begin{align*}
    &~\Ex{f(z)\cdot \Xi\cdot 
    \prod_{i\in [n]} \zeta_{i}} \\
    &=~\sum_{\mathclap{\beta \in \degmindex{d}}} f_\beta \cdot\Ex{z^{\beta}\cdot \Xi\cdot 
    \prod_{i\in [n]} \zeta_{i}} \\ 
    &=~\sum_{\mathclap{\substack{\beta \in \degmindex{d} \\ \beta \geq 2\alpha}}}\,f_\beta \cdot
    \Ex{z^{\beta}\cdot \Xi\cdot \prod_{i\in [n]} \zeta_{i}} 
    &&(\text{by~\ref{lem:poly:rand:rtou:moments}}) \\
    &=~\sum_{\mathclap{\gamma \in \degbmindex{k}}}\,f_{2\alpha +\gamma} \cdot 
    \Ex{z^{2\alpha +\gamma}\cdot \Xi\cdot\prod_{i\in [n]} \zeta_{i}}~+~\sum_{\mathclap{\substack{
    \gamma \in \degmindex{k} \\ \gamma \not\leq \one}}}\,f_{2\alpha +\gamma} \cdot 
    \Ex{z^{2\alpha +\gamma}\cdot \Xi\cdot\prod_{i\in [n]} \zeta_{i}} \\
    &=~ 
		\sum_{
		\mathclap{
		\gamma \in \degbmindex{k} 
		}
		}
		\,f_{2\alpha +\gamma} \cdot 
		\Ex{z^{2\alpha +\gamma}\cdot \Xi\cdot 
	    \prod_{i\in [n]} \zeta_{i}} 
		~+~ 
		r(p) 
		&&\text{(by~\cref{lem:poly:bernoulli:moments})} \\
		& \text{where } r(p) \text{ is some univariate polynomial in $p$, s.t. } \mathrm{deg}(r)<k \\
		&=~ 
		\sum_{
		\mathclap{
		\gamma \in \degbmindex{k} 
		}
		}
		\,f_{2\alpha +\gamma} \cdot 
		p^{k}\cdot \bx^{\gamma} \cdot 
	    \prod_{\mathclap{i\in \supp{\alpha}}} \,\frac{\alpha_i^{\alpha_i}}{t^{\alpha_i}} 
	    ~~+~~ 
	    r(p)\\
	    &=~p^{k}\cdot \multif{2\alpha}(\bx)\cdot 
	    \prod_{\mathclap{i\in \supp{\alpha}}} \,\frac{\alpha_i^{\alpha_i}}{t^{\alpha_i}} 
	    +r(p)&&(\text{where } \mathrm{deg}(r)<k)
	\end{align*}
	
	Lastly we have, 
	\begin{align*}
		\ftwo{f}
		&\geq 
		\ftwo{f}^c\cdot 2^{-O(d)} 
		&&\text{by~\cref{lem:poly:complex:to:real}} \\
		&\geq 
		\max_{p\in [0,1]}
		\Ex{|f(z)|}\cdot 2^{-O(d)} 
		&&(\|z\|=O(1)) \\
		&= 
		\max_{p\in [0,1]}
		\Ex{\cardin{f(z)\cdot \Xi\cdot 
	    \prod_{i\in [n]} \zeta_{i}}}\cdot 2^{-O(d)}  \\
		&\geq 
		\max_{p\in [0,1]}
		\cardin{
		\Ex{f(z)\cdot \Xi\cdot 
	    \prod_{i\in [n]} \zeta_{i}}
		}
		\cdot 2^{-O(d)} \\
		&\geq 
		|\multif{2\alpha}(\bx)|\cdot 
	    \prod_{\mathclap{i\in \supp{\alpha}}} \,\frac{\alpha_i^{\alpha_i}}{t^{\alpha_i}} 
		\cdot 2^{-O(d)} 
		&&(\text{by Chebyshev:~\cref{lem:poly:chebyshev}}) \\
		&= 
		\ftwo{\multif{2\alpha}} \cdot 
	    \prod_{\mathclap{i\in \supp{\alpha}}} \,\frac{\alpha_i^{\alpha_i}}{t^{\alpha_i}} 
		\cdot 2^{-O(d)} \\
		&\geq 
		\frac{\ftwo{\multif{2\alpha}}}{|\orbit{\alpha}|}\cdot 2^{-O(d)} 
	\end{align*}
	This completes the proof. 
\end{proof}

In fact, the proof of~\cref{lem:poly:gen:mult:2} yields a more general result: 
\begin{lemma}[Weak Decoupling]\label{lem:poly:weak:decoupling}
    Let $f(x)$ be a homogeneous $n$-variate degree-$d$ polynomial.  Then for any $\alpha\in
    \udmindex{d/2}$ and any unit vector $y$,
    \[
      \ftwo{f}~\geq~y^{2\alpha}\cdot \ftwo{\multif{2\alpha}}\cdot 2^{-O(d)}\mper
    \]
\end{lemma}

We are finally able to establish the multilinear reduction result that is the focus of this section. 

\begin{theorem}\label{thm:poly:sos:gen:mult}
  Let $f(x)$ be a homogeneous $n$-variate degree-$d$ (for even $d$) polynomial. Then 
  \[
    \frac{\hssos{f}}{\ftwo{f}}~\leq~2^{O(d)}\max_{\alpha \in \udmindex{d/2}} 
    \frac{\hssos{\multif{2\alpha}}}{\ftwo{\multif{2\alpha}}} \mper
  \]
\end{theorem}

\begin{proof}
  Combining~\cref{lem:poly:gen:mult:sp} with~\cref{lem:poly:gen:mult:2} yields the claim. 
\end{proof}

\subsection{${(n/q)}^{d/4}$-Approximation for Non-negative Coefficient Polynomials}

\begin{theorem}\label{thm:poly:nnc:mult}
  Consider any homogeneous multilinear $n$-variate degree-$d$ 
  polynomial $f(x)$ with non-negative coefficients. We have, 
  \[
    \frac{\hssos{f}}{\ftwo{f}}~\leq~2^{O(d)}~\frac{n^{d/4}}{d^{d/4}}\mper
  \]
\end{theorem}
\begin{proof}
	Let $\sfM_f$ be the SoS-symmetric matrix representation of $f$. 
	Let $I^*=(i_1,\dots ,i_{d/2})\in [n]^{d/2}$ be the multi-index of 
	any row of $\sfM_f$ with maximum row sum. Let $S_I$ for $I\in 
	[n]^{d/2}$, denote the sum of the row $I$ of $\sfM_f$. By Perron-Frobenius
	theorem, $\|\sfM_f\|\leq S_{I^*}$. Thus, $\hssos{f}
	\leq S_{I^*}$. \medskip
	
	We next proceed to bound $\ftwo{f}$ from below. To this end, let 
	$x^* := y^*/\|y^*\|$ where, 
	\[
		y^* := \frac{\one}{\sqrt{n}} ~+~ 
		\frac{1}{\sqrt{d/2}}\sum_{i\in I^*} e_{i}
	\]
	Since $f$ is multilinear, $I^*$ has all distinct elements, and so 
	the second term in the definition of $y^*$ is of unit length. 
	Thus $\|y^*\| = \Theta(1)$, which implies that 
	$\ftwo{f}\geq f(x^*)\geq f(y^*)/2^{O(d)}$. Now we have, 
	\begin{align*}
		f(y^*) 
		&= 
		((y^*)^{\otimes d/2})^T \sfM_f\,(y^*)^{\otimes d/2} \\
		&\geq 
		\sum_{I\in \orbit{I^*}} \frac{1}{(nd)^{d/4}}~
		e^T_{I(1)}\otimes \dots 
		\otimes e^T_{I(d/2)} \,\sfM_f\,\one^{\otimes d/2} 
		&&(\text{by non-negativity of } \sfM_f)\\
		&= 
		\sum_{I\in \orbit{I^*}} \frac{1}{(nd)^{d/4}}~
		e^T_{I}\,\sfM_f\,\one ~(\in \Re^{[n]^{d/2}}) \\
		&=
		\sum_{I\in \orbit{I^*}} 
		\frac{S_{I}}{(nd)^{d/4}} \\
		&= 
		\sum_{I\in \orbit{I^*}} 
		\frac{S_{I^*}}{(nd)^{d/4}}
		&&(\text{by SoS-symmetry of } \sfM_f) \\
		&=
		\frac{(d/2)!S_{I^*}}{(nd)^{d/4}} 
		&&(|\orbit{I^*}|=(d/2)!\text{ by multilinearity of } f)\\
		&\geq 
		\frac{d^{d/4}S_{I^*}}{n^{d/4}\,2^{O(d)}} ~
		\geq 
		\frac{d^{d/4}\hssos{f}}{n^{d/4}\,2^{O(d)}}.
	\end{align*}
	This completes the proof. 
\end{proof}

\begin{theorem}\label{thm:poly:nnc:d/4}
	Let $f(x)$ be a homogeneous $n$-variate degree-$d$ 
	polynomial  with non-negative coefficients. 
	Then for any even $q$ such that $d$ divides $q$, 
	\[
		\frac{(\hssos{f^{q/d}})^{d/q}}{\ftwo{f}} ~\leq~ 
		2^{O(d)}~\frac{n^{d/4}}{q^{d/4}}.
	\]
\end{theorem}

\begin{proof}
	Applying~\cref{thm:poly:sos:nnc:mult} to $f^{q/d}$ and combining 
	this with~\cref{thm:poly:nnc:mult} yields the claim. 
\end{proof}

\subsection{${(n/q)}^{d/2}$-Approximation for General Polynomials}

\begin{theorem}\label{thm:poly:gen:mult}
  Consider any homogeneous multilinear $n$-variate degree-$d$ (for even $d$)	polynomial $f(x)$. We have, 
  \[
    \frac{\hssos{f}}{\ftwo{f}}~\leq~2^{O(d)}~\frac{n^{d/2}}{d^{d/2}}\mper
  \]
\end{theorem}
\begin{proof}
Let $\sfM_f$ be the SoS-symmetric matrix representation of $f$, i.e. 
\[
\sfM_f[I,J] = \frac{f_{\alpha(I)+\alpha(J)}}{\abs{\orbit{\alpha(I)+\alpha(J)}}}.
\] 
By the Gershgorin circle theorem, we can bound
$\norm{2}{\sfM_f}$, and hence $\hssos{f}$ by $n^{d/2} \cdot (\max_{\beta}\abs{f_{\beta}}/d!)$.
Here, we use the multilinearity of $f$.
On the other hand for a multilinear polynomial, 
using $x = \beta/\sqrt{\abs{\beta}}$ (where $|\beta|=d$ by multilinearity),
gives $\ftwo{f} \geq d^{-d/2} \cdot \abs{f_{\beta}}$. Thus, we easily get
\[
\hssos{f} ~\leq~ \frac{d^{d/2}}{d!} \cdot n^{d/2} \cdot \ftwo{f} ~=~ 2^{O(d)}~\frac{n^{d/2}}{d^{d/2}} \mper
\]
\end{proof}

\iffalse
\begin{proof}
	Let $\sfM_f$ be the SoS-symmetric matrix representation of $f$. 
	Let $I^*=(i_1,\dots ,i_{d/2}),
	J^*=(j_1,\dots ,j_{d/2})\in [n]^{d/2}$ respectively be the row 
	and column multi-indices of an entry of $\sfM_f$ with 
	maximum absolute value. By Gershgorin circle theorem, 
	$\|\sfM_f\|\leq n^{d/2}\,|\sfM_f[I^*,J^*]|$. Thus,
	$\hssos{f}\leq n^{d/2}\,|\sfM_f[I^*,J^*]|$. \medskip
	
	We next proceed to bound $\ftwo{f}$ from below. To this end, let 
	$x^* := y^*/\|y^*\|$ where, 
	\[
		y^* := 
		\frac{1}{\sqrt{d}}\inparen{\mathrm{sgn}(\sfM[I^*,J^*])\,e_{i_1}
		~+
		\sum_{i\in I^*, i\neq i_1} e_{i} ~+
		\sum_{j\in J^*} e_{j}}
	\]
	Since $f$ is multilinear, $i_1, \dots, i_{d/2}, j_1, \dots, j_{d/2}$ are all distinct, and so $y^*$ is 
	of unit length. Now we have, 
	\begin{align*}
		& \quad f(y^*) \\
		&= 
		((y^*)^{\otimes d/2})^T \sfM_f\,(y^*)^{\otimes d/2} \\
		&= 
		\sum_{I \oplus J\in \orbit{I^* \oplus J^*}} 
		\frac{\mathrm{sgn}(\sfM[I^*,J^*])}{d^{d/2}}~
		e^T_I \sfM_f\, e_{J}
		&&(\text{by multilinearity of } f)\\
		&= 
		\sum_{I \oplus J\in \orbit{I^* \oplus J^*}} 
		\frac{|\sfM[I,J]|}{d^{d/2}} \\
		&= 
		\sum_{I \oplus J\in \orbit{I^* \oplus J^*}} 
		\frac{|\sfM[I^*,J^*]|}{d^{d/2}}
		&&(\text{by SoS-symmetry of } \sfM_f) \\
		&=
		\frac{d! ~|\sfM[I^*,J^*]|}{d^{d/2}} 
		&&(|\orbit{I^* \oplus J^*}|=d!
		\text{ by multilinearity of } f)\\
		&\geq 
		\frac{d^{d/2} ~|\sfM[I^*,J^*]|}{2^{O(d)}} ~
		\geq 
		\frac{d^{d/2} ~\hssos{f}}{n^{d/2}~2^{O(d)}}.
	\end{align*}
	This completes the proof. 
\end{proof}
\fi

\begin{theorem}\label{thm:poly:gen:d/2}
  Let $f(x)$ be a homogeneous $n$-variate degree-$d$ 
  polynomial, and assume that $2d$ divides $q$. Then  
  \[
    \frac{(\hssos{f^{q/d}})^{d/q}}{\ftwo{f}}~\leq~2^{O(d)}~\frac{n^{d/2}}{q^{d/2}}\mper
	\]
\end{theorem}

\begin{proof}
	Applying~\cref{thm:poly:sos:gen:mult} to $f^{q/d}$ and combining 
	this with~\cref{thm:poly:gen:mult} yields the claim. 
\end{proof}

\subsection{$\sqrt{m/q}$-Approximation for $m$-sparse polynomials}

\begin{lemma}\label{lem:poly:sparse:mult}
  Consider any homogeneous multilinear $n$-variate degree-$d$ (for even $d$)
  polynomial $f(x)$ with $m$ non-zero coefficients. We have, 
  \[
    \frac{\hssos{f}}{\ftwo{f}}~\leq~2^{O(d)}\sqrt{m}\mper
  \]
\end{lemma}
\begin{proof}
  Let $\sfM_f$ be the SoS-symmetric matrix representation of $f$, \ie\
  \[
    \sfM_f[I,J] = \frac{f_{\alpha(I)+\alpha(J)}}{\abs{\orbit{\alpha(I)+\alpha(J)}}}.
  \] 
Now $\hssos{f}\leq \|\sfM_f\| \leq \|\sfM_f\|_F$. Thus, we have:
\begin{align*}
    \|\sfM_f\|^2_F 
    &= 
    \sum_{I,J\in [n]^{d/2}} \sfM_f[I,J]^2 \\
    &= 
    \sum_{\beta \in \degbmindex{d}} 
    \frac{f_{\beta}^{2}}{\abs{\orbit{\beta}}} \\
    &=  
    \sum_{\beta \in \degbmindex{d}} 
    \frac{f_{\beta}^{2}}{d!} \\
    &\leq 
    \frac{m}{d!}\cdot \max_{\beta}\abs{f_{\beta}}
\end{align*}

On the other hand, since $f$ is multilinear, using $x = \beta/\sqrt{\abs{\beta}}$ (where $|\beta|=d$ by
multilinearity), implies $\ftwo{f} \geq d^{-d/2} \cdot \abs{f_{\beta}}$ for any $\beta$.  This implies
the claim. 
\end{proof}

\begin{theorem}\label{thm:poly:sparse:q}
	Let $f(x)$ be a homogeneous $n$-variate degree-$d$ 
	polynomial with $m$ non-zero coefficients, and assume that $2d$ divides $q$. Then
	\[
		\frac{(\hssos{f^{q/d}})^{d/q}}{\ftwo{f}} ~\leq~ 
		2^{O(d)}\sqrt{m/q}.
	\]
\end{theorem}

\begin{proof}
	Combining~\cref{thm:poly:sos:gen:mult} and~\cref{lem:poly:sparse:mult}, yields that for any degree-$q$ 
	homogeneous polynomial $g$ with sparsity $\overbar{m}$, we have 
	\[
          \frac{(\hssos{g})}{\ftwo{g}}~\leq~2^{O(q)}\sqrt{\overbar{m}}.
	\]
	Lastly, taking $g = f^{q/d}$ and observing that the sparsity of $g$ is at most 
	$\multichoose{m}{q/d}$ implies the claim. 
\end{proof}

%% file: poly_opt/5.folding.tex
\section{Improved Approximations via Folding}\label{sec:poly:folding}
Recall that we call a folded polynomial multilinear if all its monomials are multilinear. In particular,
there's no restriction on the folds of the polynomial. 

\begin{lemma}[Folded Analogue of~\cref{lem:poly:split:gen:mult}]\label{lem:poly:fold:split:gen:mult}
	Let $\FPR{d_1}{d_2} \ni f(x) :=
	\sum_{\beta\in \degmindex{d_1}} \fold{f}{\beta}(x)\cdot x^\beta$ be a 
	$(d_1,d_2)$-folded polynomial.
	$f$ can be written as 
	\[
		\sum_{\alpha\in \udmindex{d_1\!/2}} 
		\multif{2\alpha}(x) \cdot x^{2\alpha}
	\] 
	where for any $\alpha\in \udmindex{d_1\!/2}$, $\multif{2\alpha}(x)$ is a 
	multilinear $(d_1-2|\alpha|,d_2)$-folded polynomial. 
\end{lemma} 

\begin{proof}
Simply consider the folded polynomial 
	\[
	\multif{2\alpha}(x) = 
	\sum_{\gamma\in \degbmindex{d_1-2|\alpha|}} 
	\fold{(\multif{2\alpha})}{\gamma} \cdot x^{\gamma}
	\]
	where 
	$
		\fold{(\multif{2\alpha})}{\gamma}
		= 
		\fold{f}{2\alpha+\gamma}
	$.
\end{proof}

\subsection{Reduction to Multilinear Folded Polynomials}

Here we will prove a generalized version of~\cref{lem:poly:gen:mult:sp}, 
which is a generalization in two ways; firstly it allows for folds 
instead of just coefficients, and secondly it allows a more 
general set of constraints than just the hypersphere since we will need 
to add some additional non-negativity constraints for the case of non-negative 
coefficient polynomials (so that $\hscsos{C}{}$ satisfies monotonicity 
over NNC polynomials which will come in handy later). 

Recall that $\hscsos{C}{}$ is defined in~\cref{sec:poly:hscsos} and that 
$\ftwo{f}$ and $\hscsos{C}{f}$ for a folded polynomial $f$, are applied 
to the unfolding of $f$. 

\subsubsection{Relating $\hscsos{C}{f}$ to $\hscsos{C}{\multif{2\alpha}}$}
\begin{lemma}[Folded Analogue of~\cref{lem:poly:gen:mult:sp}]\label{lem:poly:fold:gen:mult:sos}
	Let $C$ be a system of polynomial constraints of the form 
	$\{\|x\|_2^2 = 1\}\cup C'$ where $C'$ is a moment non-negativity 
	constraint set. Let $f \in \FPR{d_1}{d_2}$ 
	be a $(d_1,d_2)$-folded polynomial. We have, 
	\[
	\hscsos{C}{f} ~\leq~ 
	\max_{\alpha\in \udmindex{d_1\!/2}} 
	\frac{\hscsos{C}{\multif{2\alpha}}}{|\orbit{\alpha}|}\,(1+d_1/2)
	\]
\end{lemma}

\begin{proof}
	Consider any degree-$(d_1+d_2)$ pseudo-expectation operator $\PExc$. 
	We have, 
	\begin{align*}
		\PExc{f} 
		&= 
		\sum_{\mathclap{\alpha\in \udmindex{d_1\!/2}}} 
		\PExc{\multif{2\alpha}(x)\cdot x^{2\alpha}}
		&&\text{(by~\cref{lem:poly:fold:split:gen:mult})} \\
		&\leq 
		\sum_{\mathclap{\alpha\in \udmindex{d_1\!/2}}} 
		\PExc{x^{2\alpha}} \cdot \hscsos{C}{\multif{2\alpha}}
		&&\text{(by~\cref{lem:poly:sos:replace})} \\
		&= 
		\sum_{0\leq t\leq \frac{d_1}{2}}
		\sum_{\alpha\in \degmindex{t}} 
		\PExc{x^{2\alpha}}\cdot 
		\hscsos{C}{\multif{2\alpha}} \\
		&= 
		\sum_{0\leq t\leq \frac{d_1}{2}}
		\sum_{\alpha\in \degmindex{t}} 
		\PExc{|\orbit{\alpha}|x^{2\alpha}}\cdot 
		\frac{\hscsos{C}{\multif{2\alpha}}}{|\orbit{\alpha}|} \\
		&\leq 
		\sum_{0\leq t\leq \frac{d_1}{2}}
		\sum_{\alpha\in \degmindex{t}} 
		\PExc{|\orbit{\alpha}|x^{2\alpha}}\cdot 
		\max_{\mathclap{\beta\in \udmindex{d_1\!/2}}} \,
		\frac{\hscsos{C}{\multif{2\beta}}}{|\orbit{\beta}|}
		&&(\PExc{x^{2\alpha}}\geq 0) \\
		&= 
		\sum_{0\leq t\leq \frac{d_1}{2}}
		\PExc{
		\sum_{\alpha\in \degmindex{t}} 
		|\orbit{\alpha}|x^{2\alpha}
		}
		\cdot 
		\max_{\mathclap{\beta\in \udmindex{d_1\!/2}}} \,
		\frac{\hscsos{C}{\multif{2\beta}}}{|\orbit{\beta}|} \\
		&= 
		\sum_{0\leq t\leq \frac{d_1}{2}}
		\PExc{\|x\|_2^{2t}}
		\cdot 
		\max_{\beta\in \udmindex{d_1\!/2}} 
		\frac{\hscsos{C}{\multif{2\beta}}}{|\orbit{\beta}|}  \\
		&= 
		\sum_{0\leq t\leq \frac{d_1}{2}} ~
		\max_{\beta\in \udmindex{d_1\!/2}} 
		\frac{\hscsos{C}{\multif{2\beta}}}{|\orbit{\beta}|} \\
		&= ~~
		\max_{\beta\in \udmindex{d_1\!/2}}
		\frac{\hscsos{C}{\multif{2\beta}}}{|\orbit{\beta}|} 
		(1+d_1/2) 
	\end{align*}
\end{proof}

\subsection{Relating Evaluations of $f$ to Evaluations of $\multif{2\alpha}$}

Here we would like to generalize~\cref{lem:poly:nnc:mult:2} and~\cref{lem:poly:gen:mult:2} to 
allow folds, however for technical reasons related to decoupling of the domain of the 
folds from the domain of the monomials of a folded polynomial, we instead generalize 
claims implicit in the proofs of~\cref{lem:poly:nnc:mult:2} and~\cref{lem:poly:gen:mult:2}. 

Let $f \in \FPR{d_1}{d_2}$ be a $(d_1,d_2)$-folded polynomial. Recall that an evaluation 
of a folded polynomial treats the folds as coefficients and only substitutes values 
in the monomials of the folded polynomial. Thus, for any fixed $y\in \Re^{n}$, $f(y)$ 
(sometimes denoted by $(f(y))(x)$ for contextual clarity) is a degree-$d_2$ polynomial 
in $x$, i.e. $f(y)\in \PR{d_2}$. 

\begin{lemma}[Folded Analogue of~\cref{lem:poly:nnc:mult:2}]\label{lem:poly:fold:nnc:mult:eval}
    Let $f \in \NFPR{d_1}{d_2}$ be a $(d_1,d_2)$-folded polynomial whose folds have non-negative 
    coefficients. Then for any $\alpha\in\udmindex{d_1/2}$ and any $y\geq 0$, 
    \[
        \pth{f\pth{y+\frac{\sqrt{\alpha}}{\sqrt{|\alpha|}}}}\!(x) 
        ~\geq~ 
        \frac{(\multif{2\alpha}(y))(x)}{|\orbit{\alpha}|} \cdot 2^{-O(d_1)}
    \]
    where the ordering is coefficient-wise. 
\end{lemma}

\begin{proof}
    Identical to the proof of~\cref{lem:poly:nnc:mult:2}. 
\end{proof}

\begin{lemma}[Folded Analogue of~\cref{lem:poly:gen:mult:2}]\label{lem:poly:fold:gen:mult:eval}
    Let $f \in \FPR{d_1}{d_2}$ be a $(d_1,d_2)$-folded polynomial. 
    Consider any $\alpha\in\udmindex{d_1/2}$ and any $y$, and let 
    \[
        z \quad := \quad
        \Xi \cdot y\circ \frac{1}{2\alpha+\one}\circ b ~~+~ 
		\frac{\sqrt{\alpha}\circ \zeta}{\sqrt{|\alpha|}} 
	\]  
	where $\Xi$ is an independent and uniformly randomly chosen 
	$(d_1-2|\alpha|+1)$-th root of unity, and for any $i\in [n]$,  
	$\zeta_i$ is an independent and uniformly randomly chosen 
	$(2\alpha_i+1)$-th root of unity, and $b_i$ is an 
	independent $\mathrm{Bernoulli}(p)$ random variable ($p$ is an arbitrary  
	parameter in $[0,1]$). Then 
    \[
        \Ex{(f(z))(x)\cdot \Xi\cdot 
	    \prod_{i\in [n]} \zeta_{i}} 
        ~=~ 
        p^{\,d_1-2|\alpha|}\cdot \frac{(\multif{2\alpha}(y))(x)}{|\orbit{\alpha}|}\cdot 2^{-O(d_1)}
        ~+~
        r(p)
    \]
    where $r(p)$ is a univariate polynomial in $p$ with degree less than $d_1 - 2|\alpha|$ 
    (and whose coefficients are in $\PR{d_2}$). 
\end{lemma}

\begin{proof}
    This follows by going through the proof of~\cref{lem:poly:gen:mult:2} for every fixed $x$. 
\end{proof}

\subsection{Bounding $\hscsos{C}{}$ of Multilinear Folded Polynomials}

Here we bound $\hscsos{C}{}$ of a multilinear folded polynomial in terms of 
properties of the polynomial that are inspired by treating the folds as coefficients 
and generalizing the coefficient-based approximations for regular (non-folded) polynomials 
from~\cref{thm:poly:gen:mult} and~\cref{thm:poly:nnc:mult}.

\subsubsection{General Folds: Bounding $\hssos{}$ in terms of $\hssos{}$ of the "worst" fold}
Here we will give a folded analogue of the proof of~\cref{thm:poly:gen:mult} wherein 
we used Gershgorin-Circle theorem to bound SOS value in terms of the 
max-magnitude-coefficient. 

\begin{lemma}[Folded Analogue of Gershgorin Circle Bound on Spectral Radius]\label{lem:poly:fold:gen:coefficient:sos}
	For even $d_1,d_2$, let $d=d_1+d_2$, let $f \in \FPR{d_1}{d_2}$ 
	be a multilinear $(d_1,d_2)$-folded polynomial. We have, 
	\[
		\hssos{f} ~\leq~ 
		2^{O(d)}~\frac{n^{d_1/2}}{d_1^{\,d_1}}\,
		\max_{\gamma\in \degbmindex{d_1}}
		\|\fold{f}{\gamma}\|_{sp}.
	\]
\end{lemma}

\begin{proof}
	Since $\hssos{f}\leq \|f\|_{sp}$, it is sufficient to bound 
	$\|f\|_{sp}$. 
	
	Let $M_{\fold{f}{\gamma}}$ be the matrix representation of 
	$\fold{f}{\gamma}$ realizing $\|\fold{f}{\gamma}\|_{sp}$. Let 
	$M_{f}$ be an $[n]^{d_1/2}\times [n]^{d_1/2}$ block matrix 
	with $[n]^{d_2/2}\times [n]^{d_2/2}$ size blocks, where for 
	any $I,J\in [n]^{d_1/2}$ the block of $M_f$ at index $(I,J)$ is 
	defined to be $\frac{1}{d_1!}\cdot M_{\fold{f}{\mi{I}+\mi{J}}}$. 
	Clearly $M_f$ (interpreted as an $[n]^{d/2}\times [n]^{d/2}$) is 
	a matrix representation of the unfolding of $f$ since $f$ is a 
	multilinear folded polynomial. Lastly, applying Block-Gershgorin 
	circle theorem to $M_f$ and upper bounding the sum of spectral norms 
	over a block row by $n^{d_1/2}$ times the max term implies the claim. 
\end{proof}

\subsubsection{Non-Negative Coefficient Folds: Relating SoS Value to the SoS Value of the $d_1/2$-collapse}
Observe that in the case of a multilinear degree-$d$ polynomial, the $d/2$-collapse corresponds (up to scaling) 
to the sum of a row of the SOS symmetric  matrix representation of the polynomial. We will next develop 
a folded analogue of the proof of~\cref{thm:poly:nnc:mult} wherein we employed Perron-Frobenius theorem to 
bound SOS value in terms of the $d/2$-collapse. 

\medskip

The proof here however, is quite a bit more subtle than in the general case above. This is because 
one can apply the block-matrix analogue of Gershgorin theorem (due to Feingold \etal \cite{feingold1962block}) 
to a matrix representation of the folded polynomial (whose spectral norm is an upper bound on $\hssos{}$) in the 
general case. Loosely speaking, this corresponds to bounding $\hssos{f}$ in terms of 
\[
	\max_{\gamma\in\degbmindex{k}} \sum_{\theta\in \degbmindex{k}} \hssos{\fold{f}{\gamma+\theta}}
\]
where $k=d_1/2$. This however is not enough in the nnc case as in order to win the $1/2$ in the exponent, 
one needs to relate $\hscsos{C}{f}$ to 
\[
	\max_{\gamma\in\degbmindex{k}} \hssos{ \sum_{\theta\in \degbmindex{k}} \fold{f}{\gamma+\theta}}.
\]
This however, cannot go through Block-Gershgorin since it is \textbf{not} true that the spectral norm of a non-negative block 
matrix is upper bounded by the max over rows of the spectral norm of the sum of blocks in that row. It instead, can only be upper bounded by the max over rows of the sum of spectral norms of the blocks in that row. 

To get around this issue, we skip the intermediate step of bounding $\hscsos{C}{f}$ by the spectral norm of a matrix 
and instead prove the desired relation directly through the use of pseudoexpectation operators. This involved first 
finding a pseudo-expectation based proof of Gershgorin/Perron-Frobenius bound on spectral radius that generalizes to folded 
polynomials in the right way. 

\begin{lemma}[Folded analogue of Perron-Frobenius Bound on Spectral Radius]\label{lem:poly:fold:nnc:collapse:sos}
	For even $d_1 = 2k$, let $f \in \NFPR{d_1}{d_2}$ be a multilinear 
	$(d_1,d_2)$-folded polynomial whose folds have non-negative 
	coefficients. Let $C$ be the system of polynomial constraints 
	given by $\{\|x\|_2^2 = 1; \forall \beta\in \degmindex{d_2}, 
	x^{\beta}\geq 0\}$. We have, 
	\[
		\hscsos{C}{f} ~\leq~ 
		\max_{\gamma\in \degbmindex{k}}
		\hscsos{C}{\fold{g}{\gamma}}\cdot \frac{1}{k!}
	\]
	where 
	\[
		\fold{g}{\gamma}(x) := 
		\fold{\collapse{k}{f}}{\gamma}
		=
		\sum_{
		\mathclap{
		\substack{
		\theta\leq \one - \gamma \\
		\theta\in \degmindex{k}
		}
		}
		} 
		\fold{f}{\gamma+\theta}(x).
	\]
\end{lemma}

\begin{proof}	
	Consider any pseudo-expectation operator $\PExc$ of degree 
	at least $d_1+d_2$. Note that since $\PExc$ satisfies 
	$\{\forall \beta\in \degmindex{d_2}, x^{\beta}\geq 0\}$, 
	by linearity $\PExc$ must also satisfy $\{h\geq 0\}$ for any 
	$h\in \NPR{d_2}$ - a fact we will use shortly. 
	
	Since $f$ is a multilinear folded polynomial, $\fold{f}{\alpha}$
	is only defined when $0\leq \alpha\leq \one$. If 
	$\alpha\not\leq \one$, we define $\fold{f}{\alpha}:= 0$
	We have, 
	\begin{align*}
		\PExc{f} 
		&= 
		\sum_{\mathclap{\alpha\in \degbmindex{d_1}}} 
		\PExc{\fold{f}{\alpha}\cdot x^{\alpha}} 
		&&(f \text{ is a multilinear folded polynomial})\\
		&= 
		\sum_{I\in [n]^{k}} 
		\sum_{J\in [n]^{k}}
		\PExc{\fold{f}{\mi{I}+\mi{J}}\cdot x^{I}x^{J}} 
		\cdot \frac{1}{d_1!}
		&&\text{(by multilinearity)} \\
		&\leq 
		\sum_{I\in [n]^{k}} 
		\sum_{J\in [n]^{k}}
		\PExc{\fold{f}{\mi{I}+\mi{J}}\cdot 
		\frac{(x^{I})^2+(x^{J})^2}{2}} 
		\cdot \frac{1}{d_1!}
		&&(\PExc \text{ satisfies } \fold{f}{\alpha} \geq 0) \\
		&= 
		\sum_{I\in [n]^{k}} 
		\sum_{J\in [n]^{k}}
		\PExc{\fold{f}{\mi{I}+\mi{J}}\cdot (x^{I})^2} 
		\cdot \frac{1}{d_1!} \\
		&= 
		\sum_{I\in [n]^{k}} 
		\PExc{(x^{I})^2\cdot 
		\sum_{J\in [n]^{k}}\fold{f}{\mi{I}+\mi{J}} } 
		\cdot \frac{1}{d_1!} \\
		&= 
		\sum_{I\in [n]^{k}} 
		\PExc{(x^{I})^2\cdot 
		\sum_{
		\mathclap{
		\substack{
		\theta\leq \one - \mi{I} \\
		\theta\in \degmindex{k}
		}
		}
		} 
		\fold{f}{\mi{I}+\theta} 
		} 
		\cdot \frac{k!}{d_1!} 
		&&\text{(by multilinearity)} \\
		&= 
		\sum_{I\in [n]^{k}} 
		\PExc{(x^{I})^2\cdot 
		g_{\mi{I}}
		} 
		\cdot \frac{k!}{d_1!} \\
		&\leq 
		\sum_{I\in [n]^{k}} 
		\PExc{(x^{I})^2}
		\cdot \hscsos{C}{\fold{g}{\mi{I}}}
		\cdot \frac{1}{k!}
		&&\text{(by~\cref{lem:poly:sos:replace})}\\
		&\leq 
		\sum_{I\in [n]^{k}} 
		\PExc{(x^{I})^2}
		\cdot \max_{\gamma\in \degbmindex{k}}\hscsos{C}{\fold{g}{\gamma}}
		\cdot \frac{1}{k!}
		&&(\PExc{(x^{I})^2} \geq 0)\\
		&= 
		\PExc{\|x\|_2^{d_1}}
		\cdot \max_{\gamma\in \degbmindex{k}}\hscsos{C}{\fold{g}{\gamma}}
		\cdot \frac{1}{k!} \\
		&= 
		\max_{\gamma\in \degbmindex{k}}\hscsos{C}{\fold{g}{\gamma}}
		\cdot \frac{1}{k!} 
	\end{align*}
\end{proof}

We are finally equipped to prove the main results of this section.

\subsection[Final Approxiation Result for NNC Polynomials]
{$(n/q)^{d/4 - 1/2}$-Approximation for Non-negative Coefficient Polynomials} 

\begin{theorem}\label{thm:poly:nnc:d/4-1/2}
	Consider any $f \in \NPR{d}$ for $d\geq 2$, and any $q$ 
	divisible by $2d$. Let $C$ be the system of 
	polynomial constraints given by $\{\|x\|_2^2 = 1; \forall 
	\beta\in \degmindex{2q/d}, x^{\beta}\geq 0\}$. Then we have, 
	\[
		\frac{\hscsos{C}{f^{q/d}}^{d/q}}{\ftwo{f}} ~\leq~ 
		2^{O(d)}~\frac{n^{d/4-1/2}}{q^{d/4-1/2}}.
	\]
\end{theorem}

\begin{proof}
    Let $h$ be any $(d-2,2)$-folded polynomial whose unfolding yields 
	$f$ and whose folds have non-negative coefficients and 
	let $s$ be the $(\barq,2q/d)$-folded 
	polynomial given by $h^{q/d}$ where $\barq:= (d-2)q/d$. 
	Finally, consider any $\alpha\in \udmindex{\barq/2}$ and let  
	$S_{2\alpha}$ be the multilinear component of $s$ as defined in~\cref{lem:poly:fold:split:gen:mult}. 
	We will establish that for any $\gamma\in \degbmindex{k}$ (where 
	$k:= \barq/2 - |\alpha|$), 
	\begin{equation}\label{eq:poly:nnc:rte}
		\ftwo{f}^{q/d} 
		\geq  
		\frac{2^{-O(q)}\cdot \hscsos{C}{\fold{\collapse{\barq/2 - |\alpha|}{S_{2\alpha}}}{\gamma}}}
		{(\barq/2 - |\alpha|)^{\barq/4 - |\alpha|/2}
		\cdot |\orbit{\alpha}|\cdot n^{\barq/4-|\alpha|/2}}
	\end{equation}
    which on combining with the application of~\cref{lem:poly:fold:gen:mult:sos} 
    to $s$ and its composition with~\cref{lem:poly:fold:nnc:collapse:sos}, yields 
    the claim. To elaborate, we apply~\cref{lem:poly:fold:gen:mult:sos} to $s$ with 
    $d_1 = \barq, d_2=2q/d$ and then for every $\alpha\in \udmindex{\barq/2}$ 
    we apply~\cref{lem:poly:fold:nnc:collapse:sos} with 
    $d_1 = \barq-2|\alpha|, d_2 = 2q/d$, to get 
    \[
        \hscsos{C}{f^{q/d}} 
        =
        \hscsos{C}{s}
        \leq 
        2^{O(q)}\cdot 
        \max_{\alpha\in \udmindex{\barq/2}} ~
        \max_{\gamma\in\degbmindex{\barq/2 - |\alpha|}} 
        \frac{\hscsos{C}{\fold{\collapse{\barq/2 - |\alpha|}{S_{2\alpha}}}{\gamma}}}
        {(\barq/2 - |\alpha|)!\cdot |\orbit{\alpha}|}
    \]
    which on combining with~\cref{eq:poly:nnc:rte} yields the claim. 
    \medskip 

    It remains to establish~\cref{eq:poly:nnc:rte}. 
    So fix any $\alpha, \gamma$ satisfying the above conditions. 
    Let $t:= |\alpha|$ and let $k:= \barq/2 - |\alpha|$. 
    Clearly $\ftwo{f} \geq f(y/\|y\|_2)$ where $y:=a+z$, and 
	\[
		z := 
		\frac{\one}{\sqrt{n}} +
		\frac{\gamma}{\sqrt{k}} +
		\frac{\sqrt{\alpha}}{\sqrt{t}} 
	\]
	and $a$ is the unit vector that maximizes the quadratic polynomial
	\[
		(h(z))(x).
	\]
	Since $\|y\|_2 = O(1)$, $\ftwo{f} \geq f(y)/2^{O(d)}$. 
	Now clearly by non-negativity we have, 
	\begin{align*}
		f(y) 
		\geq 
		(h(z))(a)
		= \ftwo{h(z)}
	\end{align*}
	Thus we have, 
	\begin{align*}
	    \ftwo{f}^{q/d} 
	    &\geq 
            \ftwo{(h(z))(x)}^{q/d} \cdot 2^{-O(q)} \\
            &=\ftwo{ {h(z)}^{q/d}(x)} \cdot 2^{-O(q)} \\
	    &= \hscsos{C}{ {h(z)}^{q/d}(x)} \cdot 2^{-O(q)} 
	    &&(\text{SOS exact on powered quadratics})\\
	    &= \hscsos{C}{s(z)(x)} \cdot 2^{-O(q)} \\
	    &\geq 
	    \hscsos{C}{S_{2\alpha}(\one/\sqrt{n}+\gamma/\sqrt{k})(x)} \cdot 
	    \frac{2^{-O(q)}}{|\orbit{\alpha}|}
	    &&(\text{by~\cref{lem:poly:sos:nnc:monotonicity}
	    and~\cref{lem:poly:fold:nnc:mult:eval}}) \\
	    &\geq 
	    \frac{\hscsos{C}{\fold{\collapse{k}{S_{2\alpha}}}{\gamma}}}{k^{k/2}\cdot n^{k/2}} \cdot 
	    \frac{2^{-O(q)}}{|\orbit{\alpha}|} 
	    &&(\text{by~\cref{lem:poly:sos:nnc:monotonicity}, and} \\
	    &~&& 
	    S_{2\alpha}(\frac{\one}{\sqrt{n}}+\frac{\gamma}{\sqrt{k}})
	    \geq 
	    \fold{\collapse{k}{S_{2\alpha}}}{\gamma} \text{ coefficient-wise}) 
	\end{align*}
	which completes the proof since we've established~\cref{eq:poly:nnc:rte}.
\end{proof}

\subsection[Final Approximation for General Polynomials]{$(n/q)^{d/2 - 1}$-Approximation for General Polynomials}
\begin{theorem}\label{thm:poly:gen:d/2-1}
	Consider any $f \in \NPR{d}$ for $d\geq 2$, and any $q$ divisible 
	by $2d$. Then we have, 
	\[
		\frac{\hssos{f^{q/d}}^{d/q}}{\ftwo{f}} ~\leq~ 
		2^{O(d)}~\frac{n^{d/2-1}}{q^{d/2-1}}.
	\]
\end{theorem}

\begin{proof}
    Let $h$ be the unique $(d-2,2)$-folded polynomial whose unfolding yields 
	$f$ and such that for any $\beta\in\degmindex{d-2}$, the fold $\fold{h}{\beta}$ of 
	$h$ is equal up to scaling, to the quadratic form	of the corresponding $(n\times n)$ block 
	of the SOS-symmetric matrix representation $\sfM_f$ of $f$. That is, for any 
	$I,J\in [n]^{d/2 -1}$, s.t. $\mi{I}+\mi{J} = \beta$, 
	\[
	    \fold{h}{\beta}(x) = \frac{x^T\sfM_f[I,J] x}{|\orbit{\beta}|}.
	\]
	Let $s$ be the $(\barq,2q/d)$-folded 
	polynomial given by $h^{q/d}$ where $\barq:= (d-2)q/d$.  
	Consider any $\alpha\in \udmindex{\barq/2}$ and 
	$\gamma\in \degbmindex{\barq - 2|\alpha|}$, 
	and let 	$S_{2\alpha}$ be the multilinear component of $s$ as defined 
        in~\cref{lem:poly:fold:split:gen:mult}. Below, we will show
	\begin{equation}\label{eq:poly:gen:rte}
		\ftwo{f}^{q/d} \geq \frac{2^{-O(q)}\cdot \fsp{\fold{(S_{2\alpha})}{\gamma}}}
		{(\barq - 2|\alpha|)^{\barq/2 - |\alpha|}\cdot |\orbit{\alpha}|}
	\end{equation}
	which would complete the proof after applying~\cref{lem:poly:fold:gen:mult:sos} to 
	$s$ and composing the result with~\cref{lem:poly:fold:gen:coefficient:sos}. 
	To elaborate, we apply~\cref{lem:poly:fold:gen:mult:sos} to $s$ with 
    $d_1 = \barq, d_2=2q/d$ and then for every $\alpha\in \udmindex{\barq/2}$ 
    we apply~\cref{lem:poly:fold:gen:coefficient:sos} with 
    $d_1 = \barq-2|\alpha|, d_2 = 2q/d$, to get 
    \[
        \hssos{f^{q/d}} 
        =
        \hssos{s}
        \leq 
        2^{O(q)}\cdot 
        \max_{\alpha\in \udmindex{\barq/2}} ~
        \max_{\gamma\in\degbmindex{\barq - 2|\alpha|}} 
        \frac{\|\fold{(S_{2\alpha})}{\gamma}\|_{sp}}
        {(\barq - 2|\alpha|)^{\barq - 2|\alpha|}\cdot |\orbit{\alpha}|}
    \]
    which on combining with~\cref{eq:poly:gen:rte} yields the claim. 
	\medskip 
	
	Fix any $\alpha, \gamma$ satisfying the above conditions. Let $k:=\barq-2\alpha$. 
	Let $t:= |\alpha|$, and let 
	\[
		z := 
		\Xi \cdot \frac{1}{\sqrt{k}}\cdot \gamma\circ \frac{1}{2\alpha+\one}\circ b ~~+~~ 
		\frac{\sqrt{\alpha}\circ \zeta}{\sqrt{t}} 
	\]  
	$\Xi$ is an independent and uniformly randomly chosen 
	$(k+1)$-th root of unity, and for any $i\in [n]$,  
	$\zeta_i$ is an independent and uniformly randomly chosen 
	$(2\alpha_i+1)$-th root of unity, and for any $i\in [n]$, $b_i$ is an 
	independent $\mathrm{Bernoulli}(p)$ random variable ($p$ is a 
	parameter that will be set later). 
	By~$\cref{lem:poly:decoupled:lower:bound}$ and definition of $h$, we see that for any 
	$y$, $\ftwo{f}^c \geq \|(h(y))(x)\|^c_2$. Thus, we have, 
	\begin{align*}
		\ftwo{f}^{q/d} 
		&= \ftwo{f^{q/d}} \\
		&\geq\ftwo{f^{q/d}}^c_2 \cdot 2^{-O(q)} 
		&&\text{(by~\cref{lem:poly:complex:to:real})} \\
		&\geq 
		\max_{p\in [0,1]}
                \Ex{\ftwo{{h(z)}^{q/d}(x)}} \cdot 2^{-O(q)} 
		&&(\text{by~\cref{lem:poly:decoupled:lower:bound}}) \\
		&=\max_{p\in [0,1]}\Ex{\|h(z)^{q/d}(x)\|_{sp}} \cdot 2^{-O(q)} 
		&&(\text{SOS exact on powered quadratics}) \\
		&=
		\max_{p\in [0,1]}
		\Ex{\|h(z)^{q/d}(x) \cdot \Xi \cdot 
		\prod_{\mathclap{i\in [n]}} \zeta_{i}
		\|_{sp}} \cdot 2^{-O(q)} \\
		&\geq 
		\max_{p\in [0,1]}
		\bigg{\|}\Ex{
		h(z)^{q/d}(x) \cdot 
		\Xi \cdot 
		\prod_{\mathclap{i\in [n]}} \zeta_{i}
		}\bigg{\|}_{sp} \cdot 2^{-O(q)} \\
		&= 
		\max_{p\in [0,1]}
		\bigg{\|}\Ex{
		(s(z))(x) \cdot 
		\Xi \cdot 
		\prod_{\mathclap{i\in [n]}} \zeta_{i}
		}\bigg{\|}_{sp} \cdot 2^{-O(q)} \\
		&= 
		\max_{p\in [0,1]}
		\bigg{\|}
		\,p^{k}\cdot \frac{(S_{2\alpha}(\gamma/\sqrt{k}))(x)}{|\orbit{\alpha}|} +
		r(p)
		\bigg{\|}_{sp} \cdot 2^{-O(q)} 
		&&(\text{by~\cref{lem:poly:fold:gen:mult:eval}, }\mathrm{deg}(r)<k)\\
		&= 
		\max_{p\in [0,1]}
		\bigg{\|}
		\,p^{k}\cdot \frac{\fold{(S_{2\alpha})}{\gamma}(x)}{k^{k/2}\cdot |\orbit{\alpha}|} + 
		r(p)
		\bigg{\|}_{sp} 
		\cdot 2^{-O(q)} \\
		&\geq 
		\frac{\fsp{\fold{(S_{2\alpha})}{\gamma}}}{k^{k/2}\cdot 
		|\orbit{\alpha}|} \cdot 2^{-O(q+k)} 
		&&( \text{Chebyshev Inequality~\cref{lem:poly:chebyshev}}) 
	\end{align*}
	where the last inequality follows by the following argument: one would like to show 
	that there always exists $p\in [0,1]$ such that 
	$\|p^k\cdot h_k(x) + \dots p^0\cdot h_0(x))\|_{sp} \geq \|h_k(x)\|_{sp}\cdot 2^{-O(k)}$. 
	So let $p$ be such that $|p^k\cdot u^T M_k v + \dots p^0\cdot u^T M_0 v| \geq 
	|u^T M_k v|\cdot 2^{-O(k)}$ (such a $p$ exists by Chebyshev inequality) where $M_k$ is the 
	matrix representation of $h_k(x)$ realizing $\|h_k\|_{sp}$ and $u,v$ are the 
	maximum singular vectors of $M_k$. $M_{k-1},\dots ,M_0$ are arbitrary matrix representations of 
	$h_{k-1}, \dots h_0$ respectively. But $p^k\cdot M_k  + \dots p^0\cdot M_0$ is a matrix 
	representation of $p^k\cdot h_k + \dots p^0\cdot h_0$. 
	Thus $\fsp{p^k\cdot h_k + \dots p^0\cdot h_0}\geq |u^T M_k v|/2^{-O(k)} = 
	\fsp{h_k}\cdot 2^{-O(q)}$. 
	
	This completes the proof as we've established~\cref{eq:poly:gen:rte}. 
\end{proof}

\subsection{Algorithms}

It is straightforward to extract algorithms 
from the proofs of~\cref{thm:poly:nnc:d/4-1/2} and~\cref{thm:poly:gen:d/2-1}.

\subsubsection{Non-negative coefficient polynomials}
Let $f$ be a degree-$d$ polynomial with non-negative 
coefficients and let $h$ be a $(d-2,2)$-folded polynomial 
whose unfolding yields $f$. Consider any $q$ divisible by 
$2d$ and let $\barq:= (d-2)q/d$. 
Pick and return the best vector from the set 
\[
	\inbraces{\frac{\one}{\sqrt{n}} \!+\!
		\frac{\sqrt{\alpha}}{\sqrt{|\alpha|}} \!+\!
		\frac{\gamma}{\sqrt{|\gamma|}} + 
		\mathop{\arg\max}~ 
		\bigg{\|}
		h\pth{\frac{\one}{\sqrt{n}} +
		\frac{\sqrt{\alpha}}{\sqrt{|\alpha|}} +
		\frac{\gamma}{\sqrt{|\gamma|}}}\!\!(x)
		\bigg{\|}_2 
		\sep{\alpha\in \udmindex{\barq/2}, 
		\gamma\in \degmindex{\barq/2-|\alpha|}}}
\]

\subsubsection{General Polynomials}
Let $f$ be a degree-$d$ polynomial and let $h$ be the unique $(d-2,2)$-folded polynomial whose unfolding yields $f$ and such that for any $\beta\in\degmindex{d-2}$, the fold $\fold{h}{\beta}$ of 
$h$ is equal up to scaling, to the quadratic form	of the corresponding $(n\times n)$ block 
of the SOS-symmetric matrix representation $\sfM_f$ of $f$. That is, for any 
$I,J\in [n]^{d/2 -1}$, s.t. $\mi{I}+\mi{J} = \beta$, 
\[
    \fold{h}{\beta}(x) = \frac{x^T\sfM_f[I,J] x}{|\orbit{\beta}|}.
\] 
Consider any $q$ divisible by $2d$ and let $\barq:= (d-2)q/d$. 
Let the set $S$ be defined by,  
\[
    S:= 
    \inbraces{
    \Xi \cdot \frac{1}{\sqrt{|\gamma|}}\cdot \gamma\circ \frac{1}{2\alpha+\one}\circ b +
	\frac{\sqrt{\alpha}\circ \zeta}{\sqrt{|\alpha|}} 
	~\sep{
        ~\begin{array}{l}
	\Xi\in \Omega_{k+1},~
	\zeta_i\in \Omega_{2\alpha_i+1},~
	b\in \{0,1\}^n, \\[5 pt]
	\alpha\in \udmindex{\barq/2},~
	\gamma\in \degbmindex{\barq-2|\alpha|} 
         \end{array}
       }
	}
\]
where $\Omega_p$ denotes the set of $p$-th roots of unity. 
Pick and return the best vector from the set 
\[
  \inbraces{c_1\cdot y + c_2\cdot \mathop{\arg\max}\ftwo{(h(y))(x)}\sep{y\in S,~c_1\in[-(d-2),
  (d-2)],~c_2\in [-2,2]}}
\]
Note that one need only search through all roots of unity 
vectors $\zeta$ supported on $\supp{\gamma}$ and all 
$\{0,1\}$-vectors $b$ supported on $\supp{\alpha}$. 
The~\cref{lem:poly:decoupled:lower:bound} can trivially be made constructive in 
time $2^{O(q)}$. Lastly, to go from complexes to reals, the~\cref{lem:poly:complex:to:real}
can trivially be made constructive using $2^{O(d)}$ time. Thus, the algorithm runs in time $n^{O(q)}$.

%% file: poly_opt/6.oraclelowerbound.tex
\section{Oracle Lower Bound}\label{sec:poly:query-lower-bound}

\newcommand{\ball}{\mathbb{B}}
\newcommand{\qqq}{\mathbb{Q}}
\newcommand{\rrr}{\mathbb{R}}
\newcommand{\sss}{\mathbb{S}}
\newcommand{\conv}{\operatorname{conv}}
Khot and Naor~\cite{KN08} observed that the problem of maximizing a polynomial over
unit sphere can be reduced to computing diameter of centrally symmetric convex body.
This observation was also used by So~\cite{So11} later. We recall the reduction here:
For a convex set $K$, let $K^\circ$ denote the polar of $K$, \ie~$K^\circ=\{y\suchthat 
\forall x\in K\innerprod{x,y}\le 1\}$. For a degree-$3$ polynomial $P(x,y,z)$ on $3n$
variables, let $\norm{P}{x}=\norm{sp}{P(x,\cdot,\cdot)}$ where $P(x,\cdot,\cdot)$ is a degree-2
restriction of $P$ with $x$ variables set. Let $\ball_P=\{x\suchthat \norm{P}{x}\le 1\}$.
From the definition of polar and $\norm{sp}{\cdot}$, we have:
\begin{align*}
 \max_{\norm{2}{x},\norm{2}{y},\norm{2}{z}\le 1} P(x,y,z) & = \max_{x\in\ball_2} \norm{P}{x}\\
  & = \max_{x\in\ball_P^\circ} \norm{2}{x}
\end{align*}

For general convex bodies, a lower bound for number of queries with ``weak separation oracle''
for approximating the diameter of the convex body was proved by Brieden \etal~\cite{BGKKLS01} and later
improved by Khot and Naor~\cite{KN08}. We recall the
definition:
\begin{definition}
  For a given a convex body $P$, a {\deffont weak separation oracle} $A$ is an algorithm which on input
  $(x,\eps)$ behaves as following:
  \begin{itemize}
    \item If $x\in A+\eps\ball_2$, $A$ accepts it.
    \item Else $A$ outputs a vector $c\in\qqq^n$ with $\norm{\infty}{c}=1$ such that for all $y$
      such that $y+\eps\ball_2\subset P$ we have $c^T x +\eps\ge c^T y$.
  \end{itemize}
\end{definition}

Let $K_{s,v}$ be the convex set $K^{(n)}_{s,v}=\conv\left(\ball_2^n\cup \{sv,-sv\}\right)$, for unit vector $v$.
Brieden \etal~\cite{BGKKLS01} proved the following theorem:
\begin{theorem}\label{}
  Let $A$ be a randomized algorithm, for every convex set $P$, with access to a weak separation oracle
  for $P$. Let $\calK(n,s)=\{K^{(n)}_{s,u}\}_{u\in\sss^{n-1}_2}\cup \{\ball^n_2\}$. If for every
  $K\in \calK(n,s)$ and $s=\frac{\sqrt n}{\lambda}$, we have:
  \[
    \Pr{A(K)\le \diam(K)\le \frac{\sqrt n}{\lambda}A(K)}\ge \frac{3}{4}
  \]
  where $\diam(K)$ is the diameter of $K$, then $A$ must use at least $\bigoh(\lambda 2^{\lambda^2/2})$
  oracle queries for $\lambda\in[\sqrt 2,\sqrt{n/2}]$.
\end{theorem}
Using $\lambda=\log n$, we get that to get $s=\frac{\sqrt n}{\log n}$ approximation to diameter,
$A$ must use super-polynomial number of queries to the weak separation oracle. We note that this was
later improved to give analogous lower bound on the number of queries for an approximation factor
$\sqrt{\frac{ n}{\log n}}$ by Khot and Naor~\cite{KN08}.

Below, we show that the family of hard convex bodies considered by Brieden \etal~\cite{BGKKLS01}
can be realized as $\{\ball_P^\circ\}_{P\in\calP}$ by a family of polynomials $\calP$ -- which,
in turn, establishes a lower bound of $\bigomega{\frac{\sqrt n}{\log n}}$ on the approximation
for polynomial optimization, achievable using this approach, for the case of $d=3$.
%%Consider the degree-$3$ multi-linear polynomial
%%\[
%%  P(x,y,z)=\sum_{i=1}^n x_iy_iz_1+s\cdot x_1y_nz_n
%%\]
%%where $x,y,z$ are $n$-tuple of variables.
For a unit vector $u\in\sss_2^{n-1}$, let $P_u$ be the polynomial defined as:
\[
P_u(x,y,z)=\sum_{i=1}^n x_iy_iz_1+s\cdot (u^T x) y_nz_n.
\]
%Now we prove: $\norm{P}{x}=\norm{sp}{P(x,\cdot,\cdot)}=\max\{\norm{2}{x},s\abs{x_1}\}$, where
%$\bar{x}=(x_1,\ldots,x_{n-1})$, \ie projection of $x$ to first $n-1$ co-ordinates.
A matrix representation of $P_u(x,\cdot,\cdot)$, with rows indexed by $y$ and columns
indexed by $z$ variables is as follows:
\[
  A_u=
  \begin{pmatrix}
    x_1 & 0 & \cdots & 0 & 0\\
    x_2 & 0 & \cdots & 0 & 0\\
    \vdots & \vdots & \ddots & \vdots & \vdots\\
    x_{n-1} & 0 & \cdots & 0 & 0\\
    x_n & 0 & \cdots & 0 & s\cdot (u^T x)\\
  \end{pmatrix} \text { and so, }
  A_u^TA_u=
  \begin{pmatrix}
    \norm{2}{x}^2 & 0 & \ldots & 0 & 0\\
    0 & 0 & \cdots & 0 & 0\\
    \vdots & \vdots & \ddots & \vdots & \vdots\\
    0 & 0 & \cdots & 0 & 0\\
    0 & 0 & \cdots & 0 & s^2\cdot\abs{u^T x}^2\\
  \end{pmatrix}.
\]
This proves: $\norm{P_u}{x}=\norm{sp}{P_u(x,\cdot,\cdot)}=\norm{sp}{A_u}=\max\{\norm{2}{x},s\abs{u^Tx}\}$.

Let $B=\{x\suchthat \norm{2}{x}\le 1\}$ and $C_u=\{x\suchthat s\cdot \abs{u^Tx}\le 1\}$.
We note that, $B^\circ=\{y\in\rrr^n\suchthat
\norm{2}{y}\le 1\}$ %-- since $n$-th co-ordinate of $B_0$ is unbounded. Similarly,
and, $C_u^\circ=\{\lambda\cdot u\suchthat \lambda\in[-s,s]\}=\conv(\left\{ -s\cdot u,s\cdot u \right\})$.

Next, we observe: $\ball_{P_u}=B\cap C_u$. It follows from De Morgan's law of polars 
that: $\ball_{P_u}^\circ=(B\cap C_u)^\circ=\conv(B^\circ \cup C_u^\circ)=
\conv(\ball_2^n \cup \{-s\cdot u,s\cdot u\})=K_{s,u}^{(n)}$.
Finally, we observe that for the polynomial $P_0=\sum_{i=1}^n x_iy_iz_1$,
we have: $\ball_{P_0}=\ball_2^n$.

Hence, for polynomial $Q\in \calP=\{P_u\}_{u\in \sss_2^{n-1}}\cup \{ P_0 \}$, 
no randomized polynomial
can approximate $\diam{\ball_Q}$ within factor $\frac{\sqrt n}{q}$ 
without using more than $2^{\Omega(q)}$ number of queries. Since
the algorithm of Khot and Naor \cite{KN08} reduces the problem of optimizing polynomial $Q$ to computing
$\diam(\ball_Q)$, $\calP$ shows that their analysis is almost tight.

%% file: poly_opt/7.nnc-lowerbound.tex
\section{Constant Level Lower Bounds for Polynomials with Non-negative Coefficients}\label{sec:poly:nnc-lowerbound}
%%%%%%%%%%%%%%%%%%%%%%%%%%%%%%%%%%%%%%%%%%%%%%%%%%%%%%%%
Let $G = (V, E)$ be a random graph drawn from the distribution $G_{n, p}$ for $p \geq n^{-1/3}$.  Let $\cliques \subseteq \binom{V}{4}$ be
the set of $4$-cliques in $G$. 
The polynomial $f$ is defined as 
\[
f(x_1, \dots, x_n) := \sum_{\{ i_1, i_2, i_3, i_4 \} \in \cliques} x_{i_1} x_{i_2} x_{i_3} x_{i_4}. 
\]
Clearly, $f$ is multilinear and every coefficient of $f$ is nonnegative. 
In this section, we prove the following two lemmas that establish a polynomial gap 
between $\ftwo{f}$ and $\hssos{f}$. 
\begin{lemma}[Soundness]\label{lem:poly:cliques}
With probability at least $1-\frac1{n}$ over the choice of the graph $G$, we have $\ftwo{f} \leq n^2 p^6
\cdot \inparen{\log n}^{O(1)}$.
\end{lemma}

\begin{lemma}[Completeness]\label{lem:poly:nncsoslower}
With probability at least $1 - \frac1{n}$ over the choice of the graph $G$, we have
\[
\hssos{f} ~\geq~ \Omega \Bigl(  \frac{n^{1/2} \cdot p}{\log^{2} n} \Bigr)
\] 
when $p \in [n^{-1/3}, n^{-1/4}]$.
\end{lemma}
Note that the gap between the two quantities if $\tilde{\Omega}(n^{1/6})$ when $p = n^{-1/3}$, which is the choice we make.

\subsection[Upper on the Clique Polynomial]{Upper Bound on $\ftwo{f}$}

\subsubsection{Reduction to counting shattered cliques}
We say that an ordered $4$-clique $(i_1, \ldots, i_4)$ is {\em shattered} by $4$ disjoint sets $Z_1, \ldots,
Z_4$ if for each $k \in [4]$, $i_k \in Z_k$. Let $Y_{j_1}, \ldots, Y_{j_4}$ be the sets containing
the coordinates $i_1, \ldots, i_4$. Let $\cliques_G$ denote the set of (ordered) 4-cliques in $G$, and let
$\cliques_G(Z_1,Z_2,Z_3,Z_4)$ denote the set of cliques shattered by $Z_1, \ldots, Z_4$.

We reduce the problem of bounding $\ftwo{f}$, to counting shattered 4-cliques.
\begin{claim}\label{claim:poly:shattering}
There exist disjoint sets $Z_1, \ldots, Z_4 \subseteq [n]$ such that
\[
\abs{\cliques_G(Z_1, Z_2, Z_3, Z_4)} ~\geq~ \inparen{\prod_{k=1}^4 \abs{Z_k}}^{1/2} \cdot
O\inparen{ \frac{\ftwo{f}}{(\log n)^4} } \mper 
\]
\end{claim}
\begin{proof}
Let $x^* \in \SSS^{n-1}$ be the vector that maximizes $f$. Without loss of generality, 
assume that every coordinate of $x^*$ is nonnegative. 
Let $y^*$ be another unit vector defined as 
\[
y^* ~:=~ \frac{\inparen{x^* + {\one}/{\sqrt{n}}}}{\norm{2}{x^* + {\one}/{\sqrt{n}}}} \mper
\]
Since both $x^*$ and $\frac{\one}{\sqrt{n}}$ are unit vectors, the denominator is at most $2$. 
This implies that $f(y^*) \geq \frac{f(x^*)}{2^4}$, and each coordinate of $y^*$ is at least
$\frac{1}{2\sqrt{n}}$.
For $1 \leq j \leq \log_2 n$, let $Y_j$ be the set 
\[
Y_j~\defeq~\inbraces{i \in [n] ~\mid~ 2^{-j} < y^*_i \leq 2^{-(j - 1)}} \mper
\]
The sets $Y_1, \dots, Y_{\log_2 n}$ partition $[n]$.
Since $1 = \sum_{i \in [n]} y_i^2 > |Y_j| \cdot 2^{-2j}$, we have for each $j$, $|Y_j| \leq 2^{2j}$. 
Let $Z_1, Z_2, Z_3$, and $Z_4$ be pairwise disjoint random subsets of $[n]$ chosen as follows:
\begin{itemize}
\item Randomly partition each $Y_j$ to $Y_{j, 1}, \dots, Y_{j, 4}$ where each element of $Y_j$ is
  put into exactly one of $Y_{j, 1}, \dots, Y_{j, 4}$ uniformly and independently. 
\item Sample $r_1, \ldots, r_4$ independently and randomly from $\{1, \dots, \log_2 n\}$.
\item For $k = 1, \dots, 4$, take $Z_k := Y_{r_k, k}$
\end{itemize}
We use $\parts$ to denote random partitions $\inbraces{\inparen{Y_{j, 1}, \dots, Y_{j, 4}}}_{j \in [\log n]}$ and
$r$ to denote the random choices $r_1, \ldots, r_4$. Note that the events $i_k \in Z_k$ are
independent for different $k$, and that $Z_1, \ldots, Z_4$ are independent given $\parts$. Thus, we have
\begin{align*}
\Ex{\parts, r}{\frac{\ind{(i_1,i_2,i_3,i_4)~\text{is shattered}} }{\sqrt{\abs{Z_1} \abs{Z_2} \abs{Z_3}
  \abs{Z_4}} } }  
&~=~\Ex{\parts}{\prod_{k=1}^4 \Ex{r_k}{\frac{\ind{i_k \in Z_k}}{\sqrt{\abs{Z_k}}} } } \\
&~=~\Ex{\parts}{\prod_{k=1}^4 \Ex{r_k}{\frac{\ind{r_k = j_k} \cdot \ind{i_k \in Y_{j_k,
  k}}}{\sqrt{\abs{Y_{j_k, k}}}} } } \\
&~\geq~\Ex{\parts}{\prod_{k=1}^4 \Ex{r_k}{\frac{\ind{r_k = j_k} \cdot \ind{i_k \in Y_{j_k,
  k}}}{\sqrt{\abs{Y_{j_k}}}} } } \\
&~=~\Ex{\parts}{\prod_{k=1}^4 \inparen{\frac{1}{\log n} \cdot \frac{\ind{i_k \in Y_{j_k,
  k}} }{\sqrt{\abs{Y_{j_k}}} } } } \\
&~=~\Ex{\parts}{\prod_{k=1}^4 \inparen{\frac{1}{\log n} \cdot \frac{\ind{i_k \in Y_{j_k,
  k}} }{\sqrt{\abs{Y_{j_k}}} } } } \\
&~=~\frac{1}{\inparen{4 \log n}^4} \cdot \frac{1}{\sqrt{\abs{Y_{j_1}} \abs{Y_{j_2}} \abs{Y_{j_3}}
  \abs{Y_{j_4}}}} \\
&~\geq~\frac{1}{\inparen{4 \log n}^4} \cdot 2^{j_1 + j_2 + j_3 + j_4} \\
&~\geq~\frac{1}{\inparen{8 \log n}^4} \cdot y^*_{i_1}y^*_{i_2}y^*_{i_3}y^*_{i_4} \mper
\end{align*}
Then, by linearity of expectation,
\begin{align*}
\Ex{\parts, r}{\frac{\abs{\cliques_G(Z_1, Z_2, Z_3, Z_4)}}{\sqrt{\abs{Z_1} \abs{Z_2} \abs{Z_3}
  \abs{Z_4}}} }
&~\geq~\frac{1}{(8 \log n)^4} \cdot \sum_{(i_1,\ldots,i_4) \in \cliques_G}
  y^*_{i_1}y^*_{i_2}y^*_{i_3}y^*_{i_4} \\
&~=~\frac{4!}{(8 \log n)^4} \cdot f\inparen{y^*} \\
&~\geq~\frac{4!}{(16 \log n)^4} \cdot f\inparen{x^*}
~=~\frac{4!}{(16 \log n)^4} \cdot \ftwo{f} \mcom
\end{align*}
which proves the claim. 
\end{proof}
We will show that with high probability, $G$ satisfies the property that every four disjoint sets
$Z_1, \dots, Z_4 \subseteq V$ shatter at most $O\inparen{\sqrt{|Z_1||Z_2||Z_3||Z_4|} \cdot n^2 p^6 \cdot (\log
n)^{O(1)}}$ cliques, proving~\cref{lem:poly:cliques}.   

\subsubsection{Counting edges and triangles}
For a vertex $i \in [n]$, we use $\nbr(i)$ to denote the set of vertices in the graph $G$. For ease
of notation, we use $a \lsim b$  to denote $a \leq b \cdot (\log n)^{O(1)}$.
We first collect some simple consequences of Chernoff bounds. 
\begin{claim}\label{claim:poly:chernoff-counts}
Let $G \sim G_{n, p}$ with $p \geq n^{-1/3}$. Then, with probability $1 - \frac{1}{n}$, we have
\begin{itemize}
\item For all distinct $i_1, i_2 \in [n]$, $\abs{\nbr(i_1) \cap \nbr(i_2)} ~\lsim~ np^2$.
\item For all distinct $i_1, i_2, i_3 \in [n]$, $\abs{\nbr(i_1) \cap \nbr(i_2) \cap \nbr(i_3)} ~\lsim~
  np^3$.
\item For all sets $S_1, S_2 \subseteq [n]$, $\abs{E(S_1,S_2)} ~\lsim~
  \max\inbraces{\abs{S_1}\abs{S_2}p, \abs{S_1} + \abs{S_2}}$.
\end{itemize}
\end{claim}
We also need the following bound on the number of triangles shattered by three disjoint 
sets $S_1, S_2$ and $S_3$, denoted by $\triangles_G(S_1,S_2,S_3)$. 
As for $4$-cliques, a triangle is said to be shattered if it has exactly 
one vertex in each the sets.
\begin{claim}\label{claim:poly:triangles}
Let $G \sim G_{n, p}$ with $p \geq n^{-1/3}$. Then, with probability $1 - \frac{1}{n}$, for all
disjoint sets $S_1, S_2, S_3 \subseteq [n]$
\[
\abs{\triangles_G(S_1,S_2,S_3)} ~\lsim~ \abs{S_3} + \abs{E(S_1,S_2)} \cdot \inparen{np^3
  \cdot\abs{S_3}}^{1/2} \mper
\] 
\end{claim}
\begin{proof}
With probability at least $1 - \frac{1}{n}$, $G$ satisfies the conclusion of~\cref{claim:poly:chernoff-counts}.
Fix such a $G$, and consider arbitrary subsets $S_1, S_2, S_3 \subseteq V$. 
Consider the bipartite graph $H$ where the left side vertices correspond to edges in $E(S_1, S_2)$,
the right side vertices correspond to vertices in $S_3$, and there is an edge from $(i_1, i_2) \in
E(S_1, S_2)$ to $i_3 \in S_3$ when both $(i_1, i_3), (i_2, i_3) \in E$.  Clearly,
$\abs{\triangles_G(S_1,S_2,S_3)}$ is equal to the number of edges in $H$. 

Consider two different edges $(i_1, i_2), (i_1', i_2') \in E(S_1, S_2)$. These two edges are
incident on at least $3$ distinct vertices, say $\{i_1,i_2,i_1'\}$. Hence, the number of vertices
$i_3 \in [n]$  that are adjacent to all $\{ i_1, i_2, i_1', i_2' \}$ in $G$ is at most 
$\abs{\nbr(i_1) \cap \nbr(i_2) \cap \nbr(i_1')} \lsim np^3$. This gives that the number of pairs 
triangles sharing a common vertex in $S_3$ is at most $\abs{E(S_1,S_2)}^2 \cdot np^3 (\log
n)^{O(1)}$.

Let $d_H(i_3)$ denote the degree of a vertex $i_3$ in $H$, and let 
$\Delta$ denote the number of shattered triangles. Counting the above pairs of triangles using the
degrees gives
\[
\sum_{i_3 \in S_3} \binom{d_H(i_3)}{2} ~\lsim~ \abs{E(S_1,S_2)}^2 \cdot np^3 \mper
\]
An application of Cauchy-Schwarz gives
\[
\Delta^2 - \Delta \cdot \abs{S_3} ~\lsim~ \abs{S_3} \cdot \abs{E(S_1,S_2)}^2 \cdot np^3 \mcom
\]
which proves the claim.
\end{proof}

\subsubsection{Bounding $4$-clique Density}
Let $G \sim G_{n,p}$ be a graph satisfying the conclusions of~\cref{claim:poly:chernoff-counts} 
and~\cref{claim:poly:triangles}. Let $S_1, \ldots, S_4 \subseteq [n]$ be disjoint sets with sizes $n_1 \leq n_2
\leq n_3 \leq n_4$. We consider two cases:
\begin{itemize}
\item \textbf{Case 1}: $\abs{E(S_1,S_2)} ~\lsim~ n_1n_2p$ \\[2 pt]
Note that each edge $(i_1,i_2)$ can only participate in at most $\abs{\nbr(i_1) \cap \nbr(i_2)}$ triangles,
and each triangle $(i_1,i_2,i_3)$ can only be extended to at most $\abs{\nbr(i_1) \cap \nbr(i_2)
  \cap \nbr(i_3)}$ 4-cliques. Thus,~\cref{claim:poly:chernoff-counts} gives
\[
\abs{\cliques_G(S_1,S_2,S_3,S_4)} ~\lsim~ n_1n_2p \cdot np^2 \cdot np^3 ~\lsim~ (n_1n_2n_3n_4)^{1/2}
\cdot n^2 p^6 \mper
\]
\item \textbf{Case 2}: $\abs{E(S_1,S_2)} ~\lsim~ n_1 + n_2$ \\[2 pt]
The~\cref{claim:poly:triangles} gives
\[
\abs{\Delta_G(S_1,S_2,S_3)} ~\lsim~ n_3 + (n_1+n_2) \cdot \inparen{n_3 \cdot np^3}^{1/2} \mcom
\]
which together with~\cref{claim:poly:chernoff-counts} implies
\[
\abs{\cliques_G(S_1,S_2,S_3,S_4)} ~\lsim~ n_3 \cdot np^3 + (n_1+n_2) \cdot n_3^{1/2} \cdot 
\inparen{np^3}^{3/2} \mper
\]
Considering the first term, we note that
\[
n_3 \cdot np^3 ~\leq~ (n_3 n_4)^{1/2} \cdot n^2p^6 ~\leq~ (n_1 n_2 n_3 n_4)^{1/2} \cdot n^2p^6 \mcom
\]
since $n_3 \leq n_4$ and $np^3 \geq 1$. Similarly, for the second term, we have
\[
(n_1+n_2) \cdot n_3^{1/2} \cdot  \inparen{np^3}^{3/2} 
~\leq~ 2(n_2 n_3 n_4)^{1/2} \cdot \inparen{np^3}^{3/2}
~\leq~ 2 \cdot (n_1 n_2 n_3 n_4)^{1/2} \cdot n^2 p^6 \mper
\]
\end{itemize}
Combined with~\cref{claim:poly:shattering}, this completes the proof of~\cref{lem:poly:cliques}.

\subsection{Lower Bound on $\hssos{f}$}
Recall that given a random graph $G = ([n], E)$ drawn from the distribution $G_{n, p}$, 
the polynomial $f$ is defined as 
\[
f(x_1, \dots, x_n) := \sum_{\{ i_1, i_2, i_3, i_4 \} \in \cliques} x_{i_1} x_{i_2} x_{i_3} x_{i_4},
\]
where $\cliques \subseteq \binom{[n]}{4}$ is the set of $4$-cliques in $G$. 
Let $\sfA \in \RR^{[n]^2 \times [n]^2}$ be the natural matrix representation of $24f$ (corresponding
to ordered copies of cliques) with
\[
\sfA[(i_1, i_2), (i_3, i_4)] ~=~ 
\begin{cases}
1 & \text{if}~ \{i_1, \dots, i_4\} \in \cliques \\[3 pt]
0 & \text{otherwise}
\end{cases}
\]
Let $E' \subseteq {[n]}^2$ be the set of ordered edges \ie~$(i_1, i_2) \in E'$ if and only if
$\inbraces{i_1, i_2} \in E$. Note that $|E'| = 2m$ where $m$ is the number of edges in $G$. 
All nonzero entries of $\sfA$ are contained in the principal submatrix $\sfA_{E'}$, formed by the
rows and columns indexed by $E'$.

\subsubsection[A Simple Lower Bound on the Relaxation]{A simple lower bound on $\fsp{f}$}
We first give a simple proof that $\fsp{f} \geq \sqrt{n^2p^5}$ with high probability. 
\begin{lemma}\label{lem:poly:nncsoslower2}
$\fsp{f} \geq \Omega(\sqrt{n^2p^5}) = \Omega(n^{1/6})$ with high probability. 
\end{lemma}
\begin{proof}
Consider any matrix representation $M$ of $24f$ and its principal submatrix $\sfM_{E'}$. 
It is easy to observe that the Frobenius norm of $M_{E'}$ satisfies $\norm{F}{M_{E'}}^2 \geq
24\abs{\cliques}$, minimized when $M = \sfA$. 
Since $\norm{F}{M_{E'}}^2 \leq \abs{E'} \cdot \norm{2}{A_{E'}}^2$, we have that 
with high probability, 
\[
\norm{2}{A} 
~\geq~ 
\norm{2}{A_{E'}} 
~\geq~ 
\sqrt{\frac{24\abs{\cliques}}{2|E|}} = \Omega\inparen{\frac{\sqrt{n^4 p^6}}{\sqrt{n^2 p}}} 
~=~ 
\Omega\inparen{\sqrt{n^2 p^5}}.  
\]
\end{proof}

\subsubsection{Lower bound for the stronger relaxation computing $\hssos{f}$}
We now prove~\cref{lem:poly:nncsoslower}, which says that $\hssos{f} \geq \frac{n^{1/6}}{\log^{2} n}$
with high probability.
In order to show a lower bound, we look at the dual SDP for computing $\hssos{f}$, which is a
maximization problem over positive semidefinite, SoS-symmetric matrices $\sfM$ with $\Tr{\sfM} =
1$. We exhibit such a matrix $\sfM \in \RR^{[n]^2 \times [n]^2}$ 
for which the value of the objective $\iprod{\sfA}{\sfM}$ is large.

% present a moment matrix 
% $\sfM \in \RR^{[n]^2 \times [n]^2}$ that is positive semidefinite, SoS-symmetric, and $\Tr{\sfM} = 1$, that has a large inner product $\langle \sfA, \sfM \rangle \geq n^{1/6} / \log^{O(1)} n$.  

%Let $E' \subseteq [n]^2$ be the set of ordered edges --- an ordered pair $(i_1, i_2)$ is in $E'$ if
%and only if $(i_1, i_2) \in E$ as an unordered pair. Note that $|E'| = 2m$ where $m$ is the number
%of edges in $G$. All nonzero entries of $\sfA$ is contained in the principal submatrix $\sfA_{E'}$,
%formed by the rows and columns indexed by $E'$.
%
For large $\langle \sfA, \sfM \rangle$, one natural attempt is to take $\sfM$ to be $\sfA$ and
modify it to satisfy other conditions. Note that $\sfA$ is already SoS-symmetric. However, 
$\Tr{\sfA}= 0$, which implies that the minimum eigenvalue is negative.

Let $\lambda_{\min}$ be the minimum eigenvalue of $\sfA$, which is also the minimum eigenvalue of
$\sfA_{E'}$. 
Let $I_{E'} \in \RR^{[n]^2 \times [n]^2}$ be such that $I[(i_1, i_2), (i_1, i_2)] = 1$ if $(i_1,
i_2) \in E'$ and all other entries are $0$. Note that $I_{E'}$ is a diagonal matrix with $\Tr{I_{E'}} = 2m$. 
Adding $-\lambda_{\min} \cdot I_{E'}$ to $\sfA$ makes it positive semidefinite, so setting 
\begin{equation}\label{eq:poly:M-first-attemp}
\sfM 
~=~ \frac{\sfA - \lambda_{\min} I_{E'}}{\Tr{\sfA - \lambda_{\min} I_{E'}} } 
~=~  
\frac{\sfA - \lambda_{\min} I_{E'}}{ - 2 m \lambda_{\min} }
~=~
\frac{\sfA + \abs{\lambda_{\min}} \cdot I_{E'}}{ 2 m \cdot \abs{\lambda_{\min}} }
\end{equation}
makes sure that $\sfM$ is positive semidefinite, $\Tr{\sfM} = 1$, and $\iprod{\sfA}{\sfM} =
\frac{12 |\cliques|}{m \cdot \abs{\lambda_{\min}}}$ (each $4$-clique in $\cliques$ contributes $24$). 
Since $\abs{\cliques} = \Theta(n^4 p^6)$ and $m = \Theta(n^2 p)$ with high probability, if
$\abs{\lambda_{\min}} = O(np^{5/2})$, $\iprod{\sfA}{\sfM} = \Theta\inparen{n^2p^{5/2}}$, which is
$\Omega(n^{1/6})$ when $p = \Omega(n^{-1/3})$. 

The $\sfM$ defined in~\cref{eq:poly:M-first-attemp} 
does not directly work since it is not SoS-symmetric. However, the following claim
proves that this issue can be fixed by losing a factor $2$ in $\langle \sfA, \sfM \rangle$. 
\begin{claim}\label{claim:poly:fix}
There exists $\sfM$ such that it is SoS-symmetric, positive semidefinite with $\Tr{\sfM} = 1$, 
and $\langle \sfA, \sfM \rangle \geq \frac{6 |\cliques|}{m \cdot \abs{\lambda_{\min}}}$.
\end{claim}
\begin{proof}
Let $Q_{E'} \in \RR^{[n]^2 \times [n]^2}$ be the matrix such that
\begin{itemize}
\item For $(i_1, i_2) \in E'$, $Q_{E'}[(i_1, i_1), (i_2, i_2)] = Q_{E'}[(i_2, i_2), (i_1, i_1)] = 1$. 
\item For $i \in [n]$, $Q_{E'}[(i, i), (i, i)] = \deg_G(i)$, where $\deg_G(i)$ denotes the degree of $i$ in $G$.  
\item All other entries are $0$. 
\end{itemize}
We claim that $I_{E'} + Q_{E'}$ is SoS-symmetric: $(I_{E'} + Q_{E'})[(i_1, i_2), (i_3, i_4)]$ has a
nonzero entry if and only if $i_1 = i_2 = i_3 = i_4$ or two different numbers $j_1, j_2$ appear
exactly twice and $(j_1, j_2) \in E'$ (in this case $(I_{E'} + Q_{E'})[(i_1, i_2), (i_3, i_4)] =
1$). 
Since $\sfA$ is SoS-symmetric, so $\sfA + \abs{\lambda_{\min}} \cdot (I_{E'} + Q_{E'})$ is also
SoS-symmetric. 

It is easy to see that $Q_{E'}$ is diagonally dominant, and hence positive semidefinite. 
Since we already argued that $\sfA + \abs{\lambda_{\min}} \cdot I_{E'}$ is positive semidefinite,
$\sfA + \abs{\lambda_{\min}} \cdot (I_{E'} + Q_{E'})$ is also positive semidefinite. Also,
$\Tr{Q_{E'}} = \sum_{i \in [n]} \deg_G(i) = 2m$. 
Thus, we take
\[
\sfM ~=~ 
\frac{\sfA + \abs{\lambda_{\min}} \cdot (I_{E'} + Q_{E'})}{\Tr{\sfA + \abs{\lambda_{\min}} \cdot
    (I_{E'} + Q_{E'})}
 }
 ~=~  
\frac{\sfA + \abs{\lambda_{\min}} \cdot I_{E'}}{ 4 m \cdot \abs{\lambda_{\min}} } \mper
\]
By the above arguments, we have that $\sfM$ that is PSD, SoS-symmetric with $\Tr{\sfM} = 1$, and
\[
\langle \sfA, \sfM \rangle = \frac{6 |\cliques|}{m \cdot \abs{\lambda_{\min}}}
\]
as desired. 
\end{proof}

It only remains to bound $\lambda_{\min}$, which is the minimum eigenvalue of $\sfA$ and
$\sfA_{E'}$. For $p$ in the range $[n^{-1/3}, n^{-1/4}]$, we will show  a bound of
$\tilde{O}(n^{3/2}p^4)$ below, which when combined with the above claim, completes the proof
of~\cref{lem:poly:nncsoslower}.

\subsubsection{Bounding the smallest eigenvalue via the trace method}
Our estimate $\abs{\lambda_{\min}} = O(np^{5/2})$ is based on the following
observation: $\sfA_{E'}$ is a $2m \times 2m$ random matrix where each row and column is expected to
have $\Theta(n^2 p^5)$ ones (the expected number of $4$-cliques an edge participates in). An adjacency
matrix of a random graph with average degree $d$ has a minimum eigenvalue $- \Theta(\sqrt{d})$,
hence the estimate $\abs{\lambda_{\min}} = O(n p^{5/2})$. Even though $\sfA_{E'}$ is not sampled
from a typical random graph model (and even $E'$ is a random variable), we will be able to prove the
following weaker estimate, which suffices for our purposes.
\begin{lemma}\label{lem:poly:eigenvalue}
With high probability over the choice of the graph $G$, we have
\[
\abs{\lambda_{\min}} ~=~ 
\begin{cases}
\tilde{O}\inparen{n^{3/2} \cdot p^4} & \text{for}~ p \in \insquare{n^{-1/3}, n^{-1/4}} \\[5 pt]
\tilde{O}\inparen{n^{5/3} \cdot p^{14/3}} & \text{for}~ p \in \insquare{n^{-1/4}, 1/2} 
\end{cases}
\]
% $\lambda_{\min} = -O(np^{2.5} \log^{2} n) = - O(n^{1/6} \log^{2} n)$ with high probability. 
\end{lemma}
% \mtnote{Modify this lemma statement}
%
\begin{proof}
Instead of $\sfA_{E'}$, we directly study $\sfA$ to bound $\lambda_{\min}$. 
For simplicity, we consider the following matrix $\hA$, where each row and column is indexed by an
unordered pair $\{i, j\} \in \binom{[n]}{2}$, and $\hA[\{i_1, i_2\}, \{i_3, i_4\}] = 1$ if and only
if $i_1, i_2, i_3, i_4$ form a $4$-clique.
$\sfA$ has only zero entries in the rows or columns indexed by $(i, i)$ for all $i \in [n]$, 
and for two pairs $i_1 \neq i_2$ and $i_3 \neq i_4$, we have 
\begin{align*}
\hA[\{i_1, i_2\}, \{i_3, i_4\}] ~\defeq~ &\frac{1}{4} \cdot \left\{ \sfA[(i_1, i_2), (i_3, i_4)] + 
\sfA[(i_1, i_2), (i_4, i_3)] \right\}  \\ 
+~ &\frac14 \cdot \left\{\sfA[(i_2, i_1), (i_3, i_4)] + 
\sfA[(i_2, i_1), (i_4, i_3)] \right\} \mper
\end{align*}
Therefore, $\abs{\lambda_{\min}\inparen{\sfA}} \leq 4\cdot \abs{\lambda_{\min}\inparen{\hA}}$
and it suffices to bound the minimum eigenvalue of $\hA$. 
We consider the matrix $\hN_E := \hA - p^{4} \cdot \hJ_E$, where $\hJ_E \in \RR^{\binom{[n]}{2}
  \times \binom{[n]}{2}}$ is such that
\[
\hJ_E[\{i_1, i_2\}, \{i_3, i_4\}] ~=~ 
\begin{cases}
1 & \text{if}~ \{i_1, i_2\}, \{i_3, i_4\} \in E \\
0 & \text{otherwise}
\end{cases} \mper
\]
Since $\hJ_E$ is a rank-$1$ matrix with a positive eigenvalue, the minimum eigenvalues of $\hA$ and
$\hN_E$ are the same. In summary, $\hN_E$ is the following matrix. 
\[
\hN_E[\{i_1, i_2\}, \{i_3, i_4\}] ~=~ 
\begin{cases}
1 - p^4 & \text{if}~\{i_1, i_2, i_3, i_4\} \in \cliques \\
-p^4 &  \text{if}~\{i_1, i_2, i_3, i_4\} \notin \cliques ~\text{but}~ \{i_1,i_2\}, \{i_3,i_4\} \in E
\\
0 & \text{otherwise}
\end{cases}
\]
We use the trace method to bound 
$\norm{2}{\hN_{E}}$, based on the observation that for every even $r \in \NN$, 
$\norm{2}{\hN_{E}} \leq \inparen{\Tr{(\hN_E)^r}}^{1/r}$.
Fix an even $r \in \NN$. The expected value of the trace can be represented as
\[
\Ex{\Tr{(\hN_E)^r}}
~=~ 
\Ex{\sum_{I^1, \dots, I^r \in \binom{[n]}{2}} \prod_{k = 1}^r \hN_E[I^k, I^{k + 1}]}
~=~ 
\sum_{I^1, \dots, I^r \in \binom{[n]}{2}} \Ex{ \prod_{k = 1}^r \hN_E[I^k, I^{k + 1}] }
\]
where each $I^j  = \inbraces{i^j_1, i^j_2} \in \binom{[n]}{2}$ is an edge of the complete graph on $n$ vertices (call it a {\em potential edge}) and $I^{r + 1} := I^1$. 

Fix $r$ potential edges $I^1, \dots, I^r$, let $t := \prod_{k = 1}^r \hN_E[I^k, I^{k + 1}]$, and
consider $\E[t]$. 
Let $E_0 := \{ I^1, \dots, I^r \}$ be the set of distinct edges represented by $I^1, \dots, I^r$. 
Note that the expected value is $0$ if one of $I^j$ does not become an edge. 
Therefore, $\E[t] = p^{|E_0|} \cdot \Ex{t ~|~ E_0 \subseteq E}$. 

Let $D \subseteq [r]$ be the set of $j \in [r]$ such that all four vertices in $I^j$ and $I^{j+1}$
are distinct \ie
\[
D ~\defeq~ \inbraces{j \in [r] ~\mid~ \left|\inbraces{i^j_1, i^j_2, i^{j+1}_1, i^{j+1}_2} \right| = 4} \mper
\]
For $j \in [r] \setminus D$, $\{ i^j_1, i^j_2, i^{j+1}_1, i^{j+1}_2 \}$ cannot form a $4$-clique, so
given that $I^j, I^{j + 1} \in E$, we have $\hN_E[I^j, I^{j + 1}] = -p^4$. For $j \in D$, let 
$E_j := \inbraces{ 
\inbraces{i^j_1, i^{j + 1}_1},
\inbraces{i^j_1, i^{j + 1}_2},
\inbraces{i^j_2, i^{j + 1}_1},
\inbraces{i^j_2, i^{j + 1}_2}
} \setminus E_0$ 
be the set of edges in the $4$-clique created by $\inbraces{ i^j_1, i^j_2, i^{j+1}_1, i^{j+1}_2 }$ except
ones in $E_0$. 
Then
\[
\Ex{t} 
~=~ 
p^{|E_0|} \cdot \Ex{t | E_0 \subseteq E} 
~=~ 
p^{|E_0|} \cdot (-p^4)^{r - |D|} \cdot \Ex{\prod_{k \in D} \hN_E[I^k, I^{k + 1}] \, \mid \, E_0 \subseteq E }.
\]
Suppose there exists $j \in D$ such that $\abs{E_j} = 4$ and $E_j \cap (\cup_{j' \in D \setminus \{
  j \}} E_{j'}) = \emptyset$. Then, given that $E_0 \subseteq E$, $\hN_E[I^j, I^{j + 1}]$ is
independent of all $\inbraces{ \hN_E[I^k, I^{k + 1}] }_{k \in D \setminus \{ j \} }$, and 
\[
\Ex{\hN_E[I^j, I^{j + 1}] | E_0 \subseteq E} ~=~ p^4 (1 - p^4) + (1 - p^4)(- p^4) ~=~ 0 \mper
\] 
Therefore, $\E[t] = 0$ unless for all $j \in D$, either $|E_j| \leq 3$ or there exists $j' \in D
\setminus \{ j \}$ with $E_j \cap E_{j'} \neq \emptyset$. 

Let $E_D := \bigcup_{j \in D} E_j$. Note that $E_0$ and $E_D$ completely determines $t$. $\E[t]$ can
be written as 
\begin{align*}
% &\Ex{t}\\
% = \,
&\, p^{|E_0|} \cdot (-p^4)^{r - |D|} \cdot 
\Ex{\prod_{k \in D} \hN_E[I^k, I^{k + 1}] \, \mid \, E_0 \subseteq E } \\
= \, & \,  p^{|E_0|} \cdot (-p^4)^{r - |D|} \cdot \sum_{F \subseteq E_D} \inparen{ p^{|F|} (1 -
       p)^{|E_D| - |F|} \cdot \Ex{\prod_{k \in D} \hN_E[I^k, I^{k + 1}] \, | \, E_0 \subseteq E,
       E_D \cap E = F } } \\
= \, &\,  p^{|E_0|} \cdot (-p^4)^{r - |D|} \cdot \sum_{F \subseteq E_D} \inparen{ p^{|F|} (1 -
       p)^{|E_D| - |F|} \cdot (1 - p^4)^{|D| - a(F)} (-p^4)^{a(F)} },
\end{align*}
where $a(F)$ denotes the number of $j \in D$ with $E_j \not \subseteq F$. 
Since $E_D \subseteq F \cup \inparen{\bigcup_{j : E_j \not \subseteq F} E_j) }$ and $4a(F) + |F|
\geq |E_D|$, we have
\begin{align*}
\Ex{t}
&~=~ 
p^{|E_0|} \cdot (-p^4)^{r - |D|} \cdot \sum_{F \subseteq E_D} \inparen{ p^{|F|} (1 -
       p)^{|E_D| - |F|} \cdot (1 - p^4)^{|D| - a(F)} (-p^4)^{a(F)} } \\
&~\leq~
p^{|E_0|} \cdot (p^4)^{r - |D|} \cdot 2^{|E_D|} \cdot p^{|E_D|} \\[3 pt]
&~\leq~ 
2^{4r} \cdot p^{4(r - D) + |E_0| + |E_D|} \mper
\end{align*}

We now count the number of terms which contribute to the sum.
Fix a graph $H$ with $r$ labelled edges $I^1, \dots, I^r$ (possibly repeated) and $q := q(H)$
vertices, without any isolated vertex (so $q \leq 2r$). There are at most $\binom{q}{2}^r \leq
(2r)^{2r}$ such graphs. Then $I^1, \dots, I^r$, as edges in $\binom{[n]}{2}$, are determined by a
map $V_H \to [n]$. There are at most $n^q$ such mappings. 

Let $E_0 := E_0(H), D := D(H), E_j := E_j(H), E' := E'(H)$ be defined as before. Note that $E_0$ is
set the edges of $H$.  As observed before, the contribution from $H$ is $0$ if there exists $j \in
D$ such that $|E_j| = 4$ and $E_j$ is disjoint from $\{ E_{j'} \}_{j' \in D \setminus \{ j \}}$. Let
$\calH$ be the set of $H$ that has nonzero contribution. 
Then,
\begin{align*}
\Ex{\Tr{\inparen{\hN_E}^r}} 
& ~=~ 
\sum_{I^1, \dots, I^r \in \binom{[n]}{2}} \Ex{ \prod_{k = 1}^r \hN_E[I^k, I^{k + 1}] } \\
&~\leq~ 
\sum_{H \in \calH} n^{q(H)} \cdot 2^{4r} \cdot p^{4(r - D(H)) + |E_0(H)| + |E_D(H)|} \\
& ~\leq~ 
(2r)^{2r}  \cdot \max_{H \in \calH} \inparen{ n^{q(H)} 2^{4r} \cdot p^{4(r - D(H)) + |E_0(H)| + |E_D(H)|} } \\
& ~\leq~ 
(8r)^{2r}   \cdot \max_{H \in \calH} \inparen{ n^{q(H)} p^{4(r - D(H)) + |E_0(H)| + |E_D(H)|} }
% & = (8r)^{2r}   \max_{H \in \calH} \big( n^{q(H) - (4(r - D(H)) +  |E_0(H)| + |E'(H)|) / 3} 
%  \big).
\end{align*}
We will prove the following bound on the maximum contribution of any $H \in \calH$.
\begin{claim}\label{claim:poly:contribution}
Let $\calH$ be defined as above. Then, for all $H \in \calH$, we have
\[
n^{q(H)} p^{4(r - D(H)) + |E_0(H)| + |E_D(H)|} ~\leq~ n^2 \cdot B_p^r \mcom
\]
where
\[
B_p
~=~ 
\begin{cases}
n^{3/2} \cdot p^{4} & \text{for}~ p \in \insquare{n^{-1/3}, n^{-1/4}} \\[5 pt]
n^{5/3} \cdot p^{14/3} & \text{for}~ p \in \insquare{n^{-1/4}, 1/2} 
\end{cases} \mper
\]
\end{claim}
Using the above claim, we can bound $\Ex{\Tr{\inparen{\hN_E}^r}}$ as 
\begin{align*}
\Ex{\Tr{\inparen{\hN_E}^r}}
& ~\leq~ (8r)^{2r}   \cdot \max_{H \in \calH} \big( n^{q(H)} p^{4(r - D(H)) + |E_0(H)| + |E_D(H)|}
 \big) \\
& ~\leq~ (8r)^{2r} \cdot n^2 \cdot B_p^r \mcom
\end{align*}
where $B_p$ is given by~\cref{claim:poly:contribution} for different ranges of $p$. By Markov's
inequality, we get that with probability $1 - \frac{1}{n}$, we have $\Tr{\inparen{\hN_E}^r} \leq
(8r)^{2r} \cdot n^3 \cdot B_p^r$, which gives
\[
\norm{2}{\hN_E} ~\leq~ \inparen{8r}^2 \cdot B_p \cdot n^{3/r} \mper
\]
Choosing $r = \Theta(\log n)$ then proves the lemma.
%
%
% and $\Tr{(\hN_E)^r} \leq (8r)^{2r} n^{r/6 + 3}$ with probability $1 - \frac{1}{n}$ (we indeed proved $\E[|\Tr{(\hN_E)^r}|]$, so we can apply Markov's inequality). This implies that $\| \hN_E \|_2 \leq (\Tr{(\hN_E)^r})^{1/r} \leq (8r)^{2} n^{1/6 + 3/r}$. Setting $r = \log n$ gives 
% $\| \hN_E \|_2 \leq O(n^{1/6} \log^2 n )$, proving the lemma. 
\end{proof}

It remains to prove~\cref{claim:poly:contribution}.

\subsubsection{Analyzing Contributing Subgraphs}
Recall that  graphs $H \in \calH$ were constructed from edges $\{I^1, \ldots, I^r\}$, with edge
$I^j$ consisting of vertices $\{i_1^j, i_2^j\}$. Also, we define $q(H) = \abs{V(H)}$.
Moreover, we defined the following sets for graph $H$
\begin{align*}
  E_0(H) &~\defeq~\{I^1, \ldots, I^r\}  \quad\text{(counting only distinct edges)}\\
  D(H) &~\defeq~\inbraces{j \in [r] \suchthat \abs{\inbraces{i_1^j, i_2^j, i_1^{j+1}, i_2^{j+1}}} = 4}\\
  E_j(H) &~\defeq~\inbraces{\{i_1^j, i_1^{j+1}\}, \{i_1^j, i_2^{j+1}\}, \{i_2^j, i_1^{j+1}\},\{i_2^j, i_2^{j+1}\}}
  \setminus E_0(H) \\
  E_D(H) &~\defeq~\bigcup_{j \in D} E_j(H)\mper
\end{align*}
Moreover, the graph $H$ is in $\calH$ only if for every $j \in D$, either $\abs{E_j(H)} \leq 3$ or
there exists $j' \in D\setminus\{j\}$ such that $E_j(H) \cap E_{j'}(H) \neq \emptyset$. 
The~\cref{claim:poly:contribution} then follows from the following combinatorial claim (taking $b =
\log(1/p)/\log n$).
\begin{claim}
Any graph $H \in \calH$ satisfies, for all $b \in [0,1/3]$
\[
q(H) ~\leq~  2 +  b \cdot \inparen{4(r - \abs{D(H)}) +  |E_0(H)| + |E_D(H)|} + c \cdot r \mcom
\]
where $c = 5/3 - 14b/3$ ~for $b \in [0,1/4]$ and $c = 3/2 - 4b$ ~for $b \in [1/4, 1/3]$.
\end{claim}
%
% The following claim bounds $q(H) - (4(r - D(H)) +  |E_0(H)| + |E'(H)|) / 3$.
% \begin{claim}
% Any graph $H \in \calH$ satisfies 
% \[
% q(H) \leq 2 +  \frac{4(r - D(H)) +  |E_0(H)| + |E'(H)|}{3} + \frac{r}{6}.
% \]
% \end{claim}
\begin{proof}
Fix a graph $H \in \calH$. 
Let $j = 1, \dots, r$, let $V_j := \inbraces{ i^{j}_1, i^{j}_2 }_{1 \leq j \leq r}$ (i.e., the set
of vertices covered by $I^1, \dots, I^j$). 
For each $j = 2, \dots, r$, let $v_j := |V_j| - |V_{j - 1}|$ and classify the index $j$ to one of
the following types.
\begin{itemize}
\item Type $-1$: $I^j \cap I^{j - 1} \neq \emptyset$ (equivalently, $j-1 \notin D$).
\item Type $k$ ($0 \leq k \leq 2$): $I^j$ and $I^{j - 1}$ are disjoint, and $v_j = k$ (i.e., adding
  $I^j$ introduces $k$ new vertices). 
\end{itemize}
Let $T_k$ ($-1 \leq k \leq 2$) be the set of indices of Type $k$, and let $t_k := |T_k|$. 
The number of vertices $q$ is bounded by 
\[
q 
~\leq~ 
2 + 1 \cdot t_{-1} + 0 \cdot t_0 + 1 \cdot t_1 + 2 \cdot t_2 
~=~ 
2 + t_{-1} + t_1 + 2t_2.
\]

Let $H_j$ be the graph with $V_j$ as vertices and edges
\[
E(H_j) ~=~ \inbraces{I^1, \ldots, I^j} \bigcup \inparen{\bigcup_{k \in D \cap [j-1]} E_k} \mper
\]
%
% $\{ I^k \}_{1 \leq k \leq j } \cup (\cup_{1 \leq k
%   \leq j - 1} E_{j})$ as edges.  
%
For $j = 2, \dots, r$, let $e_j = \abs{E(H_j)} - \abs{E(H_{j-1})}$. 
% %
% The total number of edges of $H_j$ is at least $e_2 + \dots + e_r$. 
% %
For an index $j \in T_2$, adding two vertices ${i^j_1, i^j_2}$ introduces at least $5$ edges in
$H_j$ compared to $H_{j - 1}$ (i.e., six edges in the $4$-clique on $\inbraces{i^{j-1}_1, i^{j-1}_2,
  i^j_1, i^j_2 }$ except $I^{j-1}$), so $e_j \geq 5$. Similarly, we get $e_j \geq 3$ for each $j \in T_1$. 

The lemma is proved via the following charging argument. For each index $j = 2, \dots, r$, we get
value $b$ for each edge in $H_j \setminus H_{j - 1}$ and get value $c$ for the new index. 
If $j \in T_{-1}$, we get an additional value of $4b$. We give this value to vertices in $V_{j}
\setminus V_{j-1}$.
If we do not give more value than we get and each vertex in $V(H) \setminus V_1$ 
gets more than $1$, this means 
\[
q - 2 ~\leq~ b \cdot \inparen{|E_0| + |E_D| + 4(r - \abs{D(H)})} + c \cdot {r},
\]
proving the claim. 
For example, if $j$ is an index of Type $1$, it gets a value at least $3b+c$ and needs to give value
$1$, such a charging can be done if $3b+c \geq 1$. Similarly, a type 0 vertex does not need to give
any value and has a surplus. We will choose parameters so that each $j$ of types $-1$, $1$ or $0$ 
can distribute the value to vertices added in $V_j \setminus V_{j-1}$.
However, if $j$ is an index of Type $2$, it needs to distribute the value it gets ($5b+c$) to two
vertices, and we will allow it to be ``compensated'' by vertices of other types, which may have a
surplus.

% we can get as small as $5/3 + 1/6 = 11/6$, while we need to give value $2$. 
% We fix this problem by performing the distribution for such Type $2$ indices with indices of other Types. 

Consider an index $j \in T_2$. The fact that $j \in T_2$ guarantees that earlier edges $I^1, \dots,
I^{j - 1}$ are all vertex disjoint from $I^j$. 
If later edges $I^{j+1}, \dots, I^{r}$ are all vertex disjoint from $I^j$, then $|E_{j-1}| = 4$ and
$E_{j-1}$ is disjoint from $\{ E_{j'} \}_{j' \in D \setminus \{ j-1 \}}$, and this means that $H
\not \in \calH$. 
Thus, there exists $j' > j$ such that $I^{j'}$ and $I^j$ share a vertex. Take the smallest $j' >
j$, and say that $j'$ {\em compensates} $j$. Note that $j' \notin T_2$.

We will allow a type 1 index to compensate at most one of type 2 index, and a type -1 or 0 index to
compensate at most two of type 2 indices. We consider below the constraints implied by each kind of
compensation.

% Note that Type $2$ indices cannot compensate another Type $2$ index, Type $1$ indices compensate at most one Type $2$ index, and Type $-1$ and Type $0$ indices can compensate at most two Type $2$ indices. Thus, there are six different types of compensators. For each type, we check whether the collective distribution scheme works so that we give each vertex value $1$. 

\begin{enumerate}
\item \emph{One Type $1$ index $j'$ compensates one Type $2$ index $j$} \\[3 pt]
$v_{j'} + v_{j} = 3$ and $e_{j'} + e_{j} \geq 8$ ($5$ from $e_j$ and $3$ from
  $e_{j'}$). This is possible if $8b + 2c \geq 3$.
%
% Distribution works since $3 \leq 8/3 + 2/6$. 
%
% \end{itemize}
\item \emph{One Type $0$ index $j'$ compensates one Type $2$ index $j$} \\[3 pt]
$v_{j'} + v_{j} = 2$ and $e_{j'} + e_{j} \geq 5$ ($5$ from $e_j$).
This is possible if $5b + 2c \geq 2$.
\item \emph{One Type $0$ index $j'$ compensates two Type $2$ indices $j_1$ and $j_2$ (say $j_1 < j_2$)}. \\[3 pt]
There are two cases.
\begin{enumerate}
\item $e_{j'} + e_{j_1} + e_{j_2} \geq 11$: 
$v_{j'} + v_{j_1} + v_{j_2} = 4$. 
This is possible if $11b + 3c \geq 4$.
% Distribution works since $4 \leq 11/3 + 3/6$. 
% 
\item $e_{j'} + e_{j_1} + e_{j_2} = 10$: since $e_{j_1}, e_{j_2} \geq 5$, 
this means that $e_{j'} = 0$.

First, we note that since $j_1$ is a type 2 index and $j'$ is the smallest index $j$ such that
$I^{j_1} \cap I_{j} \neq \emptyset$, in the graph $H_{j'-1}$, vertices in $I^{j_1}$ only have edges
to vertices in $I^{j_1 - 1}$ and $I^{j_1 + 1}$. Similarly, vertices in $I^{j_2}$ only have edges
to vertices in $I^{j_2 - 1}$ and $I^{j_2 + 1}$.

Since $I^{j'}$ shares one vertex each with $I^{j_1}$ and $I^{j_2}$, let $I^{j'} =
\{i^{j'}_1, i^{j'}_2\}$ with $i^{j'}_1 \in I^{j_1}$ and $i^{j'}_2 \in I^{j_2}$.
Since $e_{j'} = 0$ means that $I^{j'} = \{i^{j'}_1, i^{j'}_2\}$ was in $H_{j' - 1}$. However,
this is an edge between vertices in $I^{j_1}$ and $I^{j_2}$. By the above argument, this is only
possible if $j_2 = j_1 + 1$. Also, since $j'$ is type 0 and $I^{j'}$ shares a vertex with $I^{j_2}$,
we must have $j' > j_2 + 1$ (otherwise $j'$ would be type -1).

Consider $I^{j' - 1}$, which are vertex disjoint from both $I^{j_1}$ and $I^{j_2}$. If $I^{j' - 1}
\neq I^{j_1 - 1}$, at least one edge between $I^{j' - 1}$ and $I^{j}$ was not in $H_{j' - 1}$,
contradicting the assumption $e_{j'} = 0$. Therefore, $I^{j' - 1} = I^{j_1 - 1}$. For the same
reason, $I^{j' - 1} = I^{j_2 + 1}$. Thus, in particular, we have $I^{j_2+1} = I^{j_1-1}$. Thus,
$j_2+1$ is also a type 0 index. Moreover, it cannot compensate any previous index, since any such
index would already be compensated by $j_1-1$.

In this case we consider that $I^{j_2 + 1}$ and $I^{j'}$ jointly compensate $j_1$ and $j_2$. 
$v_{j'} + v_{j_2+1} + v_{j_1} + v_{j_2} = 4$ and $e_{j_2+1} + e_{j'} + e_{j_1} + e_{j_2} \geq 10$. 
Compensation is possible if $10b + 4c \geq 4$.
\end{enumerate}
\item \emph{One Type $-1$ index $j'$ compensates one Type $2$ index $j$}. \\[3 pt]
$v_{j'} + v_{j} \leq 3$ and $e_{j'} + e_{j} \geq 5$  ($5$ from $e_j$). Compensation is possible if
$5b + 4b + 2c \geq 2$. 
\item \emph{One Type $-1$ index $j'$ compensates two Type $2$ indices $j_1$ and $j_2$} \\[3 pt]  
We have $v_{j'} + v_{j_1} + v_{j_2} \leq 5$ and $e_{j'} + e_{j_1} + e_{j_2} \geq 10$.
Compensation is possible if $10b + 4b + 3c \geq 5$.
\end{enumerate}
Each index $j$ of Type $2$ is compensated by exactly one other index $j'$. We also require indices
of types $1$ and $-1$ which do not compensate any other index, to have value at least 1 (to account
for the one vertex added). This is true if $3b+c \geq 1$ and $4b+c \geq 1$.

Aggregating the above conditions (and discarding the redundant ones), we take
\[
c = \max \inbraces{\frac32 - 4b, ~1 - \frac{5b}{2}, ~\frac43 - \frac{11b}{3}, ~\frac53 - \frac{14b}{3} }
\]
It is easy to check that the maximum is attained by $c = 5/3 - 14b/3$ when $b \in [0,1/4]$ and $c =
3/2 - 4b$ when $b \in [1/4, 1/3]$.
\end{proof}

%%% Local Variables: 
%%% mode: latex
%%% TeX-master: "main"
%%% End: 

%% file: poly_opt/8.fsp-lower-bounds.tex
\section[Lifting Spectral Lower Bounds to Higher Levels]{Lifting $\fsp{\cdot}$ Lower Bounds to Higher
Levels}\label{sec:poly:fsp-lowerbound}
For a matrix $B\in\Re^{[n]^{q/2}\times [n]^{q/2}}$, let 
$B^S$ denote the matrix obtained by symmetrizing $B$, i.e. 
for any $I,J\in [n]^{q/2}$, 
\[
B^S[I,J] ~~
:= ~~
\frac{1}{|\orbit{\mi{I}+\mi{J}}|}
\cdot ~~~
\sum_{
\mathclap{
\substack{
\mi{I'}+\mi{J'} \\
= \mi{I}+\mi{J}
}
}
} ~
B[I',J'] 
\]
Equivalently, $B^S$ can be defined as follows: 
\[
    B^S = \frac{1}{q!}\cdot \sum_{\pi \in \Sym_{q}} B^{\pi}
\]
where for any $K\in [n]^{q}$, $B^{\pi}[K] := B[\pi(K)]$. 

\medskip

For a matrix $M\in\Re^{[n]^2\times [n]^2}$ let $T\in \RR^{[n]^4}$ denote 
the tensor given by, $T[i_1,i_2,i_3,i_4] = M[(i_1,i_2),(i_3,i_4)]$. 
Also for any non-negative integers $x,y$ satisfying $x+y = 4$, let 
$M_{x,y}\in \Re^{[n]^{x}\times [n]^{y}}$ denote the matrix given by, 
$M[(i_1,\dots ,i_x),(j_1,\dots j_y)] = T[i_1,\dots ,i_x,j_1,\dots j_y]$. 
We will use the following result that we prove in~\cref{sec:poly:lifting:stable:lbs}. 

\begin{theorem}[Lifting ``Stable'' Spectral Lower Bounds]\label{thm:poly:lift:stable:sp:lb}
    Let $M\in \Re^{[n]^2 \times [n]^2}$ be a degree-$4$-SOS-symmetric matrix satisfying  
    \[
        \norm{S_1}{M},~ 
        \norm{S_1}{M_{3,1}}
        \leq 1.
    \]
    Then for any $q$ divisible by $4$,
    \[
    	\norm{S_1}{
        \pth{M^{\otimes q/4}}^{S}
        }
        = 2^{O(q)}
    \]
\end{theorem}

\subsection[Gap Between Spectral Relaxation and Optimal for NNC Polynomials]{Gap Between $\fsp{\cdot}$ and $\ftwo{\cdot}$
for Non-Neg. Coefficient Polynomials}

\begin{lemma}\label{lem:poly:simple:sp:lb}
    Consider any homogeneous polynomial $g$ of even degree-$t$ 
    and let $M_g\in \Re^{[n]^{t/2}\times [n]^{t/2}}$ be its SoS-symmetric matrix 
    representation. Then $\fsp{g}\geq \norm{F}{M_g}^2/\norm{S_1}{M_g}$. 
\end{lemma}

\begin{proof}
    We know by strong duality, that 
    \[
        \fsp{g} = 
        \max\inbraces{
        \iprod{X}{M_g} 
        \sep{
        \norm{S_1}{X}=1,~~
        X\text{ is SoS-Symmetric},~~
        X\in \Re^{[n]^{t/2}\times [n]^{t/2}}
        }
        }.
    \]
    The claim follows by substituting $X:= M_g/\norm{S_1}{M_g}$. 
\end{proof}

\begin{theorem}
    For any $q$ divisible by $4$ and $f$ as defined in~\cref{sec:poly:nnc-lowerbound}, we have that 
    w.h.p. \[\frac{\fsp{f^{q/4}}}{\ftwo{f^{q/4}}} ~\geq~ 
    \frac{n^{q/24}}{(q \log n)^{O(q)}}.\]
\end{theorem}

\begin{proof}
    Let $f$ be the degree-$4$ homogeneous polynomial as defined in~\cref{sec:poly:nnc-lowerbound} and 
    let $M=M_f$ be its SoS-symmetric matrix representation. 
    Let $g:= f^{q/4}$ and let $M_g$ be its SoS-symmetric matrix representation. 
    Thus $M_g = (M^{\otimes q/4})^{S}$ and it is easily verified that 
    w.h.p., $\norm{F}{M}^2 \ge \widetilde{\Omega}(n^4p^6) = \widetilde{O}(n^2)$ and also  
    $\norm{F}{M_g}^2 \ge \widetilde{\Omega}((n^4p^6)^{q/4}/q^{O(q)}) = \widetilde{\Omega}(n^{q/2}/q^{O(q)})$.
   
    It remains to estimate $\norm{S_1}{M_g}$ so that we may apply~\cref{lem:poly:simple:sp:lb}. 
    Implicit, in the proof of~\cref{lem:poly:eigenvalue}, is that w.h.p. $M$ has one eigenvalue of 
    magnitude $O(n^2p) = O(n^{5/3})$ and at most $O(n^2p) = O(n^{5/3})$ eigenvalues of 
    magnitude $\widetilde{O}(n^{3/2} p^4) = \widetilde{O}(n^{1/6})$. Thus,
    $\norm{S_1}{M} = \widetilde{O}(n^{11/6})$ w.h.p. 
    Now we have that $\norm{S_1}{M_{1,3}} \leq \sqrt{n}\cdot \norm{F}{M_{1,3}} = 
    \sqrt{n}\cdot \norm{F}{M} = \widetilde{O}(n^{3/2})$ w.h.p. Thus, on applying~\cref{thm:poly:lift:stable:sp:lb}
    to $M/\widetilde{O}(n^{11/6})$, we get that 
    $\norm{S_1}{M_g}/\widetilde{O}(n^{11q/24}) \leq 2^{O(q)}$ w.h.p. 
    
    Thus, applying~\cref{lem:poly:simple:sp:lb} yields the claim. 
\end{proof}

\subsection{Tetris Theorem}
    Let $M\in \Re^{[n]^2\times [n]^2}$ be a degree-$4$ SoS-Symmetric matrix, 
    let $M_A:= M_{3,1}\otimes M_{0,4}\otimes M_{3,1}$, let $M_B := M_{3,1}\otimes M_{1,3}$, 
    let $M_C:=M$ and let $M_D:= \Vector{M}\Vector{M}^T = M_{0,4}\otimes M_{4,0}$. 
    For any permutation $\pi\in \Sym_{q/2}$ let $\overline{\pi}\in \Sym_{n^{q/2}}$ denote 
    the permutation that maps any $i\in [n]^{q/2}$ to $\pi(i)$. Also let 
    $P_{\pi} \in \Re^{[n]^{q/2} \times [n]^{q/2}}$ denote the row-permutation 
    matrix induced by the permutation $\overline{\pi}$. Let $P:=\sum_{\pi\in \Sym_{q/2}} 
    P_{\pi}$. Let $\mathcal{R}(a,b,c,d) := {(c!^{\,2} \,2!^{\,2c})~(b!\,(2a+b)!\,3!^{\,2a+2b})
    ~(d!\,(a+d)!\,4!^{\,a+2d})}$.
    Define 
    \begin{align*}
        \tetmat(a,b,c,d) &:= 
        \frac{M_A^{\otimes a}\otimes M_B^{\otimes b}\otimes M_C^{\otimes c}\otimes 
        M_D^{\otimes d}}{\mathcal{R}(a,b,c,d)}\\
        \overline{\tetmat}(a,b,c,d) &:= 
        \frac{(M_A^T)^{\otimes a}\otimes 
        M_B^{\otimes b}\otimes M_C^{\otimes c}\otimes M_D^{\otimes d}}{\mathcal{R}(a,b,c,d)}\\
        S^\tetmat &:= \inbraces{P\cdot \tetmat(a,b,c,d) \cdot P^T\sep{12a+8b+4c+8d=q}}~\bigcup \\
        &\phantom{:=\,\,}
        \inbraces{P \cdot \overline{\tetmat}(a,b,c,d) \cdot P^T\sep{12a+8b+4c+8d = q}} \mper
    \end{align*}
    
\begin{theorem}\label{thm:poly:tetris}
    Let $M\in \Re^{[n]^2 \times [n]^2}$ be a degree-$4$-SOS-symmetric 
    matrix. Then 
    \begin{align}\label{eq:poly:tetris:permute}
        (q/4)!\cdot 4!^{\,q/4}\cdot 
        \sum\limits_{\tetmat\in S^{\mathfrak{M}}} \tetmat 
        \quad = \quad 
        \sum\limits_{\pi \in \Sym_q} 
        \pth{M^{\otimes q/4}}^{\pi}
        \quad = \quad 
        q!\cdot 
        \pth{M^{\otimes q/4}}^{S}
    \end{align}
\end{theorem}

We shall prove this claim in~\cref{sec:poly:tetris:proof} after first exploring its 
consequences.

\subsection{Lifting Stable Degree-$4$ Lower Bounds}\label{sec:poly:lifting:stable:lbs}

\begin{theorem}[Lifting ``Stable'' Spectral Lower Bounds: Restatement
  of~\cref{thm:poly:lift:stable:sp:lb}]\label{thm:poly:nuclear:after:lift}
  Let $M\in \Re^{ {[n]}^2 \times {[n]}^2}$ be a degree-$4$-SOS-symmetric matrix satisfying  
    \[
        \norm{S_1}{M},~\norm{S_1}{M_{3,1}}
        \leq 1.
    \]
    Then for any $q$ divisible by $4$,
    \[
    	\norm{S_1}{
        \pth{M^{\otimes q/4}}^{S}
        }
        = 2^{O(q)}
    \]
\end{theorem}

\begin{proof}
	Implicit in the proof of~\cref{thm:poly:tetris} is the following:
	\begin{align}\label{eq:poly:tri:ineq}
		&\quad\quad q!\cdot \pth{M^{\otimes q/4}}^{S}
		\nonumber\\
		&=
		\sum_{12a+8b+4c+8d=q}~
        \frac{(q/4)!\cdot 4!^{\,q/4}}{\mathcal{R}(a,b,c,d)}
        ~\sum\limits_{\sigma_1,\sigma_2 \in \Sym_{q/2}}~ 
        \pth{P_{\sigma_1}^T
        \pth{
        M_A^{\otimes a}\otimes M_B^{\otimes b}\otimes M_C^{\otimes c}\otimes M_D^{\otimes d}
        }
        P_{\sigma_2}}
        \nonumber\\
        &+
        \sum_{12a+8b+4c+8d=q}~
        \frac{(q/4)!\cdot 4!^{\,q/4}}{\mathcal{R}(a,b,c,d)}
        ~\sum\limits_{\sigma_1,\sigma_2 \in \Sym_{q/2}}~ 
        \pth{P_{\sigma_1}^T
        \pth{
        (M_A^T)^{\otimes a}\otimes M_B^{\otimes b}\otimes M_C^{\otimes c}\otimes M_D^{\otimes d}
        }
        P_{\sigma_2}}
	\end{align}
	
	First note that $(q/4)!\cdot 4!^{\,q/4}/\mathcal{R}(a,b,c,d) 
	\leq 2^{O(q)}$ since for any integers $i,j,k,l$ 
	\[(i+j+k+l)!/(i!\cdot j!\cdot k!\cdot l!) 
	\leq 4^{i+j+k+l}.\] 
	
	Next note that $\norm{S_1}{M_{0,4}} = \norm{F}{M} \leq \norm{S_1}{M} \leq 1$. 
	Combining this with the fact that $\norm{S_1}{M_{1,3}}, \norm{S_1}{M}\leq 1$, 
	we get that $\norm{S_1}{M_A}, \norm{S_1}{M_B}, \norm{S_1}{M_C}, 
	\norm{S_1}{M_D}\leq 1$, since $\norm{S_1}{X\otimes Y} = 
	\norm{S_1}{X}\cdot \norm{S_1}{Y}$ for any (possibly rectangular) 
	matrices $X$ and $Y$. 
	Further note that Schatten $1$-norm is invariant to multiplication by a permutation 
	matrix. Thus, the claim follows by applying triangle inequality to the $O_q(q^{q})$ 
	terms in~\cref{eq:poly:tri:ineq}. 
\end{proof}

\iffalse

\begin{corollary}[Objective Value after Lifting]\label{cor:poly:objective:after:lift}
    Let $f$ be a degree-$4$ homogeneous polynomial and let $M^f$ be its SoS-symmetric 
    matrix representation. Let $M\in \Re^{[n]^2 \times [n]^2}$ be a degree-$4$-SOS-symmetric 
    matrix satisfying  
    \[
        \delta := \iprod{M^f}{M} \geq 0
    \]
    Then for any $q$ divisible by $4$,
    \[
    	    \iprod{(M^f)^{\otimes q/4}}{(M^{\otimes q/4})^S} 
    	    ~~\geq ~~
    	    \frac{\delta^{q/4}}{q^{O(q)}}
    \]
\end{corollary}

\begin{proof}
    We start with the observation that for any symmetric matrices $G,H\in \Re^{[n]^{q/2}
    \times [n]^{q/2}}$, we have 
    \begin{align*}
        \frac{\iprod{PGP}{PHP}}{(q/2)!^{4}} 
        &=
        \frac{\iprod{GP^2H}{P^2}}{(q/2)!^{4}} 
        &&\text{(cyclic property of Trace)} \\
        &\geq 
        \frac{\iprod{GP^2H}{P^2}}{(q/2)!^{4}} 
        &&\text{(cyclic property of Trace)} \\
    \end{align*}

    For any $12a+8b+4c+8d = q$, let 
    \begin{align*}
        G &:= (M^f_A)^{\otimes a}\otimes (M^f_B)^{\otimes b}
        \otimes (M^f_C)^{\otimes c}\otimes (M^f_D)^{\otimes d} \\
        H &:= M_A^{\otimes a}\otimes M_B^{\otimes b}
        \otimes M_C^{\otimes c}\otimes M_D^{\otimes d} 
    \end{align*}
    Observe that $•$
    
\end{proof}

\fi

%

\begin{corollary}[Lifting ``Stable'' SoS Lower Bounds]\label{cor:poly:lift:stable:sos:lb}
    Let $M\in \Re^{[n]^2 \times [n]^2}$ be a degree-$4$-SOS-symmetric matrix satisfying  
    \[
        M \succeq 0, \quad\quad
        M_A := M_{3,1}\otimes M_{0,4}\otimes M_{3,1} \succeq 0, \quad\text{ and }\quad 
        M_B := M_{3,1}\otimes M_{1,3} \succeq 0.
    \]
    Then for any $q$ divisible by $4$,
    \[
        \pth{M^{\otimes q/4}}^{S} \succeq 0
    \]
    (i.e. $\pth{M^{\otimes q/4}}^{S}$ is a degree-$q$ SOS moment matrix).
\end{corollary}

\begin{proof}
    Observe that~\cref{thm:poly:tetris} implies the 
    claim since $M_A,M_B,M_C,$ and $M_D$ are PSD and the set of PSD 
    matrices is closed under transpose, Kronecker product, scaling, 
    conjugation, and addition.     
\end{proof}

\subsection{Proof of Tetris Theorem}\label{sec:poly:tetris:proof}

We start with defining a ~\emph{hypergraphical matrix} which will allow a more 
intuitive paraphrasing of~\cref{thm:poly:tetris}. By now, this is an important 
formalism in the context of SoS, and closely-related objects have been defined in 
several works, including \cite{DM15}, \cite{RRS16}, \cite{BHKKMP16}. 

\subsubsection{Hypergraphical Matrix}

\begin{definition}\label{def:poly:template-hypergraph}
    For symbolic sets $L=\{\ell_1,\dots \ell_{q_1}\},
    R=\{r_1,\dots r_{q_2}\}$, a $d$-uniform \emphi{template-hypergraph} represented by 
    $(L,R,E)$, is a $d$-uniform hypergraph on vertex set $L\uplus R$ with $E$ 
    being the set of hyperedges. 
    
    For $I=(i_1, \dots i_{q_1}) [n]^{q_1}, J=(j_1, \dots j_{q_2})\in [n]^{q_2}$, 
    we also define a related object called \emphi{edge-set instantiation} 
    (and denoted by $E(I,J)$) as the set of size-$d$ multisets induced by $E$ 
    on substituting $\ell_t = i_t$ and $r_t=j_t$. 
\end{definition}

\paragraph{Remark.}
There is a subtle distinction between $E$ and $E(I,J)$ above, in that $E$ is a set of 
$d$-sets and $E(I,J)$ is a set of size-$d$ multisets (i.e. $e\in E(I,J)$ 
can have repeated elements). 

\begin{definition}\label{def:poly:hypergraphical-matrix}
    Given an SoS-symmetric order-$d$ tensor $\sfT$  and a $d$-uniform template-hypergraph 
    $H=(L,R,E)$ with $|L|=q_1,|R|=q_2$, we define the $d$-uniform degree-$(q_1,q_2)$ 
    hypergraphical matrix $\hgm{\sfT}{H}$ as 
    \[
        \hgm{\sfT}{H}[I,J] 
        = 
        \prod_{e\in E(I,J)} \sfT[e]
    \]    
    for any $I\in [n]^{q_1}, J\in [n]^{q_2}$. 
\end{definition}

In order to represent~\cref{thm:poly:tetris} in the language of hypergraphical matrices, 
we first show how to represent $M^{\otimes q/4}$ and 
$M_A^{\otimes a}\otimes M_B^{\otimes b}\otimes M_C^{\otimes c}\otimes M_D^{\otimes d}$ 
in this language. 

\subsubsection{Kronecker Products of Hypergraphical Matrices}\label{sec:poly:kron:prod}

We begin with the observation that the kronecker product of hypergraphical matrices 
yields another hypergraphical matrix (corresponding to the "disjoint-union" of the 
template-hypergraphs). 

\begin{definition}\label{def:poly:th:disjoint:union}
    let $H=(L,R,E)$, 
    $H'=(L',R',E')$ be template-hypergraphs with 
    $|L|=q_1, |R|=q_2,|L'|=q_3,|R'|=q_4$. Let 
    $\overbar{H}=(\overbar{L},\overbar{R},\overbar{E})$ be a template-hypergraph 
    with $|\overbar{L}|=q_1+q_3, |\overbar{R}|=q_2+q_4$, where 
    $\overbar{\ell}_t = \ell_t$ ~if $t\in [q_1]$, 
    $\overbar{\ell}_t = \ell'_t$ ~if $t\in [q_1+1,q_1+q_3]$, 
    $\overbar{r}_t = r_t$ ~if $t\in [q_2]$, 
    $\overbar{r}_t = r'_t$ ~if $t\in [q_2+1,q_2+q_4]$, and 
    $\overbar{E}=E\uplus E'$. We call $\overbar{H}$ the \emphi{disjoint-union} 
    of $H$ and $H'$, which we denote by $H\uplus H'$. 
\end{definition}

\begin{observation}\label{obs:poly:hgm:kron:prod}
    Let $\sfT$ be an SOS-symmetric order-$d$ tensor and let $H=(L,R,E)$, 
    $H'=(L',R',E')$ be template-hypergraphs. Then, 
    \[
        \hgm{\sfT}{H} \otimes \hgm{\sfT}{H'}
        =
        \hgm{\sfT}{H\uplus H'}
    \]
\end{observation}

\paragraph{Remark.}
Note that the disjoint-union operation on template-hypergraphs does 
not commute, i.e. $\hgm{\sfT}{H\uplus H'}\neq \hgm{\sfT}{H'\uplus H}$ 
(since kronecker-product does not commute). 
\bigskip

\noindent
Now consider a degree-$4$ SoS-symmetric matrix $M$ 
(as in the statment of~\cref{thm:poly:tetris}) and let $\sfT$ be the SoS-symmetric 
tensor corresponding to $M$. Then for any $x+y=4$ we have that 
$M_{x,y} = \hgm{\sfT}{H_{x,y}}$, where $H_{x,y}=(L,R,E)$ is the template-hypergraph 
satisfying $L=\{\ell_1,\dots \ell_x\},R=\{r_1,\dots r_y\}$ and 
$E=\{\{\ell_1,\dots \ell_x,r_1,\dots r_y\}\}$. Combining this observation 
with~\cref{obs:poly:hgm:kron:prod} yields that 
$M_A = \hgm{\sfT}{H_A},~ M_B = \hgm{\sfT}{H_B},~ 
M_C = \hgm{\sfT}{H_C},~ M_D = \hgm{\sfT}{H_D}$, where 
$H_A := H_{3,1}\uplus H_{0,4}\uplus H_{3,1}$,~$H_B := H_{3,1}\uplus H_{1,3}$,~ 
$H_C := H_{2,2}$,~and $H_D := H_{4,0}\uplus H_{0,4}$. 
Lastly, another application of~\cref{obs:poly:hgm:kron:prod} to the above, yields 
\begin{observation}\label{obs:poly:hgm:kron:prod:tetris}
For a template-hypergraph $H$, let $H^{\uplus t}$ denote $\biguplus_{g\in [t]} H$. 
Then, 
\begin{enumerate}[label=(\arabic*)]
        \item $M^{\otimes q/4} = \hgm{\sfT}{H_{2,2}^{\uplus q/4}}$.
        
        \item $M_A^{\otimes a}\otimes M_B^{\otimes b}\otimes 
        M_C^{\otimes c}\otimes M_D^{\otimes d} 
        = \hgm{\sfT}{H(a,b,c,d)}$ where 
        \[
        H(a,b,c,d)~:=H_A^{\uplus a}
        \uplus
        H_B^{\uplus b}
        \uplus
        H_C^{\uplus c}
        \uplus
        H_D^{\uplus d}
        \]
    \end{enumerate}
\end{observation}
For technical reasons we also define the following related template-hypergraph: 
\[
    \overbar{H}(a,b,c,d) 
    ~:= 
    \overbar{H}_A^{\uplus a}
    \uplus
    H_B^{\uplus b}
    \uplus
    H_C^{\uplus c}
    \uplus
    H_D^{\uplus d},
\]
where $\overline{H}_A$ is the template-hypergraph whose 
corresponding hypergraphical matrix is $M_A^{T}$. 
\smallskip 

To finish paraphrasing~\cref{thm:poly:tetris}, we are left with studying the effect 
of permutations on hypergraphical matrices - which is the content of the 
following section. 

\subsubsection{Hypergraphical Matrices under Permutation}
Recall that for any matrix $B\in \Re^{[n]^{q/2}\times [n]^{q/2}}$ and $\pi\in \Sym_q$, 
$B^{\pi}$ is the matrix satisfying $B^{\pi}[K] := B[\pi(K)]$ where 
$K\in [n]^{q/2}\times [n]^{q/2}$. 

Also recall that for any permutation $\sigma\in \Sym_{q/2}$, $\overline{\sigma}\in 
\Sym_{n^{q/2}}$ denotes the permutation that maps any $i\in [n]^{q/2}$ to $\sigma(i)$ 
and also that $P_{\sigma} \in \Re^{[n]^{q/2} \times [n]^{q/2}}$ denotes the 
$[n]^{q/2}\times [n]^{q/2}$ row-permutation matrix induced by the permutation 
$\overline{\sigma}$. 

We next define a permuted template hypergraph in order to capture how permuting a 
hypergraphical matrix (in the senses above) can be seen 
as permutations of the vertex set of the hypergraph. 

\begin{definition}[Permuted Template-Hypergraph]
    For any $\pi\in \Sym_{q}$ (even $q$), and a $d$-uniform template-hypergraph 
    $H=(L,R,E)$ with $|L|=|R|=q/2$, 
    let $H^{\pi} = (L',R',E)$ denote the template-hypergraph obtained by setting 
    $\ell'_t := k_t$~ and $r_t = k_{t+q/2}$~ for $t\in [q/2]$, where 
    $K=(k_1, \dots k_q) = \pi(L\oplus R)$. 
    
    Similarly, for any $\sigma_1,\sigma_2\in \Sym_{q/2}$, let 
    $H^{\sigma_1,\sigma_2} = (L',R',E)$ denote the 
    template-hypergraph obtained by setting $\ell'_t := \sigma_1(\ell_t)$~ and 
    $r'_t = \sigma_2(r_t)$. 
\end{definition}

\noindent
We then straightforwardly obtain 
\begin{observation}\label{obs:poly:hgm:permutation}
    For any $\pi\in \Sym_{q}$, $\sigma_1,\sigma_2\in \Sym_{q/2}$, SoS-symmetric 
    order-$d$ tensor $\sfT$ and any $d$-uniform template-hypergraph $H$, 
    \begin{enumerate}[label=(\arabic*)]
        \item $\pth{\hgm{\sfT}{H}}^{\pi} = \hgm{\sfT}{H^{\pi}}$
        
        \smallskip
        
        \item $P_{\sigma_1}\cdot \hgm{\sfT}{H} \cdot P^T_{\sigma_2} = 
        \hgm{\sfT}{H^{\sigma_1,\sigma_2}}$
    \end{enumerate}
\end{observation}

Thus, to prove~\cref{thm:poly:tetris}, it remains to establish 
\begin{align}\label{eq:poly:tetris:hypergraphical}
    &\frac{1}{4!^{q/4}\cdot (q/4)!}\cdot 
    \sum_{\pi\in \Sym_{q}} \hgm{\sfT}{(H_{2,2}^{\uplus q/4})^{\pi}} 
    \nonumber\\
    =& 
    \sum_{12a+8b+4c+8d=q}~
    \frac{1}{\mathcal{R}(a,b,c,d)}\cdot 
    \sum_{\sigma_1,\sigma_2\in\Sym_{q/2}} 
    \hgm{\sfT}{H(a,b,c,d)^{\sigma_1,\sigma_2}} ~~+ 
    \nonumber\\
    ~&
    \sum_{12a+8b+4c+8d=q}~
    \frac{1}{\mathcal{R}(a,b,c,d)}\cdot 
    \sum_{\sigma_1,\sigma_2\in\Sym_{q/2}} 
    \hgm{\sfT}{\overbar{H}(a,b,c,d)^{\sigma_1,\sigma_2}}
\end{align}
We will establish this in the next section by comparing the template-hypergraphs 
generated (with multiplicities) in the LHS with those generated in the RHS.

\subsubsection{Proof of~\cref{eq:poly:tetris:hypergraphical}}
We start with some definitions to track the template-hypergraphs generated 
in the LHS and RHS of~\cref{eq:poly:tetris:hypergraphical}. For any 
$12a+8b+4c+8d=q$, let 
\begin{align*}
    \mathcal{F}(a,b,c,d) 
    &:= 
    \inbraces{
    H(a,b,c,d)^{\sigma_1,\sigma_2}
    \sep{\sigma_1,\sigma_2\in \Sym_{q/2}}
    } 
    \\ 
    \overbar{\mathcal{F}}(a,b,c,d) 
    &:= 
    \inbraces{
    \overbar{H}(a,b,c,d)^{\sigma_1,\sigma_2}
    \sep{\sigma_1,\sigma_2\in \Sym_{q/2}}
    } 
    \\
    \mathcal{F}
    &:= 
    \inbraces{
    (H_{2,2}^{\uplus q/4})^{\pi}
    \sep{\pi\in \Sym_{q}}
    } 
\end{align*}

\noindent
Firstly, it is easily verified that whenever 
$(a,b,c,d)\neq (a',b',c',d')$, 
$\mathcal{F}(a,b,c,d)\cap \mathcal{F}(a',b',c',d') = \phi$, and that 
$\mathcal{F}(a,b,c,d)\cap \overbar{\mathcal{F}}(a,b,c,d) = \phi$. 
It is also easily verified that for any $12a+8b+4c+8d=q$, 
and any $H\in \mathcal{F}(a,b,c,d)$, 
\[
    \mathcal{R}(a,b,c,d) 
    ~=~ 
    \cardin{\inbraces{
    (\sigma_1,\sigma_2)\in \Sym_{q/2}^{2}
    \sep{H(a,b,c,d)^{\sigma_1,\sigma_2} = H}
    }}
\]
and for any $H\in \mathcal{F}$,
\[
    4!^{q/4}\cdot (q/4)!
    ~=~ 
    \cardin{\inbraces{
    \pi\in \Sym_{q}
    \sep{(H_{2,2}^{\uplus q/4})^{\pi} = H}
    }}.
\]
Thus in order to prove~\cref{eq:poly:tetris:hypergraphical}, it is sufficient 
to establish that 
\begin{equation}\label{eq:poly:tetris:th}
    \mathcal{F} 
    ~ = \qquad
    \biguplus_{\mathclap{12a+8b+4c+8d=q}} ~~
    (\mathcal{F}(a,b,c,d) \uplus \overbar{\mathcal{F}}(a,b,c,d))
\end{equation}
It is sufficient to establish that 
\begin{equation}\label{eq:poly:tetris:th:subset}
    \mathcal{F} 
    ~ \subseteq \qquad
    \biguplus_{\mathclap{12a+8b+4c+8d=q}} ~~
    (\mathcal{F}(a,b,c,d) \uplus \overbar{\mathcal{F}}(a,b,c,d))
\end{equation}
since the other direction is straightforward. To this end, 
consider any $H=(L,R,E)\in \mathcal{F}$, and for any $x+y=4$, define 
\[
    s_{x,y} := \cardin{\inbraces{e\in E\sep{|e\cap L|=x,~|e\cap R|=y}}}.
\]
Now clearly $H\in \mathcal{F}(a,b,c,d)$ iff 
\begin{equation}\label{eq:poly:satisfy:new}
    s_{0,4}=a+d,~ s_{3,1}=2a+b,~ s_{1,3}=b,~ s_{2,2}=c,~ s_{4,0}=d.
\end{equation} 
and $H\in \overbar{\mathcal{F}}(a,b,c,d)$ iff 
\begin{equation}\label{eq:poly:satisfy2:new}
    s_{4,0}=a+d,~ s_{1,3}=2a+b,~ s_{3,1}=b,~ s_{2,2}=c,~ s_{0,4}=d.
\end{equation} 
Thus we need only find $12a'+8b'+4c'+8d' = q$, such that~\cref{eq:poly:satisfy:new} 
or~\cref{eq:poly:satisfy2:new} is satisfied. 

We will assume w.l.o.g. that $s_{0,4}\geq s_{4,0}$ and show that one can satisfy~\cref{eq:poly:satisfy:new},
since if $s_{(0,4)} < s_{(4,0)}$, an identical argument 
allows one to show that~\cref{eq:poly:satisfy2:new} is satisfiable. 
So let $d' = s_{4,0}$, $c'=s_{2,2}$, $b'=s_{1,3}$ and $a' = (s_{3,1}-s_{1,3})/2$. 
Since $H\in \mathcal{F}$, it must be true that $4s_{4,0}+3s_{3,1}+2s_{2,2} + s_{(1,3)} 
= q/2$. Thus, $12a'+8b'+4c'+8d' = 8s_{4,0}+6s_{3,1}+4s_{2,2} + 2s_{1,3} = q$ as 
desired. We will next see that $(a',b',c',d')$ and $s_{x,y}$ satisfy~\cref{eq:poly:satisfy:new}.
We have by construction that $s_{4,0} = d'$, $s_{2,2} = c'$, 
$s_{1,3}=b'$ and $s_{3,1} = 2a'+b'$. It remains to show that 
$s_{0,4} = a'+d'$. Now we know that $4s_{4,0}+3s_{3,1}+2s_{2,2} + s_{1,3} = q/2$ 
and $4s_{0,4}+3s_{1,3}+2s_{2,2} + s_{3,1} = q/2$. Subtracting the two equations yields 
$s_{0,4}-s_{4,0} = (s_{3,1}-s_{1,3})/2$. This implies $a'+d' = s_{4,0}+
(s_{3,1}-s_{1,3})/2 = s_{0,4}$, and furthermore it implies that $a'$ is non-negative 
since we assumed $s_{0,4}\geq s_{4,0}$. So~\cref{eq:poly:satisfy:new} is satisfied. 
Thus we have established~\cref{eq:poly:tetris:th:subset}, which completes the proof of~\cref{thm:poly:tetris}.

%% file: PtoQ/PtoQ_concat.tex
\chapter{Matrix Norm: Hardness and Approximation}

\newcommand{\alfplain}[2]{\widetilde{f}_{#1,\, #2}}
\newcommand{\fplain}[3]{\widetilde{f}_{#1,\, #2}(#3)}
\newcommand{\fparen}[3]{\widetilde{f}_{#1,\, #2}\inparen{#3}}
\newcommand{\fin}[3]{\widetilde{f}^{-1}_{#1,\, #2}(#3)}
\newcommand{\alfin}[2]{\widetilde{f}^{-1}_{#1,\, #2}}
\newcommand{\finabs}[3]{\widetilde{h}_{#1,\, #2}(#3)}
\newcommand{\alfinabs}[2]{\widetilde{h}_{#1,\, #2}}
\newcommand{\hin}[3]{\widetilde{h}^{-1}_{#1,\, #2}(#3)}
\newcommand{\alhin}[2]{\widetilde{h}^{-1}_{#1,\, #2}}
\newcommand{\nf}[3]{f_{#1,\, #2}(#3)}
\newcommand{\alnf}[2]{f_{#1,\, #2}}
\newcommand{\nfext}[3]{f^{+}_{#1,\, #2}(#3)}
\newcommand{\alnfext}[2]{f^{+}_{#1,\, #2}}
\newcommand{\nfin}[3]{f^{-1}_{#1,\, #2}(#3)}
\newcommand{\alnfin}[2]{f^{-1}_{#1,\, #2}}
\newcommand{\nh}[3]{h_{#1,\, #2}(#3)}
\newcommand{\alnh}[2]{h_{#1,\, #2}}
\newcommand{\nhin}[3]{h^{-1}_{#1,\, #2}(#3)}
\newcommand{\alnhin}[2]{h^{-1}_{#1,\, #2}}

\newcommand{\absolute}[1]{\mathrm{abs}\inparen{#1}}

\newcommand{\kfc}[1]{\widetilde{f}_{#1}}
\newcommand{\knfc}[1]{f_{#1}}
\newcommand{\kfinc}[1]{\widetilde{f}^{-1}_{#1}}
\newcommand{\ngc}{f^{\,-1}_{k}}

\newcommand{\fa}[1]{\widetilde{f}^{\,(a)}(#1)}
\newcommand{\fb}[1]{\widetilde{f}^{\,(b)}(#1)}
\newcommand{\fwild}[2]{\widetilde{f}^{\,(#1)}(#2)}
\newcommand{\alfa}{\widetilde{f}^{\,(a)}}
\newcommand{\alfb}{\widetilde{f}^{\,(b)}}
\newcommand{\alfwild}[1]{\widetilde{f}^{\,(#1)}}
\newcommand{\hfac}[1]{\widehat{f}^{\,\,(a)}_{#1}}
\newcommand{\hfbc}[1]{\widehat{f}^{\,\,(b)}_{#1}}
\newcommand{\hfwildc}[2]{\widehat{f}^{\,\,(#1)}_{#2}}

\newcommand{\intfull}[1]{\mathrm{I}(#1)}
\newcommand{\intone}[1]{\mathrm{I}_{1}(#1)}
\newcommand{\inttwo}[1]{\mathrm{I}_{2}(#1)}

\newcommand{\betaterm}{\mathrm{B}\inparen{\frac{1-b}{2},1+\frac{b}{2}}}

\newcommand{\confHG}{{}_{1}F_{1}}
\newcommand{\hypergeometric}{{}_{2}F_{1}}
\newcommand{\hgp}[1]{#1\cdot \hypergeometric\!\inparen{\frac{1-a}{2},\frac{1-b}{2}\,;\,\frac{3}{2}\,;\,{#1}^2}}
\newcommand{\RisingFactorial}[1]{(#1)_{k}}

\newcommand{\baru}{\overline{u}}
\newcommand{\barv}{\overline{v}}
\newcommand{\tildeu}{\widetilde{u}}
\newcommand{\tildev}{\widetilde{v}}
\newcommand{\tildeU}{\widetilde{U}}
\newcommand{\tildeV}{\widetilde{V}}

\newcommand{\barx}{\overline{x}}
\newcommand{\bary}{\overline{y}}

\newcommand{\tildex}{\widetilde{x}}
\newcommand{\tildey}{\widetilde{y}}

\newcommand{\hatU}{\widehat{U}}
\newcommand{\hatV}{\widehat{V}}
\newcommand{\hatu}{\widehat{u}}
\newcommand{\hatv}{\widehat{v}}

\newcommand{\seriesLeft}[2]{S_L(#1,#2)}
\newcommand{\seriesRight}[2]{S_R(#1,#2)}

\newcommand{\gvec}{\bfg}

\newcommand{\holderdual}[2]{\Psi_{\!#1}(#2)}

\newcommand{\id}{\mathrm{I}}

\newcommand{\CP}[1]{\mathsf{CP}(#1)}
\newcommand{\CPD}[1]{\mathrm{DP}(#1)}

\newcommand{\HD}{\emph{\Holder Dual Rounding}~}

\newcommand{\noise}[1]{T_{\rho} \,{#1}}

\newcommand{\factorConst}[1]{\Phi(#1)}
\newcommand{\threeFactorConst}[1]{\Phi_{3}(#1)}

\newcommand{\apxConst}[2]{\alpha(#1,#2)}
\newcommand{\factorSpConst}[2]{\Phi(#1,#2)}
\newcommand{\threeFactorSpConst}[2]{\Phi_{3}(#1,#2)}

\newcommand{\sqnormX}{\mathcal{F}_{X}}
\newcommand{\sqnormY}{\mathcal{F}_{Y}}
\newcommand{\sqrtnormX}{\sqrt{\mathcal{F}}_{X}}
\newcommand{\sqrtnormY}{\sqrt{\mathcal{F}}_{Y}}

\newcommand{\bbX}{\mathbb{X}}
\newcommand{\bbY}{\mathbb{Y}}
\newcommand{\bbW}{\mathbb{W}}

\newcommand{\support}[2]{\xi_{#1}(#2)}

\newcommand{\bars}{\overline{s}}
\newcommand{\bart}{\overline{t}}

\newcommand{\CVGP}{\emph{CVGP}~}

\newcommand{\Diag}[1]{D_{#1}}

\newcommand{\Norm}[1]{\left\lVert#1\right\rVert}

\newcommand{\mydot}[2]{\ensuremath{\left\langle #1\,, #2 \right\rangle}}
\newcommand{\ip}[1]{\left\langle #1 \right\rangle}

\newcommand{\emysmalldot}[2]{\ensuremath{\langle #1, #2 \rangle}_E} %% Euiwoong 20171030
\newcommand{\cmysmalldot}[2]{\ensuremath{\langle #1, #2 \rangle}_C} %% Euiwoong 20171030

\newcommand{\eip}[1]{\left\langle #1 \right\rangle} %% Euiwoong 20171030

%%%%%%%%%%% Vectors
\def\bfx {{\bf x}}
\def\bfy {{\bf y}}
\def\bfz {{\bf z}}
\def\bfu{{\bf u}}
\def\bfv{{\bf v}}
\def\bfw{{\bf w}}
\def\bfa{{\bf a}}
\def\bfb{{\bf b}}
\def\bfc{{\bf c}}
\def\bfg{{\bf g}}

\newcommand{\zero}{\mathbf 0}
\newcommand{\zeroone}{{0/1}\xspace}
\newcommand{\minusoneone}{{-1/1}\xspace}
\newcommand{\pmone}{\ensuremath{\inbraces{-1,1}}\xspace}
\newcommand{\yes}{{\sf Yes}\xspace}
\newcommand{\no}{{\sf No}\xspace}

\renewcommand{\gauss}[2]{{\mathcal{N}\inparen{#1, #2}}}
\renewcommand{\gaussian}[2]{{\mathcal{N}\inparen{#1, #2}}}

\newcommand{\Kwapien}{Kwapie\'n\xspace}
\newcommand{\Maurey}{Maurey\xspace}
\newcommand{\Pisier}{Pisier\xspace}

\input{PtoQ/1.intro}

\input{PtoQ/2.prelims}

\input{PtoQ/3.hardness}

\input{PtoQ/4.algorithm}

\input{PtoQ/5.hermite_analysis}

\input{PtoQ/6.wrap_up}

\input{PtoQ/7.coefficient_bound}

\input{PtoQ/8.factorization}

%% file: PtoQ/1.intro.tex
\section{Introduction}\label{sec:intro}
We consider the problem of finding the \onorm{p}{q} of a given matrix $A \in \R^{m \times n}$,
which is defined as 
\[
\norm{p}{q}{A}~\defeq~\max_{x \in \R^n \setminus \{0\}} \frac{\norm{q}{Ax}}{\norm{p}{x}} \mper 
\]
The quantity $\norm{p}{q}{A}$ is a natural generalization of the well-studied spectral norm, which
corresponds to the case $p=q=2$. For general $p$ and $q$, this quantity computes the maximum
distortion (stretch) of the operator $A$ from the normed space $\ell_p^n$ to $\ell_q^m$.

The case when $p = \infty$ and $q = 1$ is the well known Grothendieck problem~\cite{KN12, Pisier12},
where the goal is to maximize $\ip{y, Ax}$ subject to $\norm{\infty}{x}, \norm{\infty}{y} \leq 1$. In
fact, via simple duality arguments (see \cref{sec:p_to_q_prelims}), the general problem computing
$\norm{p}{q}{A}$ can be seen to be equivalent to the following variant of the Grothendieck problem
(and to $ \norm{q^*}{p^*}{A^T} $)
\[
\norm{p}{q}{A} ~=~ 
\max_{ \substack{\norm{p}{x} \leq 1 \\ \norm{q^*}{y} \leq 1}} \ip{y, Ax} 
~=~\norm{q^*}{p^*}{A^T} 
\mcom
\]
where $p^*, q^*$ denote the dual norms of $p$ and $q$, satisfying $1/p + 1/p^* = 
1/q + 1/q^* = 1$.

\paragraph{Hypercontractive norms.}
The case when $p < q$, known as the case of \emph{hypercontractive} norms, also has a special
significance to the analysis of random walks, expansion and related problems in hardness of
approximation~\cite{Biswal11, BBHKSZ12}. 
The problem of computing $\norm{2}{4}{A}$ is also known to be equivalent to
determining the maximum acceptance probability of a quantum protocol with multiple unentangled
provers, and is related to several problems in quantum  information theory~\cite{HM13, BH15}.

Bounds on hypercontractive norms of operators are also used to prove expansion of small sets in
graphs.
Indeed, if $f$ is the indicator function of set $S$ of measure $\delta$ in a graph with adjacency
matrix $A$, then we have that for any $p \leq q$,
\[
\Phi(S) 
~=~ 
1 - \frac{\ip{f, A f}}{\norm{2}{f}^2} 
~\geq~
1 - \frac{\norm{q^*}{f} \cdot \norm{q}{Af}}{\delta}
~\geq~
1 - \norm{p}{q}{A} \cdot \delta^{1/p - 1/q} \mper
\]
It was proved by Barak \etal~\cite{BBHKSZ12} that the above connection to small-set expansion can in
fact be made two-sided for a special case of the \onorm{2}{q}.
They proved by that  to resolve the promise version of the small-set expansion (SSE) problem, it
suffices to distinguish the cases $\norm{2}{q}{A} \leq c \cdot \sigma_{\min}$ and $\norm{2}{q}{A}
\geq C\cdot \sigma_{\min}$, where $\sigma_{\min}$ is the least non-zero singular value of $A$ and
$C > c > 1$ are appropriately chosen constants based on the parameters of the SSE problem. 
Thus, the approximability of \onorm{2}{q} is closely related to the small-set expansion problem. In
particular,  proving the NP-hardness of approximating the \onorm{2}{q} is (necessarily) an
intermediate goal towards proving the Small-Set Expansion Hypothesis of Raghavendra and Steurer~\cite{RS10}.
However,  relatively few results algorithmic and hardness results are known for approximating
hypercontractive norms.  
Interpolation algorithms developed in~\cite{Steinberg05,GMU23} give an upper bound of
$O(\max\inbraces{m,n}^{3-2\sqrt{2}})$ on the best possible approximation factor.
For the case of \onorm{2}{q} (for any $q > 2$), Barak \etal~\cite{BBHKSZ12} give an approximation
algorithm for the promise version of the problem described above, running  
in time $\exp\inparen{\tilde{O}(n^{2/q})}$. 
They also provide an additive approximation algorithm
for the \onorm{2}{4} (where the error depends on \onorm{2}{2} and \onorm{2}{\infty} of
$A$), which was extended to the \onorm{2}{q} by Harrow and Montanaro~\cite{HM13}. Some average case
algorithms based on \sos~hierarchy was recently developed in~\cite{DHPT25}. The result concerns
\sos~certifiability of sub-gaussian matrices, which can be viewed as finding \sos~upper-bound for $2\to
q$ norms, for even $q$s, of i.i.d sub-gaussian matrices.
Barak \etal~also prove NP-hardness of approximating
$\norm{2}{4}{A}$ within a factor of $1 + \tilde{O}(1/ n^{o(1)})$,  and hardness of approximating
better than $\exp{O((\log n)^{1/2-\eps})}$ in polynomial time, assuming the Exponential Time
Hypothesis (ETH). 
This reduction was also used by Harrow, Natarajan and Wu~\cite{HNW16} to prove that $\tilde{O}(\log n)$
levels of the Sum-of-Squares SDP hierarchy cannot approximate $\norm{2}{4}{A}$ within any constant
factor.

It is natural to ask if the bottleneck in proving (constant factor) hardness of approximation for
\onorm{2}{q} arises from the fact from the nature of the domain (the $\ell_2$ ball) or from
hypercontractive nature of the objective.  As discussed in~\cref{sec:p_to_q:overview:hardness}, \emph{all}
hypercontractive norms present a barrier for gadget reductions, since if a ``true'' solution $x$ is
meant to encode the assignment to a (say) label cover problem with consistency checked via local
gadgets, then (for $q > p$), a ``cheating solution'' may make the value of $\norm{q}{Ax}$ very
large by using a sparse $x$ which does not carry any meaningful information about the underlying
label cover problem.

We show that (somewhat surprisingly, at least for the authors)  it is indeed possible to
overcome the barrier for gadget reductions for hypercontractive norms, for any $2 < p < q$ (and by
duality, for any $p < q< 2$). 
This gives the first NP-hardness result for hypercontractive norms (under randomized reductions). 
Assuming ETH, this also rules out a constant factor approximation algorithm that runs in $2^{n^{\delta}}$
for some $\delta := \delta(p, q)$. 
%
% For $2 \to 4$ norms, the reduction of~\cite{BBHKSZ12} rules out a constant factor approximation
% algorithm that runs in time $2^{\log^2 n}$ time. 
%
%
\begin{theorem}
For any $p, q$ such that $1 < p \leq q < 2$ or $2 < p \leq q < \infty$ and a constant $c > 1$,  it
is NP-hard to approximate \onorm{p}{q} within a factor of $c$. The reduction runs in time $n^{B_p
\cdot q}$ for $2 < p < q$, where $B_p = \operatorname{poly}(1 / (1 - \gamma_{p^*}))$. 
\label{thm:main_hype_vanilla}
\end{theorem}
We show that the above hardness can be strengthened to any constant factor via a simple tensoring
argument. 
In fact, this also shows that it is hard to approximate $\norm{p}{q}{A}$ within almost polynomial factors
unless NP is in randomized quasi-polynomial time. This is the content of the following theorem.
\begin{theorem}\label{thm:p_to_q_main_hype}
For any $p, q$ such that $1 < p \leq q < 2$ or $2 < p \leq q < \infty$ and $\epsilon > 0$,  there
is no polynomial time algorithm that approximates the \onorm{p}{q} of an $n \times n$ matrix within
a factor $2^{\log^{1 - \epsilon} n}$ unless $\NP \subseteq \BPTIME{2^{(\log n)^{O(1)}}}$. When $q$
is an even integer, the same inapproximability result holds unless $\NP \subseteq \DTIME{2^{(\log
    n)^{O(1)}}}$
\end{theorem}

We also note that the operator $A$ arising in our reduction in  \cref{thm:main_hype_vanilla}
 satisfies $\sigma_{\min}(A) \approx 1$ (and is in fact a product of a carefully chosen 
projection and a scaled random Gaussian matrix).
For such an $A$, we prove the hardness of distinguishing $\norm{p}{q}{A} \leq c$ and
$\norm{p}{q}{A} \geq C$, for constants $C > c > 1$. 
For the corresponding problem in
the case of \onorm{2}{q}, Barak \etal \cite{BBHKSZ12} gave a subexponential algorithm running in
time $\exp\inparen{O(n^{2/q})}$ (which works for every $C > c > 1$). 
On the other hand, since the running time of our reduction is $n^{O(q)}$, we get that 
assuming ETH, we show that no algorithm can distinguish the above cases for \onorm{p}{q} in
time $\exp\inparen{n^{o(1/q)}}$, for any $p \leq q$ when $2 \notin [p,q]$.

While the above results give some possible reductions for working with hypercontractive norms, it
remains an interesting problem to understand the role of the domain as a barrier to proving hardness
results for the \onorm{2}{q} problems. In fact, no hardness results are available even for the more
general problem of polynomial optimization over the $\ell_2$ ball.
We view the above theorem as providing some evidence that while hypercontractive norms have been
studied as a single class so far, the case when $2 \in [p,q]$ may be qualitatively different (with respect to
techniques) from the case when $2 \notin [p,q]$. 
This is indeed known to be true in the \emph{non-hypercontractive case} with $p \geq q$. In fact,
our results are obtained via new hardness results for the case $p \geq q$, as described below.

\paragraph{The non-hypercontractive case.}
In this case, $q\ge p \ge 1$.
For non-negative matrix, this can be efficiently approximated. The work~\cite{BV11} showed a polynomial time
algorithm and recently~\cite{OV25} showed a near linear time algorithm. So we are
mainly interested in the general matrix case.
Several results are known in the case when $p \geq q$, and we summarize known results for matrix
norms  in \cref{fig:normbounds}, for the both the hypercontractive and non-hypercontractive cases. 

While the case of $p=q=2$ corresponds to the spectral norm,
the problem is also easy when $q = \infty$ (or equivalently $p = 1$) since this corresponds to
selecting the row of $A$ with the maximum $\ell_{p^*}$ norm. Note that in general,
\cref{fig:normbounds} is symmetric about the principal diagonal. Also note that if $\norm{p}{q}{A}$
is a hypercontractive norm ($p < q$) then so is the equivalent $\norm{q^*}{p^*}{A^T}$ (the
hypercontractive and non-hypercontractive case are separated by the non-principal diagonal).
%
% \medskip\noindent\textbf{The non-hypercontractive case $(p \geq  q)$.} 

%
\begin{figure}[H]
\begin{center}
\includegraphics[width=3.6in]{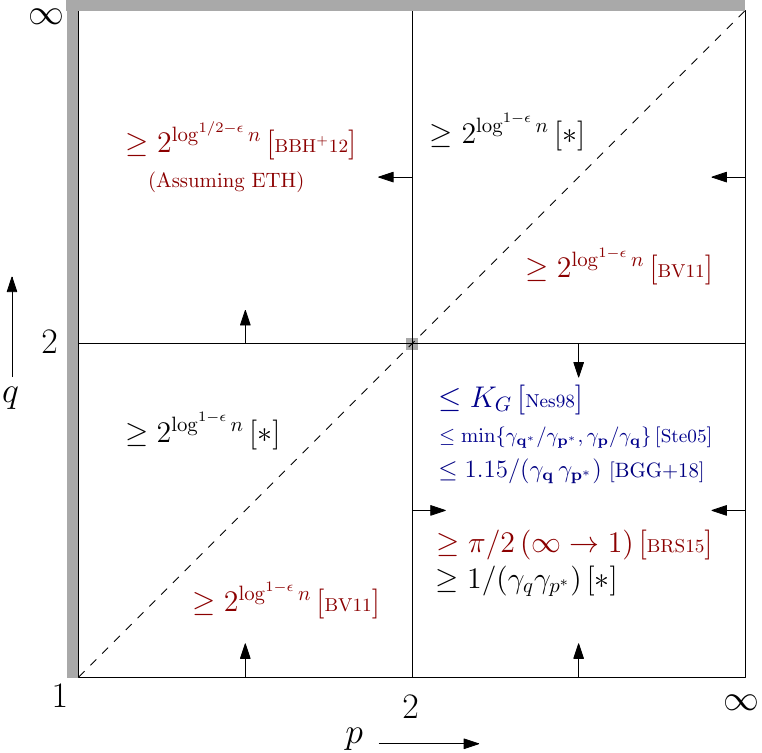}
\end{center}
\vspace{-15 pt}
\caption{Upper and lower bounds for approximating $\|A\|_{p \rightarrow q}$. Arrows
  indicate the region to which a boundary belongs and thicker shaded regions represent exact
  algorithms. Our results are indicated by [$*$]. We omit UGC-based hardness results in the figure.}
\label{fig:normbounds}
\end{figure}

As is apparent from the figure, the problem of approximating $\norm{p}{q}{A}$ for $p \geq q$ 
admits good  approximations when $2 \in [q,p]$, and is hard otherwise. 
For the case when $2 \notin [q,p]$,
an upper bound of $O(\max\{m,n\}^{3-2\sqrt{2}})$ on the approximation
ratio was proved in~\cite{Steinberg05,GMU23}.  
Bhaskara and Vijayaraghavan~\cite{BV11} showed NP-hardness of approximation within any
constant factor, and hardness of approximation within an $O\inparen{2^{(\log n)^{1-\eps}}}$ factor
for arbitrary $\eps > 0$ assuming $\NP \not\subseteq \DTIME{2^{(\log n)^{O(1)}}}$.

Determining the right
constants in these approximations when $2 \in [q,p]$ has been of considerable interest in the
analysis and optimization community.
For the case of \onorm{\infty}{1}, Grothendieck's theorem~\cite{Grothendieck56} shows that the
integrality gap of a semidefinite programming (SDP) relaxation is bounded by a constant, and the
(unknown) optimal value is now called the Grothendieck constant $K_G$. Krivine \cite{Krivine77} proved
an upper bound of $\pi/(2\ln(1+\sqrt{2})) = 1.782 \ldots$ on $K_G$, and it was later shown by
Braverman \etal that $K_G$ is strictly smaller than this bound. The best known lower bound on $K_G$
is about $1.676$, due to (an unpublished manuscript of) Reeds \cite{Reeds91} (see also \cite{KO09}
for a proof).

An upper bound of $K_G$ on the approximation factor also
follows from the work of Nesterov~\cite{Nesterov98} for any $p \geq 2 \geq q$. A later work of
Steinberg~\cite{Steinberg05} also gave an upper bound of 
$\min\inbraces{\gamma_p/\gamma_q, \gamma_{q^*}/\gamma_{p^*}}$, where $\gamma_p$ denotes
$p^{th}$  norm of a standard normal random variable (\ie~the $p$-th root of the $p$-th Gaussian moment). 
Note that Steinberg's bound is less than 
$K_G$ for some values of $(p,q)$, in particular for all values of the form $(2,q)$ with $q \leq 2$
(and equivalently $(p,2)$ for $p \geq 2$), where it equals $1/\gamma_q$ (and $1/\gamma_{p^*}$ for
$(p,2)$).

On the hardness side, \Briet, Regev and Saket~\cite{BRS15} showed NP-hardness of $\pi/2$ for the
\onorm{\infty}{1},  strengthening a hardness result of Khot and Naor based on the Unique Games
Conjecture (UGC)~\cite{KN08} (for a special case of the Grothendieck problem when the matrix $A$ is
positive semidefinite). Assuming UGC, a hardness result matching Reeds' lower bound was proved by Khot and
\Odonnell~\cite{KO09}, and hardness of approximating within $K_G$ was proved by Raghavendra and
Steurer~\cite{RS09}. 

For a related problem known as the $L_p$-Grothendieck problem, where the goal is to maximize
$\ip{x,Ax}$ for $\norm{p}{x} \leq 1$, results by  Steinberg~\cite{Steinberg05} and 
Kindler, Schechtman and Naor~\cite{KNS10} give an upper bound of $\gamma_p^2$, and a matching lower
bound was proved assuming UGC by~\cite{KNS10}, which was strengthened to NP-hardness by Guruswami
\etal~\cite{GRSW16}. 
However, note that this problem is quadratic and not necessarily bilinear, and is in
general much harder than the Grothendieck problems considered here. In particular, the case of 
$p = \infty$ only admits an $\Theta(\log n)$ approximation instead of $K_G$ for the bilinear
version~\cite{AMMN06, ABHKS05}.

We extend the hardness results of~\cite{BRS15} for the $\infty \to 1$
and $2 \to 1$ norms of a matrix to any $p \geq 2 \geq q$. 
The hardness factors obtained match the performance of known algorithms (due to Steinberg~\cite{Steinberg05})
for the cases of $2 \to q$ and $p \to 2$.
\begin{theorem}\label{thm:main_nhc_vanilla}
  For any $p, q$ such that $\infty \geq p \geq 2 \geq q \geq 1$ and $\epsilon > 0$, 
  it is NP-hard to approximate the \onorm{p}{q} within a factor $1 / (\gamma_{p^*} \gamma_q) - \epsilon$. 
\end{theorem}
Both~\cref{thm:main_hype_vanilla} and~\cref{thm:main_nhc_vanilla} are consequences of a more technical theorem,
which proves hardness of approximating $\norm{2}{r}{A}$ for $r < 2$ (and hence $\norm{r^*}{2}{A}$
for $r^* > 2$) while providing additional structure in the matrix $A$ produced by the
reduction. This is proved in~\cref{sec:nhc}.
We also show our methods can be used to provide a simple proof (albeit via randomized reductions)
of the $2^{\Omega((\log n)^{1-\eps})}$ hardness 
for the non-hypercontractive case when $2 \notin [q,p]$, which was
proved by~\cite{BV11}. This is presented in~\cref{sec:reverse}.

Some of these hardness results imply approximation lower bounds on some natural versions of sparse  Principal
Component Analysis (PCA) as discussed  in~\cite{Awasthi21a}. Technique similar to the one used in this
hardness result was also issued to prove a hardness result for a \emph{minimization} version of $p\to q$
norm problem in~\cite{BBBKLW19}.

The candidate hard instance arising from our approach (as well as~\cite{BRS15}) leads to a PSD matrix.
For the norm $\infty\to 1$\ie~the little Grothendieck problem, our hardness factor is $\nfrac{\pi}{2}$.
It was shown in a subsequent work~\cite{BLT22} that the hardness factor for the \emph{general} matrix is
at least $\nfrac{\pi}{2}+\varepsilon_0$ for some $\varepsilon_0$. This required them to go beyond our
technique here.

We also give an improved approximation for \onorm{p}{q} when $2 \in [q,p]$ (inspired by the above hardness result) 
which achieves an approximation factor of $C_0\cdot (1 / (\gamma_{p^*} \gamma_q))$, where 
$C_0 \approx 1/(\ln(1+\sqrt{2}))$ is a constant comparable to that arising in Krivine's upper bound on 
the Grothendieck constant~\cite{Krivine77}. 

The work~\cite{BRW21}, demonstrated a use of this algorithm for approximately solving a variant of low rank
approximation problem. In their iterative algorithm, each rank one update turns out to be the outer
product of the vector achieving  $p\to 2$-norm of a matrix ($p\ge 2$).

\subsection{Hardness Results: Collected}
For the sake accessibility, we collect the hardness results proved in this chapter:
\begin{restatable}{reptheorem}{maintheoremptoq}\label{thm:mainthmptoq}
  We group the results using various range of $p$ and $q$ and hypercontractivity:
  \begin{enumerate}[label=(\arabic*)]
    \item\label{thm:main_hype}
      For any $p, q$ such that $1 < p \leq q < 2$ or $2 < p \leq q < \infty$ and $\epsilon > 0$:
      there is no polynomial time algorithm that approximates the \onorm{p}{q} of an $n \times n$ matrix
      within a factor $2^{\log^{1 - \epsilon} n}$ unless $\NP \subseteq \BPTIME{2^{(\log n)^{O(1)}}}$.
      When $q$ is an even integer, the same inapproximability result holds unless $\NP \subseteq
      \DTIME{2^{(\log n)^{O(1)}}}$
    \item\label{thm:main_reverse}
      For any $p, q$ such that $\infty \geq p \geq q > 2$ or $2 > p \geq q \geq 1$ and $\epsilon >
      0$: there is no polynomial time algorithm that approximates the \onorm{p}{q} of an $n \times n$ matrix
      within a factor $2^{\log^{1 - \epsilon} n}$ unless $\NP \subseteq \BPTIME{2^{(\log n)^{O(1)}}}$.

      We note here that the hardness result proved in~\cite{BV11} is slightly stronger due to the weaker
      assumption of $\NP\not\subseteq\DTIME{2^{(\log n)^{O(1)}}}$ (approximation factor remains the
      same).

    \item\label{thm:main_nhc}
      For any $p, q$ such that $\infty \geq p \geq 2 \geq q \geq 1$ and $\epsilon > 0$:
      it is NP-hard to approximate the \onorm{p}{q} within a factor $1 / (\gamma_{p^*} \gamma_q) - \epsilon$.
  \end{enumerate}
\end{restatable}
\subsection{The Search For Optimal Constants and Optimal Algorithms}
The goal of determining the right approximation ratio for these problems is closely related to the
question of finding the optimal (rounding) algorithms. For the Grothendieck problem, the goal is to
find $y \in \R^m$ and $x \in \R^n$ with $\norm{\infty}{y}, \norm{\infty}{x} \leq 1$, and one considers
the following semidefinite relaxation:
\begin{align*}
    \mbox{maximize} \quad&~~\sum_{i,j} A_{i,j}\cdot \mysmalldot{u^i}{v^j} \quad  \text{s.t.} \\
\mbox{subject to}    \quad&~~\norm{2}{u^i} \leq 1, \norm{2}{v^j} \leq 1 & \forall i\in [m], j\in [n]\\
    &~~ u^i, v^j\in \R^{m+n} & \forall i\in [m], j\in [n]
\end{align*}
By the bilinear nature of the problem above, it is clear that the optimal $x, y$ can be taken to
have entries in $\pmone$. A bound on the approximation ratio\footnote{Since we will be dealing with problems where the optimal solution may not be integral, we
  will use the term ``approximation ratio'' instead of ``integrality gap''.}
 of the above program is then obtained by
designing a good ``rounding'' algorithm which maps the vectors $u^i, v^j$ to values in
$\pmone$. Krivine's analysis \cite{Krivine77} corresponds to a rounding algorithm which considers a
random vector $\gvec \sim \gaussian{0}{I_{m+n}}$ and rounds to $x, y$ defined as
\[
y_i ~:=~ \sgn\inparen{\ip{\phi(u^i), \gvec}} \qquad \text{and} \qquad 
x_j 
~:=~
\sgn\inparen{\ip{\psi(v^j), \gvec}} \mcom
\]
for some appropriately chosen transformations $\phi$ and $\psi$. This gives the following 
upper bound on the approximation ratio of the above relaxation, and hence on the value of the
Grothendieck constant $K_G$:
\[
K_G 
~\leq~ 
\frac{1}{\sinh^{-1}(1)} \cdot \frac{\pi}{2} 
~=~
\frac{1}{\ln(1+\sqrt{2})} \cdot \frac{\pi}{2} \mper
\]
Braverman \etal \cite{BMMN13} show that the above bound can be strictly improved (by a very small
amount) using a two-dimensional analogue of the above algorithm, where the value $y_i$ is taken to be
a function of the two-dimensional projection $(\langle \phi(u^i), \gvec_1 \rangle, \langle \phi(u^i), \gvec_2
\rangle)$ for independent Gaussian vectors $\gvec_1, \gvec_2 \in \R^{m+n}$ (and similarly for
$x$). Naor and Regev \cite{NR14} show that such schemes are optimal in the sense that it is possible
to achieve an approximation ratio arbitrarily close to the true (but unknown) value of $K_G$ by
using $k$-dimensional projections for a large (constant) $k$. A similar existential result was also
proved by Raghavendra and Steurer \cite{RS09} who proved that the there exists a (slightly
different) rounding algorithm which can achieve the (unknown) approximation ratio $K_G$.

For the case of arbitrary $p \geq 2 \geq q$, Nesterov~\cite{Nesterov98} considered the convex program
in~\cref{fig:convex}, denoted as $\CP{A}$, generalizing the one above. We reproduce it
in~\cref{sec:p_to_q:prelims:algo}.  Although the stated bound in~\cite{Nesterov98} is slightly weaker (as
it is proved for a larger class of problems), the approximation ratio of the above relaxation can be
shown to be bounded by $K_G$. By using the Krivine rounding scheme of considering the sign of a random
Gaussian projection (aka random hyperplane rounding) one can show that Krivine's upper bound on $K_G$
still applies to the above problem. 

Motivated by applications to robust optimization, Steinberg~\cite{Steinberg05} considered the dual
of (a variant of) the above relaxation, and obtained an upper bound of 
$\min\inbraces{\gamma_p/\gamma_q, \gamma_{q^*}/\gamma_{p^*}}$
on the approximation factor.
Note that while Steinberg's bound is better (approaches 1) 
as $p$ and $q$ approach 2, it is unbounded when $p, q^* \to \infty$ (as in the Grothendieck
problem).

Based on the inapproximability result mentioned above, it is 
natural to ask if this is the ``right form'' of the approximation ratio. Indeed, this ratio is $\pi/2$
when $p^* = q = 1$, which is the ratio obtained by Krivine's rounding scheme, up to a factor of
$\ln(1+\sqrt{2})$. We extend Krivine's result to all $p \geq 2 \geq q$ as below for the relaxation
$\CP{A}$ (see~\cref{fig:convex}):
\begin{theorem}
There exists a fixed constant $\eps_0 \leq 0.00863$ such that for all $p \geq 2 \geq q$, the
approximation ratio of the convex relaxation $\CP{A}$ is upper bounded by
\[
    \frac{1+\eps_0}{\sinh^{-1}(1)} \cdot \frac{1}{\gamma_{p^*} \cdot \gamma_q} 
~=~
    \frac{1+\eps_0}{\ln(1+\sqrt{2})} \cdot \frac{1}{\gamma_{p^*} \cdot \gamma_q}
\mper
\]
\end{theorem}
\begin{figure}\label{fig:bounds-comparison}
\begin{center}
\includegraphics[width=3.8in]{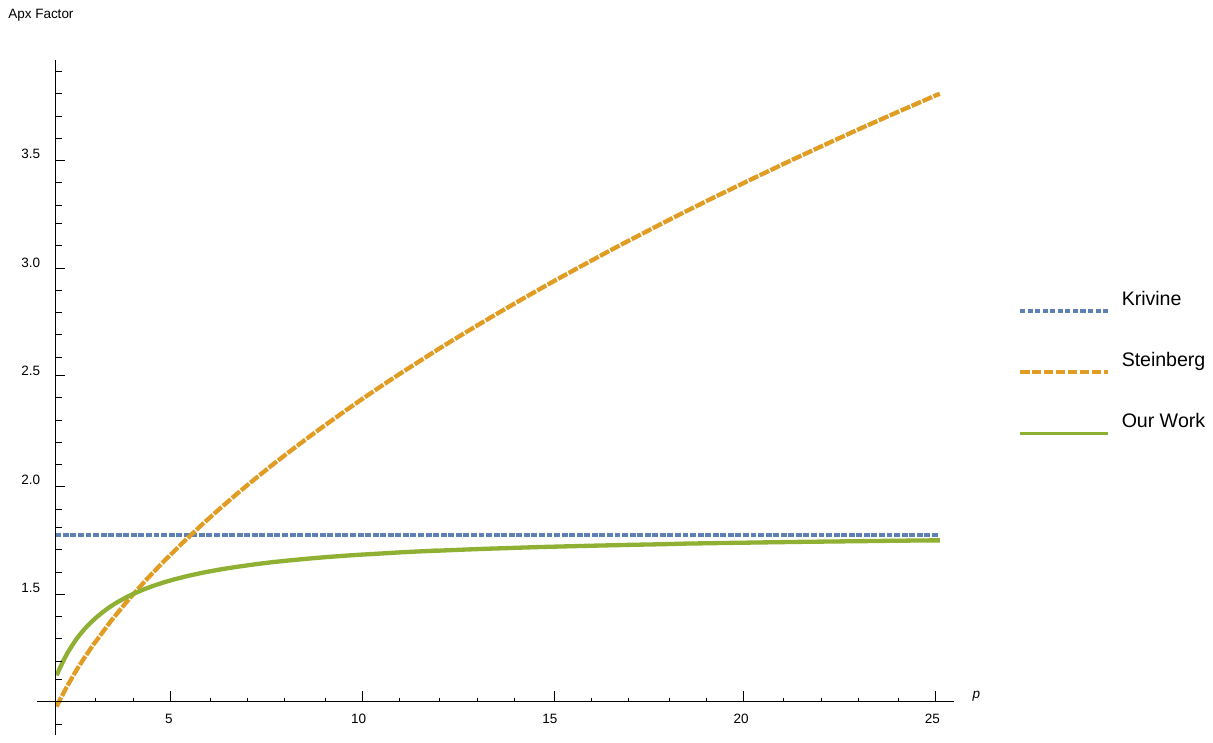}
\end{center}
\vspace{-15 pt}
\caption{{\footnotesize A comparison of the bounds for approximating \onorm{p}{p^*} obtained from Krivine's rounding
for $K_G$, Steinberg's analysis, and our bound. While our analysis yields an improved bound for $4\leq p\leq 66$, we believe that the rounding algorithm achieves an improved bound for all $p$.}}
\end{figure}
Perhaps more interestingly, the above theorem is proved via a generalization of hyperplane rounding,
which we believe may be of independent interest. Indeed, for a given collection of vectors $w^{1},
\ldots, w^{m}$ considered as rows of a matrix $W$,  Gaussian hyperplane rounding corresponds to
taking the ``rounded'' solution $y$ to be the 
\[
    y 
    ~:=~ 
    \argmax_{\norm{\infty}{y'} \leq 1} \ip{y', W\gvec} 
    ~=~ 
    \inparen{\sgn\inparen{\ip{w^i,\gvec}}}_{i \in [m]}\mper
\]
We consider the natural generalization to (say) $\ell_r$ norms, given by
\[
    y 
    ~:=~ \argmax_{\norm{r}{y'} \leq 1} \ip{y', W\gvec} 
    ~=~ \inparen{\frac{\sgn\inparen{\ip{w^i, \gvec}} \cdot 
    \abs{\langle w^i, \gvec \rangle}^{r^*-1}}{\norm{r^*}{W\gvec}^{r^*-1}}}_{i \in [m]}\mper
\]
We refer to $y$ as the ``\Holder~dual'' of $W\gvec$, since the above rounding can be obtained by
viewing $W\gvec$ as lying in the dual ($\ell_{r^*}$) ball, and finding the $y$ for which \Holder's
inequality is tight. Indeed, in the above language, Nesterov's rounding corresponds to considering
the $\ell_{\infty}$ ball (hyperplane rounding). While Steinberg used a somewhat different
relaxation, the rounding there can be obtained by viewing $W\gvec$ as lying in the primal $(\ell_r)$
ball instead of the dual one.
In case of hyperplane rounding, the analysis is motivated by the identity that for two unit vectors
$u$ and $v$, we have
\[
\Ex{\gvec}{\sgn(\ip{\gvec,u}) \cdot \sgn(\ip{\gvec,v})}~=~\frac{2}{\pi} \cdot \sin^{-1}(\ip{u,v}) \mper
\]
We prove the appropriate extension of this identity to $\ell_r$ balls (and analyze the functions
arising there) which may also be of interest for other optimization problems over $\ell_r$ balls.

\subsection{Relation to Factorization Theory}
Let $X, Y$ be  Banach spaces, and let $A: X \to Y$ be a continuous linear operator. As before,
the norm $\norm{X}{Y}{A}$ is defined as
\[
\norm{X}{Y}{A}~:=~\sup_{x \in X \setminus \{0\}} \frac{\norm{Y}{Ax}}{\norm{X}{x}} \mper
\]
The operator $A$ is said to be factorized through Hilbert space if the {\deffont factorization constant}
of $A$ defined as
\[
    \factorConst{A}~:=~\inf_{H} \inf_{BC=A} \frac{\norm{X}{H}{C}\cdot \norm{H}{Y}{B}}{\norm{X}{Y}{A}}
\]
is bounded, where the infimum is taken over all Hilbert spaces $H$ and all operators $B: H \to Y$
and $C: X \to H$. The factorization gap for spaces $X$ and $Y$ is then defined as
$\factorSpConst{X}{Y} := \sup_A \factorConst{A}$ where the supremum runs over all continuous operators 
$A:X\to Y$. 

The theory of factorization of linear operators is a cornerstone of modern functional analysis and has also 
found many applications outside the field (see \cite{Pisier86, AK06} for more information). 
An application to theoretical computer science was found by Tropp~\cite{Tropp09} who used 
the Grothendieck factorization~\cite{Grothendieck56} to give an algorithmic version of a celebrated 
column subset selection result of Bourgain and Tzafriri~\cite{BT87}. 

As an almost immediate consequence of convex programming duality, our new algorithmic results also 
imply some improved factorization results for $\ell_p^{n},\ell_{q}^{m}$ (a similar observation was already 
made by Tropp~\cite{Tropp09} in the special case of $\ell_{\infty}^{n},\ell_{1}^{m}$ and for a slightly different 
relaxation). We first state some classical factorization results, for which we will use $T_2(X)$ and $C_2(X)$ 
to respectively denote the Type-2 and Cotype-2 constants of $X$. We refer the interested reader 
to~\cref{factorization} for a more detailed description of factorization theory as well as the relevant 
functional analysis preliminaries. 

The \Kwapien-Maurey~\cite{Kwapien72a,Maurey74} theorem states that for any pair of Banach spaces 
$X$ and $Y$
\[
\factorSpConst{X}{Y} ~\leq~ T_2(X) \cdot C_2(Y) \mper
\]
However, Grothendieck's result~\cite{Grothendieck56} shows that a much better bound is possible 
in a case where $T_2(X)$ is unbounded. In particular,
\[
\factorSpConst{\ell_{\infty}^n}{\ell_{1}^{m}} ~\leq~ K_G \mcom
\]
for all $m, n \in \N$. Pisier~\cite{Pisier80} showed that if $X$ or $Y$ satisfies the approximation property 
(which is always satisfied by finite-dimensional spaces), then 
\[
\factorSpConst{X}{Y} ~\leq~  \inparen{2\cdot C_2(X^*) \cdot C_2(Y)}^{3/2} \mper
\]
We show that the approximation ratio of Nesterov's relaxation is in fact an upper bound on the
factorization gap for the spaces $\ell_p^n$ and $\ell_q^m$. Combined with our upper bound on the
integrality gap, we show an improved bound on the factorization constant, \ie
for any $p \geq 2 \geq q$ and $m,n \in \N$, we have that for $X = \ell_{p}^n$, $Y=\ell_{q}^m$ 
\[
\factorSpConst{X}{Y} ~\leq~ \frac{1+\eps_0}{\sinh^{-1}(1)} \cdot \inparen{C_2(X^*) \cdot C_2(Y)} \mcom
\]
where $\eps_0 \leq 0.00863$ as before. This improves on Pisier's bound for all $p\geq 2\geq q$, and 
for certain ranges of $(p,q)$ it also improves upon $K_G$ and the bound of \Kwapien-Maurey. 
% To our knowledge, our work is the best known upper bound on the factorization constant in these ranges. 

\subsection{Approximability and Factorizability}
Let $(X_n)$ and $(Y_m)$ be sequences of Banach spaces such that $X_n$ is over the vector space $\R^n$ 
and $Y_m$ is over the vector space $\R^m$. We shall say a pair of sequences $((X_n),(Y_m))$ factorize 
if $\factorSpConst{X_n}{Y_m}$ is bounded by a constant independent of $m$ and $n$. Similarly, 
we shall say a pair of families $((X_n),(Y_m))$ are computationally approximable if there exists a polynomial 
$R(m,n)$, such that for every $m,n \in \N$, there is an algorithm with runtime $R(m,n)$ approximating 
$\norm{X_n}{Y_m}{A}$ within a constant independent of $m$ and $n$ (given an oracle for computing the norms of vectors and a separation oracle for the unit balls of the norms). 
% in time polynomial in $m$ and $n$ . 
We consider the natural question of characterizing the families of norms that 
are approximable and their connection to factorizability and Cotype. 

The pairs $(p,q)$ for which $(\ell_p^{n},\ell_q^{m})$ is known (resp.\ not known) to factorize, are
precisely those pairs $(p,q)$ which are known to be computationally approximable (resp.\ inapproximable
assuming hardness conjectures like $\textrm{P}\neq\NP$ and \textrm{ETH}). Moreover, the Hilbertian case
which trivially satisfies factorizability, is also known to be computationally approximable (with
approximation factor 1). 

It is tempting to ask whether the set of computationally approximable pairs coincides with the set 
of factorizable pairs or the pairs for which $X^*_n,Y_m$ have bounded (independent of $m,n$) Cotype-2 
constant. 
Further yet, is there a connection between the approximation factor and the factorization constant, 
or approximation factor and Cotype-2 constants (of $X_n^*$ and $Y_m$)? 
Our work gives some modest additional evidence towards such conjectures. 
Such a result would give credibility to the appealing intuitive idea of the 
approximation factor being dependent on the ``distance'' to a Hilbert space. 

\subsection{Proof Overview: Hardness}\label{sec:p_to_q:overview:hardness}
\paragraph{The hardness of proving hardness for hypercontractive norms.}
Reductions for  various geometric problems use a ``smooth'' version of the Label Cover problem,
composed with long-code functions for the labels of the variables. In various reductions, including
the ones by Guruswami~\etal~\cite{GRSW16} and \Briet~\etal~\cite{BRS15} (which we closely follow) the
solution vector $x$ to the geometric problem consists of the Fourier coefficients of the various
long-code functions, with a ``block'' $x_v$ for each vertex of the label-cover instance. 
The relevant geometric operation (transformation by the matrix $A$ in our case)
consists of projecting to a space which enforces the consistency constraints derived from the
label-cover problem, on the Fourier coefficients of the encodings.

However, this strategy presents with two problems when designing reductions for hypercontractive
norms. Firstly, while projections maintain the $\ell_2$ norm of encodings corresponding to
consistent labelings and reduce that of inconsistent ones, their behaviour is harder to analyze for
$\ell_p$ norms for $p \neq 2$. Secondly, the \emph{global} objective of maximizing $\norm{q}{Ax}$ is
required to enforce different behavior within the blocks $x_v$, than in the full vector $x$. The
block vectors $x_v$ in the solution corresponding to a satisfying assignment of label cover are
intended to be highly sparse, since they correspond to ``dictator functions'' which have only one
non-zero Fourier coefficient. This can be enforced in a test using the fact that for a vector 
$x_v \in \R^t$, $\norm{q}{x_v}$ is a convex function of $\norm{p}{x_v}$ when $p \leq q$, and is
maximized for vectors with all the mass concentrated in a single coordinate.
However, a global objective function which tries to maximize $\sum_v \norm{x_v}_q^q$, 
also achieves a high value from global vectors $x$ which concentrate all the mass on coordinates
corresponding to  few vertices of the label cover instance, and do not carry any meaningful
information about assignments to the underlying label cover problem.
 
Since we can only check for a global objective which is the $\ell_q$ norm of some vector
involving coordinates from blocks across the entire instance, it is not clear how to enforce local
Fourier concentration (dictator functions for individual long codes) and global well-distribution
(meaningful information regarding assignments of most vertices) using the same objective function. 
While the projector $A$ also enforces a linear relation between the block vectors $x_u$ and $x_v$
for all edges $(u,v)$ in the label cover instance, using this to ensure well-distribution across blocks
 seems to require a very high density of constraints in the label cover instance, and no hardness
 results are available in this regime.
\paragraph{Our reduction.}
We show that when $2 \notin [p,q]$, it is possible to bypass the above issues using hardness of
$\norm{2}{r}{A}$ as an intermediate (for $r < 2$). Note that since $\norm{r}{z}$ is a \emph{concave}
function of $\norm{2}{z}$ in this case, the test favors vectors in which the mass is
well-distributed and thus solves the second issue. For this, we use local tests based on
the Berry-\Esseen
theorem (as in \cite{GRSW16} and \cite{BRS15}). Also, since the starting point now is the $\ell_2$
norm, the effect of projections is easier to analyze. This reduction is discussed in \cref{sec:nhc}.

By duality, we can interpret the above as a
hardness result for $\norm{p}{2}{A}$ when $p > 2$ (using $r = p^*$). We then convert this to a
hardness result for \onorm{p}{q} in the hypercontractive case by composing $A$ with an ``approximate
isometry'' $B$ from $\ell_2 \to \ell_q$ (\ie $\forall y~ \norm{q}{By} \approx \norm{2}{y}$) since we
can replace $\norm{2}{Ax}$ with $\norm{q}{BAx}$.
Milman's version of the Dvoretzky theorem \cite{Vershynin17} implies random operators to a
sufficiently high dimensional ($n^{O(q)}$) space satisfy this property, which then yields constant
factor hardness results for the \onorm{p}{q}. A similar application of Dvoretzky's theorem also
appears in an independent work of Krishnan \etal \cite{KMW18} on sketching matrix norms.

We also show that the hardness for hypercontractive norms can be
amplified via tensoring. This was known previously for the \onorm{2}{4} using an argument based
on parallel repetition for QMA \cite{HM13}, and for the case of $p = q$ \cite{BV11}. We give a
simple argument based on convexity, which proves this for all $p \leq q$, but appears to
have gone unnoticed previously. The amplification is then used to prove hardness of approximation
within almost polynomial factors.

\paragraph{Non-hypercontractive norms.}
We also use the hardness of $\norm{2}{r}{A}$ to obtain hardness for the non-hypercontractive
case of $\norm{p}{q}{A}$ with $q < 2 < p$, by using an operator that ``factorizes'' through 
$\ell_2$.
In particular, we obtain hardness results for $\norm{p}{2}{A}$ and $\norm{2}{q}{A}$ (of factors
$1/\gamma_{p^*}$and $1/\gamma_q$ respectively) using the reduction in \cref{sec:nhc}. We
then combine these hardness results using additional properties of the operator $A$ obtained in the
reduction, to obtain a hardness of factor $(1/\gamma_{p^*}) \cdot (1/\gamma_{q})$ for the
\onorm{p}{q} for $p > 2 > q$. The composition, as well as the hardness
results for hypercontractive norms, are presented in \cref{sec:hypercontractive}.

We also obtain a simple proof of the $2^{\Omega((\log n)^{1-\eps})}$ hardness for the
non-hypercontractive case when $2 \notin [q,p]$ (already proved by Bhaskara and Vijayaraghavan
\cite{BV11}) via an approximate isometry argument as used in the hypercontractive case. In the
hypercontractive case, we started from a constant factor hardness of the \onorm{p}{2} and the same
factor for \onorm{p}{q} using the fact that for a random Gaussian matrix $B$ of appropriate
dimensions, we have $\norm{q}{Bx} \approx \norm{2}{x}$ for all $x$. We then amplify the hardness via
tensoring. 
In the non-hypercontractive case, we start with a hardness for \onorm{p}{p} (obtained via the above
isometry), which we \emph{first} amplify via tensoring. We then apply another approximate isometry
result due to Schechtman \cite{Schechtman87}, which gives a samplable distribution $\calD$ over
random matrices $B$ such that with high probability over $B$, we have $\norm{q}{Bx} \approx
\norm{p}{x}$ for all $x$.

\bigskip

We thus view the above results as showing that combined with a basic hardness for \onorm{p}{2}, the
basic ideas of duality, tensoring, and embedding (which builds on powerful results from functional
analysis) can be combined in powerful ways to prove strong results in both the hypercontractive and
non-hypercontractive regimes.

\subsection{Proof Overview: Algorithm}
As discussed above, we consider Nesterov's convex relaxation and generalize the hyperplane rounding
scheme using ``\Holder~duals'' of the Gaussian projections, instead of taking the sign. 
As in the Krivine rounding scheme, this rounding is applied to transformations  of the SDP
solutions. The nature of these transformations depends on how the rounding procedure changes the
correlation between two vectors. Let $u,v \in \R^{N}$ be two unit vectors with $\ip{u,v} =
\rho$. Then, for $\gvec \sim \gauss{0}{I_{N}}$, $\ip{\gvec, u}$ and $\ip{\gvec, v}$ are
$\rho$-correlated Gaussian random variables. Hyperplane rounding then gives $\pm 1$ valued random
variables whose correlation is given by
\[
\Ex{\bfg_1 \sim_\rho \,\bfg_2}{\sgn(\bfg_1) \cdot \sgn(\bfg_2)}~=~\frac{2}{\pi} \cdot \sin^{-1}(\rho) \mper
\]
The transformations $\phi$ and $\psi$ (to be applied to the vectors $u$ and $v$) in Krivine's scheme
are then chosen depending on the Taylor series for the $\sin$ function, which is the inverse of
function computed on the correlation. For the case of \Holder-dual rounding, we prove the following
generalization of the above identity
\[
\Ex{\bfg_1 \sim_\rho \,\bfg_2}{\sgn(\bfg_1) \abs{\gvec_1}^{q-1} \cdot \sgn(\bfg_2)
  \abs{\gvec_2}^{p^*-1}}~=~\gamma_{q}^{q} \cdot  \gamma_{p^*}^{p^*} \cdot 
\rho \cdot \hypergeometric\!\inparen{1-\frac{q}{2},1-\frac{p^*}{2}\,;\,\frac{3}{2}\,;\,\rho^2} \mcom
\]
where $\hypergeometric$ denotes a hypergeometric function with the specified parameters. The proof
of the above identity combines simple tools from Hermite analysis with known integral representations from 
the theory of special functions, and may be useful in other applications of the rounding procedure.

Note that in the Grothendieck case, we have
$\gamma_{p^*}^{p^*} = \gamma_{q}^{q} = \sqrt{2/\pi}$, and the remaining part is simply the
$\sin^{-1}$ function.
In the Krivine rounding scheme, the transformations $\phi$ and $\psi$ are chosen
to satisfy $(2/\pi) \cdot \sin^{-1}\inparen{\ip{\phi(u), \psi(v)}} = c \cdot \ip{u,v}$, where the
constant $c$ then governs the approximation ratio. The transformations $\phi(u)$ and $\psi(v)$ taken
to be of the form $\phi(u) = \oplus_{i=1}^{\infty} a_i \cdot u^{\otimes i}$ such that
\[
\ip{\phi(u), \psi(v)} ~=~ c' \cdot \sin\inparen{\ip{u,v}} \qquad \text{and} \qquad \norm{2}{\phi(u)}
= \norm{\psi(v)} = 1 \mper
\]
If $f$ represents (a normalized version of) the function of $\rho$ occurring in the identity above
(which is $\sin^{-1}$ for hyperplane rounding), then the approximation ratio is governed by the
function $h$ obtained by replacing every Taylor coefficient of $f^{-1}$ by its absolute value. While
$f^{-1}$ is simply the $\sin$ function (and thus $h$ is the $\sinh$ function) in the  Grothendieck
problem, no closed-form expressions are available for general $p$ and $q$.

The task of understanding the approximation ratio thus reduces to the analytic task of understanding
the \emph{family} of the functions $h$ obtained for different values of $p$ and $q$. Concretely, the
approximation ratio is given by the value $1/(h^{-1}(1)\cdot \gamma_{q}\,\gamma_{p^*})$. 
At a high level, we prove bounds on $h^{-1}(1)$ by establishing properties of the Taylor coefficients of 
the family of functions $f^{-1}$, \ie the family given by 
\[
    \inbraces{
    f^{-1}
    ~~|~~
    f(\rho) = \rho\cdot \hypergeometric\!\inparen{a_1,b_1\,;\,3/2\,;\,\rho^2}
    ~,~
    a_1,b_1\in [0,1/2]
    } \mper
\] 
While in the cases considered earlier, the functions $h$ are easy to determine in terms of $f^{-1}$ 
via succinct formulae~\cite{Krivine77, Haagerup81, AN04} or can be truncated after the cubic term~\cite{NR14},
neither of these are true for the family of functions we consider. 
Hypergeometric functions are a rich and expressive class of functions,  capturing many of the 
special functions appearing in Mathematical Physics and various ensembles of orthogonal polynomials. 
Due to this expressive power, the set of inverses is not well understood. 
In particular, while the coefficients of $f$ are monotone in $p$ and $q$, this is not true for
$f^{-1}$. Moreover, the rates of decay of the coefficients may range from inverse polynomial to
super-exponential.
We analyze the coefficients of $f^{-1}$ using complex-analytic methods inspired by (but quite
different from) the work of Haagerup~\cite{Haagerup81} on bounding the  complex Grothendieck
constant. 
The key technical challenge in our work is in \emph{arguing systematically about a family 
of inverse hypergeometric functions} which we address by developing methods to estimate the 
values of a family of contour integrals. 

While our methods only gives a bound of the form $h^{-1}(1) \geq \sinh^{-1}(1)/(1+\eps_0)$, 
we believe this is an artifact of the analysis and the true bound should
indeed be $h^{-1}(1) \geq \sinh^{-1}(1)$.

%% file: PtoQ/2.prelims.tex
\section{Preliminaries and Notation}\label{sec:p_to_q_prelims}
\subsection{Matrix Norms}
For a vector $x\in\R^n$, throughout this chapter we will use $x(i)$ to denote its $i$-th coordinate. 
For $p \in [1, \infty)$, we define $\cnorm{p}{\cdot}$ to denote the counting $p$-norm and $\enorm{p}{\cdot}$
to denote the expectation $p$-norm; \ie for a vector $x\in\R^n$, 
\[
    \cnorm{p}{x} := \inparen{\sum_{i\in [n]} |x(i)|^{p}}^{1/p}
 \quad \mbox{ and } \quad 
    \enorm{p}{x} := \Ex{i\sim [n]}{|x(i)|^p}^{1/p} =\inparen{\frac1n \cdot \sum_{i\in [n]} |x(i)|^{p}}^{1/p}.
\]
Clearly $\cnorm{p}{x} = \enorm{p}{x}\cdot n^{1/p}$. 
For $p = \infty$, we define $\cnorm{\infty}{x} = \enorm{\infty}{x} := \max_{i \in [n]} |x(i)|$. 
We will use $p^*$ to 
denote the `dual' of $p$, i.e. $p^* = p/(p-1)$. 
Unless stated otherwise, we usually work with $\cnorm{p}{\cdot}$. 
We also define inner product $\cip{x,y}$ to denote the inner product under the counting measure unless stated
otherwise; \ie for two vectors $x, y \in \R^n$,
$\cip{x, y} := \sum_{i\in [n]} x(i)y(i)$.

Next, we record a well-known fact about $p$-norms that is used 
in establishing many duality statements.  
\begin{observation}\label{p-norm:dual}
    For any $p\in [1,\infty]$, 
$\cnorm{p}{x} = \sup_{\cnorm{p^*}{y}=1} \,\mysmalldot{y}{x}$.
\end{observation}

We now formally state two related quantities studied in this chapter:
\begin{restatable}{repdefinition}{definitionptoq}\label{def:p_to_q_norm}
  For $p,q\in [1,\infty]$, the \onorm{p}{q} problem is defined as follows:
  \begin{align*}
    \norm{p}{q}{A}=\cnorm{p}{q}{A}&\defeq\sup\frac{\cnorm{q}{Ax}}{\cnorm{p}{x}}\\
    \enorm{p}{q}{A}&\defeq\sup\frac{\enorm{q}{Ax}}{\enorm{p}{x}}\mcom
  \end{align*}
  for an $m\times n$ matrix $A$. 
\end{restatable}

Combining~\cref{p-norm:dual} and~\cref{def:p_to_q_norm}, we see that
the original Grothendieck problem is precisely \groth{\infty}{\infty}. And a natural generalization then
is as follows:
\begin{definition}
    For $p,q\in [1,\infty]$, we define a  generalization of the Grothendieck problem, namely \groth{p}{q}, as the 
    problem of computing 
    \[
        \sup_{\cnorm{p}{y}=1} \,\sup_{\cnorm{q}{x}=1} \mysmalldot{y}{Ax}
    \]
    given an $m\times n$ matrix $A$. 
\end{definition}

\subsection{Preliminaries for Hardness Results}
The symmetry of the hardness results in~\cref{thm:mainthmptoq} is explained by the following observation:
\begin{restatable}{repobservation}{observationptoqsym}\label{onorm:groth}
  For any $p,q\in [1,\infty]$ and any matrix $A$,
  \[
    \cnorm{p}{q}{A} =\sup_{\cnorm{q^*}{y}=1}\,\sup_{\cnorm{p}{x}=1}\mysmalldot{y}{Ax} = \cnorm{q^*}{p^*}{A^T}.
  \]
\end{restatable}
\begin{proof}
Using $\mysmalldot{y}{Ax} = \mysmalldot{x}{A^Ty}$, 
\begin{align*}
        \cnorm{p}{q}{A}
&=        \sup_{\cnorm{p}{x}=1}\cnorm{q}{Ax}
=         \sup_{\cnorm{p}{x}=1}\,\sup_{\cnorm{q^*}{y}=1}\mysmalldot{y}{Ax} 
=        \sup_{\cnorm{q^*}{y}=1}\,\sup_{\cnorm{p}{x}=1}\mysmalldot{y}{Ax} \\
&=    
        \sup_{\cnorm{p}{x}=1}\,\sup_{\cnorm{q^*}{y}=1}\mysmalldot{x}{A^Ty} 
        =
        \sup_{\cnorm{q^*}{y}=1}\cnorm{p^*}{A^Ty} 
	=
        \cnorm{q^*}{p^*}{A^T}\mper 
\tag*{\qedhere}
    \end{align*}
\end{proof}

The following observation will be useful for composing hardness maps for \onorm{p}{2} and \onorm{2}{q} 
to get \onorm{p}{q} hardness for when $p>q$ and $p\geq 2 \geq q$. 
\begin{observation}\label{composition:soundness}
    For any $p,q,r\in [1,\infty]$ and any matrices $B, C$, 
\[
        \cnorm{p}{q}{BC} 
        =
        \sup_x \frac{\cnorm{q}{BCx}}{\cnorm{p}{x}}
        \leq 
        \sup_x \frac{\cnorm{r}{q}{B} \cnorm{r}{Cx}}{\cnorm{p}{x}} 
\leq 
        \cnorm{r}{q}{B} \cnorm{p}{r}{C}.
\]
\end{observation}

\subsubsection{Fourier Analysis}\label{sec:fourier}
We introduce some basic facts about Fourier analysis of Boolean functions used in our proof of hardness
results. Let $R \in \N$ be a positive integer, and consider a function $f : \{ \pm 1 \}^R \to \R$. 
For any subset $S \subseteq [R]$ let $\chi_S := \prod_{i \in S} x_i$. 
Then we can represent $f$ as 
\begin{equation}\label{eq:inverse_fourier}
f(x_1, \dots, x_R) = \sum_{S \subseteq [R]} \hatf(S) \cdot \chi_S(x_1, \dots x_R),
\end{equation}
where 
\begin{equation}
\label{eq:fourier}
\hatf(S) = \Ex{x \in \{ \pm 1 \}^R}{f(x) \cdot \chi_S(x)}\mbox{ for all } S \subseteq [R].
\end{equation}
The {\em Fourier transform} refers to a linear operator $F$ that maps $f$ to $\hatf$ as defined as~\eqref{eq:fourier}. 
We interpret $\hatf$ as a $2^R$-dimensional vector whose coordinates are indexed by $S \subseteq [R]$. 
Endow the expectation norm and the expectation norm to $f$ and $\hatf$ respectively; i.e., 
\[
\enorm{p}{f} := \left( \Ex{x \in \{ \pm 1 \}^R}{|f(x)|^p} \right)^{1/p}
\quad \mbox{ and } \quad 
\cnorm{p}{\hatf} := \left( \sum_{S \subseteq [R]} | \hatf(S) |^p \right)^{1/p}.
\]
as well as the corresponding inner products $\mysmalldot{f}{g}$ and $\mysmalldot{\hatf}{\hatg}$ consistent with their $2$-norms.
We also define the {\em inverse Fourier transform} $F^T$ to be a linear operator 
that maps a given $\hatf : 2^R \to \R$ to $f : \{ \pm 1 \}^R \to \R$ defined as in~\eqref{eq:inverse_fourier}. 
We state the following well-known facts from Fourier analysis. 
\begin{observation} $F$ is an orthogonal transformation; i.e., for any $f, g : \{ \pm 1 \}^R \to \R$, 
\[\ip{f, g} = \ip{F f, F g}.\]
\end{observation}
\begin{observation} [Parseval's Theorem]
For any $f : \{ \pm 1 \}^R \to \R$, $\enorm{2}{f} = \cnorm{2}{F f}$.
\end{observation}
\begin{observation} $F$ and $F^T$ form an adjoint pair; i.e., for any $f : \{ \pm 1 \}^R \to \R$
and $\hatg : 2^R \to \R$, 
\[ 
\mysmalldot{\hatg}{Ff} = 
%\sum_{S \subseteq [R]} \hatg(S) \cdot (Ff)(S)
%=
%\Ex{x}{f(x) \cdot \sum_{S \subseteq [R]} \hatg(S) \chi_S (x)}
%=
\mysmalldot{F^T \hatg}{f}.
 \]
\end{observation}
\begin{observation}
$F^T F$ is the identity operator.
\end{observation}

In~\cref{sec:nhc}, we also consider a {\em partial} Fourier transform $F_P$
that maps a given function $f : \{ \pm 1 \}^R \to \R$ to a vector $\hatf : [R] \to \R$
defined as $\hatf(i) = \Ex{x \in \{ \pm 1 \}^R}{f(x) \cdot x_i}$ for all $i \in [R]$.
It is the original Fourier transform where $\hatf$ is further projected to $R$ coordinates corresponding to linear coefficients. 
The partial inverse Fourier transform $F_P^T$ is a transformation that maps 
a vector $\hatf : [R] \to \R$ to a function $f : \{ \pm 1 \}^R \to \R$ as in~\eqref{eq:inverse_fourier} restricted to $S = \{ i \}$ for some $i \in [R]$. 
These partial transforms satisfy similar observations as above: 
(1) $\enorm{2}{f} \geq \cnorm{2}{F_P f}$, 
(2) $\enorm{2}{F_P^T \hatf} = \cnorm{2}{\hatf}$,
(3) $F_P$ and $F_P^T$ form an adjoint pair, 
and (4) $(F_P^T F_P) f = f$ if and only if $f$ is a linear function. 

We now state and prove a useful lemma involving Fourier coefficients of linear functions that is very
useful for our proof of hardness.
\subsubsection{An Useful Lemma for Dictatorship Test}\label{sec:ptoq:dic}

\newcommand{\innernorm}{\cnorm}
\newcommand{\outernorf}{\enorm}

\begin{lemma}\label{lem:dict}
    Let $f : \{ \pm 1 \}^R \to \R$ be a linear %\mgnote{Euiwoong: please verify that the
   % in the application of the lemma, the function $g$ is linear and not boolean}
    function for
    some positive integer $R \in \N$ and $\hatf : [R] \to \R$ be its linear Fourier
    coefficients defined by 
    \[ \hatf(i) \defeq \Ex{x \in \{ \pm 1 \}^R}{x_i f(x)}.\]
    For all $\varepsilon > 0$, there exists $\delta > 0$
    such that if $\enorm{r}{f} > (\gamma_r + \varepsilon) \cnorm{2}{\hatf}$ then
    $\cnorm{4}{\hatf} > \delta \cnorm{2}{\hatf}$ for all $1 \leq r < 2$.
\end{lemma}
Before we can prove the above, we state and prove an implication of Berry-\Esseen~estimate for fractional moments
(similar to Lemma 3.3 of~\cite{GRSW16}, see also~\cite{KNS10}).
\begin{lemma}\label{lem:clt-analytic-sparse}
  There exist universal constants $c>0$ and $\delta_0>0$ such that the following statement is true.  If
  $X_1,\ldots,X_n$ are bounded independent random variables with $\abs{X_i}\le 1$, $\Ex{X_i}=0$ for
  $i\in[n]$, and $\sum_{i\in[n]}\Ex{X_i^2}=1$, $\sum_{i\in[n]}\Ex{\abs{X_i}^3}\leq \delta$ for some $0 <
  \delta < \delta_0$, then for every $p\geq 1$:
    \[
        \inparen{\Ex{\abs{\sum_{j=1}^n X_j}^p}}^{\frac{1}{p}} \leq \gamma_p\cdot
        \inparen{1+c\delta\inparen{\log\inparen{\nfrac{1}{\delta}}}^{\frac{p}{2}}}.
    \]
\end{lemma}
\begin{proof}
    The proof is almost similar to that of Lemma 2.1 of \cite{KNS10}.
    From Berry-\Esseen theorem (see \cite{Beek72} for the constant), we get that:
    \[
        \Pr{\abs{\sum_{i=1}^nX_i}\geq u} \leq \Pr{\abs{g}\geq u} +
        2\sum_{i=1}^n\Ex{\abs{X_i}^3} \leq \Pr{\abs{g}\geq u} + 2\delta\mcom
    \] for every $u>0$ and where $g \sim \gaussian{0}{1}$.
    By Hoeffding's lemma, 
\[
\Pr{\abs{\sum_{i\in[n]} X_i}\ge t} < 2 \ee^{-2t^2}
\]
    for every $t>0$.
    Combining the above observations, we get:
    \begin{align*}
        \Ex{\abs{\sum_{i=1}^n X_i}^p} & = \int_{0}^{\infty} p u^{p-1} 
            \Pr{\abs{\sum_{i=1}^nX_i}\geq u}du \\ 
        & \leq \int_{0}^{a} pu^{p-1}\Pr{\abs{g}>u} du + 2\delta a^p +
            2\int_{a}^{\infty} pu^{p-1} \ee^{-2u^2} du\\
        & = \sqrt{\frac{2}{\pi}}\int_{0}^{a} u^p \ee^{-\nfrac{u^2}{2}}du + 2\delta
        a^p + \frac{2p}{2^{\frac{p-1}{2}}}\int_{2a^2}^{\infty} z^{\frac{p+1}{2}-1} \ee^{-z} dz\\
        & =\gamma_p^p -\sqrt{\frac{2}{\pi}}\int_{a}^{\infty} u^p\ee^{-\nfrac{u^2}{2}}du + 2\delta a^p
        +\Gamma\inparen{\frac{p+1}{2},2a^2} ~ \mcom
    \end{align*}
    where $\Gamma(\cdot,\cdot)$ is the upper incomplete gamma function and $a$ is a large constant determined later depending on $\delta$ and $p$.
	The second term is bounded as 
    \begin{align*}
        \int_{a}^{\infty}u^p\ee^{-\nfrac{u^2}{2}}du = a^{p-1}
        \ee^{-\nfrac{a^2}{2}} + (p-1) \int_{a}^{\infty}
        u^{p-2}\ee^{-\nfrac{u^2}{2}} du \leq a^{p-1}
        \ee^{-\nfrac{a^2}{2}} + \frac{p-1}{a^2} \int_{a}^{\infty}
        u^{p}\ee^{-\nfrac{u^2}{2}} du \mper
    \end{align*}
    Hence $\int_{a}^{\infty}u^p\ee^{-\nfrac{u^2}{2}}du \leq
    \frac{a^{p+1}e^{-\nfrac{a^2}{2}}}{1+a^2-p}$.
    
    We know, $\Gamma(\nfrac{p+1}{2},x)\to x^{\frac{p-1}{2}}\ee^{-x}$ as $x\to\infty$.
    We choose $a=\gamma_p \sqrt{\log\frac{1}{\delta}}$. Hence, there exists $\delta_0$
    so that for all small enough
    $\delta<\delta_0$, we have $\Gamma(\nfrac{p+1}{2},2a^2) \sim 2^{\frac{p-1}{2}}a^{p-1}
    \delta^{2\gamma_p^2}\ll \delta a^p$ where the last inequality follows from 
    the fact that $2\gamma_p^2 > 1$ (as $p>1$).
    Putting all this together, we get:
    \[
        2\delta a^p +\Gamma\inparen{\frac{p+1}{2},2a^2}-
        \sqrt{\frac{2}{\pi}}\int_{a}^{\infty}u^p\ee^{-\nfrac{u^2}{2}}du 
        \ll 3\delta a^p -\sqrt{\frac{2}{\pi}}\frac{a^{p+1}e^{-\nfrac{a^2}{2}}}{1+a^2-p}
        \leq
        c\gamma_p^p \delta \inparen{\log{\frac{1}{\delta}}}^{\nfrac{p}{2}}\mcom
    \]
    where $c$ is an absolute constant independent of $a$ and $p$.
    This completes the proof of the lemma.
\end{proof}

Finally, we prove \cref{lem:dict}:
\begin{proofof}{\cref{lem:dict}}
    We will prove this lemma by the method of contradiction.
    Let us assume $\cnorm{4}{\hatf} \leq \delta\cnorm{2}{\hatf}$, for $\delta$ to be
    fixed later.

    Let us define $y_i\defeq \frac{\hatf(i)}{\cnorm{2}{\hatf}}$.
    Then, for all $x\in\on^R$,  
    \[
        g(x)\defeq \sum_{i\in[n]} x_i\cdot y_i = \frac{f(x)}{\cnorm{2}{\hatf}}\mper
    \]
    Let $Z_i=x_i\cdot y_i$ be the random variable when $x_i$ is independently uniformly
    randomly chosen from $\on$. Now
    \[
        \sum_{i\in[n]}\Ex{Z_i^2}=\sum_{i\in[n]}\frac{\hatf(i)^2}{\cnorm{2}{\hatf}^2}=1\mper
    \]
    and
    \[
        \sum_{i\in[n]}\Ex{\abs{Z_i}^3} = \sum_{i\in[n]}
        \frac{\abs{\hatf(i)}^3}{\cnorm{2}{\hatf}^3}= \sum_{i\in[n]}
        \frac{\abs{\hatf(i)}^2}{\cnorm{2}{\hatf}^2}\cdot \frac{\abs{\hatf(i)}}{\cnorm{2}{\hatf}}
        \leq \frac{\cnorm{4}{\hatf}^2}{\cnorm{2}{\hatf}^2} \leq \delta^2 \mcom
    \]
    where the penultimate inequality follows from Cauchy-Schwarz ineqality.

    Hence, by applying \cref{lem:clt-analytic-sparse}
    on the random variables $Z_1,\ldots, Z_n$, we get:
    \begin{align*}
        \frac{\enorm{r}{f}}{\cnorm{2}{\hatf}}=
        \enorm{r}{g} & = \inparen{ \Ex{x\in \on^n} {\abs{g(x)}^r}}^\frac{1}{r}\\
        & = \inparen{ \Ex{x\in \on^n} {\abs{\sum_{i\in[n]}Z_i}^r}}^\frac{1}{r}\\
        & \leq \gamma_r\inparen{1+c\delta^2 \inparen{\log{\frac{1}{\delta}}}^{r}}\\
    \end{align*}
    We choose $\delta>0$ small enough (since $1 \leq r < 2$, setting $\delta<
    \frac{\sqrt{\varepsilon}}{\min(\delta_0,\sqrt{\gamma_2}\log\frac{c\gamma_2}{\varepsilon})}
= \frac{\sqrt{\varepsilon}}{\min(\delta_0, \log\frac{c}{\varepsilon})}$
    suffices) so that $\delta^2(\log\frac{1}{\delta})^r <
    \frac{\varepsilon}{c\gamma_r}$. For this choise of $\delta$, we get:
    $\enorm{r}{f}\leq (\gamma_r+\varepsilon)\cnorm{2}{\hatf}$ -- a contradiction.
    And hence the proof follows.
\end{proofof}

\subsubsection{Starting Point of Hardness Reduction: Smooth Label Cover}
An instance of \labelcover~is given by a quadruple $\calL = (G, [R], [L], \Sigma)$ that consists of a
regular connected graph $G = (V, E)$, a label set $[R]$ for some positive integer $n$, and a collection
$\Sigma = ((\pi_{e, v}, \pi_{e, w}) : e = (v, w) \in E)$ of pairs of maps both from $[R]$ to $[L]$
associated with the endpoints of the edges in $E$. Given a {\em labeling} $\ell : V \to [R]$, we say that
an edge $e = (v, w) \in E$ is {\em satisfied } if $\pi_{e, v}(\ell(v)) = \pi_{e, w}(\ell(w))$. Let
$\opt(\calL)$ be the maximum fraction of satisfied edges by any labeling. 

The following hardness result for \labelcover, given in~\cite{GRSW16}, is a slight variant of the original construction due to~\cite{Khot02}. The theorem also describes the various structural properties, including smoothness, that are identified by the hard instances.
\begin{theorem}\label{thm:smooth_label_cover}
  For any $\xi >0$ and $J \in \N$, there exist positive integers $R = R(\xi, J), L = L(\xi, J)$ and $D =
  D(\xi)$, and a \labelcover instance $(G, [R], [L], \Sigma)$ as above such that
  \begin{itemize}
    \item (Hardness): It is NP-hard to distinguish between the following two cases:
      \begin{itemize}
        \item (Completeness): $\opt(\calL) = 1$. 
        \item (Soundness): $\opt(\calL) \leq \xi$.
      \end{itemize}
    \item (Structural Properties): 
      \begin{itemize}
        \item ($J$-Smoothness): For every vertex $v \in V$ and distinct $i, j \in [R]$, we have 
          \[
            \Pr{e : v \in e}{\pi_{e,v}(i)=\pi_{e,v}(j)} \leq 1 / J.
          \]
        \item ($D$-to-$1$): For every vertex $v \in V$, edge $e \in E$ incident on $v$, and $i \in [L]$,
          we have $|\pi^{-1}_{e, v}(i)| \leq D$; that is at most $D$ elements in $[R]$ are mapped to the
          same element in $[L]$. 
        \item (Weak Expansion): For any $\delta > 0$ and vertex set $V' \subseteq V$ such that $|V'| =
          \delta \cdot |V|$, the number of edges among the vertices in $|V'|$ is at least $(\delta^2 / 2)|E|$.
      \end{itemize}
  \end{itemize} 
\end{theorem}

\subsection{Preliminaries for Algorithm}\label{sec:p_to_q:prelims:algo}
For $p\ge 2\ge q\ge 1$, we will use the following notation: $a\defeq p^*-1$ and $b\defeq q-1$. We note that
$a,b\in [0,1]$.

For an $m\times n$ matrix $M$ (or vector, when $n=1$). For a unitary function $f$, we define $f[M]$ to 
be the matrix $M$ with entries defined as $(f[M])_{i,j}=f(M_{i,j})$ for $i\in[m],j\in[n]$. For vectors 
$u,v\in\R^{\ell}$, we denote by $u\circ v\in\R^{\ell}$ the entry-wise/Hadamard product of $u$ and $v$. 
We denote the concatenation of two vectors $u$ and $v$ by $u\oplus v$. For a vector $u$, we 
use $\Diag{u}$ to denote the diagonal matrix with the entries of $u$ forming the diagonal, and for 
a matrix $M$ we use $\diag(M)$ to denote the vector of diagonal entries. 

For a function $f(\tau)=\sum_{k\ge 0} f_k\cdot \tau^k$ defined as a power series, we denote the function
$\absolute{f}(\tau)\defeq \sum_{k\ge 0}\abs{f_k}\cdot \tau^k$. 

Next, we recall the relaxation due to Nesterov~\cite{Nesterov98}:
\begin{figure}[ht]\label{fig:convex}
\hrule
\vline
\begin{minipage}[t]{0.99\linewidth}
\vspace{-5 pt}
{\small
\begin{align*}
    \mbox{maximize} \quad&~~\sum_{i,j} A_{i,j}\cdot \mysmalldot{u^i}{v^j} ~~=~~  \ip{A, U V^T}\\
    \mbox{subject to}    \quad&~~\sum_{i\in [m]} \norm{2}{u^i}^{q^*} ~\leq~ 1 & \forall i\in [m]\\
    &~~\sum_{j\in [n]} \norm{2}{v^j}^{p} ~\leq~ 1 & \forall j\in [n]\\
    &~~ u^i, v^j\in \R^{m+n} & \forall i\in [m], j\in [n] \\
    u^i \text{ (resp. $v^j$) is the} & \text{ $i$-th (resp. $j$-th) row of $U$ (resp. $V$)} 
\end{align*}
}
\vspace{-14 pt}
\end{minipage}
\hfill\vline
\hrule
\caption{The relaxation $\CP{A}$ for approximating  \onorm{p}{q} of a matrix $A \in \R^{m \times n}$.} 
\end{figure}
Note that since $q^* \geq 2$ and $p \geq 2$, the above program is convex in the entries of the Gram
matrix of the vectors $\inbraces{u^i}_{i \in [m]}\cup  \inbraces{v^j}_{j \in [n]}$.

Next, we collect some basic ideas of Hermite Analysis that are relevant to analysis of our algorithm.
\subsubsection[Hermite Preliminaries]{Hermite Analysis Preliminaries}
Let $\gamma$ denote the standard Gaussian probability distribution. 
For this section (and only for this section), the (Gaussian) inner product for functions
$f,h\in (\R,\gamma) \rightarrow \R$ is defined as 
\[
    \mysmalldot{f}{h}:=\int_{\R} f(\tau)\cdot h(\tau)\,\,d\gamma(\tau)
    =\Ex{\tau\sim \gaussian{0}{1}}{f(\tau)\cdot h(\tau)}\mper
\]
Under this inner product there is a complete set of orthonormal polynomials $(H_k)_{k\in \N}$ defined below. 
\begin{definition}
    \label[def-hermite]{def:hermite}
    For a natural number $k$, then the $k$-th \emph{Hermite} polynomial $H_k: \R\to\R$
    \[
        H_k(\tau)=\frac{1}{\sqrt{k!}}\cdot (-1)^{\,k} \cdot \ee^{\tau^2/2}\cdot 
        \frac{d^k}{d\tau^k}\,\ee^{-\tau^2/2} \mper
    \]
\end{definition}
\noindent
Any function $f$ satisfying $\int_{\R}|f(\tau)|^{2}\,d\gamma(\tau)<\infty$ has a Hermite expansion given by 
$
    f = \sum_{k\geq 0} \widehat{f}_{k} \cdot H_{k}
$
where 
$
    \widehat{f}_k = \mysmalldot{f}{H_k} \mper
$
\medskip

\noindent
We have 
\begin{fact}
\label[fact]{even:odd:hermite}
    $H_k(\tau)$ is an even (resp. odd) function when $k$ is even (resp. odd). 
\end{fact}
We also have the Plancherel Identity (as Hermite polynomials form an orthonormal basis):
\begin{fact}
    \label[fact]{plancherel-identity}
    For two real valued functions $f$ and $h$ with Hermite coefficients $\widehat{f}_k$ and
    $\widehat{h}_k$, respectively, we have:
    \[
        \mysmalldot{f}{h} = \sum_{k\geq 0} \widehat{f}_k\cdot \widehat{h}_k \mper
    \]
\end{fact}
The generating function of appropriately normalized Hermite polynomials satisfies the following identity: 
\begin{equation}
\label[equation]{hermite:generating:function}
    e^{\,\tau\lambda - \lambda^2/2} 
    = 
    \sum_{k\geq 0} H_k(\tau)\cdot \frac{\lambda^{k}}{\sqrt{k!}} \mper
\end{equation}

Similar to the noise operator in Fourier analysis, we define the corresponding noise operator $\noise{}$
for Hermite analysis: 
\begin{definition}
    For $\rho \in [-1,1]$ and a real valued function $f$, we define the function $\noise{f}$ as:
    \[
        (\noise{f})(\tau) = \int_{\R} f\inparen{\rho\cdot \tau + \sqrt{1-\rho^2}\cdot
        \theta}\,d\gamma(\theta) =\Ex{\tau'\sim_\rho\,\tau}{f(\tau')}\mper
    \]
    \label[definition]{noise-operator}
\end{definition}
Again similar to the case of Fourier analysis, the Hermite coefficients admit the following identity:
\begin{fact}\label[fact]{hermite-noise}
    $
        \widehat{(\noise{f})}_k =  \rho^{k} \cdot \widehat{f}_k\mper
    $
\end{fact}

Next, we collect some definitions and facts from the theory of special functions useful to our analysis
of the algorithm:
\subsubsection{Gamma, Hypergeometric and Parabolic Cylinder Function Preliminaries}
For a natural number $k$ and a real number $\tau$, we denote the rising factorial as
$(\tau)_k\defeq \tau \cdot (\tau+1)\cdot \cdots (\tau+k-1)$.
We now define the following fairly general classes of functions, and we later use them we obtain a
Taylor series representation of $\fplain{a}{b}{\tau}$.
\begin{definition}
    \label[definition]{conf-hypergeometric}
    The confluent hypergeometric function with parameters $\alpha,\beta$, and $\lambda$ as:
    \[
        \confHG(\alpha\,;\beta\,;\lambda)\defeq \sum_{k} \frac{(\alpha)_k}{(\beta)_k}\cdot \frac{\lambda^k}{k!}\mper
    \]
\end{definition}
The (Gaussian) hypergeometric function is defined as follows:
\begin{definition}
    \label[definition]{hypergeometric}
    The hypergeometric function with parameters $w,\alpha,\beta$ and $\lambda$ as:
    \[
        \hypergeometric(w,\alpha\,;\beta\,;\lambda) \defeq \sum_k \frac{(w)_k\cdot (\alpha)_k}{(\beta)_k}\cdot
        \frac{\lambda^k}{k!}\mper
    \]
\end{definition}

    Next we define the $\Gamma$ function:
    \begin{definition}
	    For a real number $\tau$, we define:
	    \[
	        \Gamma(\tau)\defeq \int_{0}^{\infty}t^{\tau-1}\cdot \ee^{-t}\,dt\mper
	    \]
	    \label[definition]{Gamma-function}
	\end{definition} The $\Gamma$ function has the following property:
    \begin{fact}[Duplication Formula]
    \label[fact]{duplication:formula}
        \[
            \frac{\Gamma(2\tau)}{\Gamma(\tau)} = \frac{\Gamma(\tau+1/2)}{2^{1-2\tau}\sqrt{\pi}}
        \]
    \end{fact}
	We also note the relationship between $\Gamma$ and $\gamma_r$:
	\begin{fact}
	\label[fact]{gamma:and:Gamma}
	    For $r\in [0,\infty)$, 
	    \[
	        \gamma_r^{r}~:=~\Ex{\bfg\sim \gaussian{0}{1}}{|\bfg|^r}~=~\frac{2^{r/2}}{\sqrt{\pi}}\cdot
	        \Gamma\inparen{\frac{1+r}{2}} \mper
	    \]
	\end{fact}
	\begin{proof}
	    \begin{align*}
	        \Ex{\bfg\sim\gaussian{0}{1}}{\abs{\bfg}^r} &= \frac{\sqrt{2}}{\sqrt{\pi}}\cdot\int_{0}^{\infty}
	        \abs{\bfg}^r\cdot\ee^{-\nfrac{\bfg^2}{2}}\,d\bfg\\
	        &= \sqrt{\frac{2}{\pi}}\cdot 2^{(r-1)/2}\cdot\int_{0}^{\infty}\abs{\frac{\bfg^2}{2}
	        }^{(r-1)/2}\!\cdot\ee^{-\nfrac{\bfg^2}{2}}\!\cdot \bfg\,d\bfg\\
	        &= \frac{2^{r/2}}{\sqrt{\pi}}\cdot\Gamma\inparen{\frac{1+r}{2}}
	    \end{align*} 
	\end{proof}

Next, we record some facts about parabolic cylinder functions:
\begin{fact}[12.5.1 of \cite{Lozier03}]
    \label[fact]{integral-parabolic}
    Let $U$ be the function defined as 
    \[
        U(\alpha,\lambda)\defeq\frac{\ee^{\lambda^2/4}}{\Gamma\inparen{\frac{1}{2}+\alpha}}
        \int_{0}^{\infty} t^{\alpha-1/2}\cdot \ee^{-(t+\lambda)^2/2} \,dt\mcom
    \]for all $\alpha$ such that $\Re(\alpha)>-\frac{1}{2}\mper$
    The function $U(\alpha,\pm \lambda)$ is a parabolic cylinder function and is a standard
    solution to the differential equation: $\frac{d^2 w}{d\lambda^2}-\inparen{\frac{\lambda^2}{4}+\alpha}w=0$. 
\end{fact}
Next we quote the confluent hypergeometric representation of the parabolic cylinder function $U$ 
defined above:
\begin{fact}[12.4.1, 12.2.6, 12.2.7, 12.7.12, and 12.7.13 of \cite{Lozier03}]
    \label[fact]{power-parabolic}
    \begin{align*}
        U(\alpha,\lambda)=\frac{\sqrt{\pi}}{2^{\alpha/2+1/4}\cdot
        \Gamma\inparen{\frac{3}{4}+\frac{\alpha}{2}}}\cdot
        \ee^{\lambda^2/4}\cdot\confHG\inparen{-\frac{1}{2}\alpha+\frac{1}{4}\,;\,
        \frac{1}{2}\,;\,-\frac{\lambda^2}{2}}\\
        -\frac{\sqrt{\pi}}{2^{\alpha/2-1/4}\cdot
        \Gamma\inparen{\frac{1}{4}+\frac{\alpha}{2}}}\cdot\lambda\cdot\ee^{\lambda^2/4}
        \cdot\confHG\inparen{-\frac{\alpha}{2}+\frac{3}{4}\,;\,\frac{3}{2}\,;\,-\frac{\lambda^2}{2}}
    \end{align*}
\end{fact}
Combining the previous two facts, we get the following:
\begin{corollary}
    \label[corollary]{parabolic-conf-hypergeometric}
    For all real $\alpha>-\frac{1}{2}$, we have:
    \begin{align*} 
        \int_{0}^{\infty} t^{\alpha-1/2}\cdot \ee^{-(t+\lambda)^2/2} \,dt ~~=~
        \frac{\sqrt{\pi}\cdot\Gamma\inparen{\frac{1}{2}+\alpha}}{2^{\alpha/2+1/4}
        \cdot\Gamma\inparen{\frac{3}{4}+\frac{\alpha}{2}}}
        \cdot\confHG\inparen{-\frac{\alpha}{2}+\frac{1}{4}\,;\,\frac{1}{2}\,;\,-\frac{\lambda^2}{2}}
        \\-\frac{\sqrt{\pi}\cdot\Gamma\inparen{\frac{1}{2}+\alpha}}{2^{\alpha/2-1/4}
        \cdot\Gamma\inparen{\frac{1}{4}+\frac{\alpha}{2}}}
        \cdot \lambda\cdot\confHG\inparen{-\frac{\alpha}{2}+\frac{3}{4}\,;\,\frac{3}{2}\,;\,-\frac{\lambda^2}{2}}
        \mper
    \end{align*}
\end{corollary}
%%We next observe that the integral in \cref{gen:function:integral:rep} is closely related to 
%%a special function called the parabolic cylinder function. We then use known facts about the 
%%relation between parabolic cylinder functions and confluent hypergeometric functions, to show 
%%that the Hermite coefficients of $\alfwild{c}$ can be obtained from the Taylor coefficients of a 
%%confluent hypergeometric function.

%% file: PtoQ/3.hardness.tex
\section{Hardness of \onorm{2}{r} with $r < 2$}\label{sec:nhc}

This section proves the following theorem that serves as a starting point of our hardness results. 
The theorem is stated for the expectation norm for consistency with the current literature, but the same statement holds for the counting norm, since if 
$A$ is an $n \times n$ matrix,  $\cnorm{2}{r}{A} = n^{1/r -1/2} \cdot \enorm{2}{r}{A}$. 
Note that the matrix $A$ used in the reduction below does not depend on $r$.
\begin{theorem} For any $\varepsilon > 0$, 
    there is a polynomial time reduction that takes a 3-CNF formula $\phi$ and 
    produces a symmetric matrix $A \in \R^{n \times n}$ with $n = |\phi|^{\poly(1 / \eps)}$ such that 
	\begin{itemize}
        \item (Completeness) If $\phi$ is satisfiable, there exists $x \in
            \R^{n}$ with $|x(i)| = 1$ for all $i \in [n]$ and $Ax = x$.  In particular,
            $\enorm{2}{r}{A} \geq 1$ for all $1 \leq r \leq \infty$.
        \item (Soundness) $\enorm{2}{r}{A} \leq \gamma_r + \varepsilon^{2 - r}$ for all $1 \leq r < 2$. 
	\end{itemize}
    \label{thm:brs}
\end{theorem}
We adapt the proof by \Briet, Regev and Saket for the hardness of $2 \to 1$ and $\infty \to 1$
norms to prove the above theorem. A small difference is that, 
unlike their construction which starts with a Fourier encoding of the long-code functions, we start
with an evaluation table (to ensure that the resulting matrices are symmetric). We also analyze
their dictatorship tests for the case of fractional $r$. 
\subsection{Reduction and Completeness}
Let $\calL = (G, [R], [L], \Sigma)$ be an
instance of \labelcover with $G = (V, E)$. 
In the rest of this section, $n = |V|$ and our reduction will construct a self-adjoint linear operator 
$\bfA : \R^N \to \R^N$ with $N = |V| \cdot 2^R$, which yields a symmetric $N \times N$ matrix
representing $\bfA$ in the standard basis. 
This section concerns the following four Hilbert spaces based on the standard Fourier
analysis composed with $\calL$.
\begin{enumerate}
    \item Evaluation space $\R^{2^R}$. Each function in this space is
        denoted by $f : \{ \pm 1 \}^R \to \R$.  The inner product is defined as
        $\ip{f, g} := \Ex{x \in \{ \pm 1 \}^R}{f(x) g(x)}$, which
        induces $\norm{2}{f} := \enorm{2}{f}$. We also define $\enorm{p}{f} := \Ex{x}{|f(x)|^p}^{1/p}$ in this space. 
    \item Fourier space $\R^{R}$. Each function in this space is denoted by $\hatf : [R]
        \to \R$.  The inner product is defined as $\mysmalldot{\hatf}{\hatg} := \sum_{i
        \in [R]} \hatf(i) \hatg(i)$, which induces $\norm{2}{\hatf} := \cnorm{2}{\hatf}$. 

    \item Combined evaluation space $\R^{V \times 2^R}$. Each function in this space is
        denoted by $\bff : V \times \{ \pm 1 \}^R \to \R$.  The inner product is defined as
        $\mysmalldot{\bff}{\bfgg} := \Ex{v \in V}{\Ex{x \in \{ \pm 1 \}^R}{\bff(v, x)
        \bfgg(v, x)}}$, which induces $\enorm{2}{\bff} := \enorm{2}{\bff}$. 
We also define $\norm{p}{\bff} := \Ex{v, x}{|\bff(v, x)|^p}^{1/p}$ in this space.
    \item Combined Fourier space $\R^{V \times R}$. Each function in this space is denoted
        by $\bfhatf : V \times [R] \to \R$.  The inner product is defined as
        $\mysmalldot{\bfhatf}{\bfhatg} := \Ex{v \in V}{\sum_{i \in [R]} \bfhatf(v, i)
    \bfhatg(v, i) }$, which induces $\norm{2}{\bfhatf}$, which is neither a counting nor an
    expectation norm.
\end{enumerate}
Note that $\bff \in \R^{V \times 2^R}$ and a vertex
$v \in V$ induces $f_v \in \R^{2^R}$ defined by $f_v(x) := \bff(v, x)$, and similarly
$\bfhatf \in \R^{V \times R}$ and a vertex $v \in V$ induces $\hatf_v \in \R^{R}$
defined by $\hatf_v(x) := \bfhatf(v, x)$.  
As defined in \cref{sec:fourier}, 
we use the standard
following (partial) {\em Fourier transform} $F$ that maps $f \in \R^{2^R}$ to $\hatf \in
\R^{R}$ as follows.  \footnote{We use only {\em linear Fourier coefficients} in this
work. $F$ was defined as $F_P$ in~\cref{sec:fourier}.} \begin{equation} \hatf(i) = (Ff) (i) := \Ex{x \in \{\pm 1 \}^R}{x_i f(x)}.  \end{equation} 

%Let $F : \R^{2^R} \to \R^R$ be the linear operator representing the Fourier transform
%(i.e., $(F f)(i) = \hatf (i)$ for all $f$ and $i$).  
The (partial) {\em inverse Fourier
transform} $F^T$ that maps $\hatf \in \R^{R}$ to $f \in \R^{2^R}$ is defined by
\begin{equation} f(x) = (F^T \hatf)(x) := \sum_{i \in [R]} x_i \hatf(i). \end{equation}

This Fourier transform can be naturally extended to combined spaces by defining $\bfF :
\bff \mapsto \bfhatf$ as $f_v \mapsto \hatf_v$ for all $v \in V$.  Then $\bfF^T$ maps
$\bfhatf$ to $\bff$ as $\hatf_v \mapsto f_v$ for all $v \in V$.

Finally, let $\bfhatP : \R^{V \times R} \to \R^{V
\times R}$ be the orthogonal projector to the following subspace of the combined
Fourier space: 
\begin{equation} 
    \bfhatL := \left\{ \bfhatf \in \R^{V \times R} :
    \sum_{j \in \pi^{-1}_{e, u}(i)} \hatf_u(i) = \sum_{j \in \pi^{-1}_{e, v}(i)}
    \hatf_v(j) \mbox{ for all } (u, v) \in E \mbox{ and } i \in [L] \right\}\mper
\end{equation}
Our transformation $\bfA : \R^{V \times 2^R} \to \R^{V \times 2^R}$ is
defined by \begin{equation} \bfA := (\bfF^T) \bfhatP \bfF.  \end{equation} In other
words, given $\bff$, we apply the Fourier transform for each $v \in V$, project the
combined Fourier coefficients to $\bfhatL$ that checks the \labelcover consistency, and
apply the inverse Fourier transform. Since $\bfhatP$ is a projector, $\bfA$ is
self-adjoint by design. 

We also note that a similar reduction that produces $(\bfF^T) \bfhatP$ 
was used in Guruswami \etal~\cite{GRSW16} and \Briet \etal~\cite{BRS15} 
for subspace approximation and Grothendieck-type problems, 
and indeed this reduction suffices for~\cref{thm:brs}
except the self-adjointness and additional properties in the completeness case. 

\paragraph{Completeness.}
We prove the following lemma for the completeness case. 
A simple intuition is that if $\calL$ admits a good labeling, 
we can construct a $\bff$ such that each $f_v$ is a linear function and 
$\bfhatf$ is already in the subspace $\bfhatL$. 
Therefore, each of Fourier transform, projection to $\bfhatL$, and inverse Fourier transform 
does not really change $\bff$.

\begin{lemma} [Completeness] Let $\ell : V \to [R]$ be a labeling that satisfies every
    edge of $\calL$. There exists a function $\bff \in \R^{V \times 2^R}$ such that
    $\bff(v, x)$ is either $+1$ or $-1$ for all $v \in V, x \in \{ \pm 1 \}^R$ and
    $\bfA \bff = \bff$.
\end{lemma}
\begin{proof} Let $\bff(v, x) := x_{\ell(v)}$ for
    every $v \in V, x \in \{ \pm 1 \}^R$.  Consider $\bfhatf = \bfF \bff$.  For each
    vertex $v \in V$, $\bfhatf(v, i) = \hatf_v(i) = 1$ if $i = \ell(v)$ and $0$
    otherwise.  Since $\ell$ satisfies every edge of $\calL$, $\bfhatf \in \bfhatL$ and
    $\bfhatP \bfhatf = \bfhatf$.  Finally, since each $f_v$ is a linear function, the
    partial inverse Fourier transform $F^T$ satisfies $(F^T) \hatf_v = f_v$, which
    implies that $(\bfF^T) \bfhatf = \bff$. Therefore, $\bfA \bff = (\bfF^T \bfhatP
    \bfF) \bff = \bff$.
\end{proof}

\subsection{Soundness}
We prove the following soundness lemma. 
This finishes the proof of~\cref{thm:brs} since~\cref{thm:smooth_label_cover} 
guarantees NP-hardness of \labelcover for arbitrarily small $\xi > 0$ and arbitrarily large $J \in \N$. 

\begin{lemma}[Soundness] \label{lem:soundness} For every $\varepsilon > 0$, 
	there exist $\xi > 0$ (that determines $D = D(\xi)$ as in~\cref{thm:smooth_label_cover}) 
	and $J \in \N$ such that if $\opt(\calL) \leq \xi$, $\calL$ is $D$-to-$1$, and $\calL$ is $J$-smooth, 
	$\enorm{2}{r}{\bfA} \leq \gamma_{r} + 4\varepsilon^{2-r}$ for every $1 \leq r < 2$. 
    %For any $\epsilon > 0$ there exists $\xi > 0$, $J \in \N$ such that $\norm{}{\calf}^p
    %> \tau^p + 4\epsilon^{2 - p}$ then there exists an assignment that satisfies at least
    %$\xi$-fraction of the edges of $G$. 
\end{lemma}

The proof closely follows that of Lemma 3.3 of~\cite{BRS15}. 
We (partially) generalize their proof for $1$-norm (of general norms) to $r$-norms when $1 \leq r < 2$. 

\begin{proof}
Let $\bff \in \R^{V \times 2^R}$ be an arbitrary vector such that $\enorm{2}{\bff} =
1$.  Let $\bfhatf = \bfF \bff$, $\bfhatg = \bfhatL \bfhatf$, and $\bfgg = \bfF^T
\bfhatg$ so that $\bfgg = (\bfF^T \bfhatL \bfF) \bff = \bfA \bff$.  By Parseval's
theorem, $\cnorm{2}{\hatf_v} \leq \enorm{2}{f_v}$ for all $v \in V$ and
$\norm{2}{\bfhatf} \leq \enorm{2}{\bff} \leq 1$.  Since $\bfhatL$ is an orthogonal
projection, $\norm{2}{\bfhatg} \leq \norm{2}{\bfhatf} \leq 1$.  
Fix $1 \leq r < 2$ and suppose
\begin{equation}
    \enorm{r}{\bfgg}^r = \Ex{v \in V}{\enorm{r}{g_v}^r} \geq \gamma_r^r +
    4 \epsilon^{2 - r}\mper
    \label{eq:soundness}
\end{equation}

Use~\cref{lem:dict} to obtain $\delta = \delta(\epsilon)$ such that $\enorm{p}{g_v}^p >
(\gamma_p^p + \epsilon) \cnorm{2}{\hatg_v}^p$ implies $\cnorm{4}{\hatg} > \delta
\cnorm{2}{\hatg}$ for all $1 \leq p < 2$ (so that $\delta$ does not depend on $r$), 
and consider \begin{equation} V_0 := \{ v \in V : \cnorm{4}{\hatg_v}
    > \delta \epsilon \mbox{ and } \cnorm{2}{\hatg_v} \leq 1/\epsilon \}.
    \label{eq:vzero} \end{equation} We prove the following lemma that lower bounds the
size of $V_0$.
\begin{lemma}%[Lemma 3.4 of BRS15]
    For $V_0 \subseteq V$ defined as in~\eqref{eq:vzero}, we have $|V_0| \geq \epsilon^2
    |V|$.
\end{lemma}
\begin{proof} The proof closely follows the proof of Lemma 3.4 of~\cite{BRS15}.
    Define the sets
    \begin{align*}
        V_1 &= \{ v \in V : \cnorm{4}{\hatg_v}
        \leq \delta \epsilon \mbox{ and } \cnorm{2}{\hatg_v} < \epsilon \}, \\ V_2 &= \{ v
        \in V : \cnorm{4}{\hatg_v} \leq \delta \epsilon \mbox{ and } \cnorm{2}{\hatg_v}
        \geq \epsilon \}, \\ V_3 &= \{ v \in V : \cnorm{2}{\hatg_v} > 1/\epsilon \}.
    \end{align*}
    From~\eqref{eq:soundness}, we have
    \begin{equation}
        \sum_{v \in V_0} \enorm{r}{g_v}^r + \sum_{v \in V_1} \enorm{r}{g_v}^r + \sum_{v \in V_2}
        \enorm{r}{g_v}^r + \sum_{v \in V_3} \enorm{r}{g_v}^r \geq (\gamma_r^r + 4
        \epsilon^{2 - r}) |V|\mper
        \label{eq:soundness_sum}
    \end{equation}
    We bound the four sums on the left side of~\eqref{eq:soundness_sum} individually.
    Parseval's theorem and the fact that $r < 2$ implies $\enorm{r}{g_v} \leq
    \enorm{2}{g_v}=\cnorm{2}{\hatg_v}$, and since $\cnorm{2}{\hatg_v} \leq 1/\epsilon$
    for every $v \in V_0$, the first sum in~\eqref{eq:soundness_sum} can be bounded by
    \begin{equation}
        \sum_{v \in V_0} \enorm{r}{g_v}^r \leq |V_0| / \epsilon^r.
        \label{eq:soundness_zero}
    \end{equation}
    Similarly, using the definition of $V_1$
    the second sum in~\eqref{eq:soundness_sum} is at most $\epsilon^r |V|$.
    By~\cref{lem:dict}, for each $v \in V_2$, we have $\enorm{r}{g_v}^r \leq (\gamma_r^r +
    \epsilon) \cnorm{2}{\hatg_v}^r$. Therefore, the third sum in~\eqref{eq:soundness_sum}
    is bounded as
    \begin{align}
        \sum_{v \in V_2} \enorm{r}{g_v}^r &\leq (\gamma_r^r +
            \epsilon) \sum_{v \in V_2} \cnorm{2}{\hatg_v}^r    &   \nonumber \\
            &= (\gamma_r^r + \epsilon) |V_2| \Ex{v \in V_2}{\cnorm{2}{\hatg_v}^r}
            &     \nonumber \\
            &\leq (\gamma_r^r + \epsilon) |V_2| \Ex{v \in V_2}{\cnorm{2}{\hatg_v}^2}^{r / 2} &&
            \mbox{(By Jensen using $r < 2$)}    \nonumber \\
        &= (\gamma_r^r + \epsilon) |V_2| \bigg( \frac{\sum_{v \in V_2}
            \cnorm{2}{\hatg_v}^2 }{|V_2|} \bigg) ^{r / 2} &\nonumber \\
        &\leq (\gamma_r^r + \epsilon) |V_2|^{1 - r / 2} |V|^{r / 2} &&
            (\sum_{v \in V_2} \cnorm{2}{\hatg_v}^2 \leq \sum_{v \in V} \cnorm{2}{\hatg_v}^2
             \leq |V|) \nonumber \\
        &\leq (\gamma_r^r + \epsilon) |V|. &
    \end{align}
    Finally, the fourth sum in~\eqref{eq:soundness_sum} is bounded by
    \begin{align}
        \sum_{v \in V_3} \enorm{r}{g_v}^r &\leq \sum_{v \in V_3} \enorm{2}{g_v}^r 
            && \mbox{(Since $r < 2$)}\nonumber \\
        &= \sum_{v \in V_3} \cnorm{2}{\hatg_v}^r    && \mbox{(By Parseval's
           theorem)}   \nonumber \\
        &= \sum_{v \in V_3} \cnorm{2}{\hatg_v}^{r - 2}\cnorm{2}{\hatg_v}^2 \nonumber \\
        &< \sum_{v \in V_3} \epsilon^{2 - r}\cnorm{2}{\hatg_v}^2 
            && (\cnorm{2}{\hatg_v} > 1/\epsilon \mbox{ for } v \in V_3,\mbox{ and } r < 2)
            \nonumber \\
        &= \epsilon^{2 - r} \sum_{v \in V_3}\cnorm{2}{\hatg_v}^2 \leq \epsilon^{2- r}|V|.&&
    \end{align} Combining the above with~\eqref{eq:soundness_sum} yields
    \begin{align}
        |V_0| & \geq \epsilon^{r} \sum_{v\in V_0} \enorm{r}{g_v}^r  \nonumber \\
        & \geq \epsilon^{r} \bigg( (\gamma_r^r + 4\epsilon^{2 - r}) |V| - 
            \epsilon^r |V| - (\gamma_r^r + \epsilon)|V| - \epsilon^{2-r} |V| \bigg)
            \nonumber \\
        & \geq \epsilon^{r} \epsilon^{2 - r} |V|  =   \epsilon^2 |V|, 
    \end{align} where the last inequality uses the fact that
    $\epsilon^{2- r } \geq \epsilon \geq \epsilon^r$.
\end{proof}
Therefore, $|V_0| \geq \epsilon^2 |V|$ and every vertex of $v$ satisfies 
$\cnorm{4}{\hatg_v} > \delta \epsilon$ and $\cnorm{2}{\hatg_v} \leq 1 / \epsilon$.
Using only these two facts
together with $\bfhatg \in \bfhatL$, \Briet \etal~\cite{BRS15} proved that if the
smoothness parameter $J$ is large enough given other parameters, $\calL$
admits a labeling that satisfies a significant fraction of edges.
\begin{lemma} [Lemma
    3.6 of \cite{BRS15}] Let $\beta := \delta^2 \epsilon^3$.  There exists an absolute
    constant $c' > 0$ such that if $\calL$ is $T$-to-$1$ and $T / (c' \epsilon^8
    \beta^4)$-smooth for some $T \in \N$, there is a labeling that satisfies at least $\epsilon^8 \beta^4
    /1024$ fraction of $E$.
\end{lemma}
This finishes the proof of~\cref{lem:soundness} by setting $\xi := \epsilon^8 \beta^4 /1024$
and $J := D(\xi) / (c' \epsilon^8 \beta^4)$ with $D(\xi)$ defined in~\cref{thm:smooth_label_cover}. 
Given a $3$-SAT formula, $\phi$, 
by the standard property of Smooth Label Cover, the size of the reduction is 
$|\phi|^{O(J \log (1 / \xi))} = |\phi|^{\poly(1/\eps)}$. 
\end{proof}

\section{Hardness of \onorm{p}{q}}\label{sec:hypercontractive}
In this section, we prove our main results. 
We prove \cref{thm:main_nhc} on hardness of approximating \onorm{p}{q} when $p \geq 2 \geq q$, and 
\cref{thm:main_hype} on hardness of approximating \onorm{p}{q} when 
$2 < p < q$. By duality, the same hardness is implied for the case of $p < q < 2$.  

Our result for $p \geq 2 \geq q$ in~\cref{sec:composition} follows 
from \cref{thm:brs} using additional properties in the completeness case. 
For hypercontractive norms, we start by showing 
constant factor hardness via reduction from \onorm{p}{2} (see \cref{isometry}), and then amplify the 
hardness factor by using the fact that all hypercontractive norms productivize under Kronecker product, 
which we prove in \cref{sec:productivize}. 

\subsection{Hardness for $p \geq 2 \geq q$}
\label{sec:composition}

We use~\cref{thm:brs} to prove hardness of \onorm{p}{q} for $p \geq 2 \geq q$, which proves~\cref{thm:main_nhc}.

\begin{proofof}{\cref{thm:main_nhc}}
	Fix $p, q$, and $\delta > 0$ such that $\infty \geq p \geq 2 \geq q$ and $p > q$. 
	Our goal is to prove that \onorm{p}{q} is NP-hard to approximate within a factor $1 / (\gamma_{p^*}\gamma_q + \delta)$. 
	For \onorm{2}{q} for $1 \leq q < 2$, Theorem~\ref{thm:brs} (with $\varepsilon \leftarrow \delta^{1/(2-q)}$) directly proves a hardness
     ratio of $1/(\gamma_q + \varepsilon^{2-q}) = 1/(\gamma_q + \delta)$. 
	By duality, it also gives an $1/(\gamma_{p^*} + \delta)$ hardness for \onorm{p}{2} for $p > 2$. 

    For \onorm{p}{q} for $p > 2 > q$,
	apply~\cref{thm:brs} with $\varepsilon = (\delta / 3)^{\max(1/(2-p^*), 1/(2-q))}$. 
	It gives a polynomial time reduction that produces a symmetric matrix $A \in \R^{n \times n}$ given a \threesat formula $\phi$. 
Our instance for \onorm{p}{q} is  $AA^T = A^2$.
    \begin{itemize}
        \item (Completeness) If $\phi$ is satisfiable, there exists $x \in \R^{n}$ such that $|x(i)| = 1$ for all $i \in [N]$ and $Ax = x$. Therefore, $A^2
            x = x$ and $\enorm{p}{q}{A^2} \geq 1$. 

        \item (Soundness) If $\phi$ is not satisfiable, 
\begin{align*}
\enorm{p}{2}{A} &~=~ \enorm{2}{p^*}{A} \leq  \gamma_{p^*} + \varepsilon^{2-p^*} \leq  \gamma_{p^*} +
                  \delta / 3, 
\mbox{ and } \\
\enorm{2}{q}{A} &~\leq~ \gamma_{q} + \varepsilon^{2-q} \leq \gamma_{q} + \delta / 3. 
\end{align*}
This implies that 
            \[ \enorm{p}{q}{A^2} \leq \enorm{p}{2}{A}  \enorm{2}{q}{A} \leq
                (\gamma_{p^*} + \delta /3)(\gamma_{q} + \delta / 3) \leq
                \gamma_{p^*}\gamma_q + \delta \mper \]
    \end{itemize} 
This creates a gap of $1 / (\gamma_{p^*} \gamma_q + \delta)$ between the completeness and the soundness case. The same gap holds for the counting norm since $\cnorm{p}{q}{A^2} = n^{1/q - 1/p} \cdot \enorm{p}{q}{A^2}$. 
\end{proofof}

\subsection{Reduction from \onorm{p}{2} via Approximate Isometries}
\label{isometry}
Let $A \in \R^{n \times n}$ be a hard instance of \onorm{p}{2}. For any $q \geq 1$, if a matrix 
$B \in \R^{m \times n}$ satisfies $\cnorm{q}{Bx} = (1 \pm o(1)) \cnorm{2}{x}$ for all $x \in \R^n$,  
then $\norm{p}{q}{BA} = (1 \pm o(1)) \norm{p}{2}{A}$. Thus, $BA$ will serve as a hard instance for 
\onorm{p}{q} if one can compute such a matrix $B$ efficiently. In fact, a consequence of the 
Dvoretzky-Milman theorem is that a sufficiently tall random matrix $B$ satisfies the aforementioned 
property with high probability. 
In other words, for $m=m(q,n)$ sufficiently large, a random linear operator from 
$\ell_2^n$ to $\ell_q^m$ is an approximate isometry. 

To restate this from a geometric perspective, for $m(q,n)$ sufficiently larger than $n$, a random section 
of the unit ball in $\ell_q^m$ is approximately isometric to the unit ball in $\ell_2^n$. In the interest of 
simplicity, we will instead state and use a corollary of the following matrix deviation inequality due to 
Schechtman (see \cite{Schechtman06}, Chapter 11 in \cite{Vershynin17}). 
\begin{theorem}[Schechtman~\cite{Schechtman06}]
\label{matrix:deviation:bound}
    Let $B$ be an $m \times n$ matrix with i.i.d. $\gaussian{0}{1}$ entries. Let $f : \R^m\to \R$ be a 
    positive-homogeneous and subadditive function, and let $b$ be such that 
    $f(y)\leq b \cnorm{2}{y}$ for all $y\in \R^m$. Then for any $T\subset \R^n$, 
    \[
        \sup_{x\in T} |f(Bx)-\Ex{f(Bx)}| = O(b\cdot \gamma(T) + t\cdot \mathrm{rad}(T))
    \]
    with probability at least $1-e^{-t^2}$, 
    where $\mathrm{rad}(T)$ is the radius of $T$, and 
    $\gamma(T)$ is the Gaussian complexity of $T$ defined as 
    \[
        \gamma(T) := \Ex{g\sim \gaussian{0}{I_n}}{\sup_{t\in T} |\mysmalldot{g}{t}|}
    \]
\end{theorem}
The above theorem is established by proving that the random process given by
$X_x := f(Bx) - \Ex{f(Bx)}$  has sub-gaussian increments with respect to $L_2$ and subsequently
appealing to  Talagrand's Comparison tail bound.

We will apply this theorem with $f(\cdot) = \cnorm{q}{\cdot}$, $b=1$ and $T$ being the unit ball 
under $\cnorm{2}{\cdot}$. We first state a known estimate of $\Ex{f(Bx)} = \Ex{\cnorm{q}{Bx}}$ for any fixed $x$ satisfying $\cnorm{2}{x} = 1$. 
Note that when $\cnorm{2}{x} = 1$, $Bx$ has the same distribution as an $m$-dimensional random vector with i.i.d. $\gaussian{0}{1}$ coordinates. 

\begin{theorem}[Biau and Mason~\cite{BM15}]
\label{expected:qnorm}
    Let $X\in\R^m$ be a random vector with i.i.d. $\gaussian{0}{1}$ coordinates. Then for any $q\geq 2$, 
    \[
        \Ex{\cnorm{q}{X}} = m^{1/q}\cdot \gamma_q+O(m^{(1/q) -1)}).
    \]
\end{theorem}

We are now equipped to see that a tall random Gaussian matrix is an approximate isometry (as a linear map 
from $\ell_2^n$ to $\ell_q^m$) with high probability. 
\begin{corollary}
\label{embedding}
    Let $B$ be an $m \times n$ matrix with i.i.d. $\gaussian{0}{1}$ entries where $m = \omega(n^{q/2})$. 
    Then with probability at least $1-e^{-n}$, every vector $x\in \R^n$ %\suchthat \cnorm{2}{x} = 1$ 
    satisfies, \[\cnorm{q}{Bx} = (1\pm o(1))\cdot m^{1/q} \cdot \gamma_q  \cdot \cnorm{2}{x}. \]
\end{corollary}

\begin{proof}
    We apply \cref{matrix:deviation:bound} with function $f$ being the $\ell_q$ norm, $b=1$, and $t = \sqrt{n}$. 
    Further we set $T$ to be the $\ell_2$ unit sphere, which yields $\gamma(T) = \Theta(\sqrt{n})$ and 
    $\mathrm{rad}(T) = 1$.  Applying \cref{expected:qnorm} yields that with probability at least 
    $1 - e^{t^2} = 1-e^{-n}$, for all $x$ with $\cnorm{2}{x} = 1$, we have 
\begin{align*}
\left| \cnorm{q}{Bx} - m^{1/q}  \cdot \gamma_q \right| 
& \leq 
\left| \cnorm{q}{Bx} - \Ex{\cnorm{q}{X}} \right| 
+ 
\left| \Ex{\cnorm{q}{X}} - m^{1/q} \cdot  \gamma_q \right|  \\
& \leq O(b \cdot \gamma(T) + t  \cdot \mathrm{rad}(T) + m^{(1/q)-1})  \\
& \leq O(\sqrt{n} + \sqrt{n} + m^{(1/q)-1})\\ % \quad \quad \qquad \qquad \qquad \quad \quad \quad (t \leftarrow \sqrt{n}) \\
& \leq o(m^{1/q}). \tag*{\qedhere}
\end{align*}
\end{proof}

We thus obtain the desired constant factor hardness:
\begin{proposition}
\label{hypercontractive:const:factor}
    For any $p>2,~2\leq q< \infty$ and any $\varepsilon>0$, there is no polynomial time algorithm that 
    approximates \onorm{p}{q} (and consequently \onorm{q^*}{p^*}) within a factor of 
    $1/\gamma_{p^*}-\varepsilon$ ~unless $\NP \not\subseteq \BPP$. 
\end{proposition}

\begin{proof}
    By \cref{embedding}, for every $n\times n$ matrix $A$ and a random $m\times n$ matrix $B$ with i.i.d. 
    $\gaussian{0}{1}$ entries ($m = \omega(n^{q/2})$), with probability at least $1-e^{-n}$, we have 
    \[
        \cnorm{p}{q}{BA} = 
        (1\pm o(1))\cdot \gamma_q\cdot m^{1/q}\cdot \cnorm{p}{2}{A}.
    \]
    Thus the reduction $A\rightarrow BA$ combined with \onorm{p}{2} hardness implied by \cref{thm:brs}, 
    yields the claim. 
\end{proof}

The generality of the concentration of measure phenomenon underlying the proof of the 
Dvoretzky-Milman theorem allows us to generalize \cref{hypercontractive:const:factor}, to obtain constant 
factor hardness of maximizing various norms over the $\ell_p$ ball ($p>2$).  In this more general version, 
the strength of our hardness assumption is dependent on the Gaussian width of the dual of the norm being 
maximized. Its proof is identical to that of \cref{hypercontractive:const:factor}. 

\begin{theorem}
\label{p-to-anything:const:factor}
    Consider any $p>2, \varepsilon>0$, and any family $(f_m)_{m\in \N}$ of positive-homogeneous and  
    subadditive functions where $f_m : \R^m\to \R$. Let $(b_m)_{m\in \N}$ be such that $f_m(y) \leq b_m\cdot 
    \cnorm{2}{y}$ for all $y$ and let $N=N(n)$ be such that $\gamma_*(f_N) = \omega(b_N\cdot \sqrt{n})$, 
    where 
    \[
        \gamma^*(f_N) := \Ex{g\sim \gaussian{0}{I_N}}{f_N(g)}.
    \] 
    Then unless $\NP \not\subseteq \BPTIME{N(n)}$, there is no polynomial 
    time $(1/\gamma_{p^*}-\varepsilon)$-approximation algorithm for the problem of computing 
    $\sup_{\norm{p}{x} = 1} f_m(Ax)$, given an $m\times n$ matrix $A$. 
\end{theorem}

\subsection{Derandomized Reduction}\label{sec:derandomization}
In this section,  we show how to derandomize the reduction in \cref{hypercontractive:const:factor}
to obtain  NP-hardness when $q\geq 2$ is an even integer and $p > 2$.  Similarly to~\cref{isometry},
given $A \in \R^{n \times n}$ as a hard instance of \onorm{p}{2}, our strategy is to construct a
matrix $B \in \R^{m \times n}$ and output $BA$ as a hard instance of \onorm{p}{q}.

Instead of requiring $B$ to satisfy $\cnorm{q}{Bx} = (1 \pm o(1)) \cnorm{2}{x}$ for all $x \in
\R^n$, we show that $\cnorm{q}{Bx} \leq (1 + o(1)) \cnorm{2}{x}$ for all $x \in \R^n$  and
$\cnorm{q}{Bx} \geq (1 - o(1)) \cnorm{2}{x}$ when every coordinate of $x$ has the same absolute
value. Since~\cref{thm:brs} ensures that $\cnorm{p}{2}{A}$ is achieved by $x = Ax$ for such a
well-spread $x$ in the completeness case,  $BA$ serves as a hard instance for \onorm{p}{q}. 

 We use the following construction of $q$-wise independent sets to construct such a $B$ deterministically.

\begin{theorem}[Alon, Babai, and Itai~\cite{ABI86}]
\label{kwise:independence}
    For any $k\in\N$, one can compute a set $S$ of vectors in $\{\pm 1\}^n$ of size $O(n^{k/2})$, in time 
    $n^{O(k)}$, such that the vector random variable $Y$ obtained by sampling uniformly from $S$ satisfies 
    that for any $I\in \binom{[n]}{k}$, the marginal distribution $Y\restrict{I}$ is the uniform distribution over 
    $\{\pm 1\}^{k}$. 
\end{theorem}

For a matrix $B$ as above, a randomly chosen row behaves similarly to an $n$-dimensional
Rademacher random vector with respect to $\cnorm{q}{\cdot}$. 

\begin{corollary}
\label{derandomized:operator}
    Let $R\in\R^n$ be a vector random variable with i.i.d. Rademacher ($\pm 1$) coordinates. 
    For any even integer $q\geq 2$, there is an $m \times n$ matrix $B$ with $m = O(n^{q/2})$, computable in $n^{O(q)}$ time, 
    such that for all $x\in\R^n$, we have 
    \[
        \cnorm{q}{Bx} = m^{1/q} \cdot \Ex{R}{\mysmalldot{R}{x}^q}^{1/q}.
    \]
\end{corollary}
\begin{proof}\belowdisplayskip=-12pt
    Let $B$ be a matrix, the set of whose rows is precisely $S$.
	%Let $Y$ be sampled uniformly from $S$. 
    By \cref{kwise:independence}, 
    \begin{align*}
        \cnorm{q}{Bx}^q
        =
        \sum_{Y \in S}{\mysmalldot{Y}{x}^q}
        &= 
        m\cdot \Ex{R}{\mysmalldot{R}{x}^q}% &&\text{(by \cref{kwise:independence})} \\
%        &= 
%        \norm{q}{\mysmalldot{R}{X}}^q.
        \mper \tag*{\qedhere}
    \end{align*}
\end{proof}
%\elnote{qedhere.}

We use the following two results that will bound $\cnorm{p}{q}{BA}$ for the completeness case and the soundness case respectively.

\begin{theorem}[Stechkin~\cite{Stechkin61}]
\label{stechkin:clt}
    Let $R\in\R^n$ be a vector random variable with i.i.d. Rademacher coordinates. Then for any $q\geq 2$ 
    and any $x\in \R^n$ whose coordinates have the same absolute value, 
    \[
        \Ex{\mysmalldot{R}{x}}^{1/q} = (1-o(1))\cdot \gamma_q \cnorm{2}{x}.
    \]
\end{theorem}

\begin{theorem}[Khintchine inequality~\cite{Haagerup81}]
\label{khintchine:clt}
    Let $R\in\R^n$ be a vector random variable with i.i.d. Rademacher coordinates. Then for any $q\geq 2$ 
    and any $x \in \R^n$, 
    \[
        \Ex{\mysmalldot{R}{x}^q}^{1/q} \leq \gamma_q \cdot \cnorm{2}{x}.
    \]
\end{theorem}

We finally prove the derandomimzed version of 
\cref{hypercontractive:const:factor} for even $q \geq 2$.

\begin{proposition}
\label{hypercontractive:derandomized:const:factor}
    For any $p > 2, \varepsilon>0$, and any even integer $q\geq 2$, it is NP-hard to approximate 
    \onorm{p}{q} within a factor of ~$1/\gamma_{p^*}  - \varepsilon$.
\end{proposition}

\begin{proof}
    Apply \cref{thm:brs} with $r_1 \leftarrow p^*$ and $\varepsilon \leftarrow \varepsilon$.  Given an instance $\phi$ of \threesat, 
    \cref{thm:brs} produces a symmetric matrix $A \in \R^{n \times n}$ in polynomial time as a hard instance of \onorm{p}{2}.  Our instance for \onorm{p}{q} is $BA$ where $B$ is 
    the $m \times n$ matrix given by \cref{derandomized:operator} with $m = O(n^{q/2})$. 
    \begin{itemize}
        \item (Completeness) If $\phi$ is satisfiable, there exists a vector $x\in\{\pm \frac{ 1}{\sqrt{n}}\}^n$ such that $Ax = x$. 
        So we have $\cnorm{q}{BAx} = \cnorm{q}{Bx} = (1-o(1))\cdot m^{1/q} \cdot \gamma_q$, where the last 
        equality uses \cref{derandomized:operator} and \cref{stechkin:clt}. Thus, $\cnorm{p}{q}{BA}\geq 
        (1-o(1))\cdot m^{1/q} \cdot \gamma_q$. 
        
        \item (Soundness) If $\phi$ is not satisfiable, then for any $x$ with $\cnorm{p}{x}=1$, 
            \begin{align*}
                &\cnorm{q}{BAx} = m^{1/q} \cdot \Ex{R}{\mysmalldot{R}{Ax}^q}^{1/q} \leq m^{1/q} \cdot \gamma_q\cdot \cnorm{2}{Ax} 
                 \\
                \leq~ & m^{1/q} \cdot  \gamma_q \cdot \cnorm{p}{2}{A} 
                \leq~ m^{1/q} \cdot  \gamma_q \cdot (\gamma_{p^*}-\varepsilon)
            \end{align*}
            where the first inequality is a direct application of~\cref{khintchine:clt}. 
            \hfill\qedhere
   \end{itemize}
\end{proof}

\subsection{Hypercontractive Norms Productivize}\label{sec:productivize}
We will next amplify our hardness results using the fact that hypercontractive norms
productivize under the  natural operation of Kronecker or tensor product. 
Bhaskara and Vijayraghavan~\cite{BV11} showed this for the special case of $p=q$ and the Harrow  and 
Montanaro~\cite{HM13} showed this for \onorm{2}{4} (via parallel repetition for $\mathrm{QMA(2)}$). 
In this section we prove this claim whenever $p\leq q$. 

\begin{theorem}
\label{productivization}
    Let $A$ and $B$ be $m_1\times n_1$ and $m_2\times n_2$ matrices respectively. Then for any $1\leq p\leq q < \infty$, 
    $\cnorm{p}{q}{A\otimes B} \leq \cnorm{p}{q}{A}\cdot \cnorm{p}{q}{B}$. 
\end{theorem}

\begin{proof}
    We will begin with some notation. Let $a_{i},b_{j}$ respectively denote the $i$-th and $j$-th rows of $A$ 
    and $B$. Consider any $z\in \R^{[n_1]\times [n_2]}$ satisfying $\cnorm{p}{z}=1$. 
    For $k\in [n_1]$, let $z_k\in \R^{n_2}$ denote the vector given by $z_k(\ell) := z(k,\ell)$. 
    For $j\in [m_2]$, let $\barz_j\in \R^{n_1}$ denote the vector given by $\barz_j(k) := 
    \mysmalldot{b_{j}}{z_{k}}$. 
    Finally, for $k\in [n_1]$, let $\lambda_k := \cnorm{p}{z_k}^{p}$ and let $v_k\in \R^{m_2}$ be the vector given 
    by $v_k(j) := |\barz_j(k)|^{p}/\lambda_k$. 
    
    We begin by 'peeling off' $A$:
    \begin{align*}
        \cnorm{q}{(A\otimes B)z}^{q} 
        ~=~ 
        \sum_{i,j} |\mysmalldot{a_i\otimes b_j}{z}|^{q}  
        &~=~
        \sum_{j} \sum_{i} |\mysmalldot{a_i}{\barz_j}|^{q}  \\
        &~=~
        \sum_{j} \cnorm{q}{A\barz_j}^{q}  \\ 
        &~\leq~
        \cnorm{p}{q}{A}^{q}\cdot \sum_{j} \cnorm{p}{\barz_j}^{q}  \\
        &~=~
        \cnorm{p}{q}{A}^{q}\cdot \sum_{j} \inparen{\cnorm{p}{\barz_j}^{p}}^{q/p} 
    \end{align*}
    In the special case of $p=q$, the proof ends here since the expression is a sum of terms of the 
    form $\cnorm{p}{By}^p$ and can thus be upper bounded term-wise by 
    $\cnorm{p}{p}{B}^p\cdot \cnorm{p}{z_k}^p$ which sums to $\cnorm{q}{p}{B}^p$. To handle the case of 
    $q>p$, we will use a convexity argument: 
    \begin{align*}
        &\cnorm{p}{q}{A}^{q}\cdot \sum_{j} \inparen{\cnorm{p}{\barz_j}^{p}}^{q/p}  \\
        ~=~
        &\cnorm{p}{q}{A}^{q}\cdot \sum_{j} \inparen{\sum_k |\barz_j(k)|^{p}}^{q/p}  \\
        ~=~
        &\cnorm{p}{q}{A}^{q}\cdot \cnorm{q/p}{\sum_k \lambda_k\cdot v_k}^{q/p}  
&&(|\barz_j(k)|^{p} = \lambda_k v_k(j)) \\
        ~\leq~
        &\cnorm{p}{q}{A}^{q}\cdot \sum_k \lambda_k \cdot \cnorm{q/p}{v_k}^{q/p}  
        &&(\text{by convexity of  }\norm{q/p}{\cdot}^{q/p} \text{ when } q\geq p) \\
        ~\leq~
        &\cnorm{p}{q}{A}^{q}\cdot \max_k \cnorm{q/p}{v_k}^{q/p}  
    \end{align*}
    
    It remains to show that $\cnorm{q/p}{v_k}^{q/p}$ is precisely $\cnorm{q}{Bz_k}^q/\cnorm{p}{z_k}^q$. 
    \begin{align*}
        \cnorm{p}{q}{A}^{q}\cdot \max_k \cnorm{q/p}{v_k}^{q/p} 
        ~=~
        &\cnorm{p}{q}{A}^{q}\cdot \max_k \frac{1}{\cnorm{p}{z_k}^q}\cdot \sum_j |\barz_j(k)|^q  \\
        ~=~
        &\cnorm{p}{q}{A}^{q}\cdot \max_k \frac{1}{\cnorm{p}{z_k}^q}\cdot \sum_j |\mysmalldot{b_j}{z_k}|^q  \\
        ~=~
        &\cnorm{p}{q}{A}^{q}\cdot \max_k \frac{\cnorm{q}{Bz_k}^q}{\cnorm{p}{z_k}^q}  \\
        ~\leq~
        &\cnorm{p}{q}{A}^{q}\cdot \cnorm{p}{q}{B}^q   
    \end{align*}
    Thus we have established $\cnorm{p}{q}{A\otimes B} \leq \cnorm{p}{q}{A}\cdot \cnorm{p}{q}{B}$. 
    Lastly, the claim follows by observing that the statement is equivalent to the statement obtained by 
    replacing the counting norms with expectation norms. 
\end{proof}

We finally establish super constant NP-Hardness of approximating \onorm{p}{q}, proving \cref{thm:main_hype}.
\begin{proofof}{\cref{thm:main_hype}}
Fix $2 < p \leq q < \infty$. 
\cref{hypercontractive:const:factor} states that there exists $c = c(p, q) > 1$ such that 
any polynomial time algorithm approximating the \onorm{p}{q} of an $n \times n$-matrix $A$ within a factor of $c$ will imply $\NP \subseteq \BPP$. 
Using \cref{productivization}, for any integer $k \in \N$ and $N = n^k$, any polynomial time algorithm approximating the \onorm{p}{q} of an $N \times N$-matrix $A^{\otimes k}$ within a factor of $c^k$ 
implies that $\NP$ admits a randomized algorithm running in time $\poly(N) = n^{O(k)}$. 
Under $\NP \not\subseteq \BPP$, any constant factor approximation algorithm is ruled out by setting $k$ to be a sufficiently large constant. 
For any $\varepsilon > 0$, setting $k = \log^{1/\varepsilon} n$ rules out an approximation factor of 
$c^k = 2^{O(\log^{1 - \varepsilon} N)}$ unless $\NP \subseteq \BPTIME{2^{\log^{O(1)} n}}$.

By duality, the same statements hold for $1 < p \leq q < 2$. 
When $2 < p \leq q$ and $q$ is an even integer, all reductions become deterministic due to \cref{hypercontractive:derandomized:const:factor}.
\end{proofof}
\subsection{A Simple Proof of Hardness for the Case $2 \notin [q,p]$}\label{sec:reverse}
In this section, we prove~\cref{thm:main_reverse}. That is, we show how to prove an almost-polynomial
factor hardness for approximating \onorm{p}{q} in the non-hypercontractive case when $2 > p \geq q$ (and
the case $p \geq q > 2$ follows by duality). 
This result is already known from the work of Bhaskara and Vijayaraghavan \cite{BV11}. We
show how to obtain a more modular proof, composing our previous results with a simple embedding
argument. 
However, while the reduction in \cite{BV11} was deterministic, we will only give a randomized
reduction below.

As in \cite{BV11}, we start with a strong hardness for the \onorm{p}{p}, obtained in
\cref{thm:main_hype}. While the reduction in \cite{BV11} relied on special properties of the
instance for \onorm{\ell_p}{\ell_p}, we can simply use the following embedding result of Schechtman
\cite{Schechtman87} (phrased in a way convenient for our application). 
\begin{theorem}[Schechtman \cite{Schechtman87}, Theorem 5]
Let $q<p<2$ and $\eps>0$. Then, there exists a polynomial time samplable distribution $\calD$ 
on random matrices in $\R^{m \times
  n}$ with $m = \Omega_{\eps}(n^{3})$,  such that with probability $1-o(1)$, we have for every
$x \in R^n$,  $\norm{\ell_q}{Bx} ~=~ (1 \pm \eps) \cdot \norm{\ell_p}{x}$.
\label{thm:schechtman}
\end{theorem}
In fact the distribution $\calD$ is based on $p$-stable distributions. 
While the theorem in~\cite{Schechtman87} does not mention the high probability bound or
samplability, it is easy to modify the proof to obtain these properties. We provide a proof sketch
below for completeness. We note that Schechtman obtains a stronger bound of $O(n^{1+p/q})$ on the
dimension $m$ of the $\ell_q$ space, which requires a more sophisticated argument using ``Lewis
weights''. However, we only state weaker $O(n^3)$ bound above, which suffices for our purposes and
is easier to convert to a samplable distribution. 

We first prove the following hardness result for approximating \onorm{p}{q} in the
reverse-hypercontractive case, using \cref{thm:schechtman}.
\begin{theorem}
For any $p, q$ such that $1 < q \leq p < 2$ or $2 < q \leq p < \infty$ and $\epsilon > 0$,  there
is no polynomial time algorithm that approximates the \onorm{p}{q} of an $n \times n$ matrix within
a factor $2^{\log^{1 - \epsilon} n}$ unless $\NP \subseteq \BPTIME{2^{(\log n)^{O(1)}}}$. 
\label{thm:reverse}
\end{theorem}
\begin{proof}
We consider the case $1 < q \leq p < 2$ (the other case follows via duality).
\cref{thm:main_hype} gives a reduction from SAT on $n$ variables, approximating the \onorm{p}{p} of
matrices $A \in \R^{N \times N}$ with $N = 2^{(\log n)^{O(1/\eps)}}$, within a factor $2^{(\log
  N)^{1-\eps}}$. Sampling a matrix $B$ from the distribution $\calD$ given by \cref{thm:schechtman}
(with dimension $N$) gives that it is also hard to approximate $\norm{p}{q}{BA} \approx \norm{p}{p}{A}$,
within a factor $2^{(\log N)^{1-\eps}}$.
\end{proof}
We now give a sketch of the proof of \cref{thm:schechtman} including the samplability condition. The
key idea is to embed the space $\ell_p^n$ into the infinite-dimensional space $L_q$ (for $0 \leq q
\leq p < 2$) using $p$-stable random variables. The corresponding subspace of $L_q$ can then be
embedded into $\ell_q^m$ if the random variables (elements of $L_q$) constructed in the previous
space are bounded in $L_{\infty}$ norm. This is the content of the following claim.
\begin{claim}[Schechtman \cite{Schechtman87}, Proposition 4]
\label{clm:schechtman}
Let $\eps > 0$ and $\Omega$ be an efficiently samplable probability space and let $V$ be an
$n$-dimensional subspace of $L_q(\Omega)$, such that
\[
M ~\defeq~ \sup\inbraces{\norm{L_{\infty}}{f} ~\mid~ \norm{L_{q}}{f} \leq 1, f \in V} ~<~ \infty \mper
\]
Then there exists a polynomial time samplable distribution $\calD$ over linear operators $T:
L_q(\Omega) \to \R^{m}$ for $m = C(\eps, q) \cdot n \cdot M^q$ such that with
probability $1 - o(1)$, we have that for every $f \in V$, $\norm{\ell_q}{Tf} = (1\pm \eps) \cdot \norm{L_q}{f}$.
\end{claim}
\begin{proofsketch}
The linear operator is simply defined by sampling $x_1, \ldots, x_m \sim \Omega$ independently, and
taking
\[
Tf ~\defeq~ \frac{1}{m^{1/q}} \cdot \inparen{f(x_1), \ldots, f(x_m)} \qquad \forall f\mper
\]
The proof then follows by concentration bounds for $L_{\infty}$-bounded random variables, and a
union bound over an epsilon net for the space $V$.
\end{proofsketch}
The problem then reduces to constructing an embedding of $\ell_p^n$ into $L_q$, which is bounded in
$L_{\infty}$ norm. While a simple embedding can be constructed using $p$-stable distributions,
Schechtman uses a clever reweighting argument to control the $L_{\infty}$ norm. We show below that a
simple truncation argument can also be used to obtain a somewhat crude bound on the $L_{\infty}$
norm, which suffices for our purposes and yields an easily samplable distribution. 

We collect below the relevant facts about $p$-stable random variables needed for our argument, which
can be found in many well-known references, including \cite{Indyk06, AlbiacK06}. 
\begin{fact}
For all $p \in (0,2)$, there exist (normalized) $p$-stable random variables $Z$ satisfying the
following properties:
\begin{enumerate}
\item For $Z_1, \ldots, Z_n$ iid copies of $Z$, and for all $a \in \R^n$, the random variable
\[
S ~\defeq~ \frac{a_1 \cdot Z_1 + \cdots + a_n \cdot Z_n}{\norm{\ell_p}{a}} \mcom
\]
has distribution identical to $Z$.
\item For all $q < p$, we have 
\[
C_{p,q} ~\defeq~ \norm{L_q}{Z} ~=~ \inparen{\Ex{\abs{Z}^q}}^{1/q} ~<~ \infty \mper
\]
\item There exists a constant $C_p$ such that for all $t > 0$,
\[
\Pr{\abs{Z} \geq t} ~<~ \frac{C_p}{t} \mper
\]
\item $Z$ can be sampled by choosing $\theta \in_R [-\pi/2, \pi/2]$, $r \in_R [0,1]$, and taking
\[
Z ~=~ \frac{\sin(p\theta)}{(\cos(\theta))^{1/p}} \cdot \inparen{\frac{\cos((1-p) \cdot
    \theta)}{\ln(1/r)}}^{(1-p)/p} \mper
\]
\end{enumerate}
\end{fact}
We now define an embedding of $\ell_p^n$ into $L_q$ with bounded $L_{\infty}$, using truncated
$p$-stable random variables. Let $Z = (Z_1, \ldots, Z_n)$ be a vector of 
iid $p$-stable random variables as above, and let $B$ be a parameter to be chosen later. We consider
the random variables
\[
\Delta(Z) ~\defeq~ \indicator{\exists i \in [n] ~\abs{Z_i} > B} 
\quad \text{and} \quad
Y ~\defeq~ (1 - \Delta(Z)) \cdot Z ~=~ \indicator{\forall i \in [n] ~\abs{Z_i} \leq B} \cdot Z\mper
\]
For all $a \in \R^n$, we define the (linear) embedding
\[
\phi(a) ~\defeq~ \frac{ \ip{a,Y}}{C_{p,q}} ~=~ \frac{\ip{a,Z}}{C_{p,q}} - \Delta(Z) \cdot
\frac{\ip{a,Z}}{C_{p,q}}\mper 
\] 
By the properties of $p$-stable distributions, we know that $\norm{L_q}{\ip{a,Z}/C_{p,q}} =
\norm{\ell_p}{a}$ for all $a \in \R^n$. By the following claim, we can choose $B$ so that the second
term only introduces a small error.
\begin{claim}
For all $\eps > 0$, there exists $B = O_{p,q,\eps}(n^{1/p})$ such that for the embedding $\phi$ defined
above
\[
\abs{\norm{L_q}{\phi(a)} - \norm{\ell_p}{a}} ~\leq~ \eps \cdot \norm{\ell_p}{a} \mper
\]
\end{claim}
\begin{proof}
By triangle inequality, it suffices to bound $\norm{L_q}{\Delta(Z) \cdot \ip{a,Z}}$ by $\eps \cdot
C_{p,q} \cdot \norm{\ell_p}{a}$. Let $\delta > 0$ be such that $(1+\delta) \cdot q < p$. Using the
fact that $\Delta(Z)$ is Boolean and \Holder's inequality, we observe that
\begin{align*}
\norm{L_q}{\Delta(Z) \cdot \ip{a,Z}}
&~=~
\inparen{\Ex{\abs{\ip{a,Z}}^q \cdot \Delta(Z)}}^{1/q} \\
&~\leq~
\inparen{\Ex{\abs{\ip{a,Z}}^{q(1+\delta)}}}^{1/(q(1+\delta))} \cdot
  \inparen{\Ex{\Delta(Z)}}^{\delta/(q(1+\delta))} \\
&~=~
C_{p,(1+\delta)q} \cdot \norm{\ell_p}{a} \cdot \inparen{\Pr{\exists i \in [n] ~\abs{Z_i} \geq
  B}}^{\delta/(q(1+\delta))} \\
&~\leq~
C_{p,(1+\delta)q} \cdot \norm{\ell_p}{a} \cdot \inparen{n \cdot \frac{C_p}{B^p}}^{\delta/(q(1+\delta))}
\end{align*}
Thus, choosing $B = O_{\eps, p, q}(n^{1/p})$ such that
\[
\frac{C_{p,(1+\delta)q}}{C_{p,q}} \cdot \inparen{n \cdot \frac{C_p}{B^p}}^{\delta/(q(1+\delta))}
~\leq~ \eps
\]
proves the claim.
\end{proof}
Using the value of $B$ as above, we now observe a bound on $\norm{L_{\infty}}{\phi(a)}$.
\begin{claim}
Let $B = O_{\eps,p,q}(n^{1/p})$ be chosen as above. Then, we have that
\[
M ~\defeq~ \sup\inbraces{\norm{L_{\infty}}{\ip{a,Y}} ~\mid~ \norm{L_{q}}{\ip{a,Y}} \leq 1} ~=~
O_{\eps,p,q}(n) \mper 
\]
\end{claim}
\begin{proof}
By the choice of $B$, we have that $\norm{L_{q}}{\ip{a,Y}} \geq (1-\eps) \norm{\ell_p}{a}$. Thus, we
can assume that $\norm{\ell_p}{a} \leq 2$. \Holder's inequality then gives for all such $a$,
\begin{align*}
\abs{\ip{a,Y}} 
&~\leq~ 
\norm{\ell_1}{a} \cdot \norm{\ell_{\infty}}{Y} \\
&~\leq~
n^{1-1/p} \cdot \norm{\ell_p}{a} \cdot B \\
&~\leq~ 2 \cdot n^{1-1/p} \cdot B ~=~ O_{\eps, p, q}(n) \mcom
\end{align*}
which proves the claim.
\end{proof}
Using the above bound on $M$ in \cref{clm:schechtman} gives a bound of $m = O_{\eps,p,q}(n^{q+1}) =
O_{\eps,p,q}(n^3)$. Moreover, the distribution over embeddings is efficiently samplable, since it
obtained by truncating $p$-stable random variables. This completes the proof  of \cref{thm:schechtman}.

%% file: PtoQ/4.algorithm.tex
\iffalse
\section[Convex Programming Relaxation and Analysis]{Convex Programming Relaxation and Analysis}
\subsection[Vector Program]{Vector Program}
We first write the vector program that we optimize:
\begin{align*}
    &\sup~\sum_{i,j} A_{i,j}\cdot \mysmalldot{u^i}{v^j}\quad\text{s.t.} \\
    & \sum_{i\in [m]} \norm{2}{u^i}^{q^*} \leq 1,\quad\sum_{j\in [n]} \norm{2}{v^j}^{p} \leq 1 \\
    & u^i, v^j\in \R^{m+n} \quad \forall i\in [m], j\in [n]
\end{align*}

\subsection[Convex Program]{Convex Program $\CP{A}$}
Let $A$ be an $m\times n$ matrix. 
We use the following convex relaxation for $\norm{p}{q}{A}$ ($1\leq q\leq 2\leq p\leq \infty$) 
due to Nesterov~\cite{NWY00}: 
\begin{align*}
    \textbf{maximize} \quad
    &\frac{1}{2}\cdot 
    \mydot{
    \left[
    \begin{array}{cc}
    0 & A \\
    A^T & 0
    \end{array}
    \right]   
    }
    {
    \left[
    \begin{array}{cc}
    \bbY & \bbW \\
    \bbW^T & \bbX
    \end{array}
    \right]   
    } 
    \quad 
    \text{s.t.} \\[0.5em]
    &
    \norm{p/2}{\mathrm{diag}(\bbX)}\leq 1 ,\quad
    \norm{q^*/2}{\mathrm{diag}(\bbY)}\leq 1 \\[0.5em]
    &
    \left[
    \begin{array}{cc}
    \bbY & \bbW \\
    \bbW^T & \bbX
    \end{array}
    \right]   
    \succeq 0 ,\quad 
    \bbY\in \Sym^{m\times m},~\bbX\in \Sym^{n\times n},~\bbW\in \R^{m\times n}
\end{align*}

\noindent
We may assume that any solution has the form: 
\[
    \left[
    \begin{array}{cc}
    UU^T & UV^T \\
    VU^T & VV^T
    \end{array}
    \right]   \quad 
    \text{where ~}
    U\in \R^{m\times (m+n)},~V \in \R^{n\times (m+n)}
\]
\fi

\section{Analyzing the Approximation Ratio via Rounding}
We will show that $\CP{A}$ is a good approximation to $\norm{p}{q}{A}$ by using an appropriate 
generalization of Krivine's rounding procedure. Before stating the generalized procedure, we shall 
give a more detailed summary of Krivine's procedure. 

\subsection{Krivine's Rounding Procedure}
Krivine's procedure centers around the classical random hyperplane rounding. In this context, we define the 
random hyperplane rounding procedure on an input pair of matrices $U\in \R^{m\times \ell},~ 
V\in\R^{n\times \ell}$ as outputting the vectors $\sgn [U\bfg]$ and $\sgn [V\bfg]$ where $\bfg\in \R^{\ell}$ 
is a vector with i.i.d. standard Gaussian coordinates ($f[v]$ denotes entry-wise application of a scalar 
function $f$ to a vector $v$. We use the same convention for matrices.). 
The so-called Grothendieck identity states that for vectors $u,v\in \R^\ell$, 
\[
    \Ex{\sgn\!{\mysmalldot{\bfg}{u}}\cdot \sgn\!{\mysmalldot{\bfg}{v}}} 
    = 
    \frac{\sin^{-1}\!{\mysmalldot{\widehat{u}}{\widehat{v}}}}{\pi/2}
\]
where $\widehat{u}$ denotes $u/\norm{2}{u}$. This implies the following equality which we will call 
the hyperplane rounding identity: 
\begin{equation}
    \Ex{\sgn [U\bfg](\sgn [V\bfg])^T}
    = 
    \frac{\sin^{-1}[\hatU \hatV^T]}{\pi/2} \mper
\end{equation}
where for a matrix $U$, we use $\hatU$ to denote the matrix obtained by replacing the rows of $U$ 
by the corresponding unit (in $\ell_2$ norm) vectors. 
Krivine's main observation is that for any matrices $U,V$, there exist matrices 
$\phi(\hatU),\psi(\hatV)$ with unit vectors as rows, such that 
\[
    \phi(\hatU)\,\psi(\hatV)^T = \sin [(\pi/2) \cdot c\cdot \hatU\hatV^T]
\]
where $c = \sinh^{-1}(1)\cdot 2/\pi$. Taking $\hatU,\hatV$ to be the optimal solution to $\CP{A}$, 
it follows that 
\[
    \norm{\infty}{1}{A}
    \geq 
    \mydot{A}{\Ex{\sgn [\phi(\hatU)\,\bfg]~(\sgn [\psi(\hatV)\,\bfg])^T}} 
    = 
    \mysmalldot{A}{c\cdot \hatU\hatV^T}
    = 
    c\cdot \CP{A} \mper
\]
The proof of Krivine's observation follows from simulating the Taylor series of a scalar function using inner 
products. We will now describe this more concretely. 
\begin{observation}[Krivine]
    \label[observation]{simulating:taylor}
    Let $f:[-1,1]\to\R$ be a scalar function 
    satisfying $f(\rho) = \sum_{k\geq 1} f_k\, \rho^k$ for an absolutely convergent series $(f_k)$. 
    Let $\absolute{f}(\rho):= \sum_{k\geq 1} |f_k|\,\rho^{k}$ and further for vectors 
    $u,v\in\R^{\ell}$ of $\ell_2$-length at most $1$, let 
    \begin{align*}
        &\seriesLeft{f}{u} := (\sgn(f_1)\sqrt{f_1}\cdot u)\oplus 
        (\sgn(f_2)\sqrt{f_2}\cdot u^{\otimes 2})\oplus 
        (\sgn(f_3)\sqrt{f_3}\cdot u^{\otimes 3})\oplus \cdots \\
        &\seriesRight{f}{v} := (\sqrt{f_1}\cdot v)\oplus (\sqrt{f_2}\cdot v^{\otimes 2})\oplus 
        (\sqrt{f_3}\cdot v^{\otimes 3})\oplus \cdots 
    \end{align*}
    Then for any $U\in \R^{m\times \ell},~V\in \R^{n\times \ell}$, ~$\seriesLeft{f}{\sqrt{c_f}\cdot\hatU}$ 
    and $\seriesRight{f}{\sqrt{c_f}\cdot \hatV}$ have $\ell_2$-unit vectors as rows, and 
    \[
        \seriesLeft{f}{\sqrt{c_f}\cdot\hatU}~\seriesRight{f}{\sqrt{c_f}\cdot \hatV}^T 
        = 
        f\,[c_f\cdot \hatU\hatV^T]
    \]
    where $\seriesLeft{f}{W}$ for a matrix $W$, is applied to row-wise and $c_f := (\absolute{f}^{-1})(1)$. 
\end{observation}

\begin{proof}
    Using the facts $\mysmalldot{y^1\otimes y^2}{y^3\otimes y^4} = 
    \mysmalldot{y^1}{y^3}\cdot \mysmalldot{y^2}{y^4}$ and \\
    $\mysmalldot{y^1\oplus y^2}{y^3\oplus y^4} = 
    \mysmalldot{y^1}{y^3}+ \mysmalldot{y^2}{y^4}$,~we have  
    \begin{itemize}
        \item $\mysmalldot{\seriesLeft{f}{u}}{\seriesRight{f}{v}} = f(\mysmalldot{u}{v})$
        
        \item $\norm{2}{\seriesLeft{f}{u}} = \sqrt{\absolute{f}(\norm{2}{u}^{2})}$
        
        \item $\norm{2}{\seriesRight{f}{v}} = \sqrt{\absolute{f}(\norm{2}{v}^{2})}$ 
    \end{itemize}
    The claim follows. 
\end{proof}

Before stating our full rounding procedure, we first discuss a natural generalization of random hyperplane 
rounding, and much like in Krivine's case this will guide the final procedure. 

\subsection{Generalizing Random Hyperplane Rounding -- \Holder Dual Rounding}
Fix any convex bodies $B_1 \subset \R^m$ and $B_2\subset \R^k$. Suppose that we would like a strategy 
that for given vectors $y\in\R^m,~x\in\R^n$, outputs $\bary\in B_1,~\barx\in B_2$ so that 
$y^TA\,x = \mysmalldot{A}{y\,x^T}$ is close to $\mysmalldot{A}{\bary\, \barx^T}$ for all $A$. A natural 
strategy is to take 
\[
    (\bary,\barx) := \argmax_{(\tildey,\tildex)\in B_1\times B_2} 
    \mydot{\tildey\,\tildex^T}{y\,x^T}
    = 
    \inparen{
    \argmax_{\tildey\in B_1} 
    \mydot{\tildey}{y}
    ~,~
    \argmax_{\tildex\in B_2} 
    \mydot{\tildex}{x}
    }
\]
In the special case where $B$ is the unit $\ell_p$ ball, there is a closed form for an optimal 
solution to $\max_{\tildex\in B} \mysmalldot{\tildex}{x}$, given by $\holderdual{p^*}{x}/
\norm{p^*}{x}^{p^*-1}$, where $\holderdual{p^*}{x} := \sgn[x]\circ |[x]|^{p^*-1}$. Note that for 
$p= \infty$, this strategy recovers the random hyperplane rounding procedure. We shall call this 
procedure, \emph{Gaussian \Holder Dual Rounding} or \HD for short. 

Just like earlier, we will first understand the effect of \HD on a solution pair $U,V$.
%%Here and in the rest of this paper, we assume $a:=p^*-1$ and $b=q-1$. 
For $\rho\in [-1, 1]$, let $\bfg_1 \!\sim_{\rho} \bfg_2$ denote $\rho$-correlated standard 
Gaussians, \ie $\bfg_1 = \rho\,\bfg_2 + \sqrt{1-\rho^2}\,\bfg_3$ ~where 
$(\bfg_2,\bfg_3) \sim \gaussian{0}{\id_2}$, and let 
\[
    \fplain{a}{b}{\rho} 
    :=  
    \Ex{\bfg_1 \sim_\rho \bfg_2}{\sgn(\bfg_1)|\bfg_1|^{b}\sgn(\bfg_2)|\bfg_1|^{a}}
\]
We will work towards a better understanding of $\fplain{a}{b}{\cdot}$ in later sections. 
For now note that we have for vectors $u,v\in\R^{\ell}$, 
\[
    \Ex{\sgn\!{\mysmalldot{\bfg}{u}}\,|\mysmalldot{\bfg}{u}|^{b} \cdot 
    \sgn\!{\mysmalldot{\bfg}{v}} \,|\mysmalldot{\bfg}{v}|^{a}} 
    = 
    \norm{2}{u}^{b}\cdot\norm{2}{v}^{a}\cdot \fplain{a}{b}{\mysmalldot{\widehat{u}}{\widehat{v}\,}} \mper
\]
Thus given matrices $U,V$, we obtain the following generalization of the hyperplane rounding identity 
for \HD: 
\begin{equation}
\label[equation]{cvgp:identity}
    \Ex{\holderdual{q}{[U\bfg]}\,\holderdual{p^*}{[V\bfg]}^T}
    = 
    \Diag{(\norm{2}{u^i}^{b})_{i\in [m]}} \cdot \fplain{a}{b}{[\hatU \hatV^T]}\cdot 
    \Diag{(\norm{2}{v^j}^{a})_{j\in [n]}} 
    \mper
\end{equation}

\subsection{Generalized Krivine Transformation and the Full Rounding Procedure}
\label[subsection]{rounding}

We are finally ready to state the generalized version of Krivine's algorithm. At a high level the algorithm 
simply applies \HD to a transformed version of the optimal convex program solution pair $U,V$. 
Analogous to Krivine's algorithm, the transformation is a type of ``inverse'' of \cref{cvgp:identity}. 
\begin{enumerate}[label=(Inversion \arabic*),align=left]
    \item Let $(U,V)$ be the optimal solution to $\CP{A}$, and let $(u^i)_{i\in [m]}$ and $(v^j)_{j\in[n]}$ 
    respectively denote the rows of $U$ and $V$. 

    \item Let $c_{a,b} := \inparen{\absolute{\alfin{a}{b}}}^{-1}\!\!(1)$ and let 
        \begin{align*}
            \phi(U) &:= \Diag{(\norm{2}{u^i}^{1/b})_{i\in [m]}}\,\,\seriesLeft{\alfin{a}{b}}{\sqrt{c_{a,b}}\cdot\hatU} 
            \mcom \\
            \psi(V) &:= \Diag{(\norm{2}{v^j}^{1/a})_{j\in [n]}}\,\,\seriesRight{\alfin{a}{b}}{\sqrt{c_{a,b}}\cdot\hatV}
            \mper
        \end{align*}
\end{enumerate}
\begin{enumerate}[label=(\Holder-Dual \arabic*),align=left]
    \item Let $\bfg \sim \gaussian{0}{\id}$ be an infinite dimensional i.i.d. Gaussian vector. 
    
    \item Return $y:=\holderdual{q}{\phi(U)\,\bfg}/\norm{q}{\phi(U)\,\bfg}^{b}$ and 
    $x:=\holderdual{p^*}{\psi(V)\,\bfg}/\norm{p^*}{\psi(V)\,\bfg}^{a}$. 
\end{enumerate}

\bigskip

\begin{remark}
\label[remark]{round:feasibility}
    Note that $\norm{r^*}{\holderdual{r}{\barx}} = \norm{r}{\barx}^{r-1}$ and so the returned solution pair 
    always lie on the unit $\ell_{q^*}$ and $\ell_{p}$ spheres respectively. 
\end{remark}

\begin{remark}
    Like in~\cite{AN04} the procedure above can be made algorithmic by observing that there always exist 
    $\phi'(U)\in \R^{m\times (m+n)}$ and $\psi'(V)\in \R^{m \times (m+n)}$, whose rows have the exact 
    same lengths and pairwise inner products as those of $\phi(U)$ and $\psi(V)$ above. 
    Moreover they can be computed without explicitly computing $\phi(U)$ and $\psi(V)$ by 
    obtaining the Gram decomposition of 
    \[
        M~:=~
        \left[
        \begin{array}{cc}
        \absolute{\alfin{a}{b}}[c_{a,b}\cdot \hatV\hatV^T] & \fin{a}{b}{[c_{a,b}\cdot \hatU\hatV^T]} \\
        \fin{a}{b}{[c_{a,b}\cdot \hatV\hatU^T]} & \absolute{\alfin{a}{b}}[c_{a,b}\cdot \hatV\hatV^T]
        \end{array}
        \right] \mcom
    \]
    and normalizing the rows of the decomposition according to the definition of $\phi(\cdot)$ and 
    $\psi(\cdot)$ above. The entries of $M$ can be computed in polynomial time with exponentially 
    (in $m$ and $n$) good accuracy by implementing the Taylor series of $\alfin{a}{b}$ 
    upto $\mathrm{poly}(m,n)$ terms (Taylor series inversion can be done upto $k$ terms in time 
    $\mathrm{poly}(k)$).
\end{remark}

\begin{remark}
\label[remark]{algo:rmk}
    Note that the $2$-norm of the $i$-th row (resp. $j$-th row) of $\phi(U)$ (resp. $\psi(V)$) is 
    $\norm{2}{u^i}^{1/b}$ (resp. $\norm{2}{v^j}^{1/a}$).
\end{remark}

We commence the analysis by defining some convenient normalized functions, and we will 
also show that $c_{a,b}$ above is well-defined. 
\subsection{Auxiliary Functions}
Let ~$\nf{p}{q}{\rho} := \fplain{p}{q}{\rho}/(\gamma_{p^*}^{p^*}\,\gamma_{q}^{q})$, ~ 
$\alfinabs{a}{b} := \absolute{\alfin{a}{b}} $,~ and ~$\alnh{a}{b} := \absolute{\alnfin{a}{b}}$. 
Also note that $\nhin{a}{b}{\rho} = \hin{a}{b}{\rho}/(\gamma_{p^*}^{p^*}\,\gamma_{q}^{q})$. 

\paragraph{Well Definedness. }
By \cref{inv:coeff:bound}, $\nfin{a}{b}{\rho}$ and $\nh{a}{b}{\rho}$ are well-defined for $\rho\in [-1,1]$.  
By (M1) in \cref{monotonicity:properties}, ~
$\kfinc{1} = 1$ and hence $\nh{a}{b}{1} \geq 1$ and $\nh{a}{b}{-1}\leq -1$. Combining this with the 
fact that $\nh{a}{b}{\rho}$ is continuous and strictly increasing on $[-1,1]$, implies that 
$\nhin{a}{b}{x}$ is well-defined on $[-1,1]$. 

\medskip
We can now proceed with the analysis. 

\subsection[Bound on Approximation Factor]{$1/(\nhin{p}{q}{1}\cdot \gamma_{p^*}\,\gamma_q)$
Bound on Approximation Factor}

For any vector random variable $\bfX$ in a universe $\Omega$, and scalar valued functions 
$f_1:\Omega \to\R$ and $f_2:\Omega\to (0,\infty)$. Let $\lambda = \Ex{f_1(\bfX)}/\Ex{f_2(\bfX)}$. 
Now we have 
\begin{align*}
    &\max_{x\in\Omega} f_1(x)-\lambda\cdot f_2(x) \geq \Ex{f_1(\bfX)-\lambda\cdot f_2(\bfX)} = 0 \\
    \Rightarrow\quad  
    &\max_{x\in\Omega} f_1(x)/f_2(x) \geq \lambda = \Ex{f_1(\bfX)}/\Ex{f_2(\bfX)}\mper
\end{align*}
Thus we have 
\[
    \norm{p}{q}{A} 
    ~\geq ~
    \frac{\Ex{\mysmalldot{A}{\holderdual{q}{\phi(U)\,\bfg}~\holderdual{p^*}{\psi(V)\,\bfg}^T}}} 
    {\Ex{\norm{q^*}{\holderdual{q}{\phi(U)\,\bfg}}\cdot \norm{p}{\holderdual{p^*}{\psi(V)\,\bfg}}}} 
    ~= ~
    \frac{\mysmalldot{A}{\Ex{\holderdual{q}{\phi(U)\,\bfg}~\holderdual{p^*}{\psi(V)\,\bfg}^T}}} 
    {\Ex{\norm{q^*}{\holderdual{q}{\phi(U)\,\bfg}}\cdot \norm{p}{\holderdual{p^*}{\psi(V)\,\bfg}}}} 
    \mcom
\]
which allows us to consider the numerator and denominator separately. 
We begin by proving the equality that the above algorithm was designed to satisfy: 
\begin{lemma}
    \label[lemma]{numerator}
    $
        \Ex{\holderdual{q}{\phi(U)\,\bfg}~ \holderdual{p^*}{\psi(V)\,\bfg}^T} 
        ~=~ 
        c_{a,b}\cdot (\tildeU \tildeV^T)
    $
\end{lemma}

\begin{proof}
    \begin{align*}
        &\quad~\Ex{\holderdual{q}{\phi(U)\,\bfg}~\holderdual{p^*}{\psi(V)\,\bfg}^T}  \\
        &= 
        \Diag{{(\norm{2}{u^i})}_{i\in [m]}} \cdot 
        \fplain{a}{b}{[\seriesLeft{\alfin{a}{b}}{\sqrt{c_{a,b}}\cdot\hatU}\cdot
        \seriesRight{\alfin{a}{b}}{\sqrt{c_{a,b}}\cdot\hatV}^T]} \cdot 
        \Diag{{(\norm{2}{v^j})}_{j\in [n]}} \\
        &~(\text{by \cref{cvgp:identity} and \cref{algo:rmk}}) \\
        &= 
        \Diag{(\norm{2}{u^i})_{i\in [m]}} \cdot 
        \fplain{a}{b}{[\fin{a}{b}{[c_{a,b}\cdot \hatU\hatV^T]}]} \cdot 
        \Diag{(\norm{2}{v^j})_{j\in [n]}} \\
        &~(\text{by \cref{simulating:taylor}}) \\
        &= 
        \Diag{(\norm{2}{u^i})_{i\in [m]}} \cdot 
        c_{a,b}\cdot \hatU\hatV^T \cdot \Diag{{(\norm{2}{v^j})}_{j\in [n]}} \\
        &= 
        c_{a,b}\cdot UV^T
    \end{align*}
\end{proof}

It remains to upper bound the denominator which we do using a straightforward convexity argument. 
\begin{lemma}
    \label[lemma]{denominator}
    $
        \Ex{\norm{q}{\phi(U)\,\bfg}^{b}\cdot \norm{p^*}{\psi(V)\,\bfg}^{a}}
        \leq  
        \gamma_{p^*}^{a}\, \gamma_{q}^{b} \mper
    $
\end{lemma}

\begin{proof}
    \begin{align*}
        &\quad~\Ex{\norm{q}{\phi(U)\,\bfg}^{b}\cdot \norm{p^*}{\psi(V)\,\bfg}^{a}} \\
        &\leq~ 
        \Ex{\norm{q}{\phi(U)\,\bfg}^{q^* b}}\!^{1/q^*} \cdot \Ex{\norm{p^*}{\psi(V)\,\bfg}^{p a}}\!^{1/p}
        &&\inparen{\frac{1}{p}+\frac{1}{q^*}\leq 1} \\
        &=~ 
        \Ex{\norm{q}{\phi(U)\,\bfg}^{q}}\!^{1/q^*} \cdot \Ex{\norm{p^*}{\psi(V)\,\bfg}^{p^*}}\!^{1/p} \\
        &=~ 
        \insquare{\sum_{i\in [m]}\Ex{|\gaussian{0}{\norm{2}{u^i}^{1/b}}|^{q}}}^{1/q^*} 
        \cdot 
        \insquare{\sum_{j\in [n]}\Ex{|\gaussian{0}{\norm{2}{v^j}^{1/a}}|^{p^*}}}^{1/p}  
        &&(\text{By \cref{algo:rmk}}) \\
        &=~ 
        \insquare{\sum_{i\in [m]} \norm{2}{u^i}^{q/b}}^{1/q^*} 
        \cdot 
        \insquare{\sum_{j\in [n]} \norm{2}{v^j}^{p^*/a}}^{1/p}  
        \cdot \gamma_{q}^{q/q^*}\,\gamma_{p^*}^{p^*/p} \\
        &=~ 
        \insquare{\sum_{i\in [m]} \norm{2}{u^i}^{q^*}}^{1/q^*} 
        \cdot 
        \insquare{\sum_{j\in [n]} \norm{2}{v^j}^{p}}^{1/p}  
        \cdot \gamma_{q}^{b}\,\gamma_{p^*}^{a} \\
        &=~
        \gamma_{q}^{b}\,\gamma_{p^*}^{a} 
        &&\hspace{-40 pt}(\text{feasibility of } U,V) \qedhere
    \end{align*}
\end{proof}

\medskip

We are now ready to prove our approximation guarantee. 
\begin{lemma}
    \label[lemma]{defect:bound:implies:apx}
    Consider any $1\leq q\leq 2\leq p \leq \infty$. Then, 
    \[
        \frac{\CP{A}}{\norm{p}{q}{A}} 
        ~\leq ~ 
        1/(\gamma_{p^*}\,\gamma_q \cdot \nhin{a}{b}{1}) 
    \]
\end{lemma}

\begin{proof}
    \begin{align*}
        \norm{p}{q}{A}~
        &\geq ~
        \frac{\mysmalldot{A}{\Ex{\holderdual{q}{\phi(U)\,\bfg}~\holderdual{p^*}{\psi(V)\,\bfg}^T}}}
        {\Ex{\norm{q^*}{\holderdual{q}{\phi(U)\,\bfg}}\cdot \norm{p}{\holderdual{p^*}{\psi(V)\,\bfg}}}} \\
        &= ~
        \frac{\mysmalldot{A}{\Ex{\holderdual{q}{\phi(U)\,\bfg}~\holderdual{p^*}{\psi(V)\,\bfg}^T}}}
        {\Ex{\norm{q}{\phi(U)\,\bfg}^{b}\cdot \norm{p^*}{\psi(V)\,\bfg}^{a}}} 
        &&(\text{by \cref{round:feasibility}}) \\
        &= ~ 
        \frac{c_{a,b}\cdot \mysmalldot{A}{UV^T}}
        {\Ex{\norm{q}{\phi(U)\,\bfg}^{b}\cdot \norm{p^*}{\psi(V)\,\bfg}^{a}}} 
        &&(\text{by \cref{numerator}}) \\
        &= ~ 
        \frac{c_{a,b}\cdot \CP{A}}
        {\Ex{\norm{q}{\phi(U)\,\bfg}^{b}\cdot \norm{p^*}{\psi(V)\,\bfg}^{a}}} 
        &&(\text{by optimality of }U,V) \\
        &\geq ~ 
        \frac{c_{a,b}\cdot \CP{A}}{\gamma_{p^*}^{a}\,\gamma_q^{b}} 
        &&(\text{by \cref{denominator}}) \\
        &= ~ 
        \frac{\hin{a}{b}{1}\cdot \CP{A}}{\gamma_{p^*}^{a}\,\gamma_q^{b}} \\
        &= ~ 
        \nhin{a}{b}{1}\cdot \gamma_{p^*}\,\gamma_{q}\cdot \CP{A}
    \end{align*}
\end{proof}

We next begin the primary technical undertaking for the analysis of our algorithm, namely proving upper bounds on 
$\nhin{p}{q}{1}$.

%% file: PtoQ/5.hermite_analysis.tex
\section[Hypergeometric Representation]{Hypergeometric Representation of  $\nf{a}{b}{x}$}

In this section, we show that $\nf{a}{b}{\rho}$ can be represented using the Gaussian hypergeometric 
function $\hypergeometric$. The result of this section can be thought of as a generalization of the 
so-called Grothendieck identity for hyperplane rounding which simply states that 
\[
    \nf{0}{0}{\rho}~=~\frac{\pi}{2}\cdot \Ex{\bfg_1 \sim_\rho \,\bfg_2}{\sgn(\bfg_1)\sgn(\bfg_2)}~=~\sin^{-1}
    (\rho)
\]
We believe the result of this section and its proof technique to be of independent 
interest in analyzing generalizations of hyperplane rounding to convex bodies other than the hypercube. 

Recall that $\fplain{a}{b}{\rho}$ is defined as follows: 
\[
    \Ex{\bfg_1 \sim_\rho \,\bfg_2}{\sgn(\bfg_1)|\bfg_1|^{a}\sgn(\bfg_2)|\bfg_1|^{b}}
\]
where $a=p^*-1$ and $b=q-1$. 
Our starting point is the simple observation that the above expectation can be viewed as the noise 
correlation (under the Gaussian measure) of the functions  $\fa{\tau}:= \sgn{\tau}\cdot |\tau|^{a}$ and 
$\fb{\tau}:= \sgn{\tau}\cdot |\tau|^{b}$. Elementary Hermite analysis then implies that it suffices 
to understand the Hermite coefficients of $\alfa$ and $\alfb$ individually, in order to understand the 
Taylor coefficients of $\alnf{a}{b}$. To understand the Hermite coefficients of $\alfa$ and $\alfb$ individually, 
we use a generating function approach. More specifically, we derive an integral representation for 
the generating function of the (appropriately normalized) Hermite coefficients which fortunately turns out 
to be closely related to a well studied special function called the parabolic cylinder function. 

We can write: $\fplain{a}{b}{\rho}=\Ex{\bfg_1\sim\rho\,\bfg_2}{\alfa(\bfg_1)\cdot \alfb(\bfg_2)}$,
where $\fwild{c}{\tau}:=\sgn(\tau)\cdot \abs{\tau}^c$ for $c\in \{a,b\}$. 
Now we note that $\nf{a}{b}{\rho}$ is the noise correlation 
of $\alfa$ and $\alfb$. Thus, we can relate the Taylor coefficients of $\nf{a}{b}{\rho}$, to the Hermite 
coefficients of $\alfa$ and $\alfb$. 
\begin{claim}[Coefficients of $\fplain{a}{b}{\rho}$]
    \label[claim]{noise:correlation} For $\rho\in[-1,1]$, we have:
    \[
        \fplain{a}{b}{\rho}=\sum_{k\ge 0}\rho^{2k+1}\cdot \hfac{2k+1}\cdot \hfbc{2k+1}\mcom
    \] where $\hfac{i}$ and $\hfbc{j}$ are the $i$-th and $j$-th Hermite coefficients of
    $\alfa$ and $\alfb$, respectively. Moreover, $\hfac{2k}=\hfbc{2k}=0$ for $k\ge 0$.
    %%Let $\kfc{k}$ denote $[\rho^k]\,\fplain{a}{b}{\rho}$. Then for every integer $k\geq 0$, 
    %%$\kfc{2k} = 0$ and  $\kfc{2k+1} = \hfac{2k+1} \cdot \hfbc{2k+1}$, where for $c\in \{a,b\}$, 
    %%$\hfwildc{c}{k}$ is the $k$-th Hermite coefficient of $\fwild{c}{\tau} := \sgn(\tau)|\tau|^{c}$.
\end{claim}
\begin{proof}
    We observe that both $\alfa$ and $\alfb$ are odd functions and hence \cref{even:odd:hermite} implies that
    $\hfac{2k}=\hfbc{2k}=0$ for all $k\ge 0$ -- as $\alfa(\tau)\cdot H_{2k}(\tau)$ is an odd function of
    $\tau$.
    \begin{align*}
        \fplain{a}{b}{\rho} &=  \Ex{\bfg_1\sim\rho\,\bfg_2}{\alfa(\bfg_1)\cdot \alfb(\bfg_2))}\\
        &= \Ex{\bfg_1}{\alfa(\bfg_1)\cdot \noise{\alfb}(\bfg_1)} & (\text{\cref{noise-operator}})\\
        &= \mysmalldot{\alfa}{\noise{\alfb}}\\
    &= \sum_{k\ge 0} \hfac{k}\cdot \widehat{(\noise{\alfb})}_{k} & (\text{\cref{plancherel-identity}})\\
        &= \sum_{k\ge 0} \hfac{2k+1}\cdot \widehat{(\noise{\alfb})}_{2k+1} \\
        &= \sum_{k\ge 0} \rho^{2k+1}\cdot \hfac{2k+1}\cdot \hfbc{2k+1} &
        (\text{\cref{hermite-noise}})\mper
    \end{align*}
%%    Combining Plancherel Identity with \cref{hermite-noise} yields 
%%    $\kfc{k} = \hfac{k} \cdot \hfbc{k}$. Now 
%%    \begin{align*}
%%        \hfwildc{c}{k}
%%        &= 
%%        \frac{1}{\sqrt{2\pi}}\int_{-\infty}^{\infty} |\tau|^{c} \sgn(\tau)\cdot H_k(\tau)\cdot e^{-\tau^2/2} \,\,dx \\
%%        &= 
%%        \left\{
%%        \begin{array}{lc}
%%        \quad 0 &  k \text{ is even} \\
%%        ~&~\\
%%        \sqrt{\frac{2}{\pi}}\int_{0}^{\infty} \tau^{a_1} \cdot H_k(\tau)\cdot e^{-\tau^2/2} \,\,dx 
%%        & k \text{ is odd}
%%        \end{array}
%%        \right. 
%%        &&(\text{by \cref{even:odd:hermite}})
%%    \end{align*}
%%    Thus $\kfc{2k} = 0$.
\end{proof}

\subsection[Hermite Coefficients via Parabolic Cylinder Functions]
{Hermite Coefficients of $\alfa$ and $\alfb$ via Parabolic Cylinder Functions}
In this subsection, we use the generating function of Hermite polynomials to obtain an integral
representation for the generating function of the ($\sqrt{k!}$ normalized) odd Hermite coefficients
of $\alfa$ (and similarly of $\alfb$) is closely related to 
a special function called the parabolic cylinder function. We then use known facts about the 
relation between parabolic cylinder functions and confluent hypergeometric functions, to show 
that the Hermite coefficients of $\alfwild{c}$ can be obtained from the Taylor coefficients of a 
confluent hypergeometric function.

\subsubsection{Generating Function of Hermite Coefficients and its Confluent Hypergeometric Representation}
Using the generating function of (appropriately normalized) Hermite polynomials, we derive an integral 
representation for the generating function of the (appropriately normalized) Hermite coefficients of 
$\alfa$ (and similarly $\alfb$):
\begin{lemma}\label[lemma]{gen:function:integral:rep}
    For $c\in\{a,b\}$, let $\hfwildc{c}{k}$ denote the $k$-th Hermite coefficient of 
    $\fwild{c}{\tau} := \sgn{(\tau)}\cdot |\tau|^{c}$. 
    Then we have the following identity:
    \[
        \sum_{k\geq 0} \frac{\lambda^{2k+1}}{\sqrt{(2k+1)!}}\cdot \hfwildc{c}{2k+1}=
        \frac{1}{\sqrt{2\pi}}\int_{0}^{\infty} \tau^{c}\cdot 
        \inparen{\ee^{-(\tau-\lambda)^2/2} - \ee^{-(\tau+\lambda)^2/2}} \,d\tau \mper
    \] 
\end{lemma}
\begin{proof}
    We observe that for, $\alfwild{c}$ is an odd function and hence \cref{even:odd:hermite} implies that 
    $\fwild{c}{\tau}\cdot H_{2k}(\tau)$ is an odd function and $\fwild{c}{\tau}\cdot H_{2k+1}(\tau)$ is an 
    even function. This implies for any $k\geq 0$, that $\hfwildc{c}{2k} = 0$ and 
    \[
        \hfwildc{c}{2k+1}
        = 
        \frac{1}{\sqrt{2\pi}} \int_{-\infty}^{\infty} 
        \sgn{(\tau)}\cdot \tau^{c} \cdot H_{2k+1}(\tau) \cdot \ee^{-\tau^2/2}  \,d\tau
        = 
        \sqrt{\frac{2}{\pi}} \int_{0}^{\infty} 
        \tau^{c} \cdot H_{2k+1}(\tau) \cdot \ee^{-\tau^2/2}  \,d\tau \mper
    \]
    Thus we have 
    \begin{align*}
        &\quad~\sum_{k\geq 0} \frac{\lambda^{2k+1}}{\sqrt{(2k+1)!}}\cdot \hfwildc{c}{2k+1} \\
        &= 
        \sqrt{\frac{2}{\pi}} \cdot \sum_{k\geq 0}~\int_{0}^{\infty} 
        \tau^{c} \cdot \ee^{-\tau^2/2} \cdot H_{2k+1}(\tau) \cdot \frac{\lambda^{2k+1}}{\sqrt{(2k+1)!}}\,\,d\tau  \\
        &= 
        \sqrt{\frac{2}{\pi}} \cdot \int_{0}^{\infty} \tau^{c} \cdot \ee^{-\tau^2/2} \sum_{k\geq 0} 
        H_{2k+1}(\tau) \cdot \frac{\lambda^{2k+1}}{\sqrt{(2k+1)!}}\,\,d\tau 
        &&(\text{see below}) \\
        &= 
        \frac{1}{\sqrt{2\pi}} \cdot \int_{0}^{\infty} \tau^{c} \cdot \ee^{-\tau^2/2} \cdot 
        \inparen{\ee^{\tau\lambda  -\lambda^2/2} - \ee^{-\tau\lambda  -\lambda^2/2}} \,\,d\tau
        &&(\text{~by \cref{hermite:generating:function}}) \\
        &= 
        \frac{1}{\sqrt{2\pi}} \cdot \int_{0}^{\infty} \tau^{c} \cdot 
        \inparen{\ee^{-(\tau-\lambda)^2/2} - \ee^{-(\tau+\lambda)^2/2}} \,d\tau 
    \end{align*}
    where the exchange of summation and integral in the second equality follows by Fubini's theorem. 
    We include this routine verification for the sake of completeness.
    As a consequence of Fubini's theorem, if $(f_k:\R\to\R)_k$ is a sequence 
    of functions such that $\sum_{k\geq 0}\int_{0}^{\infty} |f_k| <\infty$, then 
    $
        \sum_{k\geq 0}\int_{0}^{\infty} f_k = \int_{0}^{\infty} \sum_{k\geq 0} f_k \mper
    $
    Now for any fixed $k$, we have 
    \[
        \int_{0}^{\infty} \tau^{c}\cdot |H_{k}(x)|\,d\gamma(\tau) 
        ~\leq ~
        \inparen{\int_{0}^{\infty} \tau^{2c}\,d\gamma(\tau)}^{1/2} \cdot 
        \inparen{\int_{0}^{\infty} |H_{k}(x)|^2\,d\gamma(\tau)}^{1/2} 
        ~\leq ~ 
        \gamma_{2c}^{c}
        < 
        \infty \mper
    \]
    Setting $f_k(\tau) := \tau^{c} \cdot \ee^{-\tau^2/2} \cdot H_{2k+1}(\tau) \cdot 
    \lambda^{2k+1}/\sqrt{(2k+1)!}\,\mcom$ we get that $\sum_{k\geq 0}\int_{0}^{\infty} |f_k| <\infty$. This 
    completes the proof. 
\end{proof}

Finally using known results about parabolic cylinder functions, we are able to relate the aforementioned 
integral representation to a confluent hypergeometric function (whose Taylor coefficients are known). 
\begin{lemma}\label[lemma]{integral:representation:confHG}
    For $\lambda \in [-1,1]$ and real valued $c>-1$, we have 
    \[
        \frac{1}{\sqrt{2\pi}}
        \int_{0}^{\infty} \tau^{c}\inparen{\ee^{-(\tau-\lambda)^2/2} - \ee^{-(\tau+\lambda)^2/2}}\,d\tau
        ~=~
        \gamma_{c+1}^{c+1}\cdot 
        \lambda\cdot \confHG\inparen{\frac{1-c}{2}\,;\,\frac{3}{2}\,;\,-\frac{\lambda^2}{2}} 
    \]
\end{lemma}

\begin{proof}
 %%   We will proceed by equating an appropriate integral representation of the parabolic cylinder 
 %%   function with its confluent hypergeometric function representation. 
 %%   Combining 12.4.1, 12.2.7, 12.5.1, 12.7.13 in \cite{Lozier03}, we get 
    We prove this by using the \cref{parabolic-conf-hypergeometric} with $a=c+\frac{1}{2}$. We note that
    $\alpha>-\frac{1}{2}$ and $\confHG\inparen{\cdot,\cdot,-\lambda^2/2}$ is an even
    function of $\lambda$. So combining the two, we get:
    \begin{align*}
        &\quad~\frac{1}{\sqrt{2\pi}}
        \int_{0}^{\infty} \tau^{c}\inparen{\ee^{-(\tau-\lambda)^2/2} - \ee^{-(\tau+\lambda)^2/2}}\,d\tau \\
        &=~\frac{2}{\sqrt{2\pi}}\cdot\frac{\sqrt{\pi}\cdot\Gamma\inparen{c+1}}{2^{c/2}
        \cdot\Gamma\inparen{\frac{c+1}{2}}}
        \cdot \lambda\cdot\confHG\inparen{-\frac{c}{2}+\frac{1}{2}\,;\frac{3}{2}\,;-\frac{1}{2}\lambda^2}\\
        &=~2^{(1-c)/2}\cdot
        \frac{\Gamma\inparen{\frac{c+1}{2}+\frac{1}{2}}}{2^{-c}\cdot\sqrt{\pi}}\cdot 
        \lambda\cdot \confHG\inparen{\frac{1-c}{2}\,;\,\frac{3}{2}\,;\,-\frac{\lambda^2}{2}} 
        &&(\text{by \cref{duplication:formula}}) \\
        &=~\gamma_{c+1}^{c+1}\cdot 
        \lambda\cdot \confHG\inparen{\frac{1-c}{2}\,;\,\frac{3}{2}\,;\,-\frac{\lambda^2}{2}} 
        &&(\text{by \cref{gamma:and:Gamma}})
    \end{align*}
\end{proof}
\subsection[Taylor Coefficients of f]{Taylor Coefficients of $\fplain{a}{b}{x}$ and Hypergeometric
  Representation}
By \cref{noise:correlation}, we are left with understanding the function whose power series is given by 
a weighted coefficient-wise product of a certain pair of confluent hypergeometric functions. This 
turns out to be precisely the Gaussian hypergeometric function, as we will see below. 
\begin{observation}
\label[observation]{1F1:circ:1F1:gives:2F1}
    Let $f_k := [\tau^k]\, \confHG(a_1,3/2,\tau)$ and $h_k := [\tau^k] \,\confHG(b_1,3/2,\tau)$. Further let \\
    $\mu_k := f_k\cdot h_k\cdot (2k+1)!/4^k$. Then for $\rho\in [-1,1]$, 
    \[
        \sum_{k\geq 0} \mu_k\cdot \rho^n  ~=~  \hypergeometric(a_1,b_1\,;\,3/2\,;\,\rho) \mper
    \]
\end{observation}

\begin{proof}
    The claim is equivalent to showing that
    $\mu_k = \RisingFactorial{a_1}\,\RisingFactorial{b_1}/(\RisingFactorial{3/2}\,k!)$. 
    Since we have $f_k = \RisingFactorial{a_1}/(\RisingFactorial{3/2}\,k!)$ and 
    $h_k = \RisingFactorial{b_1}/(\RisingFactorial{3/2}\,k!)$, it is sufficient to show that 
    $(2k+1)!/4^k = \RisingFactorial{3/2}\cdot k!$. Indeed, we have: 
    \begin{align*}
        (2k+1)! 
        &= 
        2^k\cdot k!\cdot 1\cdot 3 \cdot 5 \cdots (2k+1) \\
        &= 
        4^k\cdot k!\cdot \frac{3}{2} \cdot \frac{5}{2} \cdots \inparen{\frac{3}{2}+k-1} \\
        &= 
        4^k\cdot k!\cdot \RisingFactorial{3/2} \mper
    \end{align*}
\end{proof}

We are finally equipped to put everything together. 
\begin{theorem}
    \label[theorem]{hypergeometric:representation}
    For any $a,b\in (-1,\infty)$ and $\rho\in [-1,1]$, we have
    \[
        \nf{a}{b}{\rho}
        ~:=~
        \frac{1}{\gamma_{a+1}^{a+1}\cdot \gamma_{b+1}^{b+1}}\cdot 
        \Ex{\bfg_1 \sim_\rho \,\bfg_2}{\sgn(\bfg_1)|\bfg_1|^{a}\sgn(\bfg_2)|\bfg_1|^{b}}
        ~=~
        \hgp{\rho} \mper
    \]
    It follows that the $(2k+1)$-th Taylor coefficient of $\nf{a}{b}{\rho}$ is 
    \[
        \frac{\RisingFactorial{(1-a)/2}\,\RisingFactorial{(1-b)/2}}{(\RisingFactorial{3/2}\,k!)} \mper
    \]
\end{theorem}

\begin{proof}
    The claim follows by combining \cref{noise:correlation}, 
    \cref{gen:function:integral:rep,integral:representation:confHG}, and \cref{1F1:circ:1F1:gives:2F1}. 
\end{proof}

This hypergeometric representation immediately yields some non-trivial coefficient and 
monotonicity properties: 
\begin{corollary}
\label[corollary]{monotonicity:properties}
    For any $a,b\in [0,1]$, the function $\alnf{a}{b}:[-1,1]\to \R$ satisfies 
    \begin{enumerate}[label=(M\arabic*)]
        \item $[\rho]\,\nf{a}{b}{\rho} = 1$ ~and~ $[\rho^{3}]\,\nf{a}{b}{\rho} = (1-a)(1-b)/6$. 
        
        \item All Taylor coefficients are non-negative. Thus, $\nf{a}{b}{\rho}$ 
        is increasing on $[-1,1]$. 
        
        \item All Taylor coefficients are decreasing in $a$ and in $b$. Thus, for any fixed $\rho\in [-1,1]$,~ 
        $\nf{a}{b}{\rho}$ is decreasing in $a$ and in $b$. 
        
        \item Note that $\nf{a}{b}{0}=0$ and by (M1) and (M2), $\nf{a}{b}{1} \geq 1$. By continuity, 
        $\nf{a}{b}{[0,1]}$ contains $[0,1]$. Combining this with (M3) implies that for any fixed 
        $\rho\in [0,1]$, $\nfin{a}{b}{\rho}$ is increasing in $a$ and in $b$. 
    \end{enumerate}
\end{corollary}

%% file: PtoQ/6.wrap_up.tex
\section[Bound on Defect]{$\sinh^{-1}(1)/(1+\eps_0)$ ~Bound on $\nhin{a}{b}{1}$}
In this section we show that $p=\infty, q=1$ (the Grothendieck case) is roughly the extremal case for the 
value of $\nhin{a}{b}{1}$, \ie we show that for any $1\leq q\leq 2 \leq p \leq \infty$, $\nhin{a}{b}{1}\geq 
\sinh^{-1}(1)/(1+\eps_0)$ (recall that $\nhin{0}{0}{1} = \sinh^{-1}(1)$). While we were unable to establish as 
much, we conjecture that $\nhin{a}{b}{1} \geq \sinh^{-1}(1)$. 
\cref{behavior:of:coefficients} details some of the challenges involved in establishing that $\sinh^{-1}(1)$ is the worst case, and presents our approach to establish an approximate bound, which will be formally 
proved in \cref{bounding:coefficients}.

\subsection[Behavior of The Coefficients of The Inverse Function]
{Behavior of The Coefficients of $\nfin{a}{b}{z}$.}
\label[subsection]{behavior:of:coefficients}
Krivine's upper bound on the real Grothendieck constant, 
Haagerup's upper bound \cite{Haagerup81} on the complex Grothendieck constant and the work of 
Naor and Regev~\cite{NR14, BFV14} on the optimality of Krivine schemes are all closely related to our work in 
that each of the aforementioned papers needs to lower bound $(\absolute{f^{-1}})^{-1}(1)$ for an appropriate 
odd function $f$ (the work of Briet \etal~\cite{BFV14} on the rank-constrained Grothendieck problem is 
also a generalization of Krivine's and Haagerup's work, however they did not derive a closed form upper 
bound on $(\absolute{f^{-1}})^{-1}(1)$ in their setting). In Krivine's setting $f=\sin^{-1} x$, implying 
$(\absolute{f^{-1}})^{-1} = \sinh^{-1}$ and hence the bound is immediate. In our setting, as well as 
in~cite{Haagerup81} and~\cite{NR14, BFV14}, $f$ is  given by its Taylor coefficients and is not known to have 
a closed form. In~\cite{NR14}, all coefficients of $f^{-1}$ after the third are negligible and so one 
doesn't incur much loss by assuming that $\absolute{f^{-1}}(\rho) = c_1\rho + c_3\rho^3$.
In~\cite{Haagerup81}, the coefficient of $\rho$ in $f^{-1}(\rho)$ is $1$ and every subsequent coefficient is 
negative, which implies that $\absolute{f^{-1}}(\rho) = 2\rho - f^{-1}(\rho)$. 
Note that if the odd coefficients of $f^{-1}$ are alternating in sign like in Krivine's setting, then 
$\absolute{f^{-1}}(\rho) = -i \cdot f^{-1}(i\rho)$. 
These structural properties of the coefficients help their analyses. 

In our setting there does not appear to be such a strong relation between $(\absolute{f^{-1}})$ and $f^{-1}$. 
Consider $f(\rho) = \nf{a}{a}{\rho}$. For certain $a\in (0,1)$, the sign pattern of the coefficients of 
$f^{-1}$ is unlike that of \cite{Haagerup81} or $\sin \rho$. In fact empirical results suggest that the odd 
coefficients of $f$ alternate in sign up to some term $K=K(a)$, and subsequently the coefficients 
are all non-positive (where $K(a) \to \infty$ as $a\to 0$), \ie the sign pattern appears to be interpolating 
between that of $\sin \rho$ and that of $f^{-1}(\rho)$ in the case of Haagerup~\cite{Haagerup81}.
 
Another source of difficulty is that for a fixed $a$, the coefficients of $f^{-1}$ (with and without magnitude) 
are not necessarily monotone in $k$, and moreover for a fixed $k$, the $k$-th coefficient of $f^{-1}$ is not 
necessarily monotone in $a$. 

A key part of our approach is noting that certain milder assumptions on the coefficients are sufficient to 
show that $\sinh^{-1}(1)$ is the worst case. The proof crucially uses the monotonicity of $\nf{a}{b}{\rho}$ in 
$a$ and $b$. The conditions are as follows: \medskip 

\noindent
Let $\ngc := [\rho^k]\,\nfin{a}{b}{\rho}$. Then 
\begin{enumerate}[label=(C\arabic*)]
    \item $\ngc \leq 1/k!$ ~if~ $k\! \pmod{4} \equiv 1$. 
    
    \item $\ngc \leq 0$ ~if~ $k\! \pmod{4} \equiv 3$.
\end{enumerate}
To be more precise, we were unable to establish that the above conditions hold for all $k$ (however 
we conjecture that it is true for all $k$), and instead use 
Mathematica to verify it for the fist few coefficients. We additionally show that the coefficients of 
$\alnfin{a}{b}$ decay exponentially. Combining this exponential decay with a robust version of the previously 
advertised claim yields that $\nhin{a}{b}{1} \geq \sinh^{-1}(1)/(1 + \eps_0)$. 

We next proceed to prove the claim that the aforementioned conditions are sufficient to show that 
$\sinh^{-1}(1)$ is the worst case. We will need the following definition.
For an odd positive integer $t$, let 
\[
    h_{err}(t,\rho) := \sum_{k\geq t} |\ngc|\cdot \rho^{k}
\]

\begin{lemma}
\label[lemma]{bound:defect}
    If $t$ is an odd integer such that (C1) and (C2) are satisfied for all $k<t$, and 
    $\rho = \sinh^{-1}(1-2 h_{err}(t,\delta))$ ~for some $\delta\geq \rho$,~then $\nh{a}{b}{\rho}\leq 1$. 
\end{lemma}

\begin{proof}
    We have, 
    \begin{align*}
        &\quad~ \nh{a}{b}{\rho}  \\
        &= 
        \sum_{k\geq 1} |\ngc|\cdot \rho^k  \\
        &= 
        -\nfin{a}{b}{\rho} 
        ~+ 
        \sum_{k\geq 1} \max\{2\ngc,0\}\cdot \rho^k  \\
        &= 
        -\nfin{a}{b}{\rho} 
        ~+ 
        \sum_{\substack{1\leq k < t \\ k \text{ mod }4 \equiv 1}} \max\{2\ngc,0\}\cdot \rho^k  
        ~+~
        \sum_{k\geq t} \max\{2\ngc,0\} \cdot \rho^k
        &&(\text{by (C2)}) \\
        &\leq 
        -\nfin{a}{b}{\rho} 
        ~+ 
        \sum_{\substack{1\leq k < t \\ k \text{ mod } 4 \equiv 1}} \max\{2\ngc,0\}\cdot \rho^k  
        ~+~
        2 \, h_{err}(t,\rho) \\
        &\leq 
        -\nfin{a}{b}{\rho} 
        ~+ 
        \sin(\rho) + \sinh(\rho)
        ~+~
        2 \, h_{err}(t,\rho)
        &&(\text{by (C1)}) \\
        &\leq 
        -\nfin{a}{b}{\rho} 
        ~+ 
        \sin(\rho) + 1 + 2(h_{err}(t,\rho) - h_{err}(t,\delta))
        &&(\,\rho = \sinh^{-1}(1-2 h_{err}(t,\delta))\,) \\
        &\leq 
        -\nfin{a}{b}{\rho} 
        ~+ 
        \sin(\rho) + 1 
        &&(\rho\leq \delta) \\
        &\leq 
        -\nfin{0}{0}{\rho} 
        ~+ 
        \sin(\rho) + 1  
        &&(\text{\cref{monotonicity:properties} : (M4)})\\
        &= 
        1 
        &&(\nfin{0}{0}{\rho} = \sin(\rho))
    \end{align*}
\end{proof}

Thus, we obtain: 
\begin{theorem}
    For any $1\leq q \leq 2 \leq p \leq \infty$, let $a:=p^*-1,b=q-1$. Then for any $m,n\in \N$ and 
    $A\in \R^{m\times n}$,~~ 
    $\CP{A}/\norm{p}{q}{A} ~\leq~ 1/(\nhin{a}{b}{1}\cdot \gamma_{q}\,\gamma_{p^*})$ ~
    and moreover 
    \begin{itemize}
        \item $\nhin{1}{b}{1} = \nhin{a}{1}{1} = 1$.
        \item $\nhin{a}{b}{1}\geq \sinh^{-1}(1)/(1 + \eps_0)$ ~where $\eps_0 = 0.00863$. 
    \end{itemize}
\end{theorem}

\begin{proof}
    The first inequality follows from \cref{defect:bound:implies:apx}. 
    As for the next item, If $p=2$ or $q=2$ ~(\ie $a=1$ or $b=1$) we are trivially done since 
    $\nhin{a}{b}{\rho} = \rho$ in that case (since for $k\geq 1$,~ $\RisingFactorial{0}=0$). 
    So we may assume that $a,b\in [0,1)$. 
    
    We are left with proving the final part of the claim. 
    Now using Mathematica we verify (exactly)\footnote{
        We generated $\ngc$ as a polynomial in $a$ and $b$ and maximized it over $a,b\in [0,1]$ using 
        the Mathematica ``Maximize'' function which is exact for polynomials.}
    that $(C1)$ and $(C2)$ are true for $k\leq 29$. Now let $\delta=\sinh^{-1}(0.974203)$. 
    Then by \cref{inv:coeff:bound} (which states that $\ngc$ decays exponentially and will be proven in the 
    subsequent section), 
    \[
        h_{err}(31,\delta) :=  \sum_{k\geq 31} |\ngc|\cdot d^{k} 
        \leq  \frac{6.1831}{31}\cdot \frac{\delta^{31}}{1-\delta^2} 
        \leq  0.0128991\dots \mper
    \] 
    Now by \cref{bound:defect} we know $\nhin{a}{b}{1} \geq \sinh^{-1}(1-2h_{err}(31,\delta))$. 
    Thus, \\$\nhin{a}{b}{1}\geq \sinh^{-1}(0.974202) \geq \sinh^{-1}(1)/(1+\eps_0)$ for 
    $\eps_0 = 0.00863$, which completes the proof. 
\end{proof}

%% file: PtoQ/7.coefficient_bound.tex
\subsection{Bounding Inverse Coefficients}
\label[subsection]{bounding:coefficients}
In this section we prove that $\ngc$ decays as $1/c^k$ for some $c=c(a,b)>1$, proving 
\cref{inv:coeff:bound}. 
Throughout this section we assume $1\leq p^*,q <2$, and $a=p^*-1,~b=q-1$ (\ie $a,b\in [0,1)$).
Via the power series representation, $\nf{a}{b}{z}$ can be analytically continued to the unit complex disk. 
Let $\nfin{a}{b}{z}$ be the inverse of $\nf{a}{b}{z}$ and recall $\ngc$ denotes its $k$-th Taylor coefficient. 

We begin by stating a standard identity from complex analysis that provides a convenient contour integral 
representation of the Taylor coefficients of the inverse of a function. We include a proof for completeness. 
\begin{lemma}[Inversion Formula]
\label[lemma]{residue} 
    There exists $\delta>0$, such that for any odd $k$, 
    \begin{equation}
    \label[equation]{inversion:formula}
        \ngc = \frac{2}{\pi k} \,\Im\!\inparen{\int_{C^+_\delta} \nf{a}{b}{z}^{-k}\,dz}
    \end{equation}
    where $C^+_\delta$ denotes the first quadrant quarter circle of radius 
    $\delta$ with counter-clockwise orientation. 
\end{lemma}

\begin{proof}
    Via the power series representation, $\nf{a}{b}{z}$ can be analytically continued to the unit complex disk. 
    Thus by inverse function theorem for holomorphic functions, there exists $\delta_0 \in (0,1]$ such that 
    $\nf{a}{b}{z}$ has an analytic inverse in the open disk $|z|<\delta_0$. So for $\delta \in (0,\delta_0)$, 
    $\nf{a}{b}{C_\delta}$ is a simple-closed curve with winding number $1$ (where $C_\delta$ is the complex 
    circle of radius $\delta$ with the usual counter-clockwise orientation). Thus, by Cauchy's integral formula 
    we have:
    \[
        \ngc 
        ~=~
        \frac{1}{2\pi i}\,\int_{\nf{a}{b}{C_\delta}} \frac{\nfin{a}{b}{w}}{w^{\,k}}\,dw
        ~~=~ 
        \frac{1}{2\pi i}\,\int_{C_\delta} \,\frac{z \cdot \alnf{a}{b}'(z)}{\nf{a}{b}{z}^{k+1}} \,\,dz
    \]
    where the second equality follows from substituting $w = \nf{a}{b}{z}$. 
    
    Now by \cref{no:roots},~\, $z/\nf{a}{b}{z}^k$~ is holomorphic on the open set $|z|\in (0,1)$, which 
    contains $C_{\delta}$. Hence, by the fundamental theorem of contour integration we have 
    \[
        \int_{C_{\delta}} \frac{d}{dz}\inparen{\frac{z}{\nf{a}{b}{z}^{k}}}\,dz 
        ~=~ 
        0 
        \quad~~ \Rightarrow \quad
        \int_{C_{\delta}} \,\frac{z \cdot \alnf{a}{b}'(z)}{\nf{a}{b}{z}^{k+1}} \,\,dz
        ~=~
        \frac{1}{k}\,\int_{C_{\delta}} \,\frac{1}{\nf{a}{b}{z}^{k}}\,\,dz 
    \]
    So we get, 
    \[
        \ngc 
        ~=~
        \frac{1}{2\pi i k}\,\int_{C_\delta} \nf{a}{b}{z}^{-k}\,dz
        ~~=~ 
        \frac{1}{2\pi k}\,
        \Im\!\inparen{
        \int_{C_\delta} \nf{a}{b}{z}^{-k}\,dz
        }
    \]
    where the second equality follows since $\ngc$ is purely real. Lastly, we complete the proof of the 
    claim by using the fact that for odd $k$,~ $\nf{a}{b}{z}^{-k}$ is odd and that 
    $\overline{\nf{a}{b}{z}} = \nf{a}{b}{\overline{z}}$.
\end{proof}

We next state a standard bound on the magnitude of a contour integral that we will use in our analysis. 
\begin{fact}[ML-inequality]
\label[fact]{ML:inequality}
    If $f$ is a complex valued continuous function on a contour $\Gamma$ and $|f(z)|$ is bounded by $M$ for 
    every $z\in \Gamma$, then 
    \[
        \abs{\int_{\Gamma}f(z)} \leq M\cdot \ell(\Gamma)
    \]
    where $\ell(\Gamma)$ is the length of $\Gamma$. 
\end{fact}
Unfortunately the integrand in \cref{inversion:formula} can be very large for small $\delta$, and we cannot 
use the ML-inequality as is. To fix this, we modify the contour of integration (using Cauchy's integral 
theorem) so that the imaginary part of the integral vanishes when restricted to the sections close to 
the origin, and the integrand is small in magnitude on the sections far from the origin (thus allowing us to 
use the ML-inequality). To do this we will need some preliminaries. 

$\nf{a}{b}{z}$ is defined on the closed complex unit disk. The domain is analytically extended to 
the region $\C\setminus ((-\infty,-1) \cup (1,\infty))$, using the Euler-type integral representation of the hypergeometric function. 
\[
    \nfext{a}{b}{z} 
    := 
    \betaterm^{-1}\cdot 
    \intfull{z}
\]
where $\mathrm{B}(\tau_1,\tau_2)$ is the beta function and 
\[
    \intfull{z}
    :=  
    z \int_{0}^{1} \frac{(1-t)^{b/2} \,dt}{t^{(1+b)/2}\cdot (1-z^2 t)^{(1-a)/2}}.
\]

\begin{fact}
\label[fact]{no:roots}
    For any $a_1>0$, $\hypergeometric(a_1,b_1,c_1,z)$ has no non-zero roots in the region 
    $\C\setminus  (1,\infty)$. This implies that if $p^*<2$,~ $\nfext{a}{b}{z}$ has no non-zero roots in 
    the region $\C\setminus ((-\infty,-1) \cup (1,\infty))$. 
\end{fact}

We are now equipped to expand the contour. Our choice of contour is inspired by that of
Haagerup~\cite{Haagerup81} 
which he used in deriving an upper bound on the complex Grothendieck constant. The contour we choose 
has some differences for technical reasons related to the region to which hypergeometric 
functions can be analytically extended. The analysis is quite different from that of Haagerup 
since the functions in consideration behave differently. In fact the inverse function Haagerup considers 
has polynomially decaying coefficients while the class of inverse functions we consider have coefficients 
that have decay between exponential and factorial. 

\begin{observation}[Expanding Contour]
    For any $\alpha \geq 1$ and $\eps>0$, let $P(\alpha,\eps)$ be the four-part curve (see \cref{contour}) 
    given by 
    \begin{itemize}
        \item the line segment $\delta~ \rightarrow ~(1-\eps)$, 
        
        \item the line segment $(1-\eps)~ \rightarrow ~ (\sqrt{\alpha-\eps} + i\sqrt{\eps})$ 
        (henceforth referred to as $L_{\alpha,\eps}$), 
        
        \item the arc along $C^+_\alpha$ starting at $(\sqrt{\alpha-\eps} + i\sqrt{\eps})$ and ending at $i\alpha$ 
        (henceforth referred to as $C^+_{\alpha,\eps}$), 
        
        \item the line segment $i\alpha~ \rightarrow ~i\delta$. 
    \end{itemize}
    By Cauchy's integral theorem, combining \cref{residue} with \cref{no:roots} yields that for odd $k$, 
    \[
        \ngc = \frac{2}{\pi k}\, \Im\!\inparen{\int_{P(\alpha,\eps)} \nfext{a}{b}{z}^{-k}\,dz}
    \]
\end{observation}
\begin{figure}
\caption{The Contour $P(\alpha,\eps)$}
\label[figure]{contour}
\includegraphics{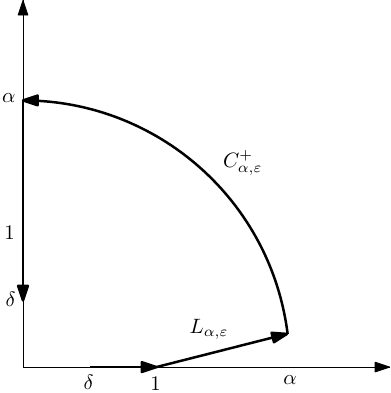}
\centering
\end{figure}

We will next see that the imaginary part of our contour integral vanishes on section of $P(\alpha,\eps)$. 
Applying ML-inequality to the remainder of the contour, combined with lower bounds on 
$|\nfext{a}{b}{z}|$ ~(proved below the fold in \cref{largeness:of:f}), allows us to derive an exponentially 
decaying upper bound on $|\ngc|$. 
\begin{lemma}
\label[lemma]{inv:coeff:bound}
    For any $1\leq p^*,q<2$,  there exists $\eps>0$ such that 
    \[
        |\ngc| \leq \frac{6.1831}{k(1+\eps)^{k}}.
    \]
\end{lemma} 

\begin{proof}
    For a contour $P$, we define $V(P)$ as  
    \[
        V(P) := \frac{2}{\pi k} \,\Im\!\inparen{\int_{P} \nfext{a}{b}{z}^{-k}\,dz}
    \]
    As is evident from the integral representation, $\nfext{a}{b}{z}$ is purely imaginary if $z$ is 
    purely imaginary, and as is evident from the power series, $\nf{a}{b}{z}$ is purely real if $z$ lies on 
    the real interval $[-1,1]$. This implies that $V(\delta \rightarrow (1-\eps)) = 
    V(i\alpha\rightarrow i\delta) = 0$. 
    
    Now combining \cref{ML:inequality} (ML-inequality) with \cref{|f(z)|:incr:in:|z|} and 
    \cref{close:to:real:contour} (which state that the integrand is small in magnitude over $C^+_{6,\eps}$ 
    and $L_{6,\eps}$ respectively), we get that for sufficiently small $\eps>0$, 
    \begin{align*}
        |V(P(6,\eps))| 
        &\leq 
        |V(C^+_{6,\eps})| + |V(L_{6,\eps})|  \\
        &\leq 
        \frac{2}{\pi k}\cdot \frac{3\pi/2}{(1+\eps)^{k}} + 
        \frac{2}{\pi k}\cdot \frac{6-1+O(\sqrt{\eps})}{(1+\eps)^k} \\
        &\leq 
        \frac{6.1831}{k (1+\eps)^{k}}. 
        &&(\text{taking } \eps \text{ sufficiently small}) 
    \end{align*}
\end{proof}

\subsubsection{Lower bounds on $|\nfext{a}{b}{z}|$ Over $C^+_{\alpha,\eps}$ and $L_{\alpha,\eps}$}
\label[section]{largeness:of:f}
In this section we show that for sufficiently small $\eps$, $|\nfext{a}{b}{z}|>1$ over $L_{\alpha,\eps}$ (regardless of the value of $\alpha$, Lemma~\ref{close:to:real:contour}), and over $C^+_{\alpha,\eps}$ when $\alpha$ is a sufficiently large constant (Lemma~\ref{|f(z)|:incr:in:|z|}). 

We will first show the claim for $C^+_{\alpha,\eps}$ by relating $|\nfext{a}{b}{z}|$ to $|z|$. While the 
asymptotic behavior of hypergeometric functions for $|z|\to \infty$ has been extensively studied (see 
for instance \cite{Lozier03}), it appears that our desired estimates aren't immediate consequences of 
prior work for two reasons. Firstly, we require relatively precise estimates for moderately large but 
constant $|z|$. Secondly, due to the expressive power of hypergeometric functions, the estimates we 
derive can only be true for hypergeometric functions parameterized in a specific range. Indeed, our proof 
crucially uses the fact that $a,b\in [0,1)$. Our approach is to use the Euler-type integral representation 
of $\nfext{a}{b}{z}$ which as a reminder to the reader is as follows: 
\[
    \nfext{a}{b}{z} 
    := 
    \betaterm^{-1}\cdot 
    \intfull{z}
\]
where $\mathrm{B}(x,y)$ is the beta function and 
\[
    \intfull{z}
    :=  
    z \int_{0}^{1} \frac{(1-t)^{b/2} \,dt}{t^{(1+b)/2}\cdot (1-z^2 t)^{(1-a)/2}}.
\]
We start by making the simple observation that the integrand of $\intfull{z}$ is always in the positive 
complex quadrant --- an observation that will come in handy multiple times in this section, in dismissing 
the possibility of cancellations. This is the part of our proof that makes the most crucial use of 
the assumption that $0\leq a<1$ (equivalently $1\leq p^* <2$). 
\begin{observation}
\label[observation]{reim:monotonicity}
    Let $z = r e^{i\theta}$ be such that either one of the following two cases is satisfied:
    \begin{enumerate}[label=(\Alph*)]
        \item $r<1$ and $\theta = 0$.
        
        \item $\theta\in (0,\pi/2]$.
    \end{enumerate}     
    Then for any $0 \leq a\leq 1$ and any $t\in \R^{+}$, 
    \[
        \arg\inparen{\frac{z}{(1 - tz^2)^{(1-a)/2}}} \in [0,\pi/2]
    \]
\end{observation}

\begin{proof}
    The claim is clearly true when $\theta = 0$ and $r<1$. It is also clearly true when $\theta=\pi/2$. 
    Thus we may assume $\theta\in (0,\pi/2)$.  
    \begin{align*}
        & \arg(z) \in (0,\pi/2) 
        \quad\Rightarrow 
        \arg(-tz^2) \in (-\pi,0) 
        \quad\Rightarrow 
        \Im(-tz^2) < 0 \\
        &\Rightarrow 
        \Im(1-tz^2) < 0
        \quad\Rightarrow 
        \arg(1 - tz^2) \in (-\pi,0) 
    \end{align*}
    Moreover since $\arg(-tz^2)=2\theta-\pi \in (-\pi,0)$, we have $\arg(1-tz^2)> 2\theta-\pi$. 
    Thus we have, 
    \begin{align*}
        & \arg(1-tz^2) \in (2\theta - \pi,0) 
        \quad\Rightarrow 
        \arg\inparen{(1-tz^2)^{(1-a)/2}} \in ((1-a)(\theta-\pi/2),0) \\
        &\Rightarrow 
        \arg\inparen{1/(1-tz^2)^{(1-a)/2}} \in (0,(1-a)(\pi/2-\theta))  \\
        &\Rightarrow 
        \arg\inparen{z/(1-tz^2)^{(1-a)/2}} \in (0,(1-a)(\pi/2-\theta) +\theta) \subseteq (0,\pi/2)
    \end{align*}
\end{proof}

\iffalse
\begin{fact}
\label[fact]{gamma:apx}
    For $0<x<1$, $\Gamma(x)\leq 1/x + 1/e$. 
\end{fact}

\begin{proof}
    \begin{align*}
        \Gamma(x) 
        ~= ~
        \int_{0}^{\infty} y^{x-1} e^{-y}\,dy 
        ~\leq ~
        \int_{0}^{1} y^{x-1}\,dy  + \int_{1}^{\infty} e^{-y}\,dy 
        ~= ~
        1/x+1/e
    \end{align*}
\end{proof}
\fi

We now show $|\nfext{a}{b}{z}|$ is large over $C^+_{\alpha,\eps}$. The main idea is to move from a 
complex integral to a real integral with little loss, and then estimate the real integral. To do this, we use 
\cref{reim:monotonicity} to argue that the magnitude of $\intfull{z}$ is within $\sqrt{2}$ of the 
integral of the magnitude of the integrand. 

\begin{lemma}[ \bf{$|\nfext{a}{b}{z}|$ is large over $C^+_{\alpha,\eps}$} ] 
\label[lemma]{|f(z)|:incr:in:|z|}
    Assume $a,b\in[0,1)$ and consider any $z\in \C$ with $|z| \geq 6$. Then $|\nfext{a}{b}{z}| > 1$. 
\end{lemma}
\begin{proof}
    We start with a useful substitution. 
    \begin{align*}
        \intfull{z}
        &=  z \int_{0}^{1} \frac{(1-t)^{b/2} \,dt}{t^{(1+b)/2}\cdot (1-z^2 t)^{(1-a)/2}} \\
        &= 
         r^{b}e^{i\theta} 
         \int_{0}^{r^2} \frac{(1-s/r^2)^{b/2} \,ds}{s^{(1+b)/2}\cdot (1-e^{2i\theta} s)^{(1-a)/2}} 
         && (\text{Subst. } s = r^2 t, \text{ where } z=re^{i\theta}) \\
         &= 
         r^{b}
         \int_{0}^{r^2} \frac{w_a(s,\theta)\cdot (1-s/r^2)^{b/2} \,ds}{s^{1+(b-a)/2}} 
    \end{align*}
    where
    \[
        w_a(s,\theta) 
        := 
        \frac{e^{i\theta}}{(1/s-e^{2i\theta})^{(1-a)/2}} \mper
    \]
    
    \medskip 
    
    We next exploit the observation that the integrand is always in the positive complex quadrant by 
    showing that $|\intfull{z}|$ is at most a factor of $\sqrt{2}$ away from the integral obtained 
    by replacing the integrand with its magnitude. 
    \begin{align*}
        &\quad~~ |\intfull{z}|  \\
        &=~ 
        \sqrt{\Re(\intfull{z})^2 + \Im(\intfull{z})^2} \\
        &\geq~ 
        (|\Re(\intfull{z})| + |\Im(\intfull{z})|)/\sqrt{2} &&(\text{Cauchy-Schwarz}) \\
        &=~ 
        (\Re(\intfull{z}) + \Im(\intfull{z}))/\sqrt{2} &&(\text{by \cref{reim:monotonicity}}) \\
        &=~
        \frac{r^{b}}{\sqrt{2}}
        \int_{0}^{r^2} 
        \inparen{ \Re(w_a(s,\theta)) + \Im(w_a(s,\theta)) } \cdot 
        \frac{(1-s/r^2)^{b/2} \,ds}{s^{1+(b-a)/2}} \\
        &=~
        \frac{r^{b}}{\sqrt{2}}
        \int_{0}^{r^2} 
        \inparen{ |\Re(w_a(s,\theta))| + |\Im(w_a(s,\theta))| } \cdot 
        \frac{(1-s/r^2)^{b/2} \,ds}{s^{1+(b-a)/2}} 
        &&(\text{by \cref{reim:monotonicity}})\\
        &\geq~
        \frac{r^{b}}{\sqrt{2}}
        \int_{0}^{r^2} 
        |w_a(s,\theta)| \cdot 
        \frac{(1-s/r^2)^{b/2} \,ds}{s^{1+(b-a)/2}} 
        &&(\norm{1}{v}\geq \norm{2}{v}) \\
        &\geq~
        \frac{r^{b}}{\sqrt{2}}
        \int_{0}^{r^2} 
        \frac{1}{(1+1/s)^{(1-a)/2}} \cdot 
        \frac{(1-s/r^2)^{b/2} \,ds}{s^{1+(b-a)/2}} 
    \end{align*}
    We now break the integral into two parts and analyze them separately. We start by analyzing the part that's  
    large when $b\rightarrow 0$.  
    \begin{align*}
        &\quad~~
        \frac{r^{b}}{\sqrt{2}}
        \int_{1}^{r^2} 
        \frac{1}{(1+1/s)^{(1-a)/2}} \cdot 
        \frac{(1-s/r^2)^{b/2} \,ds}{s^{1+(b-a)/2}} \\
        &\geq~
        \frac{r^{b}}{2}
        \int_{1}^{r^2} 
        \frac{(1-s/r^2)^{b/2} \,ds}{s^{1+(b-a)/2}} \\
         &\geq~ 
         \frac{r^b}{2} \int_{1}^{r^2/2} 
         \frac{(1-s/r^2)^{b/2} \,ds}{s^{1+(b-a)/2}} \\
         &\geq~ 
         \frac{r^b}{2\sqrt{2}} \int_{1}^{r^2/2} 
         \frac{ds}{s^{1+(b-a)/2}} 
         &&(\text{since } s\leq r^2/2) \\
         &\geq~ 
         \frac{r^b\cdot \min\{1,r^{a-b}\}}{2\sqrt{2}} \int_{1}^{r^2/2} 
         \frac{ds}{s} \\
         &=~ 
         \frac{\min\{r^a,r^b\}\cdot \log (r^2/2)}{2\sqrt{2}} \\
         &\geq~ 
         \frac{\log (r/\sqrt{2})}{\sqrt{2}} 
    \end{align*}
    We now analyze the part that's large when $b\rightarrow 1$.
    \begin{align*}
        &\quad~~
        \frac{r^{b}}{\sqrt{2}}
        \int_{0}^{1} 
        \frac{1}{(1+1/s)^{(1-a)/2}} \cdot 
        \frac{(1-s/r^2)^{b/2} \,ds}{s^{1+(b-a)/2}} \\
        &=~
        \frac{r^{b}}{\sqrt{2}}
        \int_{0}^{1} 
        \frac{(1-s/r^2)^{b/2} }{(1+s)^{(1-a)/2}} \cdot 
        \frac{ds}{s^{(1+b)/2}} \\
        &\geq~
        \frac{r^{b}\cdot \sqrt{1-1/r^2}}{2}
        \int_{0}^{1} 
        \frac{ds}{s^{(1+b)/2}} 
        &&(\text{since } s\leq 1) \\
        &=~
        \frac{r^{b}\cdot \sqrt{1-1/r^2}}{1-b} 
    \end{align*}
    Combining the two estimates above yields that if $r>\sqrt{2}$, 
    \[
        |\nfext{a}{b}{z}|
        \geq~
        \betaterm^{-1} \cdot 
        \inparen{
        \frac{\log (r/\sqrt{2})}{\sqrt{2}} 
        + 
        \frac{r^{b}\cdot \sqrt{1-1/r^2}}{1-b} 
        }
    \]
    Lastly, the proof follows by using the following estimate:
    \begin{fact}
    \label[fact]{mathematica:fact:beta}
        Via Mathematica, for $0\leq b< 1$ we have 
        \[
            \betaterm^{-1} \cdot 
            \inparen{
            \frac{\log (6/\sqrt{2})}{\sqrt{2}} 
            + 
            \frac{6^{b}\cdot \sqrt{1-1/6^2}}{1-b} 
            }
            ~~\geq~~
            1.003
        \]
    \end{fact}
\end{proof}

\begin{remark}
    The preceding proof can be used to derive the precise asymptotic behavior of 
    $|\nfext{a}{b}{z}|$ in $r$. Specifically, it grows as ~$r^{a} \log r$~ if ~$a=b$~ and as ~$r^{\,\max\{a,b\}}$ 
    ~if ~$a\neq b$. 
\end{remark}

\bigskip

We now show that $|\nfext{a}{b}{z}|>1$ over $L_{\alpha, \eps}$. To do this, it is insufficient to assume 
that $|z|\geq 1$ since there exist points $z$ (for instance $z=i$) of unit length such that 
$|\nfext{a}{b}{z}|<1$. To show the claim, we observe that $|\nfext{a}{b}{z}|$ is large when $z$ is close to 
the real line and use the fact that $L_{\alpha,\eps}$ is close to the real line. Formally, we show that if 
$z$ is of length at least $1$ and is sufficiently close to the real line, $|\nfext{a}{b}{z}|$ is close to 
$\nfext{a}{b}{1}$. Lastly, we use the power series representation of the hypergeometric function to obtain 
a sufficiently accurate lower bound on $\nfext{a}{b}{1}$. 

\begin{lemma}[ \bf{$|\nfext{a}{b}{z}|$ is large over $L_{\alpha,\eps}$} ] 
\label[lemma]{close:to:real:contour}
    Assume $a,b\in [0,1)$ and consider any $\gamma\geq 1-\eps_1$. 
    Let $\eps_2 := \sqrt{\eps_1}$ and $z := \gamma (1+i\eps_1)$. 
    Then for $\eps_1>0$ sufficiently small, $|\nfext{a}{b}{z}|>1$. 
\end{lemma}

\begin{proof}
    Below the fold we will show 
    \begin{align}
        |\intfull{z}| 
        ~\geq~ 
        (1-O(\sqrt{\eps_1}))\int_{0}^{1-\eps_2} \frac{(1-s)^{b/2} \,ds}{s^{(1+b)/2}\cdot (1 - s)^{(1-a)/2}} & 
        \label{ineq:real:closeness} 
    \end{align}
    But we know ($\mathrm{LHS},~\mathrm{RHS}$ refer to \cref{ineq:real:closeness})
    \begin{align*}
        \betaterm^{-1}\cdot \mathrm{LHS} &~=~ \nfext{a}{b}{z} \text{~~~and ~~} \\
        \betaterm^{-1}\cdot \mathrm{RHS} ~&\rightarrow~ \nf{a}{b}{1} \text{  ~as~  } \eps_1\rightarrow 0 
    \end{align*}
    Also by \cref{monotonicity:properties} : (M1), (M2),~ $\nf{a}{b}{1} \geq 1+(1-a)(1-b)/6>1$. 
    Thus for $\eps_1$ sufficiently small, we must have $|\nfext{a}{b}{z}|>1$. \medskip

    We now show \cref{ineq:real:closeness}, by comparing integrands point-wise. To do this, we will 
    assume the following closeness estimate that we will prove below the fold: 
    \begin{equation}
    \label[equation]{eq:real:close:final}
        \Re\inparen{\frac{1+i\eps_1}{(1-s(1+i\eps_1)^2)^{(1-a)/2}}}
        = 
        \frac{1-O(\eps_2)}{(1-s)^{(1-a)/2}} \mper
    \end{equation}
    We will also need the following inequality. Since $\gamma\geq 1-\eps_1 = 1-\eps_2^2$, 
    for any $0\leq s \leq 1-\eps_2$, we have 
    \begin{equation}
    \label[equation]{eq:numerator}
        (1-s/\gamma^2)^{b/2} \geq (1-O(\eps_2))\cdot (1-s)^{b/2}.
    \end{equation}   
    Given, these estimates, we can complete the proof of \cref{ineq:real:closeness} as follows: 
    \begin{align*}
        &\quad~\Re(\intfull{z}) \\
        &= 
        \Re\inparen{
         z\int_{0}^{1} \frac{(1-t)^{b/2} \,dt}{t^{(1+b)/2}\cdot (1 - tz^2)^{(1-a)/2}} 
         } \\
         &= 
         \Re\inparen{
         \gamma^{b}(1+i\eps_1)
         \int_{0}^{\gamma^2} \frac{(1-s/\gamma^2)^{b/2} \,ds}{s^{(1+b)/2}\cdot (1-s(1+i\eps_1)^2)^{(1-a)/2}} 
         }
         &&(\text{subst. } s\leftarrow \gamma^2 t)  \\
         &\geq 
         \Re\inparen{
         \gamma^{b}(1+i\eps_1)
         \int_{0}^{1-\eps_2} \frac{(1-s/\gamma^2)^{b/2} \,ds}{s^{(1+b)/2}\cdot (1-s(1+i\eps_1)^2)^{(1-a)/2}} 
         }
         &&(\text{by \cref{reim:monotonicity}}) \\
         &= 
         \gamma^{b}
         \int_{0}^{1-\eps_2} 
         \Re\inparen{ \frac{1+i\eps_1}{(1-s(1+i\eps_1)^2)^{(1-a)/2}} }
         \frac{(1-s/\gamma^2)^{b/2} \,ds}{s^{(1+b)/2}} \\ 
         &\geq 
         (1-O(\eps_2))\cdot \gamma^b
         \int_{0}^{1-\eps_2} \frac{(1-s/\gamma^2)^{b/2} \,ds}{s^{(1+b)/2}\cdot (1 - s)^{(1-a)/2}} 
         &&(\text{by \cref{eq:real:close:final}}) \\
         &\geq 
         (1-O(\eps_2))
         \int_{0}^{1-\eps_2} \frac{(1-s)^{b/2} \,ds}{s^{(1+b)/2}\cdot (1 - s)^{(1-a)/2}} 
         &&(\text{by \cref{eq:numerator}, } \gamma\geq 1-\eps_1) 
    \end{align*}
    
    It remains to establish \cref{eq:real:close:final}, which we will do by considering the numerator 
    and reciprocal of the denominator separately and subsequently using the fact that
    $\Re(z_1 z_2) = \Re(z_1)\Re(z_2) - \Im(z_1)\Im(z_2)$. In doing this, we need to show that the 
    respective real parts are large and respective imaginary parts are small for which the following 
    simple facts will come in handy. 
    \begin{fact}
        \label[fact]{real:part:under:power}
        Let $z = re^{i\theta}$ be such that $\Re{z}\geq 0$ ~(i.e. $-\pi/2\leq \theta \leq \pi/2$). 
        Then for any $0\leq \alpha\leq 1$, 
        \[
            \Re(1/z^{\alpha}) = \cos(-\alpha\theta)/r^{\alpha} = \cos(\alpha\theta)/r^{\alpha} 
            \geq \cos(\theta)/r^{\alpha} = \Re(z)/r^{1+\alpha}
        \]
    \end{fact}

    \begin{fact}
    \label[fact]{imaginary:part:under:power}
        Let $z = re^{-i\theta}$ be such that $\Re{z}\geq 0, \Im{z}\leq 0$ ~(i.e. $0\leq \theta \leq \pi/2$). 
        Then for any $0\leq \alpha\leq 1$, 
        \[
            \Im(1/z^{\alpha}) = \sin(\alpha\theta)/r^{\alpha} 
            \leq \sin(\theta)/r^{\alpha} = -\Im(z)/r^{1+\alpha}
        \]
    \end{fact}
    We are now ready to prove the claimed properties of the reciprocal of the denominator from 
    \cref{eq:real:close:final}. For any $0\leq s\leq 1-\eps_2$ we have, 
    \begin{align}
    \label[equation]{eq:real:close}
        &\quad~\Re\inparen{\frac{1}{(1-s(1+i\eps_1)^2)^{(1-a)/2}}} \nonumber \\
        &=
        \Re\inparen{\frac{1}{(1-s+s\eps_1^2-2is\eps_1)^{(1-a)/2}}} \nonumber \\
        &=
        \frac{1}{(1-s)^{(1-a)/2}}\cdot 
        \Re\inparen{\frac{1}{(1+s\eps_1^2/(1-s)-2i\eps_1/(1-s))^{(1-a)/2}}} \nonumber \\
        &\geq 
        \frac{1}{(1-s)^{(1-a)/2} \cdot (1+O(\eps_1^2/\eps_2^2))^{(3-a)/4}} 
        &&\hspace{-20 pt}(\text{by \cref{real:part:under:power}, and } 1-s\geq \eps_2) \nonumber \\
        &=
        \frac{1-O(\eps_1^2/\eps_2^2)}{(1-s)^{(1-a)/2}} \nonumber \\
        &= 
        \frac{1-O(\eps_2)}{(1-s)^{(1-a)/2}} 
    \end{align}
    Similarly, 
    \begin{align}
    \label[equation]{eq:im:small}
        &\Im\inparen{\frac{1}{(1-s+s\eps_1^2-2is\eps_1)^{(1-a)/2}}} \nonumber \\
        &=
        \frac{1}{(1-s)^{(1-a)/2}}\cdot 
        \Im\inparen{\frac{1}{(1+s\eps_1^2/(1-s)-2i\eps_1/(1-s))^{(1-a)/2}}} \nonumber \\
        &\leq 
        \frac{2\eps_1}{(1-s)^{(1-a)/2}}
        &&(\text{by \cref{imaginary:part:under:power}})
    \end{align}
    Combining \cref{eq:real:close} and \cref{eq:im:small} with the fact that 
    $\Re(z_1 z_2) = \Re(z_1)\Re(z_2) - \Im(z_1)\Im(z_2)$ yields, 
    \[
        \Re\inparen{\frac{1+i\eps_1}{(1-s(1+i\eps_1)^2)^{(1-a)/2}}}
        = 
        \frac{1-O(\eps_2)}{(1-s)^{(1-a)/2}} \mper
    \]
    This completes the proof. 
\end{proof}

\subsubsection[Challenges of Proving (C1) and (C2) for all k]{Challenges of Proving (C1) and (C2) for all $k$}
\label[subsubsection]{coefficient:challenges}
For certain values of $a$ and $b$, the inequalities in (C1) and (C2) leave very little room for error. 
In particular, when $a=b=0$, (C1) holds at equality and (C2) has $1/k!$ additive slack. 
In this special case, it would mean that one cannot analyze the contour integral (for the 
$k$-th coefficient of $\nfin{a}{b}{\rho}$) by using ML-inequality on any section of the contour that is within 
a distance of $\mathrm{exp}(k)$ from the origin. Analytic approaches would require extremely precise 
estimates on the value of the contour integral on parts close to the origin. Other challenges to naive 
approaches come from the lack of monotonicity properties for $\ngc$ (both in $k$ and in $a,b$ - see 
\cref{behavior:of:coefficients})

%% file: PtoQ/8.factorization.tex
\section{Factorization of Linear Operators}
\label[section]{factorization}
Let $X,Y,E$ be Banach spaces and let $A:X\to Y$ be a continuous linear operator. We say that $A$ 
\emph{factorizes} through $E$ if there exist continuous operators $C:X\to E$ and $B:E\to Y$ such that 
$A=BC$. 
Factorization theory has been a major topic of study in functional analysis, going as far back as 
Grothendieck's famous ``Resume'' \cite{Grothendieck56}. It has many striking applications, like 
the isomorphic characterization of Hilbert spaces and $L_p$ spaces due to 
\Kwapien~\cite{Kwapien72a, Kwapien72b}, connections to type and cotype through the work of 
\Kwapien~\cite{Kwapien72a}, Rosenthal~\cite{Rosenthal73}, Maurey~\cite{Maurey74} and 
Pisier~\cite{Pisier80}, connections to Sidon sets through the work of Pisier~\cite{Pisier86}, 
characterization of weakly compact operators due to Davis \etal~\cite{DFJP74}, connections to the 
theory of $p$-summing operators through the work of Grothendieck~\cite{Grothendieck56}, 
Pietsch~\cite{Pietsch67} and Lindenstrauss and Pelczynski~\cite{LP68}. 

Let $\factorConst{A}$ denote 
\[
    \factorConst{A}:= \inf_{H} \inf_{BC=A} \frac{\norm{X}{H}{C}\cdot \norm{H}{Y}{B}}{\norm{X}{Y}{A}}
\]
where the infimum runs over all Hilbert spaces $H$. We say $A$ factorizes through a Hilbert space if 
$\factorConst{A}< \infty$.
Further, let 
\[
    \factorSpConst{X}{Y} := \sup_{A} ~\factorConst{A}
\]
where the supremum runs over continuous operators $A:X\to Y$.  
As a quick example of the power of factorization theorems, observe 
that if $\id:X\to X$ is the identity operator on a Banach space $X$ and $\factorConst{\id}<\infty$, 
then $X$ is isomorphic to a Hilbert space and moreover the distortion (Banach-Mazur distance) is at most 
$\factorConst{\id}$ (\ie there exists an invertible operator $T:X \to H$ for some Hilbert space $H$ such 
that $\norm{X}{H}{T}\cdot \norm{H}{X}{T^{-1}}\leq \factorConst{\id}$). In fact (as observed by Maurey), 
\Kwapien gave an isomorphic characterization of Hilbert spaces by proving a factorization theorem. 

In this section we will show that our approximation results imply improved bounds on 
$\factorSpConst{\ell_{p}^{n}}{\ell_{q}^{m}}$ for certain values of $p$ and $q$. 
Before doing so, we first summarize prior work which will require the definitions of type and cotype: 
\begin{definition}
    The Type-2 constant of a Banach space $X$, denoted by $T_2(X)$, is the smallest constant $C$ such that 
    for every finite sequence of vectors $\{x^i\}$ in $X$, 
    \[
        \Ex{\norm{\sum_{i} \eps_i\cdot x^i}} \leq C\cdot \sqrt{\sum_{i} \norm{x^i}^{2}}
    \]
    where $\eps_i$ is an independent Rademacher random variable. We say $X$ is of Type-2 if 
    $T_2(X)<\infty$. 
\end{definition} 
\begin{definition}
    The Cotype-2 constant of a Banach space $X$, denoted by $C_2(X)$, is the smallest constant $C$ 
    such that for every finite sequence of vectors $\{x^i\}$ in $X$, 
    \[
        \Ex{\norm{\sum_{i} \eps_i\cdot x^i}} \geq \frac{1}{C}\cdot \sqrt{\sum_{i} \norm{x^i}^{2}}
    \]
    where $\eps_i$ is an independent Rademacher random variable. We say $X$ is of Cotype-2 if 
    $C_2(X)<\infty$. 
\end{definition} 
\begin{remark}~
    \begin{itemize}
        \item It is known that $C_2(X^*)\leq T_2(X)$. 
        
        \item It is known that for $p \geq 2$,  we have $T_2(\ell_p^n) = \gamma_p$ (while $C_2(\ell_p^n) \to
            \infty$ ~as~ $n \to \infty$) and for $q \leq 2$,~ $C_2(\ell_q^n) = \max\{2^{1/q-1/2},1/\gamma_q\}$ 
            (while~ $T_2(\ell_q^n) \to \infty$ ~as~ $n \to \infty$). 
    \end{itemize}
\end{remark}
\noindent
We say $X$ is Type-2 (resp. Cotype-2) if $T_2(X)<\infty$ (resp. $C_2(X)<\infty$). 
$T_2(X)$ and $C_2(X)$ can be regarded as measures of the ``closeness'' of $X$ to a 
Hilbert space. Some notable manifestations of this correspondence are: 
\begin{itemize}
    \item $T_2(X)=C_2(X)=1$ if and only if $X$ is isometric to a Hilbert space. 
    
    \item {\Kwapien~\cite{Kwapien72a}}:~~$X$ is of Type-2 and Cotype-2 if and only if it is isomorphic to 
    a Hilbert space. 
    
    \item {Figiel, Lindenstrauss and Milman~\cite{FLM77}}:~~If $X$ is a Banach space of Cotype-2, then any 
    $n$-dimensional subspace of $X$ has an $m=\Omega(n)$-dimensional subspace with Banach-Mazur 
    distance at most $2$ from $\ell_{2}^{m}$.  
\end{itemize}
Maurey observed that a more general factorization result underlies \Kwapien's 
work: 
\begin{theorem}[\Kwapien-Maurey]
    Let $X$ be a Banach space of Type-2 and $Y$ be a Banach space of Cotype-2. Then any 
    operator $T:X\to Y$ factorizes through a Hilbert space. Moreover, $\factorSpConst{X}{Y}\leq 
    T_2(X)C_2(Y)$. 
\end{theorem}

Surprisingly Grothendieck's work which predates the work of \Kwapien and Maurey, established that 
$\factorSpConst{\ell_\infty^n}{\ell_1^m}\leq K_G$ for all $m,n\in \N$, which is not implied by the above 
theorem since $T_2(\ell_\infty^n) \to \infty$ as $n\to \infty$. Pisier~\cite{Pisier80} unified the above 
results for the case of approximable operators by proving the following: 
\begin{theorem}[Pisier]
    Let $X,Y$ be Banach spaces such that $X^*, Y$ are of Cotype-2. Then any approximable 
    operator $T:X\to Y$ factorizes through a Hilbert space. Moreover, \\$\factorConst{T}\leq 
    (2\,C_2(X^*)C_2(Y))^{3/2}$. 
\end{theorem}

In the next section we show that for any $p^*,q\in [1,2]$, any $m,n\in \N$~  
\[
    \factorSpConst{\ell_{p}^{n}}{\ell_{q}^{m}} 
    \leq 
    \frac{1+\eps_0}{\sinh^{-1}(1)}\cdot C_2(\ell_{p^*}^{n})\cdot C_2(\ell_{q}^{m})
\] 
which improves upon Pisier's bound and for certain ranges of $(p,q)$, improves upon $K_G$ as well as the 
bound of \Kwapien-Maurey. 

\subsection{Integrality Gap Implies Factorization Upper Bound}
Known upper bounds on $\factorSpConst{X}{Y}$ involve Hahn-Banach separation arguments. 
In this section we see that for a special class of Banach spaces admitting a convex programming relaxation, 
$\factorSpConst{X}{Y}$ is bounded by the integrality gap of the relaxation as an immediate consequence 
of Convex programming duality (which of course uses a separation argument under the hood). 
A very similar observation had already been made by Tropp~\cite{Tropp09} in the special case of 
$X=\ell_{\infty}^{n}, Y=\ell_{1}^{m}$ with a slightly different convex program. 

We start by restating the relaxation in a more general setup, and stating its dual. To this end, let 
$\sqnormX,\subset\R^{n},~\sqnormY\subset \R^{m}$ be convex sets. Also let 
\[
    \sqrtnormX := \{x~|~[x]^{2}\in \sqnormX\} \quad 
    \sqrtnormY := \{y~|~[y]^{2}\in \sqnormY\} \mper
\]
Given an input matrix $A\in \R^{m\times n}$, we shall give a convex programming relaxation for the 
following problem: 
\[
    \sup_{x\in \sqrtnormX,~y\in\sqrtnormY} y^T A\,x
\]
The relaxation $\CP{A}$ (due to Nesterov \etal~\cite{NWY00}) is as follows:
\begin{align*}
    \textbf{maximize} \quad
    &\frac{1}{2}\cdot 
    \mydot{
    \left[
    \begin{array}{cc}
    0 & A \\
    A^T & 0
    \end{array}
    \right]   
    }
    {
    \left[
    \begin{array}{cc}
    \bbY & \bbW \\
    \bbW^T & \bbX
    \end{array}
    \right]   
    } 
    \quad 
    \text{s.t.} \\[0.5em]
    &
    \mathrm{diag}(\bbX)\in \sqnormX ,\quad
    \mathrm{diag}(\bbY)\in \sqnormY \\[0.5em]
    &
    \left[
    \begin{array}{cc}
    \bbY & \bbW \\
    \bbW^T & \bbX
    \end{array}
    \right]   
    \succeq 0 ,\quad 
    \bbY\in \Sym^{m\times m},~\bbX\in \Sym^{n\times n},~\bbW\in \R^{m\times n}
\end{align*}
For a vector $s$, let $D_s$ denote the diagonal matrix with $s$ as diagonal entries. 
Let 
\[
    \support{B}{s} := \sup_{x\in B} |\mysmalldot{x}{s}| \mper
\]
The dual program $\CPD{A}$ is as follows: 
\begin{align*}
    & \textbf{minimize} \quad 
    (\support{\sqnormY}{s}+\support{\sqnormX}{t})/2 \quad \text{s.t.}\\
    &\left[
    \begin{array}{cc}
    D_s & -A \\
    -A^T & D_t
    \end{array}
    \right]   
    \succeq 0 ,\quad 
    s\in \R^{m},~t\in \R^{n}\mper
\end{align*}
Strong duality is satisfied, i.e. $\CPD{A}=\CP{A}$, and a proof can be found in Theorem 13.2.3 of 
\cite{NWY00}. Assume $\sqrtnormX$ and $\sqrtnormY$ are convex and let $\norm{\sqrtnormX}{\cdot}$ 
and $\norm{\sqrtnormY}{\cdot}$ respectively denote the norms they induce. 
For Banach spaces $X$ over $\R^n$, $Y$ over $\R^m$ and an operator $A:X\to Y$, we define 
\[
    \threeFactorConst{A} := \inf_{D_1BD_2 = A} 
    \frac{\norm{X}{2}{D_2}\cdot\norm{2}{2}{B}\cdot\norm{2}{Y}{D_1}}{\norm{X}{Y}{A}} \qquad 
    \threeFactorSpConst{X}{Y} := \sup_{A:X\to Y} \threeFactorConst{A}
\]
where the infimum runs over diagonal matrices $D_1,D_2$ and $B\in\R^{m\times n}$. 
Clearly, $\factorConst{A}\leq \threeFactorConst{A}$ and therefore $\factorSpConst{X}{Y} \leq 
\threeFactorSpConst{X}{Y}$. 

Henceforth, we fix $X$ and $Y$ to be the Banach spaces $(\R^{n},\norm{\sqrtnormX}{\cdot})$ and 
$(\R^{m},\norm{\sqrtnormY^{\,\,\ast}}{\cdot})$ respectively.
As was the approach of Grothendieck, we give an upper bound on 
$\factorSpConst{X}{Y}$ by giving an upper bound on 
$\threeFactorSpConst{X}{Y}$. We do this by showing 
\begin{lemma}
    For any $A:X\to Y$, ~~
    $\threeFactorConst{A}~\leq ~\CPD{A}/\norm{\sqrtnormX}{\sqrtnormY^{\,*}}{A}$.
\end{lemma}

\begin{proof}
    Consider an optimal solution to $\CPD{A}$. We will show 
    \[
        \inf_{D_1BD_2 = A} \norm{X}{2}{D_2}\cdot\norm{2}{2}{B}\cdot\norm{2}{Y}{D_1}
        \leq ~
        \CPD{A}
    \]
    by taking $D_1 := D_{s}^{1/2}$, ~$D_2 := D_{t}^{1/2}$ and 
    $B := \inparen{D_{s}^{1/2}}^{\dagger} A \inparen{D_{t}^{1/2}}^{\dagger}$ 
    (where for a diagonal matrix $D$, $D^{\dagger}$ only inverts the non-zero diagonal entries and 
    zero-entries remain the same). 
    Note that $s_i = 0$ (resp. $t_i=0$) implies the $i$-th row (resp. $i$-th column) of $A$ is all zeroes, since 
    otherwise one can find a $2\times 2$ principal submatrix (of the block matrix in the relaxation) 
    that is not PSD. This implies that $D_1 B D_2 = A$.  
    
    It remains to show that 
    $\norm{X}{2}{D_2}\cdot\norm{2}{2}{B}\cdot\norm{2}{Y}{D_1}\leq ~\CPD{A}$. 
    Now we have, 
    \[
        \norm{X}{2}{D_{t}^{1/2}}
        =
        \sup_{x\in \sqrtnormX} \norm{2}{D_{t}^{1/2} x} 
        = 
        \sup_{x\in \sqrtnormX} \sqrt{\mysmalldot{t}{[x]^{2}}}  
        \leq  
        \sup_{x^1\in \sqnormX} \sqrt{|\mysmalldot{t}{x^1}|}  
        =
        \sqrt{\support{\sqnormX}{t}} \mper
    \]
    Similarly, since~ $\norm{2}{Y}{D_1} = \norm{Y^*}{2}{D_1}$~ we have 
    \[
        \norm{Y^*}{2}{D_{1}} 
        \leq  
        \sqrt{\support{\sqnormY}{s}} \mper
    \]
    Thus it suffices to show $\norm{2}{2}{B}\leq 1$ since 
    \[
        \norm{X}{2}{D_2}\cdot\norm{2}{Y}{D_1}
        \leq 
        \sqrt{\support{\sqnormX}{s}\cdot \support{\sqnormY}{t}} 
        \leq 
        (\support{\sqnormY}{s}+\support{\sqnormX}{t})/2
        =
        \CPD{A} \mper
    \]
    We have, 
    \begin{align*}
        &\left[
        \begin{array}{cc}
        D_s & -A \\
        -A^T & D_t
        \end{array}
        \right]   
        \succeq 
        0 \\
        \Rightarrow
        &\left[
        \begin{array}{cc}
        \inparen{D_{s}^{1/2}}^{\dagger} & 0 \\
        0 & \inparen{D_{t}^{1/2}}^{\dagger}
        \end{array}
        \right]
        \left[
        \begin{array}{cc}
        D_s & -A \\
        -A^T & D_t
        \end{array}
        \right]
        \left[
        \begin{array}{cc}
        \inparen{D_{s}^{1/2}}^{\dagger} & 0 \\
        0 & \inparen{D_{t}^{1/2}}^{\dagger}
        \end{array}
        \right]   
        \succeq 
        0 \\
        \Rightarrow 
        &\left[
        \begin{array}{cc}
        D_{\bars} & -B \\
        -B^T & D_{\bart}
        \end{array}
        \right]
        \succeq 
        0 \qquad
        \text{for some ~} \bars\in \{0,1\}^{m},~ \bart \in \{0,1\}^{n} \\
        \Rightarrow 
        &\left[
        \begin{array}{cc}
        \id & -B \\
        -B^T & \id
        \end{array}
        \right]
        \succeq 
        0 \\
        \Rightarrow 
        &~\,\norm{2}{2}{B}
        \leq 1
    \end{align*}
\end{proof}

\subsection[Improved Factorization Bounds for Certain p,q-norms]
{Improved Factorization Bounds for Certain $\ell_{p}^{n},\ell_{q}^{m}$}
Let $1\leq q \leq 2 \leq p\leq \infty$. Then taking $\sqnormX$ to be the $\ell_{p/2}^{n}$ unit ball and 
$\sqnormY$ to be the $\ell_{q^*/2}^{m}$ unit ball, we have $\sqrtnormX$ and $\sqrtnormY$ are 
respectively the unit balls in $\ell_{p}^{n}$ and $\ell_{q^*}^{m}$. Therefore, $X$ and $Y$ as defined above 
are the spaces $\ell_{p}^{n}$ and $\ell_{q}^{m}$ respectively. Hence, we obtain 
\begin{theorem}[$\ell_{p}^{n}\to \ell_{q}^{m}$ factorization]
\label[theorem]{p:q:factorization}
    If $1\leq q \leq 2\leq p\leq \infty$, then for any $m,n\in \N$ and $\eps_0= 0.00863$, 
    \[
        \factorSpConst{\ell_{p}^{n}}{\ell_{q}^{m}}
        ~\leq ~
        \frac{1+\eps_0}{\sinh^{-1}(1)\cdot \gamma_{p^*}\,\gamma_{q}}
        ~\leq ~
        \frac{1+\eps_0}{\sinh^{-1}(1)}\cdot C_2(\ell_{p^*}^{n})\cdot C_2(\ell_{q}^{m})
    \]. 
\end{theorem}
\noindent
This improves upon Pisier's bound and for a certain range of $(p,q)$, improves upon $K_G$ as well as the 
bound of \Kwapien-Maurey. 
\bigskip

On a slightly unrelated note, straightforward observations imply that the integrality gap of $\CP{A}$ for any 
pair of convex sets $\sqnormX,\sqnormY$ is $K_G$ (Grothendieck's constant). This provides a class of 
Banach space pairs for which $K_G$ is an upper bound on the factorization constant, and it would be 
interesting to get a better understanding of how this class compares to that of Pisier. We include a proof of 
this in the next section. 

\subsection[Grothendieck Bound on Approximation Ratio]{$K_G$ Bound on Approximation Ratio}
\label[subsection]{KG:worst:case}
In this subsection, we prove that for any pair of convex sets $\sqnormX$ and $\sqnormY$ such that 
$\sqrtnormX$ and $\sqrtnormY$ are convex, the approximation ratio is bounded by $K_G$. As in the previous 
section, we fix $X$ and $Y$ to be the Banach spaces $(\R^{n},\norm{\sqrtnormX}{\cdot})$ and 
$(\R^{m},\norm{\sqrtnormY^{\,\,\ast}}{\cdot})$ respectively.
\begin{lemma}
    For any $A:X\to Y$, ~~
    $\CP{A}/\norm{X}{Y}{A} ~\leq ~ K_G$.
\end{lemma}
\begin{proof}
Let 
$
B :=  \frac{1}{2}  \left[
    \begin{array}{cc}
    0 & A \\
    A^T & 0
    \end{array}
    \right]$.
The main intuition of the proof is to decompose $x \in \sqrtnormX$ as $x = |[x]| \circ \sgn{[x]}$ 
(where $\circ$ denotes Hadamard/entry-wise multiplication), 
and then use Grothendieck's inequality on $\sgn{[x]}$ and $\sgn{[y]}$. Another simple observation is that 
for any convex set $\calF$, the feasible set we optimize over is invariant under {\em factoring out the 
magnitudes of the diagonal entries}. In other words, 
\begin{align}
    &\{ \Diag{d} \ \Sigma \ \Diag{d} : d \in \sqrt{\calF} \cap \R^n_{\geq 0},~ 
    \Sigma \succeq 0,~ \diag(\Sigma) = 1 \} \nonumber \\
    = 
    &\{ \bbX : \diag(\bbX) \in \calF,~ \bbX \succeq 0 \} \label[equation]{eq:feasible_sets}
\end{align}
We will apply the above fact for $\calF = \sqnormX \oplus \sqnormY$. 
Let $\sqrtnormX^+$ denote $\sqrtnormX\cap \R_{\geq 0}^n$ (analogous for $\sqrtnormY^+$). 
Now simple algebraic manipulations yield 
\begin{align*}
    &\quad \norm{X}{Y}{A} \\
    =&\quad 
    \max_{x \in \sqrtnormX,~ y \in \sqrtnormY} (y \oplus x)^T B (y \oplus x) \\
    =&\quad 
    \max_{\substack{d_x \in \sqrtnormX^+,\  \sigma_x \in \{ \pm 1 \}^{n}, \\
    d_y \in \sqrtnormY^+, \ \sigma_y \in \{ \pm 1 \}^{m}}} 
    ((d_y \circ \sigma_y) \oplus (d_x \circ \sigma_x))^T ~B~ ((d_y \circ \sigma_y) \oplus (d_x \circ \sigma_x)) \\
    =&\quad 
    \max_{\substack{d_x \in \sqrtnormX^+, \ \sigma_x \in \{ \pm 1 \}^{n}, \\
    d_y \in \sqrtnormY^+, \ \sigma_y \in \{ \pm 1 \}^{m}}} 
    (\sigma_y \oplus \sigma_x)^T  (\Diag{d_y \oplus d_x} ~B~ \Diag{d_y \oplus d_x}) 
    (\sigma_y \oplus \sigma_x) \\
    \geq &\quad 
    (1/K_G) \cdot 
    \max_{\substack{d_x \in \sqrtnormX^+, \ d_y \in \sqrtnormY^+, \\
    \Sigma :~ \diag(\Sigma) = 1, ~\Sigma \succeq 0}}  
    \mydot{ \Sigma~}{~\Diag{d_y \oplus d_x} ~B~  \Diag{d_y \oplus d_x}} 
    &&(\text{Grothendieck})\\
    = &\quad 
    (1/K_G) \cdot 
    \max_{\substack{d_x \in \sqrtnormX^+, \ d_y \in \sqrtnormY^+, \\
    \Sigma :~ \diag(\Sigma) = 1, ~\Sigma \succeq 0}}  
    \mydot{\Diag{d_y \oplus d_x} ~\Sigma ~  \Diag{d_y \oplus d_x}~}{ ~B} \\
    =&\quad 
    (1/K_G) \cdot \CP{A} &&(\text{by \cref{eq:feasible_sets}})\mper
\end{align*}
\end{proof}